%% file: ttest.tex
\documentclass[12pt]{article}
\usepackage{"conroy_latex_v2"}

\title{\bf Cluster-robust inference with a single treated cluster using the \textbf{\textit{t}}-test\thanks{\mbox{We thank helpful comments from Stéphane Bonhomme, Ivan Canay, Azeem Shaikh, and Alex Torgovitsky.}}}
\author{Chun Pong Lau\thanks{\mbox{Kenneth C. Griffin Department of Economics, 
The University of Chicago, 
\href{mailto:ccplau@uchicago.edu}{\texttt{ccplau@uchicago.edu}}.}}
\and
Xinran Li\thanks{Department of Statistics, The University of Chicago, \href{mailto:xinranli@uchicago.edu}{\texttt{xinranli@uchicago.edu}}.}}

\usepackage{amssymb}
\usepackage{fontawesome}

\usepackage{dingbat}

\graphicspath{{./images/}}

\usepackage{centernot}
\usepackage{marginnote}
\usepackage{multicol}
\usepackage{lscape}
\usepackage{bbding}
\usepackage{framed}
\usepackage{accents}
\usepackage{makecell}

\newcommand{\Dt}{\overset{d}{\longrightarrow}}  

\newcommand{\cv}{\mathrm{cv}_{m, \alpha, k, \rho}}
\newcommand{\cSS}{\mathcal{S}_m(k, \rho)}

\newcommand{\tm}{\theta_{m+1}}

\usepackage{ragged2e}
\allowdisplaybreaks

\usepackage{xpatch}
\xpatchcmd{\proof}{\itshape}{\normalfont\proofnamefont}{}{}
\newcommand{\proofnamefont}{\bfseries\itshape}

\usepackage[textsize=tiny, textwidth = 2cm, shadow]{todonotes}
\usepackage{enumitem}
\def\bs{\bm}
\def\deri{\text{d}}
\def\I{\bm{1}}
\def\converge{\stackrel{}{\rightarrow}}
\def\convergeas{\stackrel{\text{a.s.}}{\longrightarrow}}
\def\converged{\stackrel{d}{\longrightarrow}}

\def\nob{\psi}
\def\cp{\theta}

\newcommand{\ohh}{\overline H_m(c, \rho)}

\newcommand{\hd}{\widehat{\delta}}

\def\Var{\text{Var}}

\usepackage{accents}

\usepackage{ragged2e}
\allowdisplaybreaks

\usepackage[textsize=tiny, textwidth = 2cm, shadow]{todonotes}
\usepackage{enumitem}
\def\bs{\bm}
\def\deri{\text{d}}
\def\I{\bm{1}}
\def\converge{\stackrel{}{\rightarrow}}
\def\convergeas{\stackrel{\text{a.s.}}{\longrightarrow}}
\def\converged{\stackrel{d}{\longrightarrow}}

\def\nob{\psi}

\def\cp{\theta}
\def\hd{\widehat{\theta}}

\usepackage{comment}
\theoremstyle{definition}

\newtheorem*{theorem*}{Theorem}
\newtheorem{theorem}{Theorem}
\newtheorem*{rmk*}{Remark}

\usepackage{ulem}

\usepackage[sc]{mathpazo}
\begin{document}
\fontfamily{ppl}

\maketitle

\begin{abstract}
This paper considers inference when there is a single treated cluster and a fixed number of control clusters, a setting that is common in empirical work, especially in difference-in-differences designs. 
We use the $t$-statistic and develop suitable critical values to conduct valid inference under weak assumptions allowing for unknown dependence within clusters. 
In particular, our inference procedure does not involve variance estimation. It only requires specifying the relative heterogeneity between the variances from the treated cluster and some, but not necessarily all, control clusters. 
Our proposed test works for any significance level when there are at least two control clusters.
When the variance of the treated cluster is bounded by those of all control clusters up to some prespecified scaling factor, the critical values for our $t$-statistic can be easily computed without any optimization for many conventional significance levels and numbers of clusters. 
In other cases, one-dimensional numerical optimization is needed and is often computationally efficient. 
We have also tabulated common critical values in the paper 
so researchers can use our test readily. 
We illustrate our method in simulations and empirical applications.
\par
\vspace{10pt}
\noindent \textbf{Keywords:} Cluster-robust inference, single treated cluster, $t$-test, simultaneous inference, difference-in-differences. \par
\end{abstract}

\newpage

\onehalfspacing

\section{Introduction}

In difference-in-differences designs, it is common for researchers to conduct inference using cluster-robust methods to account for correlation within clusters. However, inference becomes challenging when there is only a single treated cluster. Intuitively, this is because researchers are faced with only one estimate from the treated cluster, making its uncertainty difficult to quantify. Existing methods either impose strong assumptions on having a large number of clusters, require the variances to be homogeneous or estimable, or only work for specific significance levels depending on the number of clusters. In this paper, we develop a $t$-test associated with a suitable critical value to conduct valid inference that relaxes these assumptions.

Having a single treated cluster and a finite number of control clusters is common in empirical work. For instance, this can happen when researchers use data in the United States to examine the impact of a policy that takes place in one particular state but not the other states. The nearby states or the remaining states are commonly used as the control group. In these scenarios, it is reasonable for researchers to believe they are not in the scenario with a ``large'' number of control clusters.
Some recent examples in this context include 
\citet{wangburke2022jfe} on the effect of payday loan regulations in Texas, \citet{harrislarsen2023jhr} on the effect of Hurricane Katrina on student outcomes in New Orleans, 
\citet{dillenderetal2023jpube} on the effect of change in health care reimbursement rates in Illinois,
\citet{alpertetal2024aejep} on the impact of Kentucky's prescription drug monitoring programs on opioid prescribing, and
\citet{kumarliang2024aejep} on the labor market effects of constitutional amendments in Texas.
\citet{hagemann2024wp} has documented earlier related examples in the context of a single treated cluster. In the next section of this paper, we demonstrate that our test is also applicable to other empirical designs, in addition to difference-in-differences.

This paper contributes to the literature on cluster-robust inference. We assume that the number of clusters is fixed, as in some related work in this literature such as \citet{besteretal2011joe}, \citet{ibraimovmuller2010jbes, ibraimovmuller2016restat}, \citet{canayetal2017ecta}, \citet{hagemann2022wp}, and \citet{lau2025wp}. Although the tests in the aforementioned papers are valid when there is a finite number of clusters, they are not suitable for the problem with a single treated cluster. The first test requires certain homogeneity conditions, and the other tests cannot be applied when there is a single treated cluster.  
\citet{canaysantosshaikh2021restat} has showed that Wild cluster bootstrap popularized by \citet{cameronetal2008restat} can be valid under strong homogeneity conditions when there is a fixed number of clusters. See, for instance, \citet{cameronmiller2015jhr}, \citet{conleyetal2018jar},  \citet{mackinnonetal2023joe}, and \citet{alvarezetal2025wp} for some surveys on the literature of conducting inference with a fixed number of clusters.

Several tests have been developed to conduct inference when there is a single treated cluster, but with assumptions that can be strong in practice. \citet{conleytaber2011restat} assume homogeneous errors. \citet{fermanpinto2019restat} relax the homogeneity assumption and allow for known heteroskedasticity. \citet{alvarezferman2023wp} relax the homogeneity assumption and allows for spatial correlation. All these tests assume that there is an infinite number of control clusters. Our $t$-test neither assumes the variances are known nor assumes an infinite number of control clusters. Recently, \citet{hagemann2024wp} develops a novel rearrangement test that relaxes the assumptions mentioned earlier. \citet{hagemann2024wp} is the most related paper in that his test relies on a relative heterogeneity condition between the treated and control clusters. Both our test and \citet{hagemann2024wp} work with a fixed number of clusters, allow for arbitrary correlation within clusters, and do not assume that we can estimate or know the cluster variances. He shows his test is valid when the relative heterogeneity condition holds with all or all but one control cluster. However, the validity of his test depends on the number of clusters, the significance level, and the relative heterogeneity condition. 
For example, 
when the standard deviation of the treated cluster is bounded by $2$ times all but one of the standard deviations of the control clusters, 
\citet{hagemann2024wp} requires at least 14 control clusters to conduct a one-sided test at the 2.5\% level using his rearrangement test, and at least 17 control clusters are needed for a 1\% level test\footnote{Equivalently, to conduct two-sided tests, there have to be at least 14 control clusters for a 5\% level test.}.
This potentially limits the applicability of his test. Among the cases where his test is valid, he computes weights for his rearrangement test using numerical optimization that lead to valid tests.

There are several important distinctions between our work and  \citet{hagemann2024wp}. First, we show that our test can be valid for
any choice of significance levels and heterogeneity parameters when there are at least two control clusters. For many conventional combinations of significance levels and number of clusters, we derive a closed-form critical value for the $t$-test. No optimization is necessary for these cases, so researchers can apply the test readily. For other cases, we can derive critical values that lead to valid tests through one-dimensional optimizations. Therefore, our test can have broader applicability. Second, we relax the relative heterogeneity condition in \citet{hagemann2024wp}. Intuitively, the condition restricts the relative variance between the treated and control clusters. \citet{hagemann2024wp} requires such condition to hold between the treated and all (or all but one) control clusters to show the validity of his test. We weaken this condition by allowing the researcher to impose this restriction between the treated cluster and any number of control clusters.
This allows researchers to bound relative heterogeneity on variances depending on their choice on how to restrict the amount of heterogeneity between the treated and control clusters. Third, we allow for simultaneous inference across all the relative heterogeneity assumptions that we impose.
Specifically, under the assumption of no treatment effect, we can infer the minimum amount of relative heterogeneity required to explain away the observed association between treatment and outcomes.  If the true relative heterogeneity is unlikely to exceed this amount, the treatment is likely to have a significant nonzero effect.

While the $t$-test has been used in \citet{Bakirov:2006aa} and \citet{ibraimovmuller2010jbes, ibraimovmuller2016restat}, our proof on the validity of the $t$-test in the current context requires different proof strategies.
This is because we only have a single treated cluster and the relative heterogeneity assumption imposes additional structure on the parameter space.

We confirm the performance of our proposed test in simulations. We find that our test controls size under a wide range of significance levels and  number of clusters. We also have favorable power performance when compared to other tests in the literature.

The rest of the paper is organized as follows. Section \ref{sec:examples} outlines some popular empirical designs that our test can be applied to. Section \ref{sec:procedure-main} presents our inference procedure. Readers interested in applying the test can refer to Algorithm \ref{algo:t-test}. Section \ref{sec:theory} presents our main theory. Section \ref{sec:simu} explains how our test can be easily used for simultaneous inference. Section \ref{sec:sims} presents the simulation results. Section \ref{sec:emp} contains two empirical studies. Section \ref{sec:conclusion} concludes. All proofs can be found in the supplementary material.

\section{Motivating examples} \label{sec:examples}

We start by describing several empirically relevant designs that are common in applied work and are related to the issue of having a single treated cluster. These examples show that our method applies more broadly apart from standard difference-in-differences designs. The first three examples have also been discussed in \citet{hagemann2024wp}. \par 

In the following, $i \in \cI \equiv \{1, \ldots, n\}$ indices units, $j \in \cJ \equiv \{1, \ldots, m + 1\}$ indices clusters, and $t \in \cT \equiv \{ 1, \ldots, T\}$ indices time. Only cluster $(m+1)$ is treated, and the remaining clusters are controls. Thus, $\cJ$ can be partitioned as $\cJ_0 \cup \cJ_1$, where $\cJ_0 \equiv \{1, \ldots, m\}$ and $\cJ_1 \equiv \{m + 1\}$ denote the control clusters and treated cluster respectively. Let $D_j \equiv \ind[j \in \cJ_1]$ be a cluster-level treatment indicator across $j \in \cJ$. \par

As to be described in Section \ref{sec:procedure-main}, we require researchers to be able to run regression cluster by cluster, and write the estimator for the target parameter $\Delta$ in terms of the estimator for parameters from cluster-level regressions, denoted by $\{\widehat\cp_j\}^{m+1}_{j=1}$. The following examples explain how to connect $\{\widehat\cp_j\}^{m+1}_{j=1}$ with the estimator for the target parameters in common empirical designs.

\begin{eg}[Clustered regression] Consider the following model
\[
	Y_{ij} = \beta_0 + \Delta D_j + U_{ij},
\]
where $\beta_0, \Delta, U_{ij} \in \bR$. Under the assumption that $\bE[U_{ij}| D_j] = 0$, $\Delta$ can be estimated by
\[
	\widehat\Delta
	= \hd_{m+1}
	- \frac{1}{m} \sum^m_{j=1} \hd_j,
\]
where $\hd_j$ is the sample mean of $Y_{ij}$ for each $j \in \cJ$.
\end{eg}

\begin{eg}[Difference-in-differences]\label{eg:did}
Suppose $T = 2$ and let $\text{Post}_t \equiv \ind[t = 2]$ for all $t \in \cT$. Consider the following model
\begin{equation}
	\label{eq:did1}
	Y_{jt}
	= 
        \alpha_j
	+
	\beta \text{Post}_{t}
	+
	\Delta D_j \text{Post}_{t}
	+ 
	U_{jt},
\end{equation}
where $U_{jt} \in \bR$, and $\{\alpha_{j}\}_{j=1}^{m+1}$ are the cluster fixed effects. Under the assumption that $\bE[U_{jt} | D_j, \text{Post}_t] = 0$, $\Delta$ can be estimated by
\[
	\widehat\Delta
	= 
	\hd_{m+1}
	-
	\frac{1}{m} \sum_{j=1}^m \hd_j,
\]
where $\{\hd_j\}^{m+1}_{j=1}$ are the differences in the outcomes before and after treatment. They can be obtained as the estimates of $\{\cp_j\}^{m+1}_{j=1}$ from the following cluster-level regressions
\begin{align}
	\label{eq:did2}
	Y_{jt} 
	=
	\alpha_j
	+
	\cp_j \text{Post}_t
	+
	\epsilon_{jt}
\end{align}
for each $j \in \cJ$.
\end{eg}

\begin{eg}[Two-way fixed effects] \label{eg:twfe}
Consider the following two-way fixed effects model:
\begin{equation}
	\label{eq:twfe1}
	Y_{jt}
	= 
	\alpha_j
	+ \beta_t
	+ \Delta D_j \text{Post}_t
	+ U_{jt},
\end{equation}
where $T \geq 2$, $\{\alpha_{j}\}_{j=1}^{m+1}$ are the cluster fixed effects, $\{\beta_t\}_{t=1}^T$ are the time fixed effects, and $\text{Post}_t \equiv \ind[t \geq t_0]$ for some $1 < t_0 \leq T$. Under the assumption that $\bE[U_{jt} | D_j, \text{Post}_t]=0$, $\Delta$ can be estimated by
\[
	\widehat\Delta
	=
	\hd_{m+1} - \frac{1}{m} \sum^m_{j=1} \hd_j,
\]
using the same cluster-level regression as in \eqref{eq:did2}.
\end{eg}

For more discussion about recent advances in difference-in-differences and two-way fixed effects, see, for instance, the surveys by \citet{dechaisemartindhaultfœuille2023ej},  \citet{rothetal2023joe}, and \citet{bakeretal2025wp}.

\begin{eg}[Triple differences]
Let $C_{ij}, D_{ij} \in \{0, 1\}$ be binary indicators that depend on the individual $i \in \cI$ and cluster $j \in \cJ$. Set $D_{ij} = \ind[j = {m+1}]$ and define $\text{Post}_t\equiv \ind[t = 2]$ as in Example \ref{eg:did} with two periods. Assume for each $j \in \cJ$, there exist units with $C_{ij} = 1$ and units with $C_{ij} = 0$. Consider the following triple differences/difference-in-difference-in-differences model:
\begin{align}
	\label{eq:didid1}
	\begin{split}
	Y_{ijt}
	& =
	\beta_0 + \beta_1 C_{ij} + \beta_2 D_{ij} + \beta_3 \text{Post}_t
	+ 
	\beta_4 C_{ij} D_{ij} 	\\
	&\quad \quad + \beta_5 C_{ij} \text{Post}_t + \beta_6 D_{ij} \text{Post}_t
	+ \Delta C_{ij} D_{ij} \text{Post}_t + U_{ijt}.
	\end{split}
\end{align}
Under the assumption that $\bE[U_{ijt} | C_{ij}, D_{ij}, \text{Post}_t] = 0$, $\Delta$ can be estimated by
\[
	\widehat\Delta
	= 
	\hd_{m+1}
	-
	\frac{1}{m} \sum^m_{j=1} \hd_j,
\]
where $\{\hd_j\}^{m+1}_{j=1}$ are the estimates of $\{\cp_j\}^{m+1}_{j=1}$ from the following cluster-level regressions
\[
	Y_{ijt}
	= 
	\alpha_j
	+
	\gamma_{C,j}
	C_{ij}
	+
	\gamma_{\text{Post},j}
	\text{Post}_t
	+
	\cp_j
        C_{ij} 
	\text{Post}_t
	+
	\epsilon_{ijt},
\]
for each $j \in \cJ$. See \citet{oldenmoen2023ej} for a recent survey on triple difference estimators.
\end{eg}

\section{Inference procedure} \label{sec:procedure-main}

In this section, we present the main assumptions and our algorithm of conducting inference using the $t$-test with a single treated cluster. Readers who are interested in applying our test can directly apply Algorithm \ref{algo:t-test}.  We also present some brief intuition for the validity of our test in Section \ref{sec:3.3} to facilitate the theoretical discussion in Section \ref{sec:theory}.

\subsection{Assumptions} \label{sec:procedure-main-assu}

We follow the notation used in Section \ref{sec:examples}. Let $\{\hd_j\}^{m+1}_{j=1}$ be the cluster-level estimators, $\cJ_0 = \{1, \ldots, m\}$ index the control clusters and $\cJ_1 =\{m + 1\}$ index the treated cluster. For simplicity, we make the dependence of $\{\hd_j\}^{m+1}_{j=1}$ on the cluster size implicit.

Our procedure requires two main assumptions. First, we assume that an appropriate central limit theorem applies to $\{\hd_j\}^{m+1}_{j=1}$ as the sample size within each cluster, denoted by $n$, goes to infinity. 
A similar assumption is also imposed in related papers on a fixed number of clusters, such as \citet{ibraimovmuller2010jbes, ibraimovmuller2016restat}, \citet{canayetal2017ecta}, \citet{hagemann2022wp, hagemann2024wp} and \citet{lau2025wp}.
Intuitively, this holds when each cluster has a large number of units or consists of a panel with a long time periods. This high-level assumption is formalized as follows.

\begin{assu}
\label{assu:CLT}
The following holds as the sample size within each cluster $n \longrightarrow \infty$:
\begin{align}\label{eq:CLT}
	\sqrt{n} 
         \begin{pmatrix}
	\hd_1 - \mu_0 \\ 
	\vdots \\ 
	\hd_m - \mu_0 \\  
	\hd_{m+1} - \mu_1
	\end{pmatrix}
	\Dt 
	\cN(0, \Sigma),
\end{align}
where $\Sigma \equiv \text{diag}(\sigma_1^2, \ldots, \sigma_m^2, \sigma_{m+1}^2)$ is an diagonal matrix.
\end{assu}

Below we give a few remarks regarding Assumption \ref{assu:CLT}. 
First,
it is not necessary for all clusters to have the same size.  
In cases where clusters have varying sizes, we can define $n$ as the size of the smallest cluster, and our inference essentially requires the sample sizes in all clusters to be large. 
Second, we assume that the cluster-level estimators are independent, at least asymptotically, so that the asymptotic covariance matrix $\Sigma$ in \eqref{eq:CLT} is diagonal. 
This is trivially satisfied if the units are independent across clusters. We emphasize that, while we require independence across clusters, we allow for unknown dependence structure within clusters.
Third, we assume that the estimators from the control clusters are consistent for a common parameter $\mu_0$, which can often be justified by model assumptions or study designs as illustrated in Section \ref{sec:examples}. 

Second, we impose the following relative heterogeneity assumption, which generalizes the relative heterogeneity assumption in \citet{hagemann2024wp}. 
Let $\sigma_{(1)} \leq \sigma_{(2)} \leq \cdots \leq \sigma_{(m)}$ be the ordered values of $\{\sigma_j\}^m_{j=1}$ in Assumption \ref{assu:CLT}.

\begin{assu}
\label{assu:relative_heter}
For a given $\rho \ge 0$ and a given $k \in \{1, \ldots, m\}$, $\sigma_{m+1}	\leq \rho \sigma_{(k)}$. 
\end{assu}

The above assumption does not require the variances to be known.
It only restricts the relative heterogeneity of the standard deviations between the treated and control clusters. 
In particular, for any $\rho \geq 0$ and $k \in \{1, \ldots, m\}$, it requires the standard deviation of the treated cluster to be less than or equal to $\rho$ times the standard deviations of at least $(m - k + 1)$ control clusters. For example, $k = 1$ requires that the standard deviation of the treated cluster is smaller than or equal to $\rho$ times the standard deviations of all control clusters; when $k=m$, this means that the standard deviation of the treated cluster is smaller than or equal to $\rho$ times the largest standard deviation from the control clusters.

Assumption \ref{assu:relative_heter} is a relaxed version of \citet{hagemann2024wp}'s maximum relative heterogeneity assumption. His assumption is equivalent to Assumption \ref{assu:relative_heter} with $k = 1$ or $k = 2$. 
Our test allows a general choice of $k \in \{1, \ldots, m\}$. More importantly, as demonstrated in Section \ref{sec:simu}, the inference can be simultaneously valid for all $k \in \{1, \ldots, m\}$. 
In other words, with additional choices of $k$, our test can provide more evidences against the null hypothesis than \citet{hagemann2024wp}'s test, which focuses on $k=1$ or $2$.  

Before introducing our algorithm, we end this subsection with some discussion on choosing $(\rho, k)$ in Assumption \ref{assu:relative_heter}. If the researcher believes the standard deviation of the treated cluster cannot be larger than the standard deviations of the control clusters, then the researcher can set $k = 1$ and $\rho = 1$.

On the other hand, if the researcher does not want to commit to a particular value of $(\rho, k)$, the researcher can perform simultaneous inference and report the largest value of $\rho$ such that the null can be rejected for each $k$. 
This approach can also be interpreted as finding the largest $\rho$ such that the conclusion changes, which is related to the idea of finding breakdown points/frontiers in econometrics \citep{horowitzmanski1995ecta, klinesantos2013qe, mastenpoirier2020qe}.
We demonstrate this through two applications in Section \ref{sec:emp}.

\subsection{Algorithm}
The goal is to test the following hypothesis:
\begin{equation}
    \label{eq:algor-hypo-test}
    H_0\text{:}~\mu_1 = \mu_0
    \qquad \text{vs} \qquad
    H_1\text{:}~\mu_1 \neq \mu_0.
\end{equation}

Our procedure of conducting inference with a single treated cluster is as follows.
\begin{algo} \label{algo:t-test} \text{ } \par
\noindent \underline{Inputs:} 
\begin{itemize}
    \item Significance level $\alpha \in (0, \frac{1}{2})$.
    \item The parameters on relative heterogeneity from Assumption \ref{assu:relative_heter}, i.e., $(\rho, k)$.
    \item The cluster-level estimators $\{\hd_j\}^{m+1}_{j=1}$.
\end{itemize}
\noindent \underline{Steps:}
\begin{description} 
    \item[Step 1 (Test statistic):] Compute the $t$-statistic:
    \begin{align}\label{eq:tstat_cl_est}
	\widehat{T}_m 
        \equiv 
        \frac{\widehat\cp_{m+1} - \overline{\widehat{\cp}}_m}{\widehat S_m},
    \end{align}
    where $\overline{\widehat{\cp}}_m \equiv \frac{1}{m} \sum^m_{j=1} \hd_j$ and $\widehat S_m^2  \equiv \frac{1}{m-1} \sum^m_{j=1} (\hd_j - \overline{\widehat{\cp}}_m)^2$ represent the sample average and sample variance of the estimators from the control clusters, respectively.
    \item[Step 2 (Critical value):] We have tabulated many empirically-relevant critical values in Table \ref{tab:cv}. These critical values can be immediately used. More generally, the critical value $\cv$ can be computed as follows.
    \begin{itemize}[leftmargin=*]
        \item \textbf{Case 1:} If $k = 1$ and $\alpha$ is below the thresholds in Table \ref{tab:largest_alpha}, use the closed-form expression in \eqref{eq:max_sig_level_k1}. Compute the closed-form critical value as 
        \[
            \cv
            \equiv 
            \sqrt{\rho^2 + \frac{1}{m}} \ t_{m-1, \frac{\alpha}{2}},
        \]
         where $t_{ m-1, \frac{\alpha}{2}}$ is the $(1-\frac{\alpha}{2})$-th quantile of the $t$-distribution with $(m-1$) degrees of freedom.
        \item \textbf{Case 2:} Search $\cv$ such that $p_m(\cv; k, \rho)$ defined in Theorem \ref{thm:max_rej_prob} ahead is at most $\alpha$. This requires solving one-dimensional optimization problems.  \par
    \end{itemize}
    \item[Step 3 (Decision):] Reject \eqref{eq:algor-hypo-test} if $|\widehat{T}_m| > \cv$. \hfill $\blacksquare$
\end{description}
    
\end{algo}

\subsection{Computing \textit{p}-value and confidence interval} \label{sec:ci-p}

To compute the $p$-value, the researcher does not need to run the entire Algorithm \ref{algo:t-test}. In particular, the researcher only needs to compute the test statistic $\widehat T_m$ in Step 1 of Algorithm \ref{algo:t-test} and substitute it into the $p_m$ function in Step 2 of Algorithm \ref{algo:t-test} and return the $p$-value as $p_m(\widehat T_m; k, \rho)$. \par 

The researcher can take the critical value $\cv$ from Step 2 of Algorithm \ref{algo:t-test} and return 
$(\widehat\theta_{m+1} - \overline{\widehat{\cp}}_m) \pm \cv \widehat S_m$
as an $(1 - \alpha)$ confidence interval for $\mu_1 - \mu_0$.

\subsection{Discussion} \label{sec:3.3}

Algorithm \ref{algo:t-test} is a computationally simple three-step procedure. 
The first step is to compute the usual $t$-statistic using the cluster-level estimators. \par 

The second step finds the critical value that depends on $(m, \alpha, k, \rho)$. This step aims to find the critical value $\cv$ such that the procedure controls size for any configurations of $\{\sigma_j\}^{m+1}_{j=1}$ that satisfy Assumption \ref{assu:relative_heter}. 
 In the first case, a closed-form critical value is available when $\alpha$ is below a certain threshold in Table \ref{tab:largest_alpha}. If $\alpha$ does not satisfy the condition, we can compute the critical value $\cv$ in the second step through one-dimensional optimization. \par 

The reason for having a closed-form expression for the critical value in case 1 of Step 2 of Algorithm \ref{algo:t-test} is as follows.  
For $k=1$ and any given $(m, \rho)$, 
for a sufficiently small significance level $\alpha$ (including many conventional choices),
we find that the maximum rejection probability over all configurations of $\{\sigma_j\}^{m+1}_{j=1}$ is achieved when $\sigma_{m+1}=\rho \sigma_j$ for all $j = 1, \ldots, m$. The proof is nontrivial and we explain the technical details in Sections  \ref{sec:theory}. Knowing when the maximum rejection probability is achieved, we are able to derive a closed-form expression for the critical value using properties of the $t$-distribution. 

Case 2 of Step 2 is in fact a general version that holds for any $(m, \alpha, k, \rho)$. For a general value of $k \in \{1, \ldots, m\}$, although we cannot find the exact configuration of $\{\sigma_j\}^{m+1}_{j=1}$ that achieves the desired level of rejection probability, we are able to substantially reduce the set of possible values. In particular, there are at most $m^2$ possible cases for the worst-case configuration of $\{\sigma_j\}^{m+1}_{j=1}$, each of which involves at most one unknown parameter. Therefore, we can find the maximum rejection probability through one-dimensional optimization. In addition, we have a good initial guess for the critical value. To facilitate the use of the test in practice, we have tabulated many critical values for researchers' use in Table \ref{tab:cv}.

The last step rejects if the absolute value of the test statistic is above the critical value.

\section{Theory} \label{sec:theory}
In this section, we develop the theory that justifies the validity of Algorithm \ref{algo:t-test}. In Section \ref{sec:theory-large}, we show that the large-sample behavior of the test can be studied via a fixed number of normal random variables. In Section \ref{sec:theory-validity}, we show the general theory on computing the maximum rejection probability of the test statistic. This result is used to search for a critical value such that the test is valid. We show in Section \ref{sec:theory-prob-k=1} that we can get a closed-form expression for the critical value as in Step 2 of Algorithm \ref{algo:t-test} when $k = 1$ and the significance level is not ``too large.'' In Section \ref{sec:theory-power}, we discuss the power of our $t$-test.

\subsection{General framework with normal means} \label{sec:theory-large}

In this subsection, we show that under Assumption \ref{assu:CLT}, studying the behavior of the $t$-statistic in equation \eqref{eq:tstat_cl_est} of Algorithm \ref{algo:t-test} when $n$ is large is equivalent to studying the $t$-statistic of suitably defined $(m+1)$ normal random variables. Formally, consider the following assumption on the normal random variables.

\begin{assu}
\label{assu:normal}
Let $\{\nob_j\}^{m+1}_{j=1}$ be $m+1$ independent random variables, where $\nob_j \sim \cN(0, \sigma_j^2)$ for $1\le j \le m$ and $\nob_{m+1} \sim \cN(\delta, \sigma_{m+1}^2)$. 
\end{assu} 

Compared to the notation in Section \ref{sec:procedure-main-assu}, $\{\psi_j\}^m_{j=1}$ above correspond to the $m$ control clusters and $\psi_{m+1}$ corresponds to the single treated cluster. As before, the above assumption does not require us to know the variances $\{\sigma_j\}^{m+1}_{j=1}$. The variances can be arbitrarily heterogeneous as long as they satisfy the relative heterogeneity assumption stated in Assumption \ref{assu:relative_heter}.
Analogous to \eqref{eq:tstat_cl_est}, we define the following $t$-statistic based on these normal random variables: 
\begin{equation}
	\label{eq:pop-test-1}
	T_m
	\equiv 
	\frac{\nob_{m+1} - \overline\nob_m}{S_m},
\end{equation}
where $\overline\nob_m \equiv \frac{1}{m} \sum^m_{j=1} \nob_j$ and $S_m^2 \equiv \frac{1}{m-1} \sum^m_{j=1} (\nob_j - \overline\nob_m)^2$ denote the sample mean and variance for the control clusters, respectively.

The theorem below shows that, to achieve the desired size asymptotically, it suffices to consider a stylized setting with normally distributed random variables. 
 Note that $\widehat{T}_m$ in \eqref{eq:tstat_cl_est} depends on the sample size within each cluster $n$ implicitly.
In addition, we allow $\mu_1$ and $\mu_0$ to vary with the sample size $n$ as well, which can facilitate power investigation under local alternatives.

\begin{thm}\label{thm:size_asymp}
	Let $m\ge 2$.
	Suppose that Assumption \ref{assu:CLT} holds, 
	$\sqrt{n}(\mu_1 - \mu_0) \longrightarrow \delta$ as $n\longrightarrow \infty$ for some $\delta \in \mathbb{R}$, 
	and the variances $\{\sigma_j^2\}^{m+1}_{j=1}$ are not all zero.
	Then, for any $c > 0$ and $c\ne m^{-1/2}$,\footnote{We impose $c\ne m^{-1/2}$ to avoid the case where $|T_m|$ has a positive point mass at $c$.
	This technical requirement excluding a single value of $c$ generally does not affect the practical use of our test, since the desired critical value at a usual significance level is greater than $m^{-1/2}$.
	Specifically, as discussed later in Remark \ref{re:c_greater_1_over_sq_m}, the rejection probability $\bP[|T_m| > c]$  for any $c<m^{-1/2}$ can be as large as $1$, 
	and the rejection probability $\bP[|T_m| > c]$ for $c$ close to but larger than $m^{-1/2}$ can at least be about $0.5$. 
	}
	\begin{align*}
		\lim_{n\to\infty} \bP[|\widehat{T}_m| > c]
		= 
		\bP[|T_m| > c],
	\end{align*} 
	where $T_m$ is the $t$-statistic defined in \eqref{eq:null_normal}.
\end{thm}

Note that in the asymptotics of Theorem \ref{thm:size_asymp}, the limit is taken as $n$ goes to infinity, while the number of clusters $(m+1)$ is fixed. 
In addition, we allow a general $\delta$ for the difference between $\mu_1$ and $\mu_0$ to facilitate the later discussion on the power of the test. To derive valid tests, it suffices to consider the case where $\mu_1 = \mu_0$ and $\delta = 0$. 
Specifically, to derive large-sample valid $t$-test with a single treated cluster and a finite number of control clusters, it suffices to test the following null hypothesis in the stylized setting with exactly normal observations: 
\begin{equation}
	\label{eq:null_normal}
	\overline{H}_0: \delta = 0
	\quad 
	\text{ vs } \quad 
	\overline{H}_1: \delta \neq 0.
\end{equation}
In the remainder of this section, we will study valid $t$-test for \eqref{eq:null_normal} in this stylized setting.

\begin{re}\label{re:one_sided}
	If we are interested in an one-sided alternative, such as $\mu_1 > \mu_0$, we can reject the null hypothesis if $\widehat{T}_m > c$ for some $c > 0$. 
	Moreover, under the same condition as in Theorem \ref{thm:size_asymp}, 
	$\lim_{n\to\infty} \bP[\widehat{T}_m > c]
	= 
	\bP[T_m > c] = \frac{1}{2}\bP[|T_m| > c]$. 
	Therefore, the critical value for a level-$\alpha$ one-sided test can be derived from the critical value for a level-($2\alpha$) two-sided test, for any $\alpha\in [0, \frac{1}{2}]$.
\end{re}

\begin{re}
	The setup in this paper also works when there is a single control cluster and at least two treated clusters. All the results would follow by labeling the first $m$ clusters as the treated clusters and the $(m+1)$th cluster as the control cluster. 
\end{re}

\begin{re}
Although we focus on a fixed $m$, we discuss in Appendix \ref{sec:4.6} on the behavior of the test when $m$ is large, where we derive a closed-form valid test and approximate its power.
\end{re}

\subsection{Valid $t$-test} \label{sec:theory-validity}

The key to constructing a valid $t$-test is to find an appropriate critical value $\cv$ for Step 2 of Algorithm \ref{algo:t-test}. This critical value $\cv$ has to be chosen such that for any $\{\sigma_j\}^{m+1}_{j=1}$ that satisfy Assumption \ref{assu:relative_heter}, $\bP[|T_m|>\cv]$ under the null hypothesis \eqref{eq:null_normal} is less than or equal to a given significance level $\alpha$. In the following, we will consider the maximum rejection probability at any critical value $c$. \par 

Note that when the treated standard deviation is much larger than the control standard deviations, i.e., $\sigma_{m+1} \gg \max_{1\le j \le m}\sigma_j$, 
the rejection probability $\bP[|T_m|>c]$ will be close to $1$ for any $c>0$, under which we cannot derive a meaningful $t$-test.
Therefore, Assumption \ref{assu:relative_heter} on the relative heterogeneity of $\{\sigma_j\}^{m+1}_{j=1}$ is in some sense necessary.

For descriptive convenience, we introduce $\cSS$ to denote all possible standard deviations that satisfy Assumption \ref{assu:relative_heter} for a given $m$, $\rho \geq 0$ and $k \in \{1, \ldots, m\}$\footnote{With a slight abuse of notation, we also use $\{\sigma_j\}_{j=1}^{m+1}$ to denote generic values of the standard deviations. Note that $\{\sigma_j\}_{j=1}^{m+1}$ in Assumptions \ref{assu:CLT} and \ref{assu:relative_heter} represent the true (asymptotic) standard deviations for all the clusters.}:
\begin{align}\label{eq:S_k_rho}
    	\cSS
	\equiv
	\{ (\sigma_1, \ldots, \sigma_m, \sigma_{m+1}) \in \bR^{m+1}_{\ge 0}:
	\sigma_{m+1}
	\leq
	\rho \sigma_{(k)}
	\}.
\end{align}
Using the above notation, our goal of finding the maximum rejection probability as described above can be formalized as follows 
\begin{align}\label{eq:max_rej_prob}
	p_m(c; k, \rho) & \equiv \sup_{(\sigma_1, \ldots, \sigma_m, \sigma_{m+1}) \in \cSS} \bP_0 [|T_m| > c],
\end{align}
where the $\bP_0$ notation with subscript $0$ indicates that the null hypothesis \eqref{eq:null_normal} holds and $\delta=0$. 
To conduct a valid test at any given significance level $\alpha\in(0,1)$, it suffices to find the critical value $c$ such that $p_m(c; k, \rho) \le \alpha$ and ideally $p_m(c; k, \rho) = \alpha$. On the other hand, if the goal is to compute a $p$-value, then one can directly compute $p_m(|T_m|; k, \rho)$ without searching for the required critical value as discussed in Section \ref{sec:ci-p}.
That is,
when we set $c$ to be the observed absolute value of the $t$-statistic, $p_m(c; k, \rho)$ gives a valid $p$-value for testing the null hypothesis in \eqref{eq:null_normal}. 

In the following, we consider two cases depending on the standard deviation for the treated cluster $\sigma_{m+1}$. 
Section \ref{sec:sigma_treat_0} considers $\sigma_{m+1} = 0$ and shows that the corresponding maximum rejection probability has a closed-form solution that can be easily computed. 
Section \ref{sec:sigma_treat_0_positive} considers $\sigma_{m+1} > 0$ and shows that we have an integral representation for the rejection probability. This facilitates its numerical calculation at any given values of $\{\sigma_j\}_{j=1}^{m+1}$. However, directly evaluating the optimization problem in \eqref{eq:max_rej_prob} generally results in solving an $m$-dimensional optimization problem. We discuss how this can be reduced to solving multiple one-dimensional optimization problems in Section \ref{sec:theory-max-rej-prob-detail}.

\subsubsection{Case 1: $\sigma_{m+1} = 0$}\label{sec:sigma_treat_0}

When $\sigma_{m+1} = 0$, 
regardless of the choices of the relative heterogeneity parameters $\rho$ and $k$, the control standard deviations $\{\sigma_j\}^m_{j=1}$ can take arbitrary values in $\mathbb{R}^m_{\ge 0}$. 
Moreover, in this case, except for a multiplicative constant scaling factor of $\sqrt{m}$ in the $t$-statistic, our $t$-statistic essentially reduces to a one-sample $t$-statistic,  and our $t$-test is equivalent to
testing whether the mean of the control clusters is equal to zero. 
From \citet{Bakirov:2006aa}, we know that the maximum rejection probability with a zero $\sigma_{m+1}$ has the following form:
\begin{align}\label{eq:max_rej_prob_zero}
	p_{m,0}(c) 
    & = \sup_{\substack{(\sigma_1, \ldots, \sigma_m) \in \mathbb{R}^m_{\ge 0} \\ \text{ and } \sigma_{m+1} = 0}} \bP_0 [|T_m| > c]
    =\max_{R_m(c) < j \le m} \bP\left[ |t_{j-1}| > \sqrt{\frac{(j-1) R_m(c)}{j-R_m(c)}} \right], 
\end{align}
where $R_m(c) \equiv \frac{m^2c^2}{mc^2+m-1}$, and $t_{j-1}$ denotes a random variable following the $t$-distribution with degrees of freedom $(j-1)$. 
In \eqref{eq:max_rej_prob_zero}, the maximum rejection probability is obtained when the control standard deviations $\{\sigma_j\}_{j=1}^m$ are either zero or take some common positive value. It can be efficiently computed by calculating the tail probabilities of at most $(m-1)$ $t$-distributions with various degrees of freedom.

It follows that when $\rho = 0$, we must have $\sigma_{m+1}=0$, and the maximum rejection probability of our $t$-test becomes the same as \eqref{eq:max_rej_prob_zero}, i.e., $p_m(c; k, 0) = p_{m,0}(c)$ for any $c>0$ and $k \in \{1, \ldots, m\}$. 
Thus, in the remainder of this section, 
we focus on $\rho > 0$.

\begin{re}\label{re:c_greater_1_over_sq_m}
From \eqref{eq:max_rej_prob_zero} and as commented in \citet{Bakirov:2006aa}, $p_{m,0}(c) = 1$ for $0<c<m^{-1/2}$ and $p_{m,0}(c) = 0.5$ for $c = m^{-1/2}$. 
Furthermore, by the right-continuous property of distribution functions, we can verify that $p_{m,0}(c) \longrightarrow 0.5$ as $c$ approaches $m^{-1/2}$ from the right. These observations mean that, regardless of the values of $\rho$ and $k$, the maximum rejection probability of our $t$-test in \eqref{eq:max_rej_prob} is $1$ when $0<c<m^{-1/2}$, at least $0.5$ when $c = m^{-1/2}$, and at least about $0.5$ when $c$ is greater than but close to $m^{-1/2}$. 
Thus, it is generally innocuous to assume 
$c\ne m^{-1/2}$ (or even $c > m^{-1/2}$) as in Theorems \ref{thm:size_asymp} and later in \ref{thm:max_rej_prob} for most conventional significance levels. 
\end{re}

\subsubsection{Case 2: $\sigma_{m+1} > 0$} \label{sec:sigma_treat_0_positive}

We now consider the case where $\sigma_{m+1} > 0$.
As shown in the following lemma, the rejection probability $\bP_0[|T_m|>c]$ can be written as an integral. This not only facilitates numerical computation, but is also crucial for our later theoretical investigation.

\begin{lem}\label{lem:rej_prob_integral_main}
Suppose that $\sigma_{m+1} > 0$, and define $\gamma_i \equiv \frac{\sigma_i}{\sigma_{m+1}}$ for $i = 1, 2, \ldots, m$. 
The rejection probability $\bP_0[|T_m|>c]$ can be written as:
\begin{align}\label{eq:rej_prob_integral_main}
	\bP_0[|T_m|>c] = 
	\overline{p}_m(c; \gamma_1, \ldots, \gamma_m)
	\equiv \frac{1}{\pi} 
		 \int^{|\theta_{m+1}|}_0	\frac{s^{\frac{m-1}{2}}}{
         [- g_c(-s)]^{\frac{1}{2}}
         } \ \mathrm{d}s.
\end{align}
where $\theta_{m+1} \in [-m - \max_{1\le i \le m}\gamma_i^2, \  -m]$ is the unique negative root of $g_c(\theta)$ defined below,  
\begin{align}\label{eq:g_fun_char}
	g_c(\theta) \equiv -(m+ \theta)
        \prod_{i=1}^m ( \kappa\gamma_i^2 - \theta)  
        +
        \left( \kappa 
        + \frac{\kappa+1}{m} \theta\right)
        \cdot 
        \sum_{i=1}^m 
        \left[
        \gamma_i^2 \prod_{j\ne i} ( \kappa\gamma_j^2 - \theta ) 
        \right],
\end{align}
and $ \kappa \equiv \frac{m c^2}{m-1}$.
\end{lem}

From Lemma \ref{lem:rej_prob_integral_main}, when $\sigma_{m+1}>0$, the rejection probability in \eqref{eq:rej_prob_integral_main} depends only on the ratios between the control standard deviations and the treated standard deviation. 
The relative heterogeneity assumption in Assumption \ref{assu:relative_heter} equivalently assumes that at least $(m-k+1)$ of $\{\gamma_j\}^m_{j=1}$ is greater than or equal to $\rho^{-1}$. 
To find the maximum rejection probability, we need to solve an $m$-dimensional optimization over $(\gamma_1, \ldots, \gamma_m) \in [0, \infty)^{k-1} \times [\rho^{-1}, \infty)^{m-k+1}$.\footnote{Note that the rejection probability is invariant to any permutations of the $\{\gamma_j\}^m_{j=1}$. We can, for example, assume that the last $(m-k+1)$ of them is no less than $\rho^{-1}$ without loss of generality.}
Optimizing this integral directly can be computationally challenging even for a moderate $m$. As demonstrated in the next subsection, for any $m\ge 2$, we can substantially simplify the $m$-dimensional optimization problem into multiple one-dimensional optimization problems. The reformulated problem can often be efficiently solved numerically.

\subsubsection{Reformulating the problem of computing the maximum rejection probability} \label{sec:theory-max-rej-prob-detail}

In this subsection, we discuss how to reduce the optimization for the maximum rejection probability in \eqref{eq:max_rej_prob} into multiple one-dimensional optimization problems. We summarize the key ideas here and defer the technical lemmas to Appendix \ref{app:theory-max-rej-prob-detail} of this main text.

First, Lemma \ref{lemma:max_at_finite} shows that the maximum rejection probability in \eqref{eq:max_rej_prob} must be achieved at some finite values of $(\sigma_1, \ldots, \sigma_m, \sigma_{m+1}) \in \cSS$. This means we do not need to worry about the boundary case where some of $\{\sigma_j\}_{j=1}^{m+1}$ approach infinity. 
Next, Lemma \ref{lem:first_second_deriv_gamma} shows the necessary conditions for any $(\sigma_1, \ldots, \sigma_m, \sigma_{m+1}) \in \cSS$ to be a maximizer. 
We summarize the implications of Lemma \ref{lem:first_second_deriv_gamma} in the following theorem.

\begin{thm}\label{thm:form_maximizer_main}
    For any given $m\ge 2$, $k \in \{1, \ldots, m\}$, $\rho \ge 0$, $c>0$ and  $c\ne m^{-1/2}$, the maximum rejection probability in \eqref{eq:max_rej_prob} under Assumption \ref{assu:relative_heter} must be attained at some $\{\sigma_j\}_{j=1}^{m+1} \in \cSS$ that satisfy one of the following forms:\footnote{The result in (i) follows from \citet{Bakirov:2006aa} and implies the maximum rejection probability when $\sigma_{m+1}=0$ as shown in \eqref{eq:max_rej_prob_zero}.}
    \begin{itemize}
        \item[(i)] $\sigma_{m+1} = 0$, and for $1\le j \le m$, $\sigma_j$ is 
        either $0$ or some common value;
        \item[(ii)] $\sigma_{m+1} = \rho$, and for $1\le j \le m$, $\sigma_j$ is 
        either $0$, or $1$, or some common value. 
    \end{itemize}
\end{thm}

Note that $0$ is always a boundary value of $\sigma_j$ for $1\le j \le m$, 
and $1$ is also a boundary value for some $\sigma_j$ when $\sigma_{m+1} = \rho$ and Assumption \ref{assu:relative_heter} holds. 
Therefore, Theorem \ref{thm:form_maximizer_main} essentially indicates that, at the maximizer of the rejection probability in \eqref{eq:max_rej_prob}, each $\sigma_j$ must either lie on the boundary or share a common value for all $1 \le j \le m$.
Importantly, 
Theorem \ref{thm:form_maximizer_main} explains why we can simplify the optimization for the rejection probability into multiple one-dimensional optimization problems. This is because for all possible maximizers shown in Theorem \ref{thm:form_maximizer_main}, there is only one unknown value, which is the common value of the $\{\sigma_j\}^m_{j=1}$ that are not on the boundary.
Note that the maximum rejection probability in case (i) with $\sigma_{m+1} = 0$ can be efficiently computed as shown in \eqref{eq:max_rej_prob_zero}. 
In the following, we will therefore focus on the optimization for the rejection probability in case (ii) with $\sigma_{m+1} > 0$.

Specifically, under Assumption \ref{assu:relative_heter} with any $k \in \{1, \ldots, m\}$ and $\rho > 0$, if $(\sigma_1, \ldots, \sigma_m,$ $\sigma_{m+1}) \in \cSS$ is the maximizer for the rejection probability and $\sigma_{m+1}>0$, then, for all $1\le j \le m$, $\gamma_j \equiv \frac{\sigma_j}{\sigma_{m+1}}$ is either on the boundary (equals $0$ or $\rho^{-1}$ ), or takes some common value $\gamma \ge 0$. 
Motivated by this, with a slight abuse of notation, we introduce $\overline{p}_m(c; \rho, \gamma; m_1, m_0)$ to denote the value of $\overline{p}_m(c; \gamma_1, \ldots, \gamma_m)$ in  \eqref{eq:rej_prob_integral_main} when $m_1$ of $\{\gamma_j\}_{j=1}^m$ equal $\rho^{-1}$, $m_0$ of them equal zero, and the remaining equal a common value $\gamma$. That is, 
\begin{align}\label{eq:p_m_given_num_gamma}
	\overline{p}_m(c; \rho, \gamma; m_1, m_0) 
	\equiv 
	\overline{p}_m(c; \rho^{-1} \bs{1}_{m_1}^\top, 
	\bs{0}_{m_0}^\top, 
	\gamma \bs{1}_{m-m_1-m_0}^\top),
\end{align}
where $0 \leq m_0, m_1 \leq m$ and $m_0 + m_1 \leq m$.

Next, define the following as the supremum of $\overline p_m(c; \rho, \gamma; m_1, m_0)$ over $\gamma$, with the support of $\gamma$ depending on $m_1$ and $k$:
\begin{align}\label{eq:max_rej_m1_m0}
   \widetilde{p}_m(c; k, \rho; m_1, m_0) 
   \equiv \sup_{\gamma \in [\underline\rho, \infty)}
   \ \overline{p}_m(c; \rho, \gamma; m_1, m_0),
\end{align}
where $\underline\rho = 0$ if $m_1 \ge m-k+1$ and $\underline\rho = \rho^{-1}$ if $m_1 < m-k+1$.

With the above notations, we now state our main theorem of how to evaluate the maximum rejection probability for a given critical value $c$, heterogeneity parameters $(k, \rho)$ and the number of clusters $m$.
\begin{thm}\label{thm:max_rej_prob}
For any given $m\ge 2$, $k \in \{1, \ldots, m\}$, $c>0$ and  $c\ne m^{-1/2}$, the maximum rejection probability in \eqref{eq:max_rej_prob} under Assumption \ref{assu:relative_heter} can be written as\footnote{
From Remark \ref{re:c_greater_1_over_sq_m}, $p_m(c; k, \rho) = 1$ when $0 < c < m^{-1/2}$. For convenience, we also define $p_m(c; k, \rho) = 1$ at $c = 0$ and $c = m^{-1/2}$, so that $p_m(c; k, \rho)$ upper bounds the rejection probability for all $c \ge 0$.
}
\begin{align*}
	p_m(c; k, \rho) 
        &
        \equiv
        \begin{cases}
            \displaystyle\max\Big\{ 
            p_{m,0}(c),
            \max_{0\le m_0 \le k-1, 0\le m_1 \le m-m_0} \widetilde{p}_m(c; k, \rho; m_1, m_0) \Big\}
            & \text{ if $\rho > 0$} \\
            p_{m,0}(c)
            & \text{ if $\rho = 0$} 
        \end{cases},
\end{align*}
where $p_{m,0}(c)$ is defined in \eqref{eq:max_rej_prob_zero} and $\widetilde{p}_m(c; k, \rho; m_1, m_0)$ is defined in \eqref{eq:max_rej_m1_m0}.
\end{thm}

From Theorem \ref{thm:max_rej_prob}, the key to obtaining the maximum rejection probability is to solve the optimization problem in \eqref{eq:max_rej_m1_m0} for all combinations of $(m_1, m_0)$. 
Importantly, 
for each given $(m_1, m_0)$, the optimization in \eqref{eq:max_rej_m1_m0} is an one-dimensional optimization problem, which is computationally much simpler than the original optimization in \eqref{eq:max_rej_prob}, and we can solve it using numerical optimization. 
In particular, for any given $k \in \{1, \ldots, k\}$, we need to solve at most $\frac{1}{2}k(2m+1-k)$ one-dimensional optimization problems\footnote{When $m_1+m_0 = m$, \eqref{eq:p_m_given_num_gamma} no longer depends on $\gamma$, and thus no optimization over $\gamma$ is needed.}, which is $m$ when $k=1$, $2m-1$ when $k=2$, \ldots, and $\frac{1}{2}m(m+1)$ when $k=m$.
Moreover, at a given $m$, $\rho \geq 0$ and $c > 0$, to find the maximum rejection probability over all $k \in \{1, \ldots, m\}$, we need to solve at most $m(m+1)$ one-dimensional optimization problems of the form \eqref{eq:max_rej_m1_m0}, because the optimizations required for different values of $k$ overlap with each other.

We illustrate the above theorem through the following examples. Example \ref{eg:1} is about $k = 1$ and Example \ref{eg:2} is about $k = 2$.

\begin{eg} \label{eg:1}
Let $m=6$, $\rho = 2$, $\alpha = 0.05$, $k = 1$, and $c = t_{m-1, 1-\frac{\alpha}{2}} \sqrt{\rho^2 + \frac{1}{m}}$ where $t_{m-1, 1 - \frac{\alpha}{2}}$ is the $(1- \frac{\alpha}{2})$-th quantile of the $t$-distribution with $(m-1)$ degrees of freedom.
We show the various functions from Theorem \ref{thm:max_rej_prob} in Figure \ref{fig:thm-eg-1}.

\begin{figure}[!ht]
    \centering
    \includegraphics[scale=1]{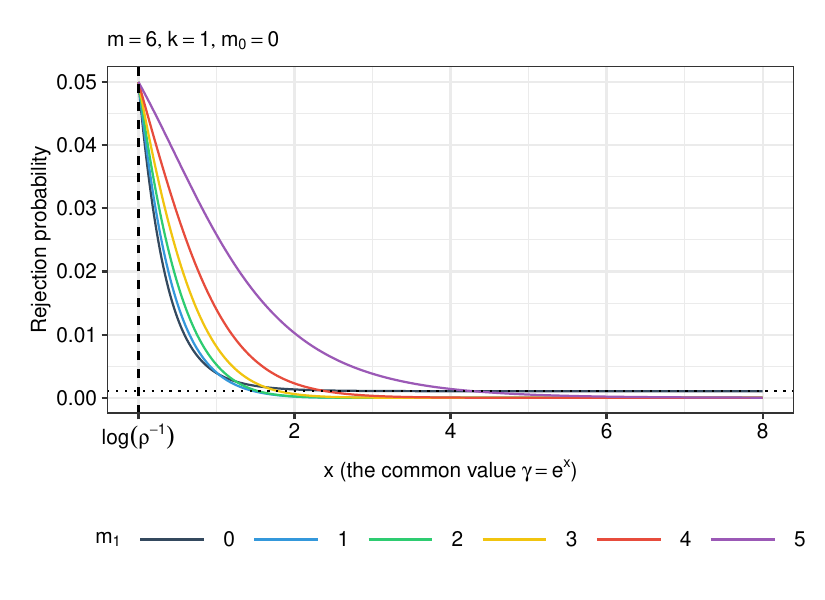}
    \caption{This above shows $\overline{p}_m(c; \rho, \gamma; m_1, m_0)$ against $\gamma$ for various values of $m_1$ with $m = 6$, $k = 1$ and $m_0 = 0$. The vertical dashed line represents $\gamma = \rho^{-1}$. The horizontal dotted line represents  $p_{m,0}(c)$. See Theorem \ref{thm:max_rej_prob} for the definitions of the functions.}
    \label{fig:thm-eg-1}
\end{figure}

A few observations from Figure \ref{fig:thm-eg-1} are as follows. First, the curves show $\overline{p}_m(c; \rho, \gamma; m_1, m_0)$ as defined in \eqref{eq:p_m_given_num_gamma} against $\gamma$, for $m_1= 0, \ldots, 5$. Recall this function means that among $\{\gamma_j\}^6_{j=1}$, $m_1$ of them equals $\rho^{-1}$ and the remaining of them equals $\gamma$. We set $\gamma = e^x$ in order to display the behavior of large $\gamma$. Each colored curve corresponds to a specific value of $m_1$. It can be seen that $\overline{p}_m(c; \rho, \gamma; m_1, m_0)$ decreases as $\gamma$ increases for each $m_1$.

The black dotted horizontal line plots $p_{m,0}(c)$ defined in \eqref{eq:max_rej_prob_zero}. It shows the maximum rejection probability when $\sigma_{m+1}=0$. Here, it is much smaller than 0.05.

As $\gamma \longrightarrow \infty$, $\overline{p}_m(c; \rho, \gamma; m_1, m_0)$ converges to values less than or equal to $p_{m,0}(c)$. 
This is not surprising, as in such scenarios, $\sigma_{m+1}$ is much smaller than some control standard deviations, essentially approximating the case where $\sigma_{m+1} = 0$.
\end{eg}

\begin{eg} \label{eg:2}
Consider the same setup as in Example \ref{eg:1} but with $k = 2$ and set the value of $c$ such that the maximum rejection probability equals $\alpha = 0.05$. 
Figure \ref{fig:thm-eg-2} shows the rejection probability under different combinations of $\gamma$, $m_0$ and $m_1$, where $m_0$ here can only take the values 0 or 1.

\begin{figure}[!ht]
    \centering
    \includegraphics[scale=1]{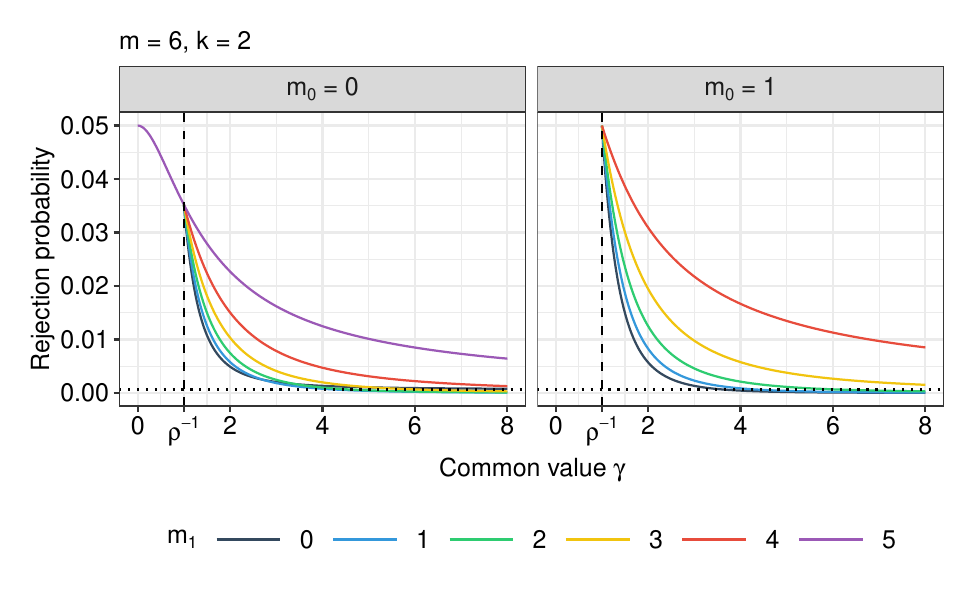}
    \caption{This above shows $\overline{p}_m(c; \rho, \gamma; m_1, m_0)$ against $\gamma$ for various $m_0$ and $m_1$ with $m = 6$ and $k = 1$. See the caption for Figure \ref{fig:thm-eg-1} for more details.}
    \label{fig:thm-eg-2}
\end{figure}

As can be seen from the figure, the maximum rejection probability is achieved when one of the $\{\gamma_j\}^m_{j=1}$ equals 0 and the remaining $(m-1)$ terms from $\{\gamma_j\}^m_{j=1}$ equals $\rho^{-1}$. For instance, in the $m_0 = 0$ panel, the purple line represents $m_1 = 5$ terms of $\{\gamma_j\}^m_{j=1}$ equals $\rho^{-1}$. For this line, the maximum rejection probability is achieved when the remaining $m - m_1 = 1$ term equals $\gamma = 0$. Similarly, in the $m_0 = 1$ panel, exactly one of $\{\gamma_j\}^m_{j=1}$ is equal to 0. Each of the lines shows that the maximum rejection probability is achieved when all the remaining terms equal $\rho^{-1}$.
\end{eg}

\begin{re}\label{rmk:guess_candidate}
As discussed in Appendix \ref{sec:4.6}, when $m\longrightarrow \infty$, under certain regularity conditions, the maximum rejection probability $p_m(c; k, \rho)$ for any $c, k$ and $\rho>0$ is achieved when $\sigma_{m+1} = 1$, $(k-1)$ of $\{\sigma_j\}_{j=1}^m$ are $0$, and the remaining $(m-k+1)$ of $\{\sigma_j\}_{j=1}^m$ are equal to $\rho^{-1}$. 
This is indeed the case for both Examples \ref{eg:1} and \ref{eg:2}.
Based on this intuition, we can first use this configuration of $\{\sigma_j\}_{j=1}^{m+1}$ to get a candidate threshold $c$ such that $\bP_0 [|T_m| > c]=\alpha$, where $\alpha$ is the significance level of interest. 
We can then verify whether this candidate threshold $c$ achieves the desired type-I error control by solving the optimization in Theorem \ref{thm:max_rej_prob}.
\end{re}

Finally, we report the critical values for different numbers of control clusters $m$ and heterogeneity parameters $\rho$ for $\alpha = 0.01$ and $\alpha = 0.05$ in Table \ref{tab:cv} when $k = 1$. We report the critical values for $k = 2$ in the supplementary material.

\begin{table}[!ht]
    \centering
    \caption{Critical values for different values of $\alpha$, $m$ and $\rho$ for $k = 1$.}
    \label{tab:cv}
    {\small \include{tables/paper_cv_k1}}
\end{table}

\subsection{Closed-form valid $t$-test when $k = 1$} \label{sec:theory-prob-k=1}

In this subsection, we consider the case where $k = 1$ in Assumption \ref{assu:relative_heter}. This means that $\sigma_{m+1} \leq \rho \sigma_j$ for $j = 1, \ldots, m$, i.e., the treated standard deviation is smaller than or equal to $\rho$ times each of the control standard deviations. \par 

In this case,  we can obtain closed-form solutions for the maximum rejection probability when the threshold $c$ is large enough. Equivalently, this corresponds to testing at a significance level that is not ``too large.'' 
The theorem is stated as follows:

\begin{thm}\label{thm:max_rej_prob_k_1_simple_closed_form}
For any given $m \ge 4$, $\rho > 0$, $c > \sqrt{\frac{3(m-1)}{m(m-3)}}$, define:
\begin{align*}
	\ohh  
	& \equiv 
        \max\left\{
        \frac{3(m\rho^2 + 1)}{ m\rho^2 + \kappa + 1}, 
        \frac{2\kappa + 3}{\kappa + 1}
        \right\}
        \notag \\
    & \qquad 
    + 
    \frac{1-\tau}{1-\tau +\min\{ (1 - 2\tau)\kappa \underline Z - \frac{1}{2} , 0\}}
    -\frac{m\kappa}{m\rho^2 + \kappa+1}
	-  1,
\end{align*}
where $\kappa \equiv \frac{m c^2}{m-1}$ and $\tau \equiv \frac{\kappa+1}{m\kappa}$ are determined by $(m, c)$,  and $\underline{Z} \equiv
    \frac{1}{2 \cdot \max\{m\rho^2 + 1, \kappa + 2\}}$ is determined by $(m, c, \rho)$. Under the above conditions and notations, the following statements hold.
\begin{itemize}
    \item[(i)] $\ohh$ is decreasing in $c$, and $\lim_{c\to \infty} \ohh < 0$. 
    \item[(ii)] Let $\underline{c}_{m, \rho} \equiv \inf\{c > \sqrt{\frac{3(m-1)}{m(m-3)}}: \ohh \le 0 \}$, which must be finite.
    Then, for any $c\ge \underline{c}_{m, \rho}$, the maximum rejection probability $p_m(c; 1, \rho)$ under Assumption \ref{assu:relative_heter} with $k=1$ and the given $(\rho, m, c)$ has the following equivalent form:
	\begin{align*}
		p_m(c; 1, \rho) = 
		\bP \left[ |t_{m-1}|\sqrt{\rho^2+\frac{1}{m}} > c \right].
	\end{align*}
\end{itemize}
\end{thm}

Theorem \ref{thm:max_rej_prob_k_1_simple_closed_form} gives a closed-form solution for the maximum rejection probability under the relative heterogeneity assumption with $k=1$ and any given $\rho > 0$. The maximum rejection probability is achieved when $\sigma_1 = \sigma_2 = \ldots = \sigma_m = \rho^{-1}$ and $\sigma_{m+1} = 1$. 
Importantly, for any given $m$ and $\rho>0$, 
the cutoff $\underline{c}_{m, \rho}$ can be easily computed numerically, due to the monotonicity of the function $\ohh$ in $c$. 
Accordingly, 
Theorem \ref{thm:max_rej_prob_k_1_simple_closed_form} gives a closed-form critical value of our $t$-test for any significance level less than or equal to  
\begin{align}\label{eq:max_sig_level_k1}
    \underline{\alpha}_{m, \rho} \equiv \bP \left[ |t_{m-1}|\sqrt{\rho^2+\frac{1}{m}} > \underline{c}_{m, \rho} \right].
\end{align}
That is, for any significance level $\alpha \in (0, \underline{\alpha}_{m, \rho}]$, a valid critical value for our two-sided $t$-test is 
$\sqrt{\rho^2 + m^{-1}} \ t_{m-1,1- \frac{\alpha}{2}}$, where $t_{ m-1,1- \frac{\alpha}{2}}$ is the $(1-\frac{\alpha}{2})$ quantile of the $t$-distribution with degree of freedom $m-1$.

Table \ref{tab:largest_alpha} reports the largest significance level $\underline{\alpha}_{m, \rho}$ such that our two-sided $t$-test has a simple closed-form expression for the critical value, under Assumption \ref{assu:relative_heter} with $k=1$ and various values of $(m, \rho)$. 
Note that our $t$-test can handle all values of $m\ge 2$, $k \in \{1, \ldots, m\}$ and $\rho \ge 0$, but it generally involves one-dimensional optimization as described in Section \ref{sec:theory-max-rej-prob-detail}. 
We can conduct a similar theoretical investigation to simplify the optimization of the rejection probability under Assumption \ref{assu:relative_heter} with $k \geq 2$. However, in this case, the potential maximizers cannot be reduced to a single point, so one-dimensional numerical optimization is still required. We relegate the detailed discussion under Assumption
\ref{assu:relative_heter} with $k\geq 2$ to the supplementary material.

\begin{table}[!ht]
    \centering
    \caption{Largest significance value {in \eqref{eq:max_sig_level_k1} such that our two-sided $t$-test has a closed-form critical value under various values of $(m, \rho)$.} 
    }
    \label{tab:largest_alpha}
   \include{tables/paper_alphas_v6}
\end{table}

\subsection{Power of the \textit{t}-test} \label{sec:theory-power}

We now investigate the power of the proposed $t$-test under the alternative hypothesis $\overline{H}_1$ in \eqref{eq:null_normal} with $\delta \ne 0$. 
Without loss of generality, we assume that $\delta>0$. 
The theorem below gives a lower bound for the power of the $t$-test. 
\begin{thm}\label{thm:power}
    Suppose that Assumption \ref{assu:normal} holds for some $\delta>0$, 
    and define the $t$-statistic $T_m$ as in \eqref{eq:pop-test-1}. 
    Then, for any finite $c>0$,  
    \begin{align*}
        \bP[|T_m|>c] \ge \bP[T_m>c] \ge 1 - \frac{1}{\delta^2} \left[ \sigma^2_{m+1} + \frac{2(c^2 + m^{-1})}{m} \sum_{j=1}^m \sigma_j^2 \right].
    \end{align*}
\end{thm}

The lower bound of power in Theorem \ref{thm:power} increases with $\delta$, where $\delta$ corresponds to $\sqrt{n}(\mu_1 - \mu_0)$ in our large-sample inference for a finite number of clusters as shown in Theorem \ref{thm:size_asymp}. 
If the gap between the means of the treated and control clusters is bounded away from zero, then, as the sample size $n\longrightarrow \infty$, the power of our $t$-test with any finite critical value will converge to $1$. 
In addition, 
Theorem \ref{thm:power} also provides a rough power estimate when we have some information about the treatment effect size and the variances for the treated and control clusters.

\section{Simultaneous inference} \label{sec:simu}

Recall that Assumption \ref{assu:relative_heter} involves two parameters $\rho$ and $k$. They represent the allowable degree of relative heterogeneity between treated and control clusters. Specifically, larger values of $\rho$ and $k$ correspond to a greater degree of allowable heterogeneity.
In practice, specifying the values of $(\rho, k)$ might be challenging, and it is often desirable to perform multiple analyses for a wide range of values for $(\rho, k)$. 
This raises the question of how to interpret the results of the analysis under different values of $(\rho, k)$.
In this section, we show that the analyses over all possible values of $(\rho, k)$ can be simultaneously valid, without the need of any adjustment due to multiple analyses. 
Moreover, the analysis results can be easily visualized and interpreted. 
Similar simultaneous inference has been used for sensitivity analysis of observational studies \citep{cuili2025wp,wuli2025jasa}.

We first introduce several notations to denote the true relative heterogeneity of the standard deviations between treated and control clusters. 
For any $k \in \{1, \ldots, m\}$, 
define 
\begin{align}\label{eq:rho_star}
    \rho^\star_{k} 
    \equiv
    \inf \big\{\rho\ge 0: \sigma_{m+1} \le \rho \sigma_{(k)}\big\}
    = 
    \begin{cases}
        \frac{\sigma_{m+1}}{\sigma_{(k)}}, & \text{if } \sigma_{(k)} > 0, \\
        0, & \text{if } \sigma_{(k)} = 0 \text{ and } \sigma_{m+1} = 0, \\
        \infty & \text{if } \sigma_{(k)} = 0 \text{ and } \sigma_{m+1} > 0, 
    \end{cases}
\end{align}
where 
$\sigma_{m+1}$ and $\sigma_{(k)}$ 
in \eqref{eq:rho_star} denote the true standard deviation of the treated cluster and that of the control cluster at rank $k$. 
Consequently, the values of $\rho_k^\star$ reflect the true relative heterogeneity between the treated and control clusters. In particular, larger values of $\rho_k^\star$ indicate greater heterogeneity between the treated and control clusters.

We now explain the key idea for our simultaneous inference procedure. 
In Section \ref{sec:theory}, we test the null hypothesis in \eqref{eq:null_normal} of no mean difference between the treated and control clusters under Assumption \ref{assu:relative_heter} for some $k \in \{1, \ldots, m\}$ and $\rho \ge 0$, 
which is equivalent to that $\rho_k^\star \le \rho$ with $\rho_k^\star$ defined as in \eqref{eq:rho_star}. 
Importantly, we can reinterpret the $t$-test in Algorithm \ref{algo:t-test} as a valid test for the null hypothesis of $\rho_k^\star \le \rho$ about the true relative heterogeneity in \eqref{eq:rho_star}, under the assumption of no mean difference between treated and control clusters (i.e., $\delta = 0$). 
By standard test inversion, we can construct confidence intervals for $\rho_k^\star$, which must be one-sided confidence intervals with unbounded right endpoints and thus provide essentially lower bounds on the true relative heterogeneity, under the assumption that $\delta=0$. 
Moreover, these confidence intervals will be simultaneously valid across all $k \in \{1, \ldots, m\}$. 
We summarize the results in the following theorem, followed by a discussion of its practical implications and interpretation.
For any $z\in \mathbb{R}$, 
we write $(z, \infty] \equiv (z, \infty)\cup \{\infty\}$ and analogously $[z, \infty] \equiv [z, \infty)\cup \{\infty\}$.

\begin{thm}\label{thm:simu_ci_rho}
    Let $\alpha \in (0, 1)$. Suppose Assumption \ref{assu:relative_heter} holds with $\delta = 0$. 
    Let $T_m$ be the test statistic as defined in \eqref{eq:pop-test-1}. 
    \begin{itemize}
        \item[(i)] For any $k \in \{1, \ldots, m\}$, an $(1-\alpha)$-confidence set for $\rho^\star_k$ in \eqref{eq:rho_star} is
        \begin{align}\label{eq:interval_rho_k}
        \mathcal{I}_{m, \alpha, k} 
        \equiv \{
        \rho \ge 0: 
        p_m(|T_m|; k, \rho) > \alpha
        \}
        \cup \{\infty\}
        = 
        (\hat{\rho}_{m, \alpha, k}, \infty ]
        \text{ or }
        [\hat{\rho}_{m, \alpha, k}, \infty ],
        \end{align}
        which must be an one-sided confidence interval with $\hat{\rho}_{m, \alpha, k} \equiv \inf \mathcal{I}_{m, \alpha, k}$. 

        \item[(ii)] 
        The confidence intervals in (i) are simultaneously valid across all $k \in \{1, \ldots, m\}$, in the sense that 
        \begin{align*}
            \bP\left[ \rho^\star_k \in \mathcal{I}_{m, \alpha, k} \text{ for all } k \in \{1, \ldots, m\} \right] \ge 1 - \alpha. 
        \end{align*}
    \end{itemize}
\end{thm}

The interpretation of Theorem \ref{thm:simu_ci_rho} is as follows.
If there is no mean difference between treated and control clusters (i.e., $\delta=0$), then we need to believe that, with $(1-\alpha)$ confidence level, the relative heterogeneity $\rho_k^\star$ in \eqref{eq:rho_star} must be greater than (or equal to) $\hat{\rho}_{m, \alpha, k}$, for all $k \in \{1, \ldots, m\}$. 
That is, with $(1-\alpha)$ confidence level, the treated standard deviation $\sigma_{m+1}$ must be at least $\hat{\rho}_{m, \alpha, k}$ times 
the control standard deviation $\sigma_{(k)}$ at rank $k$, for all $k \in \{1, \ldots, m\}$. 
If we question any of these statements on the true relative heterogeneity, then the assumption of $\delta=0$ is likely to fail, or equivalently there is likely significant mean difference between the treated and control clusters. 
Obviously, the larger the values of $\{\hat{\rho}_{m, \alpha, k}\}^m_{k=1}$, the stronger the evidence for a nonzero treatment effect.
In practice, we can easily visualize these confidence intervals by, say, plotting $k$ against $\hat{\rho}_{m, \alpha, k}$; this is illustrated in the two empirical applications in Section \ref{sec:emp} ahead.

We now discuss the computation for the confidence intervals in Theorem \ref{thm:simu_ci_rho}. 
By the definition of the $p$-value $p_m(c; k, \rho)$ in \eqref{eq:max_rej_prob} and the fact that the set $\cSS$ in \eqref{eq:S_k_rho} increases as $k$ or $\rho$ increase, 
$p_m(c; k, \rho)$ must be nondecreasing in both $k$ and $\rho$. 
The monotonicity in $\rho$ then explains 
the equivalent one-sided form of the set in \eqref{eq:interval_rho_k}.
Thus, we can use the bisection method to find the thresholds $\{\hat{\rho}_{m, \alpha, k}\}^m_{k=1}$. 
In addition, the monotonicity in $k$ implies the threshold $\hat{\rho}_{m, \alpha, k}$ is nonincreasing with $k$. 
Hence, we can first find $\hat{c}_{m, \alpha, 1}$, and then sequentially use $\hat{c}_{m, \alpha, k-1}$ as an upper bound for $\hat{\rho}_{m, \alpha, k}$ in the bisection method, for $2\le k \le m$. 

\begin{re}
    \citet{hagemann2024wp} also reported thresholds like  $\{\hat{\rho}_{m, \alpha, k}\}$ for his test. 
    However, his test only works when $k=1$ or 2, while our test works for any $k \in \{1, \ldots, m\}$. 
    In addition, as shown in Theorem \ref{thm:simu_ci_rho}, these thresholds can be interpreted as lower confidence bounds for the true relative heterogeneity, and they are simultaneously valid, indicating that the confidence bounds for all $k$ are indeed ``free lunch'' added to those for $k=1$ or $2$. 
\end{re}

\section{Simulations} \label{sec:sims}
In this section, we consider two sets of numerical exercises to compare the performance of the $t$-test against other methods for conducting inference with a single treated cluster.

\subsection{Simulation design 1: normal means} \label{sec:sims-1}

The first simulation design generates data using normal distributions. The data generating process (DGP) is as follows. We generate normal random variables as in Assumption \ref{assu:normal}. In particular, $\{\psi_j\}^m_{j=1}$ are normally distributed with mean 0, $\psi_{m+1}$ has mean $\delta$, $\sigma_{m+1}^2 = \rho^2$, and $\{\sigma_j^2\}^m_{j=1}$ are specified below.

We consider two different DGPs to generate the random variables for the control clusters.
\begin{description}
    \item[DGP 1.] $\sigma_j^2 = 1$ for $j = 1, \ldots, m$.
    \item[DGP 2.] $\sigma_j^2 = 1 + \frac{j - 1}{m - 1}$ for $j = 1, \ldots, m$.
\end{description}

\textbf{DGP 1} 
is the homogeneous design in which all control random variables have the same varaiance. We introduce heterogeneity in \textbf{DGP 2} and allow the variance to vary between 1 and 2. The goal is to test the two-sided hypothesis \eqref{eq:null_normal}, under various heterogeneity parameter $\rho$, cluster size $m$, and significance level $\alpha$. We compare the performance of our $t$-test with \citet{hagemann2024wp} in terms of size and power. This is because both of our tests work with a single treated cluster and a finite number of control clusters, and are valid under a certain relative heterogeneity assumption.

In the simulations, we consider $\rho$ from 0.1 to 2 with step size 0.1, $m \in \{5, 10, 25, 50\}$, and $\alpha \in \{0.01, 0.05, 0.1\}$. We focus on two-sided tests and choose $\rho$ in our $t$-test and \citet{hagemann2024wp} so that it matches the one used in the DGP, i.e., the relative heterogeneity assumption is always correctly specified. The results are based on 500,000 Monte Carlo replications. Figure \ref{fig:sim1-rho-a05} shows the result for $\alpha = 0.05$ under \textbf{DGP 1} when $k = 1$. The figure plots the rejection rate against various values of $\rho$. Each facet represents a specific combination of $m$ and $\delta$. $\delta = 0$ corresponds to the results under the null. $\delta > 0$ corresponds to the results under the alternative. Note that \citet{hagemann2024wp}'s result does not appear or only partially appear in some of the facets. This is because \citet{hagemann2024wp} may not necessarily be able to find a weight for his rearrangement test such that his test can be shown to be valid for some combinations of $m$, $\rho$ and $k$.

Figure \ref{fig:sim1-rho-a05} shows that our test performs favorably when compared to \citet{hagemann2024wp}. Both of our tests control size. The $t$-test is more powerful, especially when $\rho$ is small. As $\rho$ increases above 1, the power difference decreases, but we are still more powerful.

\begin{figure}
    \centering
    \includegraphics[scale=1]{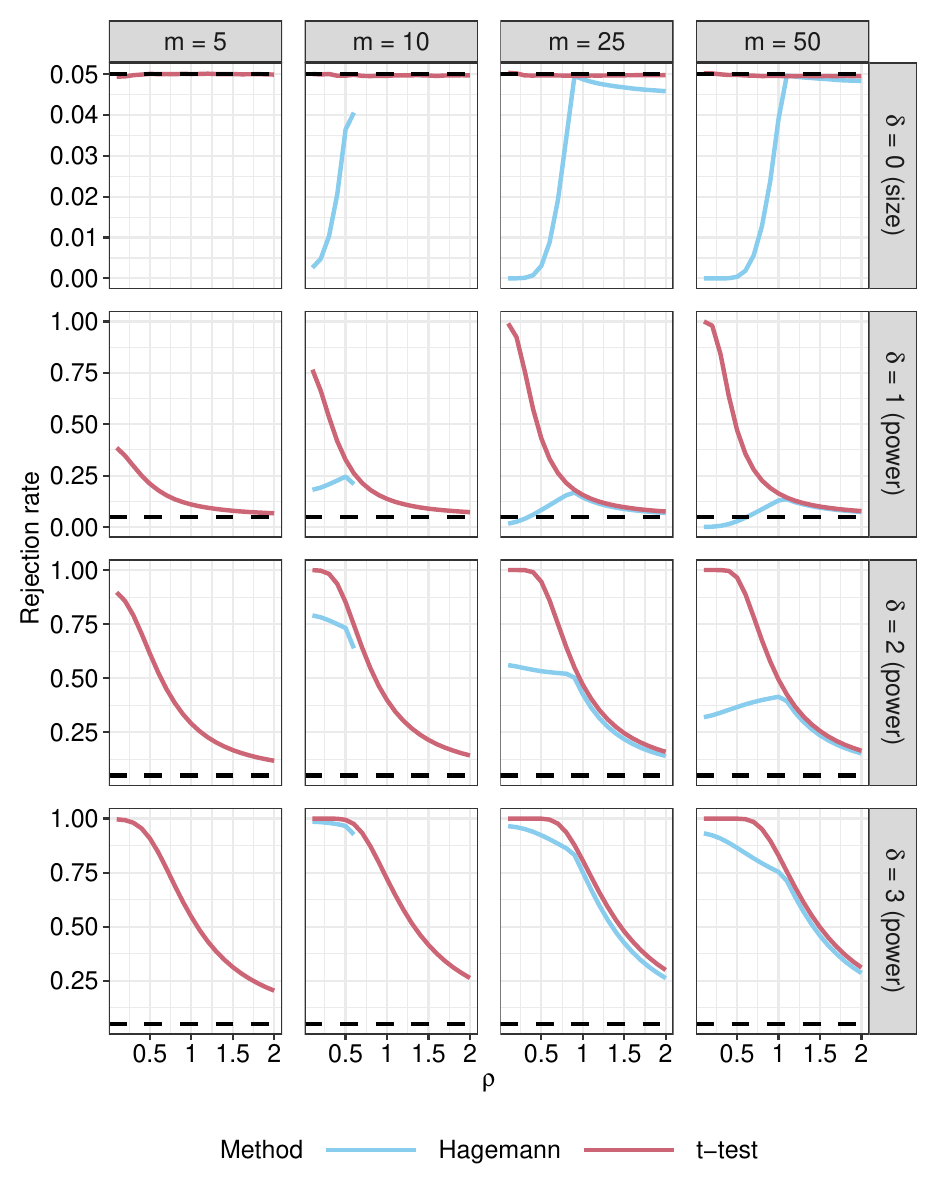}
    \caption{Probability of rejection against heterogeneity parameter $\rho$ for various cluster size $m$ and alternatives $\delta$ at $\alpha = 0.05$ for \textbf{DGP 1} of simulation design 1.}
    \label{fig:sim1-rho-a05}
\end{figure}

Next, Figure \ref{fig:sim1-dgp2-rho-a05} reports the results for \textbf{DGP 2} that includes more heterogeneity among the control random variables for $\alpha = 0.05$ and $k=1$. As predicted by the theory, our $t$-test becomes more conservative when more heterogeneity are included. Our $t$-test continues to control size and has power against the alternative. 
Moreover, our test outperforms \citet{hagemann2024wp}'s test in most cases.
In the supplementary material, we report the results for $k=1$ and $k=2$ for both DGP at various significance levels.

\begin{figure}
    \centering
    \includegraphics[scale=1]{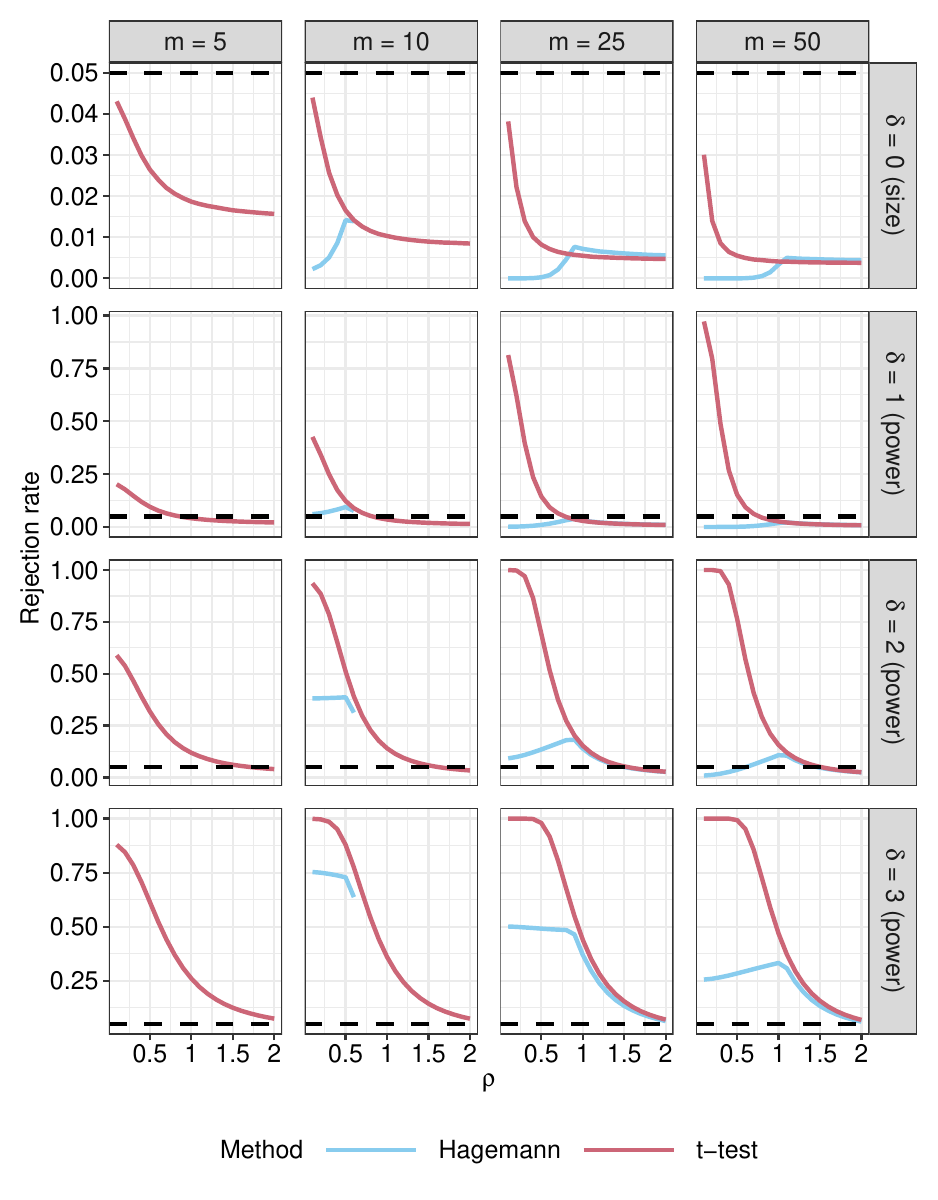}
    \caption{Probability of rejection against heterogeneity parameter $\rho$ for various cluster size $m$ and alternatives $\delta$ at $\alpha = 0.05$ for \textbf{DGP 2} of simulation design 1.}
    \label{fig:sim1-dgp2-rho-a05}
\end{figure}

\subsection{Simulation design 2: two-way fixed effects} \label{sec:sims-2}

Next, we conduct a simulation exercise with two-way fixed effects as in Example \ref{eg:twfe}. We consider a design that is based on the ones used in \citet{conleytaber2011restat} and \citet{hagemann2024wp}. As before, let $\cJ_0 \equiv \{1, \ldots, m\}$ be the control clusters and $\cJ_1 \equiv \{m+1\}$ be the treated cluster.  Let $\cT \equiv \{1, \ldots, 10\}$ be the total number of time periods. Let $t_0 = 6$ be the intervention period and  $D_{jt} \equiv \ind[j \in \cJ_1 \text{ and } t > t_0]$ for $j \in \cJ$ and $t \in \cT$. Each simulated data is generated from the following two-way fixed effects model
\begin{align}
    \label{eq:simulation2}
    \begin{split}
        Y_{jt} & = \alpha_t + \gamma_j + \cp D_{jt}  + U_{jt}, 
    \end{split}
\end{align}
where $\alpha_t = 1$ for all $t \in \cT$ and $\gamma_j = 2 \ind[j \leq \frac{m}{2}] - 1$ for all $j \in \cJ$. We test the null of $\cp = 0$ and consider $\cp \in \{0, 1, 2, 3\}$ when generating the data. We consider the following DGPs based on  model \eqref{eq:simulation2}. The first four DGPs are the same as the ones considered in \citet{hagemann2024wp}. The remaining two DGPs are based on the error distributions considered in \citet{conleytaber2011restat}.
\begin{description}
    \item[DGP 1.] Generate $U_{jt}  = \eta U_{j, t-1} + \sigma^{\ind[j = m + 1]} V_{jt}$ where $\eta = 0.5$ and $V_{jt}$ are independently distributed standard normal random variables.
    \item[DGP 2.] Same as DGP 1, but use $\eta = 0.1$.
    \item[DGP 3.] Same as DGP 1, but use $\eta = 0.9$.
    \item[DGP 4.] Same as DGP 1, but $V_{jt}$ follows a normalized $\chi_2^2$ distribution with mean 0 and variance 1.
    \item[DGP 5.] Same as DGP 1, but use $V_{jt} \sim \text{Uniform}[-\sqrt{3}, \sqrt{3}]$.
\end{description}

For each DGP, we are interested in testing the two-sided hypothesis of 
\[
    H_0: \cp = 0
    \quad \text{ vs } \quad 
    H_1 : \cp \neq 0.
\]
In addition, we consider simulations with $m \in \{5, 10, 25, 50\}$, $
\sigma
\in \{0.1, 0.5, 1, 2\}$, and $\alpha = 0.05$. We report the rejection rate curves based on 5,000 Monte Carlo simulations. For each DGP, we consider conducting inference using the following four methods:

\begin{description}
    \item[t] The $t$-test proposed in our paper.
    \item[H] The rearrangement test by \citet{hagemann2024wp}.
    \item[CT] The procedure of \citet{conleytaber2011restat} that assumes homogeneity.
    \item[FP] The bootstrap procedure of \citet{fermanpinto2019restat}.
\end{description}

Figures \ref{fig:sim2-dgp1-a05} and \ref{fig:sim2-dgp2-a05} show the results for DGPs 1 and 2 under $\alpha = 0.05$ across different numbers of clusters $m$ and true relative heterogeneity $\sigma$.
We set the relative heterogeneity parameter $\rho$ for both our and \citet{hagemann2024wp}'s method in all the DGPs to be  the corresponding $\sigma$, so that the relative heterogeneity assumption is correctly specified. 
We also incorporate the correct value of $\sigma$ in applying FP to specify the heterogeneity.

\begin{figure}[!ht]
    \centering
    \includegraphics[scale=1]{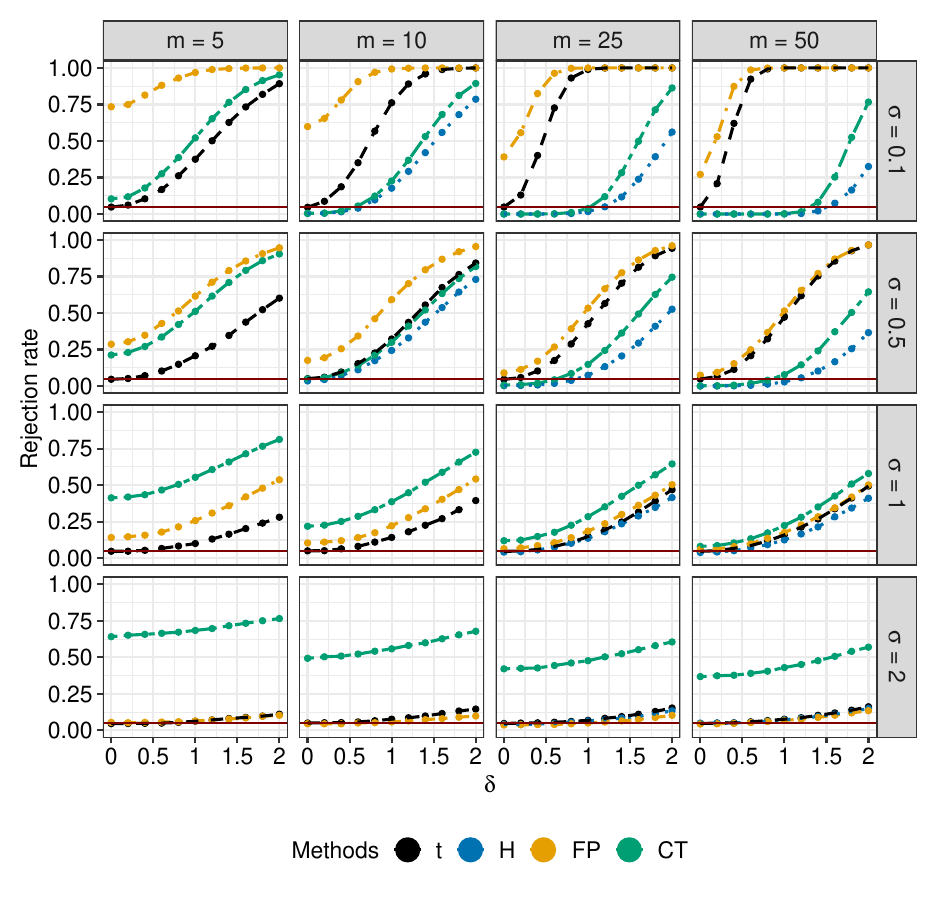}
    \caption{Simulation results for \textbf{DGP 1} of simulation design 2 at $\alpha = 0.05$.}
    \label{fig:sim2-dgp1-a05}
\end{figure}

\begin{figure}[!ht]
    \centering
    \includegraphics[scale=1]{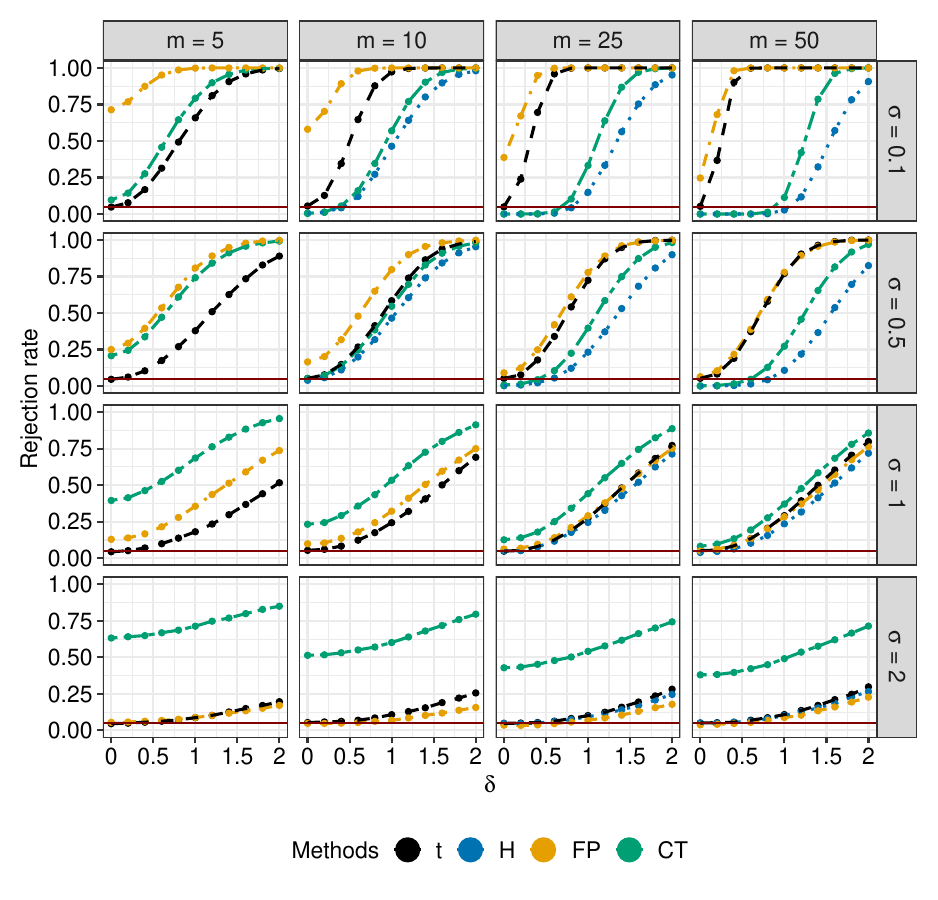}
    \caption{ Simulation results for \textbf{DGP 2} of simulation design 2 at $\alpha = 0.05$.}
    \label{fig:sim2-dgp2-a05}
\end{figure}

It can be seen that our $t$-test performs well and compares favorably to other methods. 
In particular, the CT test exhibits inflated Type I error rates when the number of control clusters $m$ is small and $\sigma$ is large. This inflation arises because the validity of CT relies on the assumption of an infinite number of clusters and homogeneous variances between the  treated and control clusters. The FP test also shows an inflated Type I error when $m$ is small, with the inflation becoming more pronounced as $\sigma$ decreases. This again reflects its reliance on the asymptotic validity under an infinite number of clusters. Notably, we assume that FP has access to the true heterogeneity between the treated cluster and each control cluster, which is a stronger assumption than those required by our $t$-test or \citet{hagemann2024wp}'s rearrangement test. Both our test and that of \citet{hagemann2024wp} successfully control Type I error. In contrast, our test is applicable regardless of the number of control clusters and demonstrates higher power across different levels of heterogeneity. The results for DGPs 3 to 5 are contained in the supplementary material.

\section{Empirical applications} \label{sec:emp}
In this section, we illustrate the $t$-test proposed in this paper with two recent empirical applications. For each of the empirical applications, we report two sets of results. First, we try to find the largest $\rho$ such that the null hypothesis is rejected when $k = 1$ among various regression specifications of the empirical examples. Among those specifications where the null can be rejected at some $\rho$,
we conduct simultaneous inference as discussed in Section \ref{sec:simu}.

\subsection{Empirical application 1: \citet{depewswensen2022ej}}

\citet{depewswensen2022ej} examines the impact of the 1911 New York State Sullivan Act on mortality rate. The Sullivan Act required New York citizens to obtain a permit and license to carry concealable weapon. We consider the following two-way fixed effects model:
	\begin{equation}
	    \label{eq:ej-app-twfe-1}
		\text{Outcome}_{jt}
		= \beta\text{Treated}_{jt}
		+ \alpha_t
		+ \gamma_j
		+ \epsilon_{jt},
	\end{equation}
where $\text{Treated}_{jt}$ equals 1 if state $j$ is New York and $t$ is the post-treatment period, $\alpha_t$ is the year fixed effects, $\gamma_j$ is the state fixed effects, and $\epsilon_{jt}$ is the residual term. They consider the following four outcome variables: homicide rate, suicide rate, gun suicide rate, and non-gun suicide rate. 

They report cluster-robust standard errors and $p$-value from the Wild cluster bootstrap. There is only nine control clusters in this empirical application. Hence, \citet{hagemann2024wp}'s test may not be applicable because there can be no weights such that his rearrangement test is valid for some relative heterogeneity parameters $\rho$. \par 

We are interested in testing
\begin{equation}
    \label{eq:application1-null}
    H_0: \beta = 0
    \quad 
    \text{ vs }
    \quad 
    H_1 : \beta \neq 0.
\end{equation}

Table \ref{tab:application1} summarizes the estimation and inference results of the baseline model described in \eqref{eq:ej-app-twfe-1} for the four outcomes described above. The point estimates and the Wild cluster bootstrap $p$-values are taken from Table 2 of \citet{depewswensen2022ej}.\footnote{The point estimates and $p$-values that \citet{depewswensen2022ej} report are based on weighted least squares.}

\begin{table}[!ht]
    \centering\setlength\tabcolsep{3.3pt}
    \caption{Regression and inference results for empirical application 1 with $k=1$.}
    \label{tab:application1}
    \include{tables/depewswensen2022ej-rcc-v5-step0.025}
\end{table}

The remaining rows of Table \ref{tab:application1} report the largest $\rho$ such that the null can be rejected for significance levels $\alpha = 0.01$, $0.05$, and $0.1$.
For entries with ``NA'' in the table, it refers to the situation where there does not exist a $\rho \geq 0$ such that the null is rejected.
If we focus on the gun suicide rate under $\alpha = 0.1$, it means that the null can be rejected when the variance in New York is at most $1.3^2 \approx 1.69$ larger than the smallest variance in the control clusters. 
The true relative level of heterogeneity such that the null can be rejected may be related to the population size of the various states. In particular, New York state has a larger population than the other control states in 1911 \citep{fred2025website}: the population of New York state was 9.249 million, and the largest control state was Massachusetts with 3.383 million and the smallest control state was Vermont with 0.358 million.

Next, we conduct simultaneous inference for our $t$-test as in Section \ref{sec:simu}. We focus on the outcome ``gun suicide rate'' as we can find $\rho$ such that the null is rejected when $k=1$. Figure \ref{fig:application1} reports the results for $\alpha \in \{0.01, 0.05, 0.1\}$ on the range of $\rho$ for various $k$ such that the null can be rejected. The shaded area represents the simultaneous confidence region for true relative heterogeneity when there is no treatment effect.

\begin{figure}[!ht]
    \centering
    \includegraphics[scale=1]{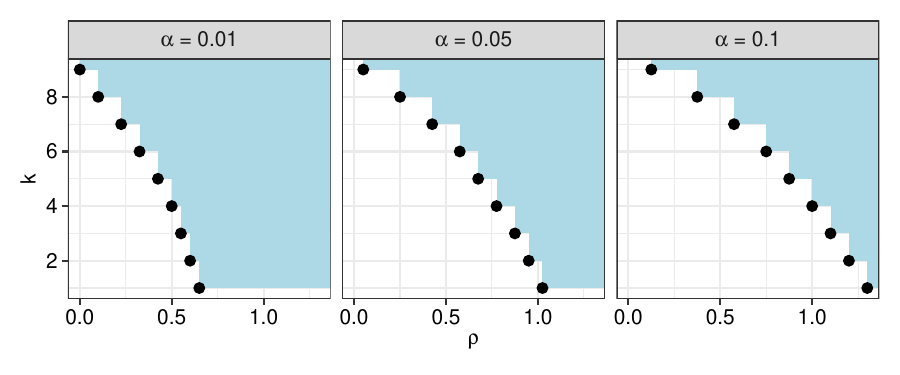}
    \caption{Simultaneous inference for empirical application 1 at different significance levels $\alpha$.}
    \label{fig:application1}
\end{figure}

\subsection{Empirical application 2: \citet{hiraiwaetal2024wp}}

\citet{hiraiwaetal2024wp} examine whether firms enforce noncompete agreements (NCAs). NCAs are restrictions that prevent workers from joining or starting competing firms. In 2020, Washington state passed a law that prohibits NCAs for workers earning below a certain threshold. The threshold was \$100,000 per year in 2020 and \$101,390 per year in 2021. They study the impact of being above or below the threshold on employment for different earnings bins. They consider the following two-way fixed effects model in their panel B of Table 1:
\[
		\log \text{Emp}_{b,t}
		= \beta \text{Treated}_{b,t}
		+ \alpha_t
		+ \gamma_b
		+ \epsilon_{b,t},
	\]
where $\text{Emp}_{b,t}$ is the employment count of bin $b$ at year $t$, $\text{Treated}_{b,t}$ equals 1 for the focal bin in the focal year, $\alpha_t$ and $\gamma_b$ are the year and bin fixed effects, and $\epsilon_{b,t}$ is the error term. They have 30 clusters (income bins). They use one-sided randomization inference and found that there is no significant effect.  

We apply our method and \citet{hagemann2024wp}. We are interested in testing the two-sided hypothesis as in \eqref{eq:application1-null}. Table \ref{tab:application2} summarizes our results. Each column contains the result for a specific focal year and definition of the treatment variable. For instance, column 2 refers to focal year 2020, and the treatment variable is defined to be equal to 1 when the income bin is just above the threshold, i.e., \$100-101.389k. The remaining columns are defined similarly. For each regression specification, we find the largest $\rho$ such that the null is rejected.  For entries with ``NA'' in the table, it refers to the situation where there does not exist a $\rho \geq 0$ such that the null is rejected.

\begin{table}[!ht]
    \centering
    \caption{Regression and inference results for empirical application 2 with $k=1$.}
    \label{tab:application2}
    \include{tables/hiralwaetal2024restat-v2}
\end{table}

Next, we conduct simultaneous inference for our $t$-test as in the first empirical application. We focus on the two variables for focal year 2021 because we can find $\rho$ such that the null is rejected when $k=1$. In particular, we report the range of $\rho$ for various $k$ as long as there exists $\rho$ such that the null can be rejected. Figure \ref{fig:application2} reports the results for $\alpha \in \{0.01, 0.05, 0.1\}$.

\begin{figure}[!ht]
    \centering
    \includegraphics[scale=1]{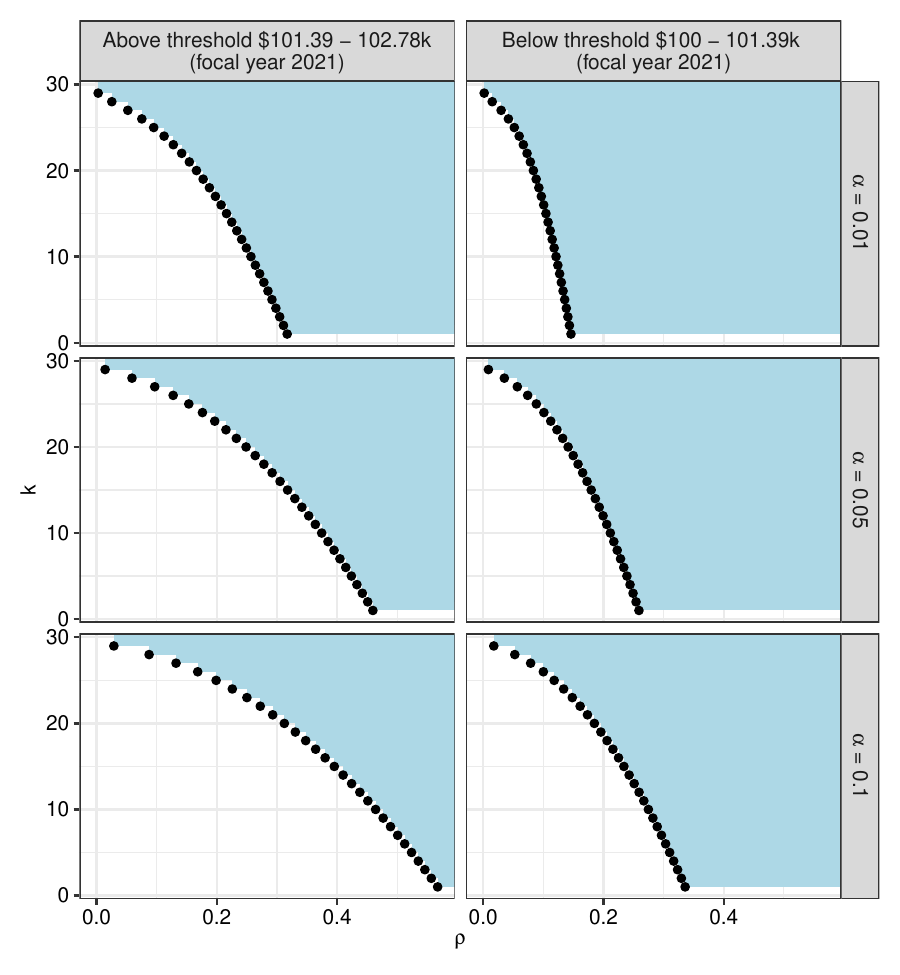}
    \caption{Simultaneous inference for  empirical application 2 for different outcomes and significance levels $\alpha$.}
    \label{fig:application2}
\end{figure}

\section{Conclusion} \label{sec:conclusion}
In this paper, we propose a $t$-test to conduct inference with a single treated cluster under a certain relative heterogeneity assumption. The $t$-statistic and the associated critical values are easy to calculate in many empirically relevant applications. We show that our test performs favorably when compared to other methods in the literature. We also show that our test is valid with weaker assumptions.

\appendix

\section{Appendix for the main text}

\subsection{Lemmas for Section \ref{sec:theory-max-rej-prob-detail}} \label{app:theory-max-rej-prob-detail}

We first state the lemma that the maximum rejection probability in \eqref{eq:max_rej_prob} must be obtained at some finite values of $\{\sigma_j\}_{j=1}^{m+1}$. 
This implies that we do not need to worry about the boundary case where some of $\{\sigma_j\}_{j=1}^{m+1}$ approach infinity.

\begin{lem}\label{lemma:max_at_finite}
For any $k \in \{1, \ldots, m\}$ and $c\ne m^{-1/2}$,\footnote{Similar to the footnote for Theorem \ref{thm:size_asymp} and as discussed in Remark \ref{re:c_greater_1_over_sq_m}, we will consider values of $c$ greater than $m^{-1/2}$ for most conventional significance levels.} 
the maximum rejection probability $p_m(c; k, \rho)$ in \eqref{eq:max_rej_prob} must be obtained at some $(\sigma_1, \ldots, \sigma_m, \sigma_{m+1}) \in \cSS$. 
\end{lem}

Next, we investigate possible maximizers for the rejection probability. From Section \ref{sec:sigma_treat_0}, when $\sigma_{m+1} = 0$, 
the rejection probability $\bP_0[|T_m|>c] $ of our $t$-test is at most $p_{m,0}(c)$ defined as in \eqref{eq:max_rej_prob_zero}, which can be achieved at some values of $\{\sigma_j\}^m_{j=1}$ that satisfy our relative heterogeneity assumption, regardless of the values of $\rho$ and $k$.
Thus, in the following, we consider only the case where $\sigma_{m+1}>0$. 

Recall the integral representation of the rejection probability in Lemma \ref{lem:rej_prob_integral_main}. Our goal is to optimize \eqref{eq:rej_prob_integral_main} over $\gamma_j \equiv \frac{\sigma_j}{\sigma_{m+1}}$, $j = 1, \ldots, m$, under the constraint imposed by the relative heterogeneity assumption.
Such an optimization can be numerically challenging even for moderate $m$. 
The following is a key lemma 
that can greatly reduce the set of possible maximizers for the rejection probability. 

\begin{lem}\label{lem:first_second_deriv_gamma}
	Consider the integral representation in \eqref{eq:rej_prob_integral_main} for the rejection probability. 
Fix the values of $c>0$, $\gamma_3, \ldots, \gamma_m \ge 0$ and $\theta_{m+1} < 0$, and view the rejection probability $\bP_0[|T_m|>c] \equiv \overline{p}_m(c; \gamma_1, \gamma_2, \gamma_3, \ldots, \gamma_m)$ as a function of only $\gamma_1$, where $\gamma_2$ is uniquely determined by $c, \gamma_1, \gamma_3, \ldots, \gamma_m$ and $\theta_{m+1}$. 
    If $c\ne m^{-1/2}$, $\gamma_1$ and $\gamma_2$ are both positive and they are not equal, i.e., $\gamma_1 \ne \gamma_2$,  and the first derivative of $\overline{p}_m(c; \gamma_1, \gamma_2, \gamma_3, \ldots, \gamma_m)$ over $\gamma_1^2$ is zero, then the second derivative of $\overline{p}_m(c; \gamma_1, \gamma_2, \gamma_3, \ldots, \gamma_m)$ over $\gamma_1^2$ must be positive.
\end{lem}

Lemma \ref{lem:first_second_deriv_gamma} has an important implication. 
Specifically, 
given values of $c, \gamma_3, \ldots, \gamma_m$ and $\theta_{m+1}$, if $\gamma_1$ is in the interior of the feasible region for optimization, and $\gamma_1 \neq \gamma_2$, then $(\gamma_1, \gamma_2, \ldots, \gamma_m)$ cannot be a maximizer, or even a local maximizer, for the rejection probability $\overline{p}_m(c; \gamma_1, \gamma_2, \gamma_3, \ldots, \gamma_m)$.
Moreover,
because $T_m$ is invariant to permutations of control clusters, 
the value of $\overline{p}_m(c; \gamma_1, \gamma_2, \ldots,$ $ \gamma_m)$ is invariant to permutations of $(\gamma_1, \ldots, \gamma_m)$.
Therefore, Lemma \ref{lem:first_second_deriv_gamma} essentially applies to any pair $(\gamma_i, \gamma_j)$ for $i \ne j$.

For any $k \in \{1, \ldots, m\}$ and $\rho > 0$, suppose now that $(\sigma_1, \ldots, \sigma_m, \sigma_{m+1})$ is a maximizer for the maximum rejection probability in \eqref{eq:max_rej_prob}, and $\sigma_{m+1}>0$. 
Then $(\gamma_1, \ldots, \gamma_m) = (\frac{\sigma_1}{\sigma_{m+1}}, \ldots, \frac{\sigma_m}{\sigma_{m+1}})$ must be a maximizer of $\overline{p}_m(c; \gamma_1, \ldots, \gamma_m)$ over all $(\gamma_1, \ldots, \gamma_m)\in \mathbb{R}_{\ge 0}^m$ satisfying that $\gamma_{(k)} \ge \rho^{-1}$, where $\gamma_{(k)}$ is the $k$th smallest value of $\{\gamma_j\}^m_{j=1}$.
From Lemma \ref{lem:first_second_deriv_gamma}, this cannot be true if there exists $i \ne j \in \{1,2,\ldots, m\}$ such that (i) $\gamma_i \neq \gamma_j$, and (ii) both $\gamma_i$ and $\gamma_j$ are not in the set $\{0, \rho^{-1}\}$, where the latter ensures that $\gamma_i$ and $\gamma_j$ are not in the boundary of the feasible region so that we can apply Lemma \ref{lem:first_second_deriv_gamma}. 
Therefore, to maximize the rejection probability, it suffices to consider only the cases where $\gamma_j$, for $1\le j \le m$, is either $0$, or $\rho^{-1}$, or some common value in $\mathbb{R}_{\ge 0}$. 
More specifically, 
under Assumption \ref{assu:relative_heter} for
any given $k$ and $\rho>0$, it suffices to consider 
the following possible values for $(\gamma_1, \ldots, \gamma_m)$ for $0\le m_0 \le k-1$ and $0\le m_1 \le m - m_0$:
\begin{align}\label{eq:possible_maximizer}
	\begin{pmatrix}
		\rho^{-1} \bs{1}_{m_1} \\
		\bs{0}_{m_0} \\
		\gamma \bs{1}_{m-m_1-m_0}
	\end{pmatrix}^\top,
	\ \ \ 
	\text{where } 
    \gamma \in 
	\begin{cases}
		\mathbb{R}_{\ge 0}, & \text{if } m_1 \ge m-k+1, \\
		[\rho^{-1}, \infty), & \text{if } m_1 < m-k+1.
	\end{cases}
\end{align}
In \eqref{eq:possible_maximizer}, we essentially enumerate the possible numbers of $\{\gamma_j\}^m_{j=1}$ that are zero, $\rho^{-1}$, and some common value $\gamma$, respectively.
Note that under Assumption \ref{assu:relative_heter} with given $\rho>0$ and $k$, the number of terms in $\{\gamma_j\}^m_{j=1}$ that are zero is at most $(k-1)$, and the possible range of the common value $\gamma$ depends on the number $m_1$ of the terms in $\{\gamma_j\}^m_{j=1}$ that are already $\rho^{-1}$. 
In particular, if $m_1 < m-k+1$, then the common value $\gamma$ is at least $\rho^{-1}$.

\subsection{Lemmas for the $k=1$ case in Section \ref{sec:theory-prob-k=1}}

In this subsection, we show that the optimization in \eqref{eq:max_rej_m1_m0} can be further simplified, by excluding values of $\gamma$ that cannot maximize $\overline{p}_m(c; \rho, \gamma; m_1, m_0)$ in \eqref{eq:p_m_given_num_gamma}. 
This relies crucially on the following lemma; we present a concise version below, and relegate the details to the supplementary material. 

\begin{lem}\label{lem:sign_deriv_gamma}
	Consider the integral representation in \eqref{eq:rej_prob_integral_main}. 
	Fix the values of $c>0$, $m\ge 3$, $\gamma_3, \ldots, \gamma_m \ge 0$ and $\theta_{m+1}<0$, and view the rejection probability $\bP_0[|T_m|>c] = \overline{p}_m(c; \gamma_1, \gamma_2, \gamma_3,$ $ \ldots, \gamma_m)$ as a function of only $\gamma_1$, where $\gamma_2$ is uniquely determined by $c, \gamma_1, \gamma_3, \ldots, \gamma_m$ and $\theta_{m+1}$.  If $\gamma_1 = \max_{1\le i \le m} \gamma_i > \gamma_2>0$
	$c\ge \sqrt{\frac{2(m-1)}{m(m-2)}}$, 
	and $H_m(c; \gamma_1, \ldots, \gamma_m) < 0$\footnote{We give the detailed expression of $H_m(\cdot)$ in the supplementary material}, then the first order derivative of $\overline{p}_m(c; \gamma_1, \gamma_2, \gamma_3, \ldots, \gamma_m)$ over $\gamma_1^2$ must be negative. 
\end{lem}

Lemma \ref{lem:sign_deriv_gamma} has an important implication.
Specifically, 
given values of $\gamma_3, \ldots, \gamma_m$ and $\theta_{m+1}$, if, for some small $\epsilon>0$, 
$[\gamma_1 - \epsilon, \gamma_1]$ is in the feasible region for optimization, and the conditions in Lemma \ref{lem:sign_deriv_gamma} holds, then $(\gamma_1, \gamma_2, \ldots, \gamma_m)$ cannot be a maximizer, or even a local maximizer, for the rejection probability $\overline{p}_m(c; \gamma_1, \gamma_2, \gamma_3, \ldots, \gamma_m)$. 
This is because we can strictly increase this rejection probability by slightly decreasing $\gamma_1$. 
Analogous to the discussion after Lemma \ref{lem:first_second_deriv_gamma}, because the rejection probability is invariant to permutations of $(\gamma_1, \ldots, \gamma_m)$, 
we can always set $\gamma_1 = \max_{1\le i \le m} \gamma_i$ without loss of generality, and apply Lemma \ref{lem:sign_deriv_gamma} to any $\gamma_j$ with $2\le j\le m$. 
As long as the conditions in Lemma \ref{lem:sign_deriv_gamma} hold for some $2\le j\le m$, we can exclude the corresponding $(\gamma_1, \ldots, \gamma_m)$ from the set of possible maximizers for the rejection probability.

We apply Lemma \ref{lem:sign_deriv_gamma} to the optimization in \eqref{eq:max_rej_m1_m0} to exclude some values of $\gamma > \rho^{-1}$ that cannot maximize $\overline{p}_m(c; \rho, \gamma; m_1, m_0)$\footnote{Specifically, when using Lemma \ref{lem:sign_deriv_gamma}, we set $\gamma_1$ to be $\gamma$ and $\gamma_2$ to be $\rho^{-1}$.
	%$\gamma$ and $\rho^{-1}$ to be $\gamma_1$ and $\gamma_2$ there, respectively.
	}.
In the ideal case, we can exclude all values of $\gamma > \rho^{-1}$, and the optimization in \eqref{eq:max_rej_m1_m0} is either solved or needs only to be maximized over $\gamma \in [0, \rho^{-1}]$. 
In Theorem \ref{thm:max_rej_prob_k_1_simple_closed_form}, we focus on the relative heterogeneity assumption with $k=1$, i.e., the treated standard deviation is smaller than or equal to $\rho$ times each of the control standard deviations.
We relegate the discussion under Assumption \ref{assu:relative_heter} with $k\ge 2$ to the supplementary material.

\subsection{\textbf{\textit{t}}-test with a large number of clusters} \label{sec:4.6}

Despite that our inference procedure focuses on the case with a finite number of control clusters, 
we briefly explore the asymptotic regime where the number of control clusters tends to infinity. This offers useful insights even when the number of control clusters is finite.
Below we first state a rather weak regularity condition. We then present the main theorem and discuss its implications. 

\begin{cond}\label{cond:large_m}
    As $m\longrightarrow \infty$, both $\sigma_{m+1}^2$ and $m^{-1}\sum_{j=1}^m \sigma_j^2$ are bounded away from zero and infinity, and $m^{-2}\sum_{j=1}^m \sigma_j^4 \longrightarrow 0$.     
\end{cond}

\begin{thm}\label{thm:limit_T_m_large}
    Under Assumption \ref{assu:normal} and 
    Condition \ref{cond:large_m}, 
    as $m\longrightarrow \infty$,  we have 
    \begin{align*}
    \sup_{c\in \mathbb{R}}\left| \bP[T_m \le c] - 
    \bP\Bigg[ \frac{\sigma_{m+1} \varepsilon + \delta}{\sqrt{m^{-1}\sum_{j=1}^m \sigma_j^2}} \le c \Bigg] \right|
    \longrightarrow 0.
    \end{align*}
    where $\varepsilon \sim \cN(0,1)$. 
\end{thm}

From Theorem \ref{thm:limit_T_m_large}, when $m$ is large, the distribution of the test statistic $T_m$ under both the null and the alternative hypotheses can be approximated by an Gaussian distribution. This has important implications, as detailed below. 

First, suppose the null hypothesis $\overline{H}_0$ in \eqref{eq:null_normal} holds (i.e., $\delta = 0$). Then, when the number of clusters is large, the rejection probability $\bP_{0}[|T_m| > c]$ depends only the ratio between the treated-cluster variance $\sigma_{m+1}^2$ and the average of the control-cluster variances $m^{-1}\sum_{j=1}^m \sigma_j^2$. It is increasing in this ratio. 
In particular, under Assumption \ref{assu:relative_heter} for some $\rho > 0$, we have the following for any $c>0$: 
\begin{align*}
        \bP_0[|T_m| > c] \approx 
        \bP\Bigg[ \frac{\sigma_{m+1} |\varepsilon|}{\sqrt{m^{-1}\sum_{j=1}^m \sigma_j^2}} > c \Bigg]
        \le 
        \bP\Bigg[ 
        \sqrt{\frac{m \rho^2}{m-k+1}} \cdot |\varepsilon| > c
        \Bigg],
\end{align*}
where the last inequality becomes equality when $\sigma_{m+1} = 1$, $(k-1)$ of $\{\sigma_j\}_{j=1}^m$ equal $0$, and the remaining $(m - k + 1)$ equal $\rho^{-1}$.
It also implies that the critical value $\cv$ of a level-$\alpha$ test is approximately 
$\sqrt{\frac{m}{m-k+1}} \cdot \rho z_{\frac{\alpha}{2}}$, where $z_{\frac{\alpha}{2}}$ is the $\frac{\alpha}{2}$ upper quantile of the standard normal distribution.

Second, under the alternative hypothesis with, say, $\delta>0$, the above discussion also implies that the power of the level-$\alpha$ $t$-test is approximately 
\begin{align}\label{eq:power_large_m}
     \bP[|T_m| > \cv] & \approx  
    \bP\Bigg[ \Bigg| \frac{\sigma_{m+1} \varepsilon + \delta}{\sqrt{m^{-1}\sum_{j=1}^m \sigma_j^2}} \Bigg| > \sqrt{\frac{m \rho^2}{m-k+1}} z_{\frac{\alpha}{2}} \Bigg]
    \notag \\
    & =
    \bP\Bigg[ \Big| \varepsilon + \frac{\delta}{\sigma_{m+1}} \Big| >
    R_{m, k, \rho}
    z_{\frac{\alpha}{2}}
    \Bigg]
    \nonumber
    \\
    & \ge \bP\Bigg[ \varepsilon + \frac{\delta}{\sigma_{m+1}} >
    R_{m, k, \rho}
    z_{\frac{\alpha}{2}}
    \Bigg],
\end{align}
where $R^2_{m, k, \rho}$ is the ratio between the maximum possible value of $\frac{\sigma_{m+1}^2}{m^{-1}\sum_{j=1}^m \sigma_j^2}$ and the corresponding true value, and $R_{m, k, \rho}$ is the square root of $R^2_{m, k, \rho}$. 
From \eqref{eq:power_large_m}, the lower bound of the approximate power in \eqref{eq:power_large_m}, which is also approximately the power of an one-sided level-$\frac{\alpha}{2}$ test, is increasing in $\frac{\delta}{\sigma_{m+1}}$ and decreasing in $R_{m, k, \rho}$. 
This is intuitive, since $\frac{\delta}{\sigma_{m+1}}$ represents the signal-to-noise ratio, whereas $R_{m, k, \rho}$ represents how conservative we are when making the relative heterogeneity assumption.

We now give a remark that is helpful for our numerical search of the critical value when $m$ is finite.

\begin{re}
Recall that $\gamma_j \equiv \frac{\sigma_j}{\sigma_{m+1}}$ for $1\le j \le m$. 
From the discussion before, the maximum rejection probability is achieved when $k-1$ of 
$\{\gamma_j\}_{j=1}^m$ equal $0$ and the remaining $(m-k+1)$ equal $1$. 
Thus, 
given any significance level, 
we suggest using this configuration of $\{\gamma_j\}_{j=1}^m$ as an initial guess for the desired critical value for any finite $m$, and then use the optimization in Theorem \ref{thm:max_rej_prob} to verify whether such an initial guess of critical value leads to the $t$-test with the desired significance level; see Remark \ref{rmk:guess_candidate}. 
Moreover, as demonstrated in Section \ref{sec:theory-prob-k=1}, under the relative heterogeneity assumption with $k=1$, 
the maximizer for the rejection probability from the asymptotic analysis with large $m$ is also valid when the number of control clusters is finite and the significance level is sufficiently small.
Our numerical experiment shows that this holds more generally; at usual significance levels, the critical value from the asymptotic analysis with large $m$ is often valid even when $m$ is finite, for general relative heterogeneity assumption with $k\ge 1$.
\end{re}

\addcontentsline{toc}{section}{References}
{\small{
	\singlespacing{
	\bibliographystyle{ecta}
	\bibliography{refs}
	}
}}

\newpage

% Command to hide proofs if needed
\newif\ifshowhidden

% Change the command here to hide or show more content  
\showhiddentrue 
%\showhiddenfalse

\newcommand{\hideif}[1]{\ifshowhidden #1\fi}

\begin{center}

{\bf \Large 
Supplementary material for ``Cluster-robust inference with a single treated cluster using the $\bm{t}$-test''}
\end{center}

\renewcommand {\theproposition} {S\arabic{proposition}}
\renewcommand {\theexample} {S\arabic{example}}
\renewcommand {\thefigure} {S\arabic{figure}}
\renewcommand {\thetable} {S\arabic{table}}
\renewcommand {\theequation} {S\arabic{equation}}
\renewcommand {\thelemma} {S\arabic{lemma}}
\renewcommand {\thesection} {S\arabic{section}}
\renewcommand {\thetheorem} {S\arabic{theorem}}
\renewcommand {\thecondition} {S\arabic{condition}}

\newcommand{\hargs}{(\gamma; c, \rho)}
\newcommand{\tha}{\widetilde H_{m,1}\hargs}
\newcommand{\thb}{\widetilde H_{m,2}\hargs}
\newcommand{\thc}{\widetilde H_{m,3}(c, \rho)}
\newcommand{\thh}{\widetilde H_m\hargs}
\newcommand{\oha}{\overline H_{m,1}(c, \rho)}
\newcommand{\ohb}{\overline H_{m,2}(c, \rho)}

% Reset section counter to 1
\setcounter{section}{0}
\setcounter{table}{0}
\setcounter{figure}{0}

\section{Technical lemmas for the $t$-test}

Throughout this section, we will study $t$-test based on independent normal observations under Assumption \ref{assu:normal}. 
In Section \ref{sec:char_poly}, we will show that the $t$-test reject the null if and only if a certain quadratic form of the normal observations is bounded by $0$, and study the property of the corresponding coefficient matrix, particularly about its characteristic polynomial. 
In Section \ref{sec:integral_rej_prob}, we give an integral representation of the rejection probability. 
In Section \ref{sec:deri_rej_prob}, we study the first and second order partial derivatives of the rejection probability. 
In Section \ref{sec:bound_rej_prob}, we study bounds on the first order partial derivative of the rejection probability. 

\subsection{Quadratic form and characteristic polynomial of the coefficient matrix}\label{sec:char_poly}

\begin{lem}\label{lem:t_test_matrix}
Recall the definition of $T_m$ in \eqref{eq:pop-test-1}. 
Let $\tilde{\bm{\nob}} \equiv (\nob_{m+1}, \nob_1, \ldots, \nob_m)^\top$. Then for any $c > 0$,
\begin{align}
	\label{eq:lem-1}
    |T_m| > c 
    \quad
    \Longleftrightarrow
    \quad
    \tilde{\bm{\nob}}^\top \bm{V} \tilde{\bm{\nob}} < 0
\end{align}
where 
\begin{align*}
    \kappa \equiv \frac{m c^2}{m-1}
    \quad 
    \text{and}
    \quad 
	\bm{V} 
	\equiv 
    \begin{pmatrix}
        -m & \bs{1}_m^\top \\
        \bs{1}_m & 
        \kappa \bs{I}_m -
        \frac{\kappa + 1}{m}
        \bs{1}_m\bs{1}_m^\top
    \end{pmatrix}.
\end{align*}

\end{lem}

\hideif{
\begin{proof}[Proof of Lemma \ref{lem:t_test_matrix}] 
By definition, 
for any $c> 0$, $|T_m|>c$ if and only if $|\nob_{m+1} - \overline{\nob}_m| > cS_{m}$, which is further equivalent to 
$mc^2S^2_{m} - m(\nob_{m+1} - \overline{\nob}_m)^2 < 0$. 
Let $\bs{1}_m$ be an $m$-dimensional vector with all elements being $1$, 
$\bs{0}_m$ be an $m$-dimensional vector with all elements being $0$,
and $\bs{I}_m$ be an $m\times m$ identity matrix. By some algebra, we can show that 
\begin{align*}
    (\nob_{m+1} - \overline{\nob}_m)^2
    & = 
    \tilde{\bm{\nob}}^\top 
    \begin{pmatrix}
        1 \\
        -m^{-1}\bs{1}_m 
    \end{pmatrix}
    \begin{pmatrix}
        1 & -m^{-1}\bs{1}_m^\top 
    \end{pmatrix} 
    \tilde{\bm{\nob}}
    = 
    \tilde{\bm{\nob}}^\top
    \begin{pmatrix}
        1 & -m^{-1}\bs{1}_m^\top \\
        -m^{-1}\bs{1}_m & 
        m^{-2} \bs{1}_m \bs{1}_m^\top
    \end{pmatrix}
    \tilde{\bm{\nob}},
\end{align*}
and 
\begin{align*}
    (m-1) S^2_{m}
    & = 
    \sum^m_{j=1} (\nob_{j} - \overline{\nob}_m)^2   = 
    \bs{\nob}_{1:m}^\top 
    ( \bs{I}_m -m^{-1} \bs{1}_m \bs{1}_m^\top ) \bs{\nob}_{1:m}  
    = 
    \bs{\nob}_{1:m}^\top 
    ( \bs{I}_m -m^{-1} \bs{1}_m \bs{1}_m^\top )
    \bs{\nob}_{1:m}\\
    & = 
    \tilde{\bm{\nob}}^\top 
    \begin{pmatrix}
        0 & \bs{0}_m^\top \\
        \bs{0}_m & \bs{I}_m -m^{-1} \bs{1}_m \bs{1}_m^\top
    \end{pmatrix}
    \tilde{\bm{\nob}},
\end{align*}
where $\bs{\nob}_{1:m} \equiv (\nob_1, \ldots, \nob_m)^\top$.
These imply that 
\begin{align*}
    & \hspace{-15pt} mc^2S^2_{m} - m (\nob_{m+1} - \overline{\nob}_m)^2 
    \notag \\
    & = 
    \frac{m c^2}{m-1}
    \tilde{\bm{\nob}}^\top 
    \begin{pmatrix}
        0 & \bs{0}_m^\top \\
        \bs{0}_m & \bs{I}_m -m^{-1} \bs{1}_m \bs{1}_m^\top
    \end{pmatrix}
    \tilde{\bm{\nob}} 
    - 
    m \tilde{\bm{\nob}}^\top
    \begin{pmatrix}
        1 & -m^{-1}\bs{1}_m^\top \\
        -m^{-1}\bs{1}_m & 
        m^{-2} \bs{1}_m \bs{1}_m^\top
    \end{pmatrix}
    \tilde{\bm{\nob}}
    \\
    & = 
    \tilde{\bm{\nob}}^\top
    \left[
    \begin{pmatrix}
        0 & \bs{0}_m^\top \\
        \bs{0}_m &  \frac{m c^2}{m-1} \bs{I}_m - \frac{c^2}{m-1}\bs{1}_m \bs{1}_m^\top
    \end{pmatrix}
    - 
    \begin{pmatrix}
        m & -\bs{1}_m^\top \\
        -\bs{1}_m & 
        m^{-1} \bs{1}_m \bs{1}_m^\top
    \end{pmatrix}
    \right]
    \tilde{\bm{\nob}}
    \\
    & = 
    \tilde{\bm{\nob}}^\top
    \begin{pmatrix}
        -m & \bs{1}_m^\top \\
        \bs{1}_m & 
        \frac{m c^2}{m-1} \bs{I}_m -
        \frac{1}{m}
        ( \frac{mc^2}{m-1} + 1 ) \bs{1}_m\bs{1}_m^\top
    \end{pmatrix}
    \tilde{\bm{\nob}}\\
    & = \tilde{\bm{\nob}}^\top \bm{V} \tilde{\bm{\nob}},
\end{align*}
where the last equality holds by definition. 
From the above, we derive Lemma \ref{lem:t_test_matrix}. 
\end{proof}
}

\begin{lem} \label{lem:character_sigma}
Let $\bm{D} \equiv \mathrm{diag}(\sigma_{m+1}, \sigma_1, \ldots, \sigma_m)$, and define $\bm{V}$ and $\kappa$ as in Lemma \ref{lem:t_test_matrix}.
For any nonnegative $\sigma_1^2, \sigma^2_2, \ldots, \sigma^2_{m+1}$, 
\begin{align*}
    f(\lambda)
    & \equiv 
    -(m\sigma_{m+1}^2+ \lambda)
        \prod_{i=1}^m ( \kappa\sigma_i^2 - \lambda)  
        +
        \left( \kappa \sigma_{m+1}^2
        + \frac{\kappa+1}{m} \lambda \right)
        \cdot 
        \sum_{i=1}^m 
        \left[
        \sigma_i^2 \prod_{j\ne i, 1\le j \le m} ( \kappa\sigma_j^2 - \lambda) 
        \right]
\end{align*}
is a characteristic polynomial of $\bm{D}\bs{V}\bs{D}$. 
\end{lem}

\hideif{
\begin{proof}[Proof of Lemma \ref{lem:character_sigma}]
We first consider the case where $\sigma_1, \sigma_2, \ldots, \sigma_{m+1}$ are all positive. 
Let $\text{diag}(\sigma^{-2}_{1:m})$ be an $m\times m$ diagonal matrix with diagonal elements $(\sigma^{-2}_1, \ldots, \sigma^{-2}_m)$. 
By definition, we have 
\begin{align*}
    \bm{V}  - \lambda \bm{D}^{-2}
    & = 
    \begin{pmatrix}
        -m & \bs{1}_m^\top \\
        \bs{1}_m & 
        \kappa \bs{I}_m -
        \frac{\kappa + 1}{m}
        \bs{1}_m\bs{1}_m^\top
    \end{pmatrix}
    - \lambda  \bm{D}^{-2}
    \\
    & 
    = \begin{pmatrix}
        -m - \lambda \sigma_{m+1}^{-2} & \bs{1}_m^\top \\
        \bs{1}_m & 
        \kappa \bs{I}_m - \lambda \text{diag}(\sigma^{-2}_{1:m})
        - 
        \frac{\kappa + 1}{m}
        \bs{1}_m\bs{1}_m^\top
    \end{pmatrix}
    \\
    & = 
    \begin{pmatrix}
			b  & 1 & 1 & \cdots & 1 \\
			1 & d + c_1 & d & \cdots & d \\
			1 & d & d + c_2 & \cdots & d \\ 
			\vdots & \vdots & \vdots & \ddots & \vdots \\
			1 & d & d & \cdots & d + c_m \\ 
	\end{pmatrix},
\end{align*}
where $b = -m- \lambda \sigma_{m+1}^{-2}$, $d = - \frac{\kappa + 1}{m}$, and $c_i = \kappa - \lambda \sigma_i^{-2}$ for $1\le i \le m$. 
Applying \citet[][Proposition 1]{bakirov1998jms}, we have 
\begin{align*}
		& \quad \ \det ( \bm{D}\bm{V} \bm{D} - \lambda \bm{I}_{m+1}) 
        \\
        & = \det (\bm{D})^2 \det( \bm{V} - \lambda \bm{D}^{-2}) \\
        & = \prod^{m+1}_{j=1} \sigma_j^2 
		\begin{vmatrix}
			b  & 1 & 1 & \cdots & 1 \\
			1 & d + c_1 & d & \cdots & d \\
			1 & d & d + c_2 & \cdots & d \\ 
			\vdots & \vdots & \vdots & \ddots & \vdots \\
			1 & d & d & \cdots & d + c_m \\ 
		\end{vmatrix}
        \\
        & = \prod^{m+1}_{j=1} \sigma_j^2  \cdot
        b \prod_{i=1}^m c_i \left( 1 + d \sum_{i=1}^m \frac{1}{c_i} - \frac{1}{b} \sum_{i=1}^m \frac{1}{c_i} \right)
        \\
        & = 
        \sigma_{m+1}^2 b  \cdot
        \prod_{i=1}^m (\sigma_i^2 c_i) 
        \cdot
        \left[ 1 + \left( d - \frac{\sigma_{m+1}^2}{\sigma_{m+1}^2 b} \right) \sum_{i=1}^m \frac{\sigma_i^2}{c_i \sigma_i^2} \right]\\
        & = 
        -(m\sigma_{m+1}^2+ \lambda)  \cdot
        \prod_{i=1}^m ( \kappa\sigma_i^2 - \lambda) 
        \cdot
        \left[ 1 + \left( -\frac{\kappa+1}{m} + \frac{\sigma_{m+1}^2}{m\sigma_{m+1}^2+ \lambda} \right) \sum_{i=1}^m \frac{\sigma_i^2}{\kappa\sigma_i^2 - \lambda} \right]
        \\
        & = 
        -(m\sigma_{m+1}^2+ \lambda)
        \prod_{i=1}^m ( \kappa\sigma_i^2 - \lambda)  
        -
        \left[ -\frac{\kappa+1}{m}(m\sigma_{m+1}^2+ \lambda) + \sigma_{m+1}^2\right]
        \cdot 
        \sum_{i=1}^m 
        \left[
        \sigma_i^2 \prod_{j\ne i, 1\le j \le m} ( \kappa\sigma_j^2 - \lambda) 
        \right]
        \\
        & = 
        -(m\sigma_{m+1}^2+ \lambda)
        \prod_{i=1}^m ( \kappa\sigma_i^2 - \lambda)  
        +
        \left( \kappa \sigma_{m+1}^2
        + \frac{\kappa+1}{m} \lambda \right)
        \cdot 
        \sum_{i=1}^m 
        \left[
        \sigma_i^2 \prod_{j\ne i, 1\le j \le m} ( \kappa\sigma_j^2 - \lambda) 
        \right]
        \\
        & = f(\lambda),
	\end{align*}
    where the last equality holds by definition. 

    We then consider the case where some of the $\{\sigma_j\}^{m+1}_{j=1}$ can be zero. 
    By definition, for any $\lambda \in \mathbb{R}$, 
    both $\det ( \bm{D}\bm{V} \bm{D} - \lambda \bm{I}_{m+1})$ and $f(\lambda)$, viewed as functions of $\{\sigma_j\}^{m+1}_{j=1}$, 
    are polynomial and thus continuous functions of $\{\sigma_j\}^{m+1}_{j=1}$. 
    By continuity and the fact that $\det ( \bm{D}\bm{V} \bm{D} - \lambda \bm{I}_{m+1})=f(\lambda)$ for any positive $\{\sigma_j\}^{m+1}_{j=1}$, 
    we can know that $\det ( \bm{D}\bm{V} \bm{D} - \lambda \bm{I}_{m+1})$ must equal $f(\lambda)$ for any nonnegative $\{\sigma_j\}^{m+1}_{j=1}$. 

    From the above, Lemma \ref{lem:character_sigma} holds. 
\end{proof}
}

\begin{lem}\label{lem:char_gamma} 
Consider any positive 
$\sigma_{m+1}$ and any nonnegative $\sigma_1, \ldots, \sigma_m$. 
Let $\kappa \equiv \frac{m c^2}{m-1}$, 
$\gamma_i \equiv \frac{\sigma_i}{\sigma_{m+1}}$ for $1\le i \le m$, 
and 
\begin{align}\label{eq:char_gamma}
    g(\theta) = -(m+ \theta)
        \prod_{i=1}^m ( \kappa\gamma_i^2 - \theta)  
        +
        \left( \kappa 
        + \frac{\kappa+1}{m} \theta\right)
        \cdot 
        \sum_{i=1}^m 
        \left[
        \gamma_i^2 \prod_{j\ne i} ( \kappa\gamma_j^2 - \theta ) 
        \right]. 
\end{align}
Then, for any $\lambda \in \mathbb{R}$, $f(\lambda)$ in Lemma \ref{lem:character_sigma} can be written equivalently as
\[
    f(\lambda) = (\sigma_{m+1}^{2} )^{m+1} 
    g\left(\frac{\lambda}{\sigma^2_{m+1}}\right).
\]
Consequently, if $\lambda_1, \ldots, \lambda_{m+1}$ are the $m+1$ roots of $f(\cdot)$, 
then $\theta_i = \frac{\lambda_i}{\sigma^2_{m+1}}$, $i=1,2,\ldots, m+1$, are the $(m+1)$ roots of $g(\cdot)$. 
\end{lem}

\hideif{
\begin{proof}[Proof of Lemma \ref{lem:char_gamma}]
For any positive $\sigma_{m+1}^2$, 
$\lambda \in \mathbb{R}$ and $\theta \equiv \frac{\lambda}{\sigma_{m+1}^2}$, 
by definition and some algebra, we can verify that  
\begin{align*}
    (\sigma_{m+1}^{2} )^{-(m+1)} \cdot f(\lambda)
    & = 
    -(m+ \theta)
        \prod_{i=1}^m ( \kappa\gamma_i^2 - \theta)  
        +
        \left( \kappa 
        + \frac{\kappa+1}{m} \theta\right)
        \cdot 
        \sum_{i=1}^m 
        \left[
        \gamma_i^2 \prod_{j\ne i} ( \kappa\gamma_j^2 - \theta ) 
        \right] \\
    &= g(\theta). 
\end{align*}
This immediately implies Lemma \ref{lem:char_gamma}. 
\end{proof}
}

\begin{lem}\label{lem:roots}
Consider any $c>0$, any positive 
$\sigma_{m+1}$, and any nonnegative $\sigma_1, \ldots, \sigma_m$. 
Define $\kappa, \gamma_1, \ldots, \gamma_m$ and $g(\cdot)$ as in Lemma \ref{lem:char_gamma}. 
Define further $\tau \equiv \frac{\kappa +1}{\kappa m}$. 
Let $\gamma_{(1)} \le \gamma_{(2)} \le \ldots \le \gamma_{(m)}$ be the sorted values of $\gamma_1, \gamma_2,  \ldots, \gamma_m$.
The $(m+1)$ roots of $g(\theta)$ in \eqref{eq:char_gamma} can then be characterized as follows: 
	\begin{enumerate}[label=(\alph*)]
		\item for each $1\le i \le m-1$, there is a nonnegative root $\theta_i \in (\kappa \gamma_{(i)}^2, \kappa \gamma_{(i+1)}^2)$ if $\gamma_{(i)} <  \gamma_{(i+1)}$ and $\theta_i = \kappa \gamma_{(i)}^2$ otherwise;
		\item there is a zero root $\theta_m = 0$;
		\item there is a negative root 
        $\theta_{m+1} < - \tau^{-1}$. 
	\end{enumerate}
\end{lem}

\hideif{
\begin{proof}[Proof of Lemma \ref{lem:roots}]
First, we give an equivalent expression for $g(\cdot)$. 
Let $\{\tilde\gamma_i\}^L_{i=1}$ be the unique ordered values of $\{\gamma_i\}^m_{i=1}$ such that $\tilde\gamma_1 < \cdots < \tilde \gamma_L$. In addition, let $m_i$ be the multiplicity of $\tilde\gamma_i$ for each $i = 1, \ldots, L$. Note that $\sum^L_{i=1} m_i = m$ by construction. 
From Lemma \ref{lem:char_gamma}, 
$g(\cdot)$ in \eqref{eq:char_gamma} has the following equivalent forms:
\begin{align}\label{eq:char_gamma_equiv}
    g(\theta) & = -(m+ \theta)
        \prod_{i=1}^m ( \kappa\gamma_i^2 - \theta)  
        +
        \left( \kappa 
        + \frac{\kappa+1}{m} \theta\right)
        \cdot 
        \sum_{i=1}^m 
        \left[
        \gamma_i^2 \prod_{j\ne i} ( \kappa\gamma_j^2 - \theta ) 
        \right]
        \nonumber
    \\
    & = 
    \prod_{i=1}^L  (\kappa\tilde\gamma_i^2 - \theta)^{m_i-1}
    \left\{ 
    -(m+ \theta)
        \prod_{i=1}^L ( \kappa\tilde\gamma_i^2 - \theta)  
        +
        \left( \kappa 
        + \frac{\kappa+1}{m} \theta\right)
        \cdot 
        \sum_{i=1}^L 
        \left[
        m_i
        \tilde\gamma_i^2 
         \prod_{j\ne i} ( \kappa\tilde\gamma_j^2 - \theta ) 
        \right]
    \right\}
    \nonumber
    \\
    & = 
    \prod_{i=1}^L  (\kappa\tilde\gamma_i^2 - \theta)^{m_i-1} \tilde{g}(\theta), 
\end{align}
where
\begin{align*}
	\tilde g(\theta)
	\equiv
	-(m+ \theta)
        \prod_{i=1}^L ( \kappa\tilde\gamma_i^2 - \theta)  
        +
        \left( \kappa 
        + \frac{\kappa+1}{m} \theta\right)
        \cdot 
        \sum_{i=1}^L 
        \left[
        m_i
        \tilde\gamma_i^2 
         \prod_{j\ne i} ( \kappa\tilde\gamma_j^2 - \theta ) 
        \right]. 
\end{align*}
From \eqref{eq:char_gamma_equiv}, it can be seen that there are $\sum^L_{i=1} (m_i - 1) = \sum^L_{i=1} m_l - L = m - L$ roots from the factor $\prod_{i=1}^L  (\kappa\tilde\gamma_i^2 - \theta)^{m_i-1}$.

Second, we prove (a) and (b) in the case where $\tilde\gamma_1 > 0$.  Consider the value of $g(\cdot)$ evaluated at  
$\kappa \tilde\gamma_l^2$ for $l=1, \ldots, L$.  
By definition, 
for $1\le l \le L$, 
we have
\begin{align}\label{eq:g_tilde_phi_l}
    \tilde g(\kappa \tilde\gamma_l^2)
    & \equiv
	-(m+ \kappa \tilde\gamma_l^2)
        \underbrace{\prod_{i=1}^L ( \kappa\tilde\gamma_i^2 - \kappa \tilde\gamma_l^2)}_{=0}
        +
        \left( \kappa 
        + \frac{\kappa+1}{m} \kappa \tilde\gamma_l^2 \right)
        \cdot 
        \sum_{i=1}^L 
        \left[
        m_i
        \tilde\gamma_i^2 
         \prod_{j\ne i} ( \kappa\tilde\gamma_j^2 - \kappa \tilde\gamma_l^2 ) 
        \right] 
        \nonumber
        \\
    & = 
    \left( \kappa 
        + \frac{\kappa+1}{m} \kappa \tilde\gamma_l^2 \right)
        \cdot 
        m_l 
        \tilde\gamma_l^2 
        \cdot
         \prod_{j\ne l} ( \kappa\tilde\gamma_j^2 - \kappa \tilde\gamma_l^2 ) 
         \nonumber
    \\
    & = 
    \underbrace{\kappa^{L-1}
    \left( \kappa 
        + \frac{\kappa+1}{m} \kappa \tilde\gamma_l^2 \right)
        \cdot 
        m_l 
        \tilde\gamma_l^2}_{>0}
    \cdot
         \prod_{j\ne l} ( \tilde\gamma_j^2 - \tilde{\gamma}_l^2 ) . 
\end{align}
By the construction of $\{\tilde\gamma_j\}_{j=1}^L$, we can know that the sign of  
$\tilde g(\kappa \tilde\gamma_l^2)$ is $(-1)^{l-1}$ for $1\le l \le L$. 
Thus, by the continuity of $\tilde{g}(\cdot)$, it must have a positive root in $(\kappa\tilde{\gamma}_l^2, \kappa\tilde{\gamma}_{l+1}^2)$, for $l=1,2,\ldots, L-1$. 
In addition, we can verify that 
\begin{align*}
    \tilde g(0)
	& \equiv
	-m
        \prod_{i=1}^L ( \kappa\tilde\gamma_i^2 )  
        +
        \kappa 
        \cdot 
        \sum_{i=1}^L 
        \left[
        m_i
        \tilde\gamma_i^2 
         \prod_{j\ne i, 1\le j \le L} ( \kappa\tilde\gamma_j^2 )
        \right] \\
    & = 
    -m \kappa^L 
    \prod_{i=1}^L \tilde\gamma_i^2  
    +
    \kappa^L  
    \cdot 
    \sum_{i=1}^L 
    \left(
    m_i
     \prod_{j=1}^L  \tilde\gamma_j^2
    \right)\\
    & = 
    \kappa^L 
    \prod_{i=1}^L \tilde\gamma_i^2 
    \cdot 
    \left( -m + \sum_{i=1}^L m_i \right) = 0. 
\end{align*}
Combined with the first part, we can know that (a) and (b) in Lemma \ref{lem:roots} hold.

Third, we prove (a) and (b) in the case where $\tilde\gamma_1 = 0$. 
In this case, $\tilde{g}(\cdot)$ simplifies to 
\begin{align}\label{eq:g_gamma_0}
	\tilde g(\theta)
	& \equiv
	-(m+ \theta)(-\theta)
        \prod_{i=2}^L ( \kappa\tilde\gamma_i^2 - \theta)  
        +
        \left( \kappa 
        + \frac{\kappa+1}{m} \theta\right)
        (-\theta)
        \cdot 
        \sum_{i=2}^L 
        \left[
        m_i
        \tilde\gamma_i^2 
         \prod_{j\ne 1, i} ( \kappa\tilde\gamma_j^2 - \theta ) 
        \right]
        \nonumber
    \\
    & = 
    -\theta \cdot 
    \left\{
        -(m+ \theta)
        \prod_{i=2}^L ( \kappa\tilde\gamma_i^2 - \theta)  
        +
        \left( \kappa 
        + \frac{\kappa+1}{m} \theta\right)
        \cdot 
        \sum_{i=2}^L 
        \left[
        m_i
        \tilde\gamma_i^2 
         \prod_{j\ne 1, i} ( \kappa\tilde\gamma_j^2 - \theta ) 
        \right]
    \right\}
    \nonumber
    \\
    & = - \theta \cdot \check{g}(\theta), 
\end{align}
where 
\begin{align*}
    \check{g}(\theta)
    & = 
    -(m+ \theta)
    \prod_{i=2}^L ( \kappa\tilde\gamma_i^2 - \theta)  
    +
    \left( \kappa 
    + \frac{\kappa+1}{m} \theta\right)
    \cdot 
    \sum_{i=2}^L 
    \left[
    m_i
    \tilde\gamma_i^2 
     \prod_{j\ne 1, i} ( \kappa\tilde\gamma_j^2 - \theta ) 
    \right]. 
\end{align*}
We can verify that
\begin{align*}
    \check{g}(\kappa \tilde\gamma_1^2) 
    & = 
    \check{g}(0)
    =
    -m
    \prod_{i=2}^L ( \kappa\tilde\gamma_i^2 )  
    +
    \kappa 
    \cdot 
    \sum_{i=2}^L 
    \left[
    m_i
    \tilde\gamma_i^2 
     \prod_{j\ne 1, i} ( \kappa\tilde\gamma_j^2) 
    \right]
    = 
    \kappa^{L-1} 
    \prod_{i=2}^L \tilde\gamma_i^2 
    \left(- m + \sum_{i=2}^L m_i\right)\\
    & = \kappa^{L-1} 
    \prod_{i=2}^L \tilde\gamma_i^2 \cdot (-m_1) < 0, 
\end{align*}
where the last equality holds because $\sum^L_{i=1} m_i = m$, 
and
for $2\le l \le L$, 
\begin{align*}
    \check{g}(\kappa \tilde{\gamma}_l^2)
    & = 
    -(m+ \kappa \tilde{\gamma}_l^2)
    \underbrace{
    \prod_{i=2}^L ( \kappa\tilde\gamma_i^2 - \kappa \tilde{\gamma}_l^2)}_{=0}
    +
    \left( \kappa 
    + \frac{\kappa+1}{m} \kappa \tilde{\gamma}_l^2\right)
    \cdot 
    \sum_{i=2}^L 
    \left[
    m_i
    \tilde\gamma_i^2 
     \prod_{j\ne 1, i} ( \kappa\tilde\gamma_j^2 - \kappa \tilde{\gamma}_l^2 ) 
    \right]
    \\
    & = 
    \underbrace{\left( \kappa 
    + \frac{\kappa+1}{m} \kappa \tilde{\gamma}_l^2\right)
    \cdot  
    m_l 
    \tilde\gamma_l^2}_{>0} 
    \cdot
     \prod_{j\ne 1, l} ( \kappa\tilde\gamma_j^2 - \kappa \tilde{\gamma}_l^2 ), 
\end{align*}
whose sign is $(-1)^{l-2}$. 
Thus, by the continuity of $\check{g}(\cdot)$, it must have a positive root in $(\kappa\tilde{\gamma}_l^2, \kappa\tilde{\gamma}_{l+1}^2)$, for $l=1,2,\ldots, L-1$. 
In addition, as shown in \eqref{eq:g_gamma_0}, $\tilde{g}(\cdot)$ also has a zero root. 
Combined with the first part, we can know that (a) and (b) in Lemma \ref{lem:roots} holds. 

Fourth, we prove (c).  
Consider first the sign of the limit of $g(\theta)$ as $\theta \longrightarrow -\infty$. 
By definition, 
\begin{align*}
    & \quad \ \lim_{\theta \rightarrow -\infty}\theta^{-(m+1)} g(\theta) \\
    & = 
    \lim_{\theta \rightarrow -\infty} \left\{ -(m/\theta + 1)
    \prod_{i=1}^m ( \kappa\gamma_i^2/\theta - 1)  
    +
    \theta^{-1} 
    \left( \kappa/\theta 
    + \frac{\kappa+1}{m} \right)
    \cdot 
    \sum_{i=1}^m 
    \left[
    \gamma_i^2 \prod_{j\ne i} ( \kappa\gamma_j^2/\theta - 1 ) 
    \right]\right\}\\
    & = (-1)^{m+1}. 
\end{align*}
Consequently, $g(\theta)$ is positive when $\theta$ is sufficiently small. 
Consider then the sign of $g(\theta)$ when $-\tau^{-1} \le \theta < 0$. 
By definition, for $\theta < 0$, 
\begin{align}\label{eq:g_simp}
    & \quad \ \underbrace{\left[\prod_{i=1}^m ( \kappa\gamma_i^2 - \theta) \right]^{-1}}_{>0} 
    g(\theta) 
    \nonumber
    \\
    & = 
    -(m+ \theta)
        +
        \left( \kappa 
        + \frac{\kappa+1}{m} \theta\right)
        \cdot 
        \sum_{i=1}^m 
        \frac{\gamma_i^2}{\kappa\gamma_i^2 - \theta} 
    = 
    -(m+ \theta)
    +
    ( 1 
    + \tau \theta)
    \cdot 
    \sum_{i=1}^m 
    \frac{\kappa  \gamma_i^2}{\kappa\gamma_i^2 - \theta}
    \nonumber
    \\
    & = 
    \sum_{i=1}^m 
    \left[ 
    \frac{ \kappa  \gamma_i^2
    + \tau \theta \kappa  \gamma_i^2}{\kappa\gamma_i^2 - \theta}
    - 1
    \right]
    - \theta 
    = 
    \sum_{i=1}^m 
    \frac{  \tau \theta \kappa  \gamma_i^2 + \theta}{\kappa\gamma_i^2 - \theta}
    - \theta 
    = 
    \theta 
    \left(
    \sum_{i=1}^m 
    \frac{  \tau \kappa  \gamma_i^2 + 1}{\kappa\gamma_i^2 - \theta}
    - 1 \right)
    \nonumber
    \\
    & = 
    \theta 
    \left[
    \sum_{i=1}^m 
    \frac{  \tau (\kappa  \gamma_i^2 + \tau^{-1} )}{\kappa\gamma_i^2 - \theta}
    - 1 \right]. 
\end{align}
When $-\tau^{-1} \le \theta < 0$, we have $\tau^{-1} \ge -\theta$, and thus $\kappa  \gamma_i^2 + \tau^{-1} \ge \kappa\gamma_i^2 - \theta$. This then implies that 
\begin{align*}
    \underbrace{\left[\prod_{i=1}^m ( \kappa\gamma_i^2 - \theta) \right]^{-1}}_{>0} 
    g(\theta)
    & = 
    \theta 
    \left[
    \sum_{i=1}^m 
    \frac{  \tau (\kappa  \gamma_i^2 + \tau^{-1} )}{\kappa\gamma_i^2 - \theta}
    - 1 \right]     \\
    & \le 
    \theta 
    \left(
    \sum_{i=1}^m 
    \tau 
    - 1 \right)     \\
    & = \theta (m\tau - 1) \\
    & = 
    \theta \left(\frac{\kappa+1}{\kappa} - 1\right)
    \\
    & < 0. 
\end{align*}
Consequently, $g(\theta)$ is negative for $\theta \in [-\tau^{-1}, 0)$. 
From the above, $g(\cdot)$ must have a negative root in $(-\infty, -\tau^{-1})$. 
Thus, (c) in Lemma \ref{lem:roots} holds. 

From the above, Lemma \ref{lem:roots} holds. 
\end{proof}
}

\subsection{Integral representation of the rejection probability}\label{sec:integral_rej_prob}

\begin{lem} \label{lem:rej_prob_integral}
Recall the test statistic $T_m$ in \eqref{eq:pop-test-1}. 
Let $\{\theta_i\}^{m+1}_{i=1}$ be the roots of the function $g(\cdot)$ in Lemma \ref{lem:char_gamma}, with $\theta_{m+1}$ being the unique negative root, and 
$g_0(\theta) \equiv \prod^{m+1}_{i=1} (\theta - \theta_i)$. Then for any $c>0$, any positive 
$\sigma_{m+1}$, and any nonnegative $\sigma_1, \ldots, \sigma_m$,
	\begin{equation*}
		\bP[ |T_m|> c ]
		= \frac{1}{\pi} 
		 \int^{|\theta_{m+1}|}_0	\frac{s^{\frac{m-1}{2}}}{
         [(-1)^m g_0(-s)]^{\frac{1}{2}}
         } \ \mathrm{d}s
         = 
         \frac{1}{\pi} 
		 \int^{|\theta_{m+1}|}_0	\frac{s^{\frac{m-1}{2}}}{
         [-g(-s)]^{\frac{1}{2}}
         } \ \mathrm{d}s
	\end{equation*}
\end{lem}

\hideif{
\begin{proof}[Proof of Lemma \ref{lem:rej_prob_integral}]
Let $\{\xi_i\}^{m+1}_{i=1}$ be i.i.d. standard normal random variables, 
and $\bs{\xi} \equiv (\xi_1, \xi_2, \ldots,$ $\xi_{m+1})^\top$. 
From Lemma \ref{lem:t_test_matrix}, 
the probability that $|T_m|>c$ is equivalent to:
\begin{align*}
    \bP[|T_m|>c]
    = 
    \bP[ \tilde{\bm{\nob}}^\top \bm{V} \tilde{\bm{\nob}} < 0 ]
    = \bP[ \bs{\xi}^\top \bs{D}\bm{V} \bs{D} \bs{\xi} < 0 ]
    = \bP\left[ \sum^{m+1}_{i=1} \lambda_i \xi_i^2 < 0 \right], 
\end{align*}
where $\lambda_1, \lambda_2, \ldots, \lambda_{m+1}$ are eigenvalues of $\bs{D}\bm{V} \bs{D}$ or equivalently the roots of the characteristic polynomial in Lemma \ref{lem:character_sigma}. 
From Lemma \ref{lem:char_gamma}, it can be further written as 
\begin{align*}
    \bP[|T_m|>c]
    = 
    \bP\left[ \sum^{m+1}_{i=1} \lambda_i \xi_i^2 < 0 \right]
    = 
    \bP\left[ \sum^{m+1}_{i=1} \lambda_i \sigma_{m+1}^{-2} \xi_i^2 < 0 \right]
    =
    \bP\left[ \sum^{m+1}_{i=1} \theta_i \xi_i^2 < 0 \right].
\end{align*}
Using \citet[Theorem 2]{makshanovshalaevskii1977some}, we can write the probability as the following integral:
\begin{equation}\label{eq:rej_prob_integral}
		\bP[|T_m|>c]
		 = \bP\left[ \frac{\xi_{m+1}^2}{\sum^{m}_{i=1} \frac{\theta_i}{|\theta_{m+1}|} \xi_i^2} > 1 \right]
		 = \frac{1}{\pi} \int^\infty_0 \frac{t^{-\frac{1}{2}} (1 + t)^{-1}}{\sqrt{\prod^{m}_{i=1} [1 + \frac{\theta_i}{|\tm|}(1 + t)]}} \ \text{d}t.
	\end{equation}
    By some algebra and the definition of $g_0(\cdot)$, 
	\begin{align*}
		& \quad \ \prod^{m}_{i=1} \left[ 1 + \frac{\theta_i}{|\theta_{m+1}|} (1 + t) \right]
        \nonumber
        \\
		& = \prod^{m}_{i=1} \left[ \left( \frac{|\theta_{m+1}|}{1 + t} + \theta_i \right) \frac{1 + t}{|\theta_{m+1}|} \right]	\notag	
        = \left( \frac{1 + t}{|\theta_{m+1}|} \right)^{m} \prod^{m}_{i=1} \left( \frac{|\theta_{m+1}|}{1 + t} + \theta_i  \right) 	\notag	\\
		& = \left( \frac{1 + t}{|\theta_{m+1}|} \right)^{m} \left( \frac{|\theta_{m+1}|}{1 + t} - |\theta_{m+1}|  \right)^{-1} \prod^{m+1}_{i=1} \left( \frac{|\theta_{m+1}|}{1 + t} + \theta_i  \right) 	\notag	\\
		& = - \frac{1}{t} \left( \frac{1 + t}{|\theta_{m+1}|} \right)^{m+1}  \prod^{m+1}_{i=1} \left( \frac{|\theta_{m+1}|}{1 + t} + \theta_i  \right)	\notag	
        = (- 1)^{m+2} \frac{1}{t} \left( \frac{1 + t}{|\theta_{m+1}|} \right)^{m+1}  \prod^{m+1}_{i=1} \left( - \frac{|\theta_{m+1}|}{1 + t} - \theta_i  \right) \notag	\\
		& =
		(- 1)^{m} \frac{1}{t} \left( \frac{1 + t}{|\theta_{m+1}|} \right)^{m+1}  g_0 \left( - \frac{|\theta_{m+1}|}{1 + t}  \right).
	\end{align*}
    We can then write \eqref{eq:rej_prob_integral} as 
	\begin{equation*}
		\bP[|T_m|>c]
		= \frac{1}{\pi} \int^\infty_0 
		 	\frac{(1+t)^{-\frac{m+1}{2} - 1} |\theta_{m+1}|^{\frac{m+1}{2}} }{
			[(-1)^{m} g_0(-\frac{|\theta_{m+1}|}{1 + t})]^{\frac{1}{2}}} \ \text{d}t.
	\end{equation*}
    After a change of variable with $s = \frac{|\theta_{m+1}|}{1+t}$ 
    and $\text{d}s = -|\theta_{m+1}|(1+t)^{-2} \ \text{d}t$, this further simplifies to 
    \begin{align*}
        \bP[|T_m|>c]
		= \frac{1}{\pi} \int^0_{|\theta_{m+1}|} 
		\frac{(1+t)^{-1} s^{\frac{m+1}{2}} }{
			[(-1)^{m} g_0(-s)]^{\frac{1}{2}}} 
        \cdot 
        - \frac{(1+t)^2}{|\theta_{m+1}|}
            \ \text{d}s
        = \frac{1}{\pi} \int^{|\theta_{m+1}|}_0 
		\frac{ s^{\frac{m-1}{2}} }{
			[(-1)^{m} g_0(-s)]^{\frac{1}{2}}} 
            \ \text{d}s.
    \end{align*}

Note that, by definition, $g(\theta)$ and $g_0(\theta)$ 
are both polynomial functions of $\theta$ and have exactly the same roots. 
Hence, they must be a scalar multiple of each other. 
Because the coefficient of $\theta^{m+1}$ in $g(\theta)$ is $(-1)^{m+1}$ and the coefficient of $\theta^{m+1}$ in $g_0(\theta)$ is 1, we have 
$
    g(\theta) = (-1)^{m+1}  \cdot g_0(\theta).
$
Consequently, 
\begin{align*}
        \bP[|T_m|>c]
        = \frac{1}{\pi} \int^{|\theta_{m+1}|}_0 
		\frac{ s^{\frac{m-1}{2}} }{
			[(-1)^{m} g_0(-s)]^{\frac{1}{2}}} 
            \ \text{d}s
        = \frac{1}{\pi} \int^{|\theta_{m+1}|}_0 
		\frac{ s^{\frac{m-1}{2}} }{
			[-g(-s)]^{\frac{1}{2}}} 
            \ \text{d}s
    \end{align*}

    From the above, Lemma \ref{lem:rej_prob_integral} holds. 
\end{proof}
}

\begin{lem} \label{lem:rej_prob_U}
Recall the test statistic $T_m$ in \eqref{eq:pop-test-1}, 
and consider any $c>0$, any positive 
$\sigma_{m+1}$, and any nonnegative $\sigma_1, \ldots, \sigma_m$. 
Define $\kappa \equiv \frac{m c^2}{m-1}$, $\tau \equiv \frac{\kappa +1}{\kappa m}$, and 
$\gamma_i \equiv \frac{\sigma_i}{\sigma_{m+1}}$ for $i \in \{1,\ldots, m\}$. 
Let $\theta_{m+1}$ be the unique negative root of the function $g(\cdot)$ in Lemma \ref{lem:char_gamma}. 
Let $x_i \equiv \kappa \gamma_i^2$ for $i \in \{1,\ldots, m\}$. Then,
\[
	\bP[|T_m|>c]
	= \frac{1}{\pi}	\int^{|\theta_{m+1}|}_0
			\frac{U(\bm{x}, s)}{\sqrt{|\theta_{m+1}|-s}} 		\ \mathrm{d}s,	
\]
where $\bm{x} \equiv (x_1, \ldots, x_m)$,
	\begin{align*}
        U(\bm{x},s)
		\equiv \frac{s^{\frac{m}{2}-1}}{
			\sqrt{ P(\bm{x},s) Q(\bm{x},s) 	}}, 
        \ \ 
		P(\bm{x},s)
		\equiv \sum^m_{i=1} \frac{1 + \tau x_i}{(x_i +s)(x_i - \theta_{m+1})},
        \ \ 
		Q(\bm{x},s)
		\equiv  \prod^m_{i=1} (x_i +s). 
	\end{align*}
\end{lem}

\hideif{
\begin{proof}[Proof of Lemma \ref{lem:rej_prob_U}] 
By the same logic as \eqref{eq:g_simp}, 
for $\theta < 0$, 
we have 
\begin{align}\label{eq:g0_equiv}
   g(\theta) 
    &  = 
        \prod_{i=1}^m ( \kappa\gamma_i^2 - \theta)
        \cdot
        \theta 
    \left(
    \sum_{i=1}^m 
    \frac{  \tau \kappa  \gamma_i^2 + 1}{\kappa\gamma_i^2 - \theta}
    - 1 \right)
     = 
        \prod_{i=1}^m ( x_i - \theta)
        \cdot
        \theta 
    \left(
    \sum_{i=1}^m 
    \frac{  \tau x_i + 1}{x_i - \theta}
    - 1 \right).
\end{align}
Because $g(\theta_{m+1})=0$ and $\theta_{m+1}<0$, we must have 
	\begin{equation}
	\label{eq:constraint_neg_root}
		1 = \sum^m_{i=1} \frac{1 + \tau x_i}{x_i - \theta_{m+1}}.
	\end{equation}
	\item Substituting \eqref{eq:constraint_neg_root} into \eqref{eq:g0_equiv}, 
    we can write $g(\theta)$ for $\theta < 0$ as
    \begin{align}\label{eq:g0_equiv_neg_root}
    g(\theta) 
    & = 
        \prod_{i=1}^m ( x_i - \theta)
        \cdot
        \theta 
    \left(
    \sum_{i=1}^m 
    \frac{  \tau x_i + 1}{x_i - \theta}
    - \sum^m_{i=1} \frac{1 + \tau x_i}{x_i - \theta_{m+1}} \right)
    \nonumber
    \\
    & = 
        \prod_{i=1}^m ( x_i - \theta)
        \cdot
        \theta \left[\sum^m_{i=1} (1+\tau x_i) \left(\frac{1}{x_i - \theta} - \frac{1}{x_i - \theta_{m+1}} \right) \right]
    \nonumber
        \\
    & = 
        \prod_{i=1}^m ( x_i - \theta)
        \cdot
        \theta (\theta - \theta_{m+1})
		\sum^m_{i=1} \frac{1 + \tau x_i}{(x_i - \theta)(x_i - \theta_{m+1})}
    \end{align}
Substituting \eqref{eq:g0_equiv_neg_root} into the integral representation in Lemma \ref{lem:rej_prob_integral}, we then have
	\begin{align*}
		\bP[|T_m| > c]
		& = \frac{1}{\pi}	\int^{|\theta_{m+1}|}_0
			\frac{s^{\frac{m-1}{2}}}{
			\sqrt{ - \prod^m_{i=1} (x_i +s)	(-s)	(-s - \theta_{m+1})	
			\sum^m_{i=1} \frac{1 + \tau x_i}{(x_i +s)(x_i - \theta_{m+1})}}
			}	\ \text{d}s	\notag \\
		& =  \frac{1}{\pi}	\int^{|\theta_{m+1}|}_0
			\frac{s^{\frac{m}{2}-1}}{
			\sqrt{ \prod^m_{i=1} (x_i +s) 	(-s - \theta_{m+1})	
			\sum^m_{i=1} \frac{1 + \tau x_i}{(x_i +s)(x_i - \theta_{m+1})}}
			}		\ \text{d}s	\notag \\
		& = \frac{1}{\pi}
			\int^{|\tm|}_0
				\frac{s^{\frac{m}{2}-1}}{\sqrt{Q(\bm{x},s) P(\bm{x},s) (|\tm|-s)}}
				\ \text{d}s	\notag	\\
		& = \frac{1}{\pi}	\int^{|\theta_{m+1}|}_0
			\frac{U(\bm{x}, s)}{\sqrt{|\theta_{m+1}|-s}} 		\ \text{d}s.	\notag
	\end{align*}
Therefore, Lemma \ref{lem:rej_prob_U} holds. 
\end{proof}
}

\subsection{Derivatives of the rejection probability}\label{sec:deri_rej_prob}

\begin{lem}\label{lem:deri_PQ}
Consider the same setup and notations in Lemma \ref{lem:rej_prob_U}. 
Fix $c, x_3, \ldots, x_m$ and $\tm$, 
and let $z=x_1>0$ and $w=x_2>0$,
where we view $w$ as a function of $z$ as implied by \eqref{eq:constraint_neg_root}. 
Define $\Delta \equiv \frac{z-w}{(z - \theta_{m+1})^2}$, $A \equiv wz - \theta_{m+1}^2$ and $B \equiv z + w - 2\theta_{m+1}$. 
We then have
	\begin{align}
		\frac{\partial w}{\partial z}
		& = - \frac{( w - \theta_{m+1})^2}{(z - \theta_{m+1})^2},	\label{eq:deri_w}	\\
		\frac{\partial P(\bm{x}, s)}{\partial z}
		& = \frac{2\Delta(1-\tau s)[A + \frac{B}{2}(s-\theta_{m+1})]}{(z+s)^2(w+s)^2},	\label{eq:deri_P}	\\
		\frac{\partial Q(\bm{x}, s)}{\partial z}
		& = \Delta (A + Bs) \prod^m_{i=3} (x_i + s).	\label{eq:deri_Q}
	\end{align}
\end{lem}

\hideif{
\begin{proof}[Proof of Lemma \ref{lem:deri_PQ}]
From Lemma \ref{lem:roots}, $\theta_{m+1} < -\tau^{-1}$, which implies that $-\tau \theta_{m+1} > 1$. Consequently, 
$\frac{1 + \tau w}{w - \theta_{m+1}}$ is increasing in $w$, because 
\begin{align*}
    \frac{\partial}{\partial w}\left( \frac{1 + \tau w}{w - \theta_{m+1}} \right)
    & = 
    \frac{-\tau \theta_{m+1}-1}{(w - \theta_{m+1})^2} > 0.
\end{align*}
This ensures that $w$ can be uniquely determined by $z$ through \eqref{eq:constraint_neg_root}, once we fix $c, x_3, \ldots, x_m$ and $\tm$. 
Below 
we prove the three derivatives in \eqref{eq:deri_w}--\eqref{eq:deri_Q} separately as follows. 
\begin{itemize}
	\item {\bf Proof of \eqref{eq:deri_w}.}  From \eqref{eq:constraint_neg_root}, we have 
    \begin{align*}
        1 = \frac{1 + \tau z}{z - \theta_{m+1}}
        + \frac{1 + \tau w}{w - \theta_{m+1}}
		+ \underbrace{\sum^m_{i=3} \frac{1 + \tau x_i}{x_i - \theta_{m+1}}}_{\text{``Constant''}}.
    \end{align*}
	Taking derivative of the above gives
	\begin{align*}
		0 
		& = 
		\frac{-\tau\tm-1}{(z-\tm)^2} \ \text{d}z
		+ 
		\frac{-\tau\tm-1}{(w-\tm)^2} \ \text{d}w.
	\end{align*}
    As discussed before, 
    $-\tau\theta_{m+1} > 1$. 
    We then have  
	\begin{align*}
		\frac{\partial w}{\partial z}
		= - \frac{( w - \theta_{m+1})^2}{(z - \theta_{m+1})^2},
	\end{align*}
	i.e., \eqref{eq:deri_w} holds.      
	\item {\bf Proof of \eqref{eq:deri_P}.}  Let $\tilde p(x) \equiv \frac{1 + \tau x}{(x + s)(x - \theta_{m+1})}$. We have
	\[
		\frac{\text{d}\tilde p(x)}{\text{d}x}
		= -\frac{\tau x^2 + 2x + (s - \theta_{m+1} + s\theta_{m+1}\tau)}{(x+s)^2(x-\theta_{m+1})^2}
		= - \frac{\tau x^2 + 2x + d}{(x+s)^2(x-\theta_{m+1})^2},
	\]
	where $d \equiv s - \theta_{m+1} + s\theta_{m+1}\tau = s(1+\theta_{m+1}\tau) - \theta_{m+1}$. Using \eqref{eq:deri_w}, we have 
	\begin{align}
		\frac{\partial P(\bm{x}, s)}{\partial z}
		& = \frac{\text{d}\tilde p(z)}{\text{d}z} + \frac{\text{d}\tilde p(w)}{\text{d}w} \frac{\partial w}{\partial z} \notag	\\
        & = 
        - \frac{\tau z^2 + 2z + d}{(z+s)^2(z-\theta_{m+1})^2} + \frac{\tau w^2 + 2w + d}{(w+s)^2(w-\theta_{m+1})^2} \frac{( w - \theta_{m+1})^2}{(z - \theta_{m+1})^2} \notag
        \\
		& = \frac{-(\tau z^2 + 2z + d)(w+s)^2 + (\tau w^2 + 2w + d)(z + s)^2}{(w+s)^2(z+s)^2(z-\theta_{m+1})^2}.	\label{eq:lem:6-4}
	\end{align}
	The numerator of \eqref{eq:lem:6-4} can be simplified as follows:
	\begin{align}
		& \quad \  
        -(\tau z^2 + 2z + d)(w+s)^2 + (\tau w^2 + 2w + d)(z + s)^2		\notag	\\
        & = (z-w)[ 2wz (1 - \tau s) - (w+z) (\tau s^2 - d) + 2ds - 2s^2 ]
		\notag	\\
		& = (z - w) \{ 2wz (1 - \tau s) - (w + z)[\tau s^2 - s(1 + \theta_{m+1} \tau) + \tm] 	
        \notag \\
        & \qquad + 2s^2 (1 + \theta_{m+1}\tau) - 2s \theta_{m+1} - 2s^2 \} 
        \notag \\
		& = (z - w) \{ \underbrace{(zw - \theta_{m+1}^2)}_{=A}[2(1-\tau s)] - \underbrace{(w + z - 2\theta_{m+1})}_{=B}[\tau s^2 - s(1+\theta_{m+1}\tau) + \theta_{m+1}] \notag	\\
    		& \qquad + \underbrace{2s^2 (1 + \theta_{m+1}\tau) - 2s\theta_{m+1} - 2s^2 + 2\theta_{m+1}^2 (1 - \tau s) - 2\theta_{m+1}[\tau s^2 - s(1 + \theta_{m+1}\tau) + \theta_{m+1}]}_{=0}\}	\notag	\\
		& = (z - w) \{ A[2(1-\tau s)] - B[\tau s^2 - s(1+\theta_{m+1}\tau) + \theta_{m+1} ] \}	\notag \\
		& = (z - w) \{A[2(1-\tau s)] 
				  - B(\tau s - 1)(s - \tm)\}	\notag	\\
		& = 2(z-w)(1 - \tau s)\left[
			A + \frac{B}{2}(s - \tm)
			\right],	\label{eq:lem:6-5}
	\end{align}
    where the definitions of $A$ and $B$ follow from the statement of the lemma. Next, substitute \eqref{eq:lem:6-5} into \eqref{eq:lem:6-4} and using the definition of $\Delta$, we have
	\begin{align*}
		\frac{\partial P(\bm{x},s)}{\partial z}
		& =
        \frac{2(z-w)(1 - \tau s)\left[
			A + \frac{B}{2}(s - \tm)
			\right]}{(w+s)^2(z+s)^2(z-\theta_{m+1})^2}
        =
        2\Delta(1-\tau s) \frac{[A  + \frac{B}{2}(s-\theta_{m+1})]}{(w + s)^2 (z + s)^2}.
	\end{align*}
    Thus, \eqref{eq:deri_P} holds. 
    
	\item {\bf Proof of \eqref{eq:deri_Q}.}  By the definition of $Q$, we have 
	\[
		Q(\bm{x}, s)
		= (w+s)(z+s) \underbrace{\prod^m_{i=3} (x_i + s)}_{\text{``Constant''}}.
	\]
	Using \eqref{eq:deri_w}, 
    we have 
    \begin{align}\label{eq:deri_Q_proof}
		\frac{\partial [(w+s)(z+s)]}{\partial z}
		& = (w + s) + (z + s) 	\frac{\partial w}{\partial z}		\notag 
        =  (w + s) - (z + s) \frac{( w - \theta_{m+1})^2}{(z - \theta_{m+1})^2}	\notag	\\
		& = \frac{(z - w)[(wz - \theta_{m+1}^2) + s(z + w - 2\theta_{m+1})]}{(z - \theta_{m+1})^2}
        \nonumber
        \\
		& = \Delta (A + Bs).
    	\end{align} 
    Thus, 
	\begin{align*}
		\frac{\partial Q(\bm{x}, s)}{\partial z}
		& = \Delta (A + Bs) \prod^m_{i=1} (x_i + s), 
    	\end{align*} 
       i.e., \eqref{eq:deri_Q} holds.
\end{itemize}
From the above, Lemma \ref{lem:deri_PQ} holds. 
\end{proof}
}

\begin{lem}	\label{lem:deri_U}
Consider the same setup and notations as in Lemmas \ref{lem:rej_prob_U} and \ref{lem:deri_PQ}. We have
\[
	\frac{\partial U(\bm{x},s)}{\partial z}
	=
	-\frac{\Delta \tilde h(\bm{x}, s) U(\bm{x}, s)}{2},
\]
where
\begin{align*}
	\tilde h(\bm{x}, s)
	& \equiv \tilde h_P(\bm{x}, s) + \tilde h_Q(\bm{x}, s),	\\
	\tilde h_P(\bm{x}, s)
	& \equiv 	
	\frac{2(1-\tau s)[A + \frac{B}{2}(s-\theta_{m+1})]}{(z+s)^2(w+s)^2 P(\bm{x},s)}	,	\\
	\tilde h_Q(\bm{x}, s)
	& \equiv 	\frac{A + Bs}{(w+s)(z+s)}	.
\end{align*}

\end{lem}

\hideif{
\begin{proof}[Proof of Lemma \ref{lem:deri_U}] By definition, 
	\begin{align}
		\frac{\partial U(\bm{x}, s)}{\partial z}
		& = -\frac{s^{\frac{m}{2} - 1}}{2[P(\bm{x},s)Q(\bm{x},s)]^{\frac{3}{2}}} \left[ 
			Q(\bm{x}, s) \frac{\partial P(\bm{x}, s)}{\partial z}
			+ P(\bm{x}, s) \frac{\partial Q(\bm{x}, s)}{\partial z}
			\right]	\notag 	\\
		& = -\frac{U(\bm{x}, s)}{2} \Big[
			\underbrace{
				\frac{1}{P(\bm{x}, s)} \frac{\partial P(\bm{x}, s)}{\partial z}
			}_{\equiv h_P(\bm{x},s)}
			+ 
			\underbrace{
				\frac{1}{Q(\bm{x}, s)} \frac{\partial Q(\bm{x}, s)}{\partial z}
			}_{\equiv h_Q(\bm{x},s)}
			\Big].	\label{eq:lem:6b-1}
	\end{align}
From Lemma \ref{lem:deri_PQ}, $h_P(\bm{x},s)$ and $h_Q(\bm{x}, s)$ have the following equivalent forms:
\begin{align*}
	h_P(\bm{x},s)
	& = \frac{2\Delta(1-\tau s)[A + \frac{B}{2}(s-\theta_{m+1})]}{(z+s)^2(w+s)^2P(\bm{x},s)}
	= \Delta \tilde h_P(\bs{x}, s),		\\
	h_Q(\bm{x}, s)
	& = \frac{\Delta (A + Bs) \prod^m_{i=3} (x_i + s)}{Q(\bm{x},s)}
	= \frac{\Delta(A + Bs)}{(w+s)(z+s)}
	= \Delta \tilde h_Q(\bs{x}, s), 
\end{align*}
where the last equality in each of the above two equations follows by definition. 
Substituting the above into \eqref{eq:lem:6b-1}, we then have
\begin{align*}
		\frac{\partial U(\bm{x}, s)}{\partial z}
		& = -\frac{U(\bm{x}, s)}{2} \left[
			h_P(\bm{x},s)
			+ 
			 h_Q(\bm{x},s)
			\right]
        \notag \\
        & = 
        -\frac{ U(\bm{x}, s)}{2}
        \Delta \left[
			\tilde h_P(\bm{x},s)
			+ 
			 \tilde h_Q(\bm{x},s)
			\right]
        \notag \\
        & = -\frac{ U(\bm{x}, s)}{2}
        \Delta\tilde h(\bm{x},s),
	\end{align*}
    where the last equality holds by definition. 
    Therefore, Lemma \ref{lem:deri_U} holds. 
\end{proof}
}

\begin{lem} \label{lem:deri_AB}
Consider the same setup and notations as in Lemmas \ref{lem:rej_prob_U} and \ref{lem:deri_PQ}. 
Let $\delta \equiv A + \frac{B}{2}(s-\theta_{m+1})$. Then, we have
	\begin{align}
		\frac{\partial A}{\partial z} & = \Delta A,	\label{eq:deri_A}	\\
		\frac{\partial B}{\partial z} & = \Delta B,	\label{eq:deri_B}	\\
		\frac{\partial \delta}{\partial z} & = \Delta \delta.	\label{eq:deri_delta}	
	\end{align}

\end{lem}

\hideif{
\begin{proof}[Proof of Lemma \ref{lem:deri_AB}] 
We prove the three derivatives in \eqref{eq:deri_A}--\eqref{eq:deri_delta} separately as follows. 
\begin{itemize}
	\item {\bf Proof of \eqref{eq:deri_A}:}
	\[
	\frac{\partial A}{\partial z}
	= w + z \frac{\partial w}{\partial z}
	= w - z \frac{( w - \theta_{m+1})^2}{(z - \theta_{m+1})^2} 
	= \Delta (wz - \theta_{m+1}^2)
	= \Delta A.
	\]
	\item {\bf Proof of \eqref{eq:deri_B}:}
	\[
	\frac{\partial B}{\partial z}
	= 1 +  \frac{\partial w}{\partial z}
	= 1 -  \frac{( w - \theta_{m+1})^2}{(z - \theta_{m+1})^2} 
	= \Delta B.
	\]
	\item {\bf Proof of \eqref{eq:deri_delta}:}
\[
	\frac{\partial \delta}{\partial z}
	= 
	\frac{\partial A}{\partial z} 
	+ 
	\frac{\partial B}{\partial z} \frac{(s-\theta_{m+1})}{2}
	= \Delta A + \frac{\Delta B}{2}(s - \tm)
	= \Delta \delta.
\]
\end{itemize}
From the above, Lemma \ref{lem:deri_AB} holds. 
\end{proof}
}

\begin{lem} \label{lem:deri_h_PQ}
Consider the same setup and notations as in Lemmas \ref{lem:rej_prob_U}--\ref{lem:deri_AB}. 
We have 
\begin{align}
		\frac{\partial \tilde h_P(\bm{x}, s)}{\partial z}
		& = \Delta \tilde  h_P(\bm{x}, s) [1 -\tilde h_P(\bm{x},s) - 2\tilde h_Q(\bm{x}, s) ]	,	
		\label{eq:deri_h_P}	\\
		\frac{\partial \tilde h_Q(\bm{x}, s)}{\partial z}
		& = \Delta [\tilde h_Q(\bm{x}, s) -\tilde h_Q(\bm{x}, s)^2]. 
		\label{eq:deri_h_Q}
\end{align}
\end{lem}

\hideif{
\begin{proof}[Proof of Lemma \ref{lem:deri_h_PQ}] 
We prove the two derivatives in \eqref{eq:deri_h_P} and \eqref{eq:deri_h_Q} separately as follows. 
\begin{itemize}
	\item {\bf Proof of \eqref{eq:deri_h_P}:} 
    Let $\tilde\delta \equiv 2(1-\tau s)\delta$. 
    By definition, 
    \begin{align*}
        \tilde h_P(\bm{x}, s)
        & = 	
	\frac{2(1-\tau s)[A + \frac{B}{2}(s-\theta_{m+1})]}{(z+s)^2(w+s)^2 P(\bm{x},s)}    \\
    & = \frac{2(1-\tau s)\delta}{(z+s)^2(w+s)^2 P(\bm{x},s)}    \\
    & = \frac{\tilde \delta}{(z+s)^2(w+s)^2 P(\bm{x},s)}.
    \end{align*}
    Thus, 
    \begin{align}
		\frac{\partial \tilde h_P(\bm{x}, s)}{\partial z}
		& = \frac{\partial}{\partial z} \frac{\tilde\delta}{(w+s)^2(z+s)^2 P(\bm{x}, s)}	\notag \\
		& = \frac{1
		}{(w+s)^4(z+s)^4 P(\bm{x}, s)^2
		} 
        \Big\{ 
        [(w+s)^2(z+s)^2 P(\bm{x}, s)] \frac{\partial \tilde\delta}{\partial z} \notag \\
        & \qquad \qquad \qquad \qquad \qquad \qquad \qquad 
        - \tilde\delta \frac{\partial[(w+s)^2(z+s)^2 P(\bm{x}, s)]}{\partial z}
        \Big\}.  \label{eq:lem:8-3c}
    \end{align}
    From Lemma \ref{lem:deri_AB}, we have
	\begin{equation}
		\label{eq:lem:8-3b}
		\frac{\partial \tilde \delta}{\partial z}
		= 2 (1 - \tau s) \frac{\partial \delta}{\partial z}
		= 2 \Delta  (1 - \tau s) \delta 
		= \Delta \tilde\delta.
	\end{equation}
     Next, using the definition of $\delta$ in Lemma \ref{lem:deri_AB} and using the expression of $\frac{\partial P(\bm{x}, s)}{\partial z}$ from Lemma \ref{lem:deri_PQ}, we obtain  
	\begin{align}
		& \quad \  \frac{\partial [(z+s)^2(w+s)^2P(\bm{x},s)]}{\partial z}	\notag\\
		& = (z+s)^2(w+s)^2 \frac{\partial P(\bm{x},s)}{\partial z}
    		+ P(\bm{x},s) \frac{\partial [ (z+s)^2(w+s)^2 ]}{\partial z}	\notag \\ 
            & = (z+s)^2(w+s)^2 \frac{\partial P(\bm{x},s)}{\partial z}
    		+ 2 P(\bm{x},s) (z+s) (w+s) \frac{\partial [ (z+s)(w+s) ]}{\partial z}	\notag \\ 
		& = (z+s)^2(w+s)^2 \frac{\Delta \tilde\delta}{(z+s)^2 (w+s)^2}
		+ 2 P(\bm{x},s) (z+s)(w+s) \Delta (A + B s) \notag \\
		& = \Delta [ \tilde\delta + 2P(\bm{x},s) (z+s)(w+s)  (A + B s) ], \label{eq:lem:8-3}
	\end{align}
    where the second last equality follows from \eqref{eq:deri_Q_proof}.     
    From \eqref{eq:lem:8-3c}--\eqref{eq:lem:8-3}, 
    we then have
    \begin{align*}
		& \quad \  \frac{\partial \tilde h_P(\bm{x}, s)}{\partial z}	\\
		& = \frac{\Delta \tilde\delta[(w+s)^2(z+s)^2 P(\bm{x}, s)]  - \Delta \tilde\delta [ \tilde\delta + 2P(\bm{x},s) (z+s)(w+s)  (A + B s) ]
		}{(w+s)^4(z+s)^4 P(\bm{x}, s)^2
		} \\
		& = \frac{\Delta \tilde\delta
			[(w+s)^2(z+s)^2 P(\bm{x}, s)
			-\tilde\delta 
			- 2P(\bm{x},s) (z+s)(w+s)  (A + B s) ]
		}{(w+s)^4(z+s)^4 P(\bm{x}, s)^2
		} 	\\
		& = \frac{\Delta \tilde\delta}{(w+s)^2(z+s)^2 P(\bm{x},s)}
		\left[
            1 
			- \frac{\tilde\delta}{(w+s)^2(z+s)^2 P(\bm{x},s)}
            - 
			\frac{2(A + B s)}{(w+s)(z+s)}
		\right] \\
		& = \Delta \tilde h_P(\bm{x}, s)
		\left[
			1 - \tilde h_P(\bm{x}, s) - 2\tilde h_Q(\bm{x}, s)
		\right],
	\end{align*}
        where the first equality follows from \eqref{eq:lem:8-3c}--\eqref{eq:lem:8-3}, the second and third equality follows from 
        rearranging and simplifying the fractions, 
        and the last equality follows from the definition of $\tilde h_P(\bm{x}, s)$ and $\tilde h_Q(\bm{x}, s)$.
	\item {\bf Proof of \eqref{eq:deri_h_Q}:} 
        From the definition of $\tilde h_Q(\bm{x}, s)$ in Lemma \ref{lem:deri_U}, we have 
        \begin{align*}
		\frac{\partial \tilde h_Q(\bm{x}, s)}{\partial z}
		& = \frac{\partial}{\partial z} \frac{(A + B s)}{(w + s)(z + s)} 	=  \frac{(w + s)(z + s) \frac{\partial(A + B s)}{\partial z}
		- (A + B s) \frac{\partial [(w + s)(z + s)]}{\partial z}}{(w + s)^2(z + s)^2} 	\\
		& = 
        \frac{(w + s)(z + s) \Delta(A + B s)
		- (A + B s) \Delta(A + B s)}{(w + s)^2(z + s)^2}
        \\
        & = 
        \frac{ \Delta(A + B s)}{(w + s)(z + s)}
        -
        \frac{\Delta(A + B s)^2}{(w + s)^2(z + s)^2}\\
        & = \Delta [\tilde h_Q(\bm{x}, s) - \tilde h_Q(\bm{x}, s)^2],
	\end{align*}
    where the third last equality follows from \eqref{eq:deri_Q_proof}, and the last equality follows from the definition of $\tilde h_Q(\bm{x}, s)$. 
\end{itemize}
From the above, Lemma \ref{lem:deri_h_PQ} holds. 
\end{proof}
}

\begin{lem} \label{lem:deri_rej_prob}
Consider the same setup and notations in Lemmas \ref{lem:rej_prob_U}--\ref{lem:deri_AB}.  
For any function  $r(\bm{x}, s)$ of $\bm{x}$ and $s$, define 
	\[
		L(r(\bm{x}, s)) \equiv \frac{1}{\pi}
			\int^{|\theta_{m+1}|}_0
			\frac{r(\bm{x},s) U(\bm{x},s)}{\sqrt{|\theta_{m+1}|-s}} \ \mathrm{d} s. 
	\]
    From Lemma \ref{lem:rej_prob_U},  $\bP[|T_m|>c]=L(1)$, where $1$ here denotes a constant function taking value one. 
    \begin{itemize}
        \item[(i)] 
        The first and second order derivatives of $\bP[|T_m|>c]$ over $z$ have the following equivalent forms:
        \begin{align}
		\frac{\partial \bP[|T_m|>c]}{\partial z}
		& = -\frac{\Delta}{2} L(\tilde h(\bm{x}, s)), \label{eq:deri_rej_prob_1}	\\
		\frac{\partial^2 \bP[|T_m|>c]}{\partial z^2}
		&
		= -\frac{L(\tilde h(\bm{x},s))}{2} \frac{\partial \Delta}{\partial z}
		- \frac{\Delta^2}{2} L ( \tilde h(\bm{x},s))
		+ \frac{3\Delta^2}{4} L(\tilde h(\bm{x},s)^2). \label{eq:deri_rej_prob_2}
	\end{align}

        \item[(ii)] If $c\ne m^{-1/2}$, $z\ne w$, and $\frac{\partial \bP[|T_m|>c]}{\partial z}=0$, then we must have $\frac{\partial^2 \bP[|T_m|>c]}{\partial z^2}>0$.
    \end{itemize}
\end{lem}

\hideif{
\begin{proof}[Proof of Lemma \ref{lem:deri_rej_prob}]
We first prove (i). We prove the two derivatives in \eqref{eq:deri_rej_prob_1} and \eqref{eq:deri_rej_prob_2} separately as follows. 
\begin{itemize}
	\item {\bf Proof of \eqref{eq:deri_rej_prob_1}:} From Lemmas \ref{lem:rej_prob_U} and \ref{lem:deri_U}, the derivative of $\bP[|T_m|>c]$ with respect to $z$ is
	\begin{align}
		\frac{\partial \bP[|T_m|>c]}{\partial z}
		& = \frac{1}{\pi} \int^{|\theta_{m+1}|}_0
			\frac{\partial U(\bm{x},s)}{\partial z} \frac{1}{\sqrt{|\theta_{m+1}|-s}}
			\ \text{d}s	\notag	\\
		& = -\frac{\Delta}{2 \pi} \int^{|\theta_{m+1}|}_0
			\frac{\tilde h(\bm{x},s) U(\bm{x},s)}{\sqrt{|\theta_{m+1}|-s}} \ \text{d} s	\notag	\\
		& = -\frac{\Delta}{2} L(\tilde h(\bm{x}, s)). \notag
	\end{align}
	\item {\bf Proof of \eqref{eq:deri_rej_prob_2}:} Using \eqref{eq:deri_rej_prob_1} and recognizing that $\Delta$, $\tilde h$, and $U$ are functions of $z$, we have
	\begin{align}
		& \quad \ \frac{\partial^2 \bP[|T_m|>c]}{\partial z^2}
		\nonumber
        \\
        & = -\frac{L(\tilde h(\bm{x},s))}{2} \frac{\partial \Delta}{\partial z}
		- \frac{\Delta}{2\pi} \int^{|\theta_{m+1}|}_0 \left[\frac{\partial \tilde h(\bm{x},s)}{\partial z}U(\bm{x},s) + \frac{\partial U(\bm{x},s)}{\partial z} \tilde h(\bm{x},s) \right]\frac{1}{\sqrt{|\theta_{m+1}|-s}} \ \text{d} s
		\notag \\
		& = -\frac{L(\tilde h(\bm{x},s))}{2} \frac{\partial \Delta}{\partial z}
		- \frac{\Delta}{2\pi} \int^{|\theta_{m+1}|}_0 \frac{\partial \tilde h(\bm{x},s)}{\partial z} \frac{U(\bm{x},s)}{\sqrt{|\theta_{m+1}|-s}} \ \text{d} s		\notag \\
		& \qquad \qquad + \frac{\Delta^2}{4\pi} \int^{|\theta_{m+1}|}_0 \tilde h(\bm{x},s)^2\frac{U(\bm{x},s)}{\sqrt{|\theta_{m+1}|-s}} \ \text{d} s
		\notag \\
		& = -\frac{L(\tilde h(\bm{x},s))}{2} \frac{\partial \Delta}{\partial z}
		- \frac{\Delta}{2} L \left(\frac{\partial \tilde h(\bm{x},s)}{\partial z} \right)
		+ \frac{\Delta^2}{4} L(\tilde h(\bm{x},s)^2), 	\label{eq:lem:9-3}
	\end{align}
    where the second last equality follows from Lemma \ref{lem:deri_U}. 
    By the definition of $\tilde h(\bm{x},s)$ in Lemma \ref{lem:deri_U} and from Lemma \ref{lem:deri_h_PQ}, we have
	\begin{align}
		\frac{\partial \tilde h(\bm{x},s)}{\partial z}
		& 
        = \frac{\partial \tilde h_P(\bm{x},s)}{\partial z} + \frac{\partial \tilde h_Q(\bm{x},s)}{\partial z}
        \nonumber
        \\
        & = \Delta \tilde  h_P(\bm{x}, s) [1 -\tilde h_P(\bm{x},s) - 2\tilde h_Q(\bm{x}, s) ] + \Delta [\tilde h_Q(\bm{x}, s) -\tilde h_Q(\bm{x}, s)^2]
        \nonumber
        \\
        & = \Delta \left\{
			[\tilde h_P(\bm{x}, s) + \tilde h_Q(\bm{x},s)]
			- [\tilde h_P(\bm{x}, s) + \tilde h_Q(\bm{x},s)]^2
		\right\}		\notag	\\
		& = \Delta [\tilde h(\bm{x}, s) - \tilde h(\bm{x},s)^2].
		\label{eq:lem:9-4}
	\end{align}
	Substituting \eqref{eq:lem:9-4} into \eqref{eq:lem:9-3} gives
	\begin{align*}
		\frac{\partial^2 \bP[|T_m|>c]}{\partial z^2}
		= -\frac{L(\tilde h(\bm{x},s))}{2} \frac{\partial \Delta}{\partial z}
		- \frac{\Delta^2}{2} L ( \tilde h(\bm{x},s))
		+ \frac{3\Delta^2}{4} L(\tilde h(\bm{x},s)^2).		
	\end{align*}
\end{itemize}

We now prove (ii). Because $z\ne w$, $\Delta$ is nonzero. 
From \eqref{eq:deri_rej_prob_1} and the condition in (ii), $L(\tilde h(\bm{x}, s))$ must be zero. 
From \eqref{eq:deri_rej_prob_2}, we then have 
\begin{align*}
    \frac{\partial^2 \bP[|T_m|>c]}{\partial z^2}
    &
    = -\frac{L(\tilde h(\bm{x},s))}{2} \frac{\partial \Delta}{\partial z}
    - \frac{\Delta^2}{2} L ( \tilde h(\bm{x},s))
    + \frac{3\Delta^2}{4} L(\tilde h(\bm{x},s)^2)
    = \frac{3\Delta^2}{4} L(\tilde h(\bm{x},s)^2)\ge 0. 
\end{align*}
From the definition of $U(\bm{x},s)$ in Lemma \ref{lem:rej_prob_U}, 
we can know that $U(\bm{x},s)$ is positive for $s \in (0, |\tm|)$. 
Therefore, 
$L(\tilde h(\bm{x},s)^2) = 0$ if and only if $\tilde h(\bm{x},s)=0$ for all $s \in (0, |\tm|)$. 
Note that, by the definition of $\tilde h(\bm{x},s)$ in Lemma \ref{lem:deri_U}, for any $s\in (0, |\tm|)$, 
\begin{align*}
    & \tilde h(\bm{x},s) = 0    \\
    \Longleftrightarrow \quad &  
    \frac{2(1-\tau s)[A + \frac{B}{2}(s-\theta_{m+1})]}{(z+s)^2(w+s)^2 P(\bm{x},s)}
    + 
    \frac{A + Bs}{(w+s)(z+s)} = 0
    \\
      \Longleftrightarrow \quad & 
    2(1-\tau s)[A + \frac{B}{2}(s-\theta_{m+1})]
    + 
    (A + Bs) (z+s) (w+s) P(\bm{x},s) = 0
    \\
      \Longleftrightarrow \quad &  
      \chi(s) = 0,
\end{align*}
where
\begin{align}
    \label{eq:h_tilde_poly}
    \chi(s)
    & \equiv 2(1-\tau s)[A + \frac{B}{2}(s-\theta_{m+1})]
    \prod_{j=1}^m (x_i+s)   \notag \\
    & \qquad + 
    (A + Bs) (z+s) (w+s) \sum^m_{i=1} \left[ \prod_{j\ne i} (x_i+s) \cdot \frac{1 + \tau x_i}{x_i - \theta_{m+1}} \right]
\end{align}
is a polynomial function of $s$. 
Note that the coefficient of $s^{m+2}$ in \eqref{eq:h_tilde_poly} is 
\begin{align*}
    -B\tau + B \sum_{i=1}^m \frac{1 + \tau x_i}{x_i - \theta_{m+1}}
    & = 
    B(1-\tau)       \\
    & = 
    B\left( 1-\frac{\kappa+1}{\kappa m} \right)       \\
    & = 
    \frac{B}{\kappa m}\left[ \kappa (m-1) - 1 \right]       \\
    & = 
    \frac{B}{\kappa m} (mc^2 - 1) \ne 0,
\end{align*}
where the first equality follows from \eqref{eq:constraint_neg_root}, 
the second equality follows from the definition of $\tau$, 
the last equality follows from the definition of $\kappa$, 
and the last inequality holds because $B = z+w-2\tm > 0$ and $c\ne m^{-1/2}$. 
Thus, \eqref{eq:h_tilde_poly} cannot be zero for all $s\in (0, |\tm|)$. 
This implies that $L(\tilde h(\bm{x},s)^2)$ must be positive. 
Consequently, $\frac{\partial^2 \bP[|T_m|>c]}{\partial z^2}$ is also positive.

From the above, Lemma \ref{lem:deri_rej_prob} holds. 
\end{proof}
}

\subsection{Bounds on the derivative of the rejection probability}\label{sec:bound_rej_prob}

\subsubsection{More accurate bounds on the unique negative root $\tm$}

\begin{lem}\label{lem:neg_root_abs_lower}
    Consider the same setup and notations in Lemma \ref{lem:rej_prob_U}. 
	For any $x \ge 0$ and integer $1\le k \le m$, define 
		\begin{align*}
			h(\theta; x, k)= -(k-1) x - \left[ m - x + (m+1-k) \tau x \right] \theta + \theta^2.
		\end{align*}
	\begin{itemize}
		\item[(a)] 
		For any $x \ge 0$ and $1\le k \le m$,
		$h(\theta; x, k)$, viewed as a function of $\theta$, has a unique positive root, which is denoted as $\theta(x,k)$. 
		 
		\item[(b)] $\theta(x,k)\ge m$ for any $x \ge 0$ and $1\le k \le m$.
		 
		\item[(c)] 
		$\theta(0, k) = m$ for any $1\le k \le m$, and $\theta(x,1) = m + x/\kappa$ for any $x\ge 0$. 
		\item[(d)] For any $1\le k \le n$, $|\tm|\ge \theta(x_{(k)},k) = \theta(\kappa \gamma^2_{(k)},k)$,   
		where $\gamma_{(1)}\le \gamma_{(2)} \le \ldots \le \gamma_{(m)}$ are the sorted values of $\{\gamma_j\}_{j=1}^m$, 
		and $x_{(1)} \le x_{(2)} \le \ldots \le x_{(m)}$ are the sorted values of $\{x_j\}_{j=1}^m$. 
	\end{itemize}
\end{lem}

\hideif{
\begin{proof}[Proof of Lemma \ref{lem:neg_root_abs_lower}(a)--(c)]
First, we consider the case when $x=0$. In this case, $h(\theta; x, k)$ simplifies to $h(\theta; x, k)= - m   \theta + \theta^2 = \theta (\theta - m)$. 
Obviously, $h(\theta; x, k)$ has a unique positive root at $\theta = m$.

Second, we consider the case when $k=1$. In this case, $h(\theta; x, k)$ simplifies to $h(\theta; x, k)= - ( m - x + m \tau x ) \theta + \theta^2 = \theta \{\theta - m - (m\tau -1)x \}$.
Note that $m\tau -1 = m \cdot (\frac{\kappa+1}{\kappa m}) - 1 = \frac{1}{\kappa}$ by definition. Consequently, 
$h(\theta; x, k)= \theta ( \theta - m - \frac{x}{\kappa} )$.
Thus, 
$h(\theta; x, k)$ has a unique positive root at $\theta = m + \frac{x}{\kappa}$.

Third, we consider the case when $x>0$ and $k>1$. 
In this case, $h(0;x,k) = -(k-1)x < 0$. 
Note that $h(\theta;x,k)$ converges to infinity as $\theta$ goes to positive or negative infinity. 
Thus, $h(\theta; x, k)$ must have one positive root and one negative root.

From the above, (a) and (c) hold. Below we prove (b).
Note that, for any $x \ge 0$ and $1\le k \le m$, 
\begin{align*}
	h(m; x, k)& = -(k-1) x - \left[ m - x + (m+1-k) \tau x \right] m + m^2\\
	& = 
	-(k-1) x  + m x -  (m+1-k) \tau m x    \\
	& = (m+1-k) x (1- \tau m) \\
	& = (m+1-k) x \left( 1- \frac{\kappa+1}{\kappa} \right) \le 0,
\end{align*}
where the last equality follows from the definition of $\tau$. 
By the property of quadratic functions, we must have $\theta(x, k) \ge m$, i.e., (b) holds. 
\end{proof}
}

\hideif{
\begin{proof}[Proof of Lemma \ref{lem:neg_root_abs_lower}(d)]
From \eqref{eq:constraint_neg_root}, 
we know that 
\begin{align}\label{eq:x_x_tm}
	1 & = \sum^m_{i=1} \frac{1 + \tau x_i}{x_i - \theta_{m+1}}
	= \sum^m_{i=1} \frac{1 }{x_i - \theta_{m+1}} 
	+ 
	\sum^m_{i=1} \frac{ \tau x_i}{x_i - \theta_{m+1}}
	=
	\frac{1}{\theta_{m+1}}
	\sum^m_{i=1} \frac{\theta_{m+1} }{x_i - \theta_{m+1}} 
	+ 
	\tau \sum^m_{i=1} \frac{  x_i}{x_i - \theta_{m+1}}    \nonumber
    \\
	& = \frac{1}{\theta_{m+1}}
	\sum^m_{i=1} \left( \frac{x_i}{x_i - \theta_{m+1}} - 1 \right)
	+ 
	\tau \sum^m_{i=1} \frac{ x_i}{x_i - \theta_{m+1}}
	= 
	\left( \frac{1}{\theta_{m+1}} + \tau \right)  \sum^m_{i=1} \frac{ x_i}{x_i - \theta_{m+1}} - \frac{m}{\theta_{m+1}}. 
\end{align}
Consequently, 
\begin{align*}
	0 =  -(m + \tm) + (1 + \tau \tm) \sum^m_{j=1} \frac{x_j}{x_j - \tm}.
\end{align*}
From Lemma \ref{lem:roots}(c), $1 + \tau \tm < 0$. 
Because $x/(x-\theta_{m+1})$ is increasing in $x\ge 0$, 
we have, for any $1\le k \le m$, 
\begin{align*}
	0 & =  -(m + \tm) + (1 + \tau \tm) \sum^m_{j=1} \frac{x_j}{x_j - \tm} \\
	& = -(m + \tm) + (1 + \tau \tm) \sum^m_{j=1} \frac{x_{(j)}}{x_{(j)} - \tm}
	\\
	& \le 
	-(m + \tm) + (1 + \tau \tm) (m+1-k) \frac{x_{(k)}}{x_{(k)} - \tm}. 
\end{align*}
Multiplying both left-hand and right-hand sides by $x_{(k)} - \tm$, we have
\begin{align*}
	0 & \le 
	-(m + \tm)(x_{(k)} - \tm) + (1 + \tau \tm) (m+1-k) x_{(k)}
	\\
	& = 
	- m x_{(k)} + m \tm  - \tm x_{(k)} + \tm^2 
	+ (m+1-k) x_{(k)} + (m+1-k) x_{(k)} \tau \tm \\
	& = -(k-1) x_{(k)} + \left[ m - x_{(k)} + (m+1-k) \tau x_{(k)} \right] \tm + \tm^2\\
	& = -(k-1) x_{(k)} - \left[ m - x_{(k)} + (m+1-k) \tau x_{(k)} \right] |\tm| + |\tm|^2, 
\end{align*}
i.e., $h(|\tm|; x, k) \ge 0$. 
From Lemma \ref{lem:neg_root_abs_lower}(a) and the properties of quadratic functions, we know that $|\tm|>0$ must be no less than $\theta(x_{(k)}, k)$. 
\end{proof}
}

\begin{lem}\label{lem:neg_root_abs_upper}
	Consider the same setup and notations in Lemmas \ref{lem:rej_prob_U} and \ref{lem:neg_root_abs_lower}. 
	$|\tm| \le m + x_{(m)}/\kappa = m + \gamma^2_{(m)}$, 
	recalling that $x_{(m)} = \max_{1\le i \le m} x_i$ and $\gamma_{(m)}=\max_{1\le i \le m} \gamma_i$. 
\end{lem}

\hideif{
\begin{proof}[Proof of Lemma \ref{lem:neg_root_abs_upper}]
	From \eqref{eq:constraint_neg_root} and by the same logic as the proof of Lemma \ref{lem:neg_root_abs_lower}, 
	\begin{align*}
		0 =  -(m + \tm) + (1 + \tau \tm) \sum^m_{j=1} \frac{x_j}{x_j - \tm}.
	\end{align*}
	From Lemma \ref{lem:roots}, $\tau |\tm| > 1$. 
	Because $\frac{x}{x+|\tm|}$ is increasing in $x$, 
	we then have
	\begin{align*}
		|\tm| -m = (\tau|\tm| - 1) \sum_{j=1}^m \frac{x_j}{x_j + |\tm|}
		\le 
		(\tau|\tm| - 1) \frac{m x_{(m)}}{x_{(m)} + |\tm|}, 
	\end{align*}
	which further implies that 
	\begin{align*}
		& (x_{(m)} + |\tm|)(|\tm| -m) \le 
		(\tau|\tm| - 1) m x_{(m)}
		\\
		\Longrightarrow \quad  
		& 
		x_{(m)} |\tm| + |\tm|^2 - m x_{(m)} - m |\tm| \le 
		\tau|\tm|m x_{(m)} - m x_{(m)} 
		\\
		\Longrightarrow \quad  
		&
		(x_{(m)} + |\tm| - m) |\tm|   \le 
		\tau|\tm|m x_{(m)}
		\\
		\Longrightarrow \quad  
		&
		x_{(m)} + |\tm| - m   \le 
		\tau m x_{(m)}
		\\
		\Longrightarrow \quad  
		&
		|\tm| \le m + (\tau m-1) x_{(m)}
		= m + \frac{x_{(m)}}{\kappa} = m + \gamma_{(m)}^2 ,
	\end{align*}
	where the last two equalities follow from definition, recalling that $\tau = \frac{\kappa+1}{\kappa m}.$ Therefore, Lemma \ref{lem:neg_root_abs_upper} holds. 
\end{proof}
}

\subsubsection{Lemmas for bounding the derivative of the rejection probability}

\begin{lem}
\label{lem:deriv_G}
Consider the same setup and notations as in Lemma \ref{lem:rej_prob_U}. 
Suppose that $z \equiv \max_{1 \leq i \leq m} x_i$ and $|\theta_{m+1}| \in [\underline{\theta}, \overline{\theta}]$ for some $\underline{\theta} \ge m$. 
For any given $\nu\in \mathbb{R}$, define  
\[
	G(s) \equiv  \frac{s^\nu U(\bm{x}, s)}{(w + s)(z + s)}.
\]
Then for any $s \in (0, |\theta_{m+1}|]$, 
\begin{align*}
    \frac{\text{d} \log G(s)}{\text{d} s}
    & \ge 
	\frac{1}{s} 
    \left\{
		\nu -
		\left[ 
		1
		+ \frac{\overline{\theta}}{2(z + \overline{\theta})}
		+ 
        \frac{\overline{\theta}}{w + \overline{\theta}} 
		-\frac{\underline{\theta} - m}{2(\tau \underline{\theta} - 1)}
		\right]
    \right\}.
\end{align*}
\end{lem}

\hideif{
\begin{proof}[Proof of Lemma \ref{lem:deriv_G}]
Using the definition of $U(\bm{x}, s)$ in Lemma \ref{lem:rej_prob_U}, we have
\begin{align}
		\label{eq:lem:10-2}
        \begin{split}
		& \hspace{-15pt} \log G(s) \\
		& = 
		\left( \frac{m}{2} - 1 + \nu \right) \log s 
		- \frac{1}{2} \log P(\bm{x}, s) 
		- \frac{1}{2} \log Q(\bm{x}, s)
		-\log (w + s)
		- \log (z + s).
        \end{split}
\end{align}
	Since, by definition, 
	\[
		\frac{\partial P(\bm{x}, s)}{\partial s} = -\sum^m_{i=1} \frac{1 + \tau x_i}{(x_i + s)^2(x_i - \theta_{m+1})},
	\]
	we have
	\begin{align}
		\label{eq:lem:10-3}
		\frac{\partial \log P(\bm{x}, s)}{\partial s}
		& =
		\frac{1}{P(\bm{x}, s)} \frac{\partial P(\bm{x}, s)}{\partial s} 
		= -\frac{\sum^m_{i=1} \frac{1 + \tau x_i}{(x_i + s)^2(x_i - \theta_{m+1})}}{\sum^m_{i=1} \frac{1 + \tau x_i}{(x_i + s)(x_i - \theta_{m+1})}}.
	\end{align}
	Recall that $z = \max_{1\leq i\leq m} x_i$. It follows that $\frac{1}{z + s}  \leq \frac{1}{x_i + s}$ for any $1 \le i \le m$. Thus,
	\begin{align}
		\label{eq:lem:10-3b}
		-\frac{\partial \log P(\bm{x}, s)}{\partial s}
		& 
		= \frac{\sum^m_{i=1} \frac{1 + \tau x_i}{(x_i + s)^2(x_i - \theta_{m+1})}}{\sum^m_{i=1} \frac{1 + \tau x_i}{(x_i + s)(x_i - \theta_{m+1})}}
		\geq \frac{\frac{1}{z+s}\sum^m_{i=1} \frac{1 + \tau x_i}{(x_i + s)(x_i - \theta_{m+1})}}{\sum^m_{i=1} \frac{1 + \tau x_i}{(x_i + s)(x_i - \theta_{m+1})}}
		= \frac{1}{z+s}.
	\end{align}
Since, by definition, $\log Q(\bm{x}, s) = \sum^m_{i=1} \log (x_i + s)$, we have
	\begin{align}
		\label{eq:lem:10-4}
		\frac{\partial \log Q(\bm{x}, s)}{\partial s}
		& = \sum^m_{i=1} \frac{1}{x_i + s}.
	\end{align}
	\item Combining the results from \eqref{eq:lem:10-3}--\eqref{eq:lem:10-4} into \eqref{eq:lem:10-2}, we have
	\begin{align}\label{eq:lem:10b-5}
		\frac{\text{d} \log G(s)}{\text{d} s}
		& = 
		\left( \frac{m}{2} - 1 + \nu \right) \frac{1}{s}
		- \frac{1}{2} \frac{\partial \log P(\bm{x}, s)}{\partial s}
		- \frac{1}{2}\frac{\partial \log Q(\bm{x}, s)}{\partial s}
		- \frac{1}{w+s}
		- \frac{1}{z+s} \nonumber
		\\
		& \geq \left( \frac{m}{2} - 1 + \nu \right) \frac{1}{s}
		+ \frac{1}{2}\frac{1}{z + s}
		- \frac{1}{2}\sum^m_{i=1} \frac{1}{x_i + s}
		- \frac{1}{w+s}
		- \frac{1}{z+s} \notag	\\
		& =  \left( \frac{m}{2} - 1 + \nu \right) \frac{1}{s}
		- \frac{1}{2}\frac{1}{z + s}
		- \frac{1}{2}\sum^m_{i=1} \frac{1}{x_i + s}
		- \frac{1}{w+s}	\notag	\\
		& = \frac{1}{2} \sum^m_{i=1} \left( \frac{1}{s} - \frac{1}{x_i + s} \right)
		- \frac{1-\nu}{s}
		- \frac{1}{2(z+s)}
		- \frac{1}{w+s}	\notag	\\
		& = \frac{1}{2s} \sum^m_{i=1} \frac{x_i}{x_i + s}
		- \frac{1-\nu}{s}
		- \frac{1}{2(z+s)}
		- \frac{1}{w+s}
        \nonumber
        \\
        & = \frac{1}{2s} 
    \left[ \sum^m_{i=1} \frac{x_i}{x_i + s}
		- 2(1-\nu)
		- \frac{s}{z+s}
		- \frac{2s}{w+s}
    \right]. 
	\end{align} 
Below we bound the terms in \eqref{eq:lem:10b-5}. 

We first bound the first term in \eqref{eq:lem:10b-5}. 
Because $s \leq |\theta_{m+1}| = -\theta_{m+1}$, we have 
\begin{equation*}
	\sum^m_{i=1} \frac{x_i}{x_i + s}
	\geq 
	\sum^m_{i=1} \frac{x_i}{x_i - \theta_{m+1}}.
\end{equation*}
Next, from \eqref{eq:constraint_neg_root} and \eqref{eq:x_x_tm},
we know that 
\begin{align*}
	1 & 
	= 
	\left( \frac{1}{\theta_{m+1}} + \tau \right)  \sum^m_{i=1} \frac{ x_i}{x_i - \theta_{m+1}} - \frac{m}{\theta_{m+1}}. 
\end{align*}
Consequently, 
\begin{align*}
	\sum^m_{i=1} \frac{x_i}{x_i + s}
	& \geq \sum^m_{i=1} \frac{x_i}{x_i - \theta_{m+1}}
	= 
    \frac{ 1 +  \frac{m}{\theta_{m+1}}}{ \frac{1}{\theta_{m+1}} + \tau}
	= \frac{-\tm - m}{-\tau \tm - 1}		
	= \frac{|\tm| - m}{\tau |\tm| - 1}		
	\\
    & = \tau^{-1} - \tau^{-1} \frac{m\tau-1}{\tau |\tm| - 1}. 
\end{align*}
Because, by definition, $m \tau = \frac{\kappa+1}{\kappa}>1$,  using the condition that $\tau |\theta_{m+1}| \ge \tau \underline{\theta} \ge \tau m > 1$, we then have 
\begin{align*}
	\sum^m_{i=1} \frac{x_i}{x_i + s}
	& 
	\ge \tau^{-1} - \tau^{-1} \frac{m\tau-1}{\tau |\tm| - 1}
	\ge 
	\tau^{-1} - \tau^{-1} \frac{m\tau-1}{\tau \underline{\theta} - 1}
	= \frac{ \underline{\theta} - m}{\tau \underline{\theta} - 1}. 
\end{align*}

We then bound the third term in \eqref{eq:lem:10b-5}.
Because $\frac{s}{z+s}$ is increasing in $s>0$
and $s \le |\tm| \le \overline{\theta}$, we have
\begin{align*}
    \frac{s}{z+s}
    \leq  \frac{|\tm|}{z + |\tm|}		 \leq 
    \frac{\overline{\theta}}{z + \overline{\theta}}. 
\end{align*}

We finally bound the last term in \eqref{eq:lem:10b-5}. 
By the same logic as before, we have 
\begin{align*}
    \frac{s}{w+s}
    \leq \frac{|\tm|}{w + |\tm|}		 \leq 
    \frac{\overline{\theta}}{w + \overline{\theta}}.
\end{align*}

From the above, we have 
\begin{align*}
    \frac{\text{d} \log G(s)}{\text{d} s}
    & \ge 
    \frac{1}{2s} 
    \left[ \sum^m_{i=1} \frac{x_i}{x_i + s}
		- 2(1-\nu)
		- \frac{s}{z+s}
		- \frac{2s}{w+s}
    \right]\\
    & \ge 
    \frac{1}{2s} 
    \left[ \frac{ \underline{\theta} - m}{\tau \underline{\theta} - 1}
		- 2(1-\nu)
		- \frac{\overline{\theta}}{z + \overline{\theta}}
		- 
        \frac{2\overline{\theta}}{w + \overline{\theta}}
    \right]
	\\
	& = 
	\frac{1}{s} 
    \left[ 
		\frac{\underline{\theta} - m}{2(\tau \underline{\theta} - 1)}
		- 1+\nu
		- \frac{\overline{\theta}}{2(z + \overline{\theta})}
		- 
        \frac{\overline{\theta}}{w + \overline{\theta}}
    \right]
	\\
	& = 
	\frac{1}{s} 
    \left\{
		\nu -
		\left[ 
		1
		+ \frac{\overline{\theta}}{2(z + \overline{\theta})}
		+ 
        \frac{\overline{\theta}}{w + \overline{\theta}} 
		-\frac{\underline{\theta} - m}{2(\tau \underline{\theta} - 1)}
		\right]
    \right\}.
\end{align*}
Therefore, Lemma \ref{lem:deriv_G} holds. 
\end{proof}
}

\begin{lem}	\label{lem:18}
	For any $\nu < 1$ and $\zeta > 0$,
	\[
		\int^{\zeta}_0
			\frac{s^{-\nu}}{\sqrt{\zeta - s}} \ \mathrm{d}s
			= \sqrt{\pi} \zeta^{\frac{1}{2}-\nu} \frac{\Gamma(1-\nu)}{\Gamma (\frac{3}{2} - \nu)}.
	\]
	\end{lem}

\hideif{
	\begin{proof}[Proof of Lemma \ref{lem:18}]
	To begin with, using the properties of beta and gemma functions,
	\begin{equation}
		\label{eq:lem:18-1}
		\int^1_0 \frac{x^{-\nu}}{\sqrt{1-x}} \ \text{d}x
		= \frac{\Gamma(1-\nu) \Gamma(\frac{1}{2})}{\Gamma (\frac{3}{2} - \nu)}
		= \frac{\Gamma(1-\nu) \sqrt{\pi}}{\Gamma (\frac{3}{2} - \nu)}.
	\end{equation}
	The integral \eqref{eq:lem:18-1} is convergent if $1 - \nu > 0$, or equivalently, $\nu < 1$.  
	By change of variables, let $s = \zeta x$. Hence, $\text{d}s = \zeta\text{d}x$. Therefore,
	\begin{align*}
		\int^{\zeta}_0
			\frac{s^{-\nu}}{\sqrt{\zeta - s}} \ \text{d}s
			& = 
			\int^1_0	\frac{\zeta^{-\nu} x^{-\nu}}{\sqrt{\zeta - \zeta x}} \zeta \ \text{d}x	
			= \zeta^{\frac{1}{2}-\nu} \int^1_0 \frac{x^{-\nu}}{\sqrt{1-x}} \ \text{d} x
			=\sqrt{\pi} \zeta^{\frac{1}{2}-\nu} \frac{\Gamma(1-\nu)}{\Gamma (\frac{3}{2} - \nu)},
	\end{align*}
	where the last equality uses \eqref{eq:lem:18-1}.
    Therefore, Lemma \ref{lem:18} holds. 
	\end{proof}
}

\begin{lem}
\label{lem:F_bound_w_min}
Consider the same setup and notations in Lemma \ref{lem:rej_prob_U}, and adopt the notations from Lemmas \ref{lem:deri_PQ} and \ref{lem:deri_U}. 
Suppose that $z \equiv \max_{1 \leq i \leq m} x_i > 0$, $w \equiv \min_{1 \leq i \leq m} x_i$, $|\theta_{m+1}| \le \overline{\theta}$ for some positive $\overline{\theta}$, and $\tau \le 1/2$. 
Define 
\begin{align*}
    F(s) \equiv \frac{s^{-\nu} (w+s) (z+s) \tilde h(\bm{x}, s)}{\sqrt{|\tm| - s}}, 
\end{align*}
and 
\begin{align*}
	C_1 & \equiv 
	(1-2\tau) \frac{wz-\overline{\theta}^2}{(w+z+2\overline{\theta})\overline{\theta}} +  
    \frac{2 \frac{wz-\overline{\theta}^2}{(w+z+2\overline{\theta})\overline{\theta}} + 1}{z}
    - \tau.
\end{align*}

\begin{itemize}
	\item[(a)] 
	For any $s\in (0, |\tm|)$, 
	\begin{align*}
	F(s) & \ge
	\frac{Bs^{-\nu} }{\sqrt{|\tm| - s}} 
	\left[ |\tm| C_1
    + (1-\tau)s \right].
	\end{align*}

	\item[(b)] 
	If 
	\begin{align*}
		\nu < \min\left\{ \frac{3}{2} - \frac{1}{2(1+\frac{C_1}{1-\tau})}, \ 1 \right\}
	\end{align*}
	then 
	\begin{align*}
		\int_0^{|\tm|} 
		\frac{s^{-\nu} }{\sqrt{|\tm| - s}} 
		\left[ |\tm| C_1
    	+ (1-\tau)s \right] \ \text{d} s > 0. 
	\end{align*}
\end{itemize}
\end{lem}

\hideif{
\begin{proof}[Proof of Lemma \ref{lem:F_bound_w_min}(a)]
From the definition of $\tilde{h}(\bm{x}, s)$ in Lemma \ref{lem:deri_U}, we have
	\begin{align}
		\label{eq:lem:11b-1}
		(w + s)(z + s) \tilde h(\bm{x}, s)
		& = 
		A + B s 
		+ \frac{2(A + \frac{s - \theta_{m+1}}{2} B)}{(w+s)(z+s)P(\bm{x}, s)}
		- \frac{2\tau s(A + \frac{s - \theta_{m+1}}{2} B)}{(w+s)(z+s)P(\bm{x}, s)}.
	\end{align}
In the following, we bound \eqref{eq:lem:11b-1} from below. 

First, because $z = \max_{1 \leq i \leq m} x_i$,
we have
\begin{align}
		(w+s)(z+s) P(\bm{x},s)
		& =(w+s)(z+s) \sum^m_{i=1} \frac{1 + \tau x_i}{(x_i + s)(x_i - \tm)}  \\
        & =(w+s) \sum^m_{i=1} \left( \frac{z+s}{x_i + s}\cdot\frac{1 + \tau x_i}{x_i - \tm} \right)
        \notag  \\
		& \geq (w+s) \sum^m_{i=1} \frac{1+\tau x_i}{x_i - \tm} \notag \\
        & = w + s	\geq s,     \label{eq:lem:11b-2}
\end{align}
where the last equality follows from \eqref{eq:constraint_neg_root} and the last inequality follows from the fact that $w \geq 0$.

Second, because $w = \min_{1 \leq i \leq m} x_i$, we have 
\begin{align}
		(w+s)(z+s) P(\bm{x},s)
		& =(w+s)(z+s) \sum^m_{i=1} \frac{1 + \tau x_i}{(x_i + s)(x_i - \tm)}  \notag \\
        & =
        (z+s) \sum^m_{i=1} 
        \left( \frac{w+s}{x_i + s}
        \cdot
        \frac{1 + \tau x_i}{x_i - \tm}\right)
        \notag  \\
		& \leq (z+s) \sum^m_{i=1} \frac{1+\tau x_i}{x_i - \tm}	= z + s,  \label{eq:lem:11b-3}
\end{align}   
where the last equality follows from \eqref{eq:constraint_neg_root}. 

Third, from the definition of $A$ and $B$ in Lemma \ref{lem:deri_PQ}, we know that 
\begin{align}\label{eq:lem:11b-5}
		A + \frac{s - \theta_{m+1}}{2} B
		&
		= wz - \tm^2 + \frac{s - \tm}{2} (z + w - 2\tm)	
        \nonumber
        \\
        & =
         wz - |\tm|^2 + \frac{s + |\tm|}{2} (z + w + 2|\tm|)
        \nonumber
        \\
        & = 
        s \left( \frac{z+w}{2} + |\theta_{m+1}| \right)
        +  wz - |\tm|^2 + |\tm|^2 + |\tm| \frac{z+w}{2}
        \nonumber
        \\
        & = s \left( \frac{z+w}{2} + |\theta_{m+1}| \right)
        +  z\left( w + \frac{|\tm|}{2} \right) +  \frac{w|\tm|}{2}
        \nonumber
        \\
        & \ge 0. 
\end{align}
In addition, 
because $\frac{A + \frac{s - \theta_{m+1}}{2} B}{z+s}
= 
\frac{(A + \frac{|\tm|}{2}B) + \frac{B}{2} s}{z+s}$ is monotone in $s$, we have, for any $s\in [0, |\tm|]$, 
\begin{align}\label{eq:A_B_z_s}
    \frac{A + \frac{s - \theta_{m+1}}{2} B}{z+s}
    & \ge 
    \min\left\{ 
    \frac{A + \frac{|\tm|}{2}B}{z}, 
    \frac{A + |\tm|B}{z+|\tm|}
    \right\}
	 \nonumber
	 \\
	 & 
	\ge 
	\min\left\{ 
    \frac{A + \frac{|\tm|}{2}B}{z}, 
    \frac{A + |\tm|B}{z+\overline{\theta}}
    \right\},
\end{align}
where the last inequality holds due to \eqref{eq:lem:11b-5} and that $|\tm| \le \overline{\theta}$. 

Fourth, from the previous three parts, we have 
\begin{align}
    \label{eq:wz_h_tilde_bound}
    & \quad \ (w + s)(z + s) \tilde h(\bm{x}, s)
    \nonumber
    \\
    & 
    = 
    A + B s 
    + \frac{2(A + \frac{s - \theta_{m+1}}{2} B)}{(w+s)(z+s)P(\bm{x}, s)}
    - \frac{2\tau s(A + \frac{s - \theta_{m+1}}{2} B)}{(w+s)(z+s)P(\bm{x}, s)}
    \nonumber
    \\
    & \geq 
    A + Bs 
    + \frac{2(A + \frac{s - \theta_{m+1}}{2} B)}{z + s}
    - 2 \tau \left(A + \frac{s - \theta_{m+1}}{2} B
		\right)
    \qquad \qquad \ \  \text{using \eqref{eq:lem:11b-2}, \eqref{eq:lem:11b-3}, and \eqref{eq:lem:11b-5}
    }
    \nonumber
    \\
    & \ge 
    A + Bs 
    + 
    2\cdot\min\left\{ 
    \frac{A + \frac{|\tm|}{2}B}{z}, 
    \frac{A + |\tm|B}{z+\overline{\theta}}
    \right\}
    - 2 \tau \left(A + \frac{s - \theta_{m+1}}{2} B
		\right)
    \nonumber
    \qquad
    \text{using \eqref{eq:A_B_z_s}}
    \\
    & = 
    (1-2\tau) A +  2\cdot\min\left\{ 
    \frac{A + \frac{|\tm|}{2}B}{z}, 
    \frac{A + |\tm|B}{z+\overline{\theta}}
    \right\}
    - \tau |\tm| B + (1-\tau)Bs
    \nonumber
    \\
    & = 
    |\tm| B \cdot
    \left[ 
    (1-2\tau) \frac{A}{B|\tm|} +  2\cdot\min\left\{ 
    \frac{\frac{A}{B|\tm|} + \frac{1}{2}}{z}, 
    \frac{\frac{A}{B|\tm|} + 1}{z+\overline{\theta}}
    \right\}
    - \tau 
    \right]
    + (1-\tau)Bs.
\end{align}
Note that 
\begin{align*}
    \frac{A}{B|\tm|}
    & = 
    \frac{wz-|\tm|^2}{(w+z+2|\tm|)|\tm|}
    = 
    \frac{wz}{(w+z+2|\tm|)|\tm|}
    - 
    \frac{|\tm|}{w+z+2|\tm|}
	\\
	& \ge 
	\frac{wz}{(w+z+2 \overline{\theta})\overline{\theta}}
    - 
    \frac{\overline{\theta}}{w+z+2\overline{\theta}}
	= 
	\frac{wz-\overline{\theta}^2}{(w+z+2\overline{\theta})\overline{\theta}}. 
\end{align*}
Because $\tau \le 1/2$, 
we then have 
\begin{align*}
	& \quad \ (1-2\tau) \frac{A}{B|\tm|} +  2\cdot\min\left\{ 
    \frac{\frac{A}{B|\tm|} + \frac{1}{2}}{z}, 
    \frac{\frac{A}{B|\tm|} + 1}{z+\overline{\theta}}
    \right\}
    - \tau 
	\\
	& \ge
	(1-2\tau) \frac{wz-\overline{\theta}^2}{(w+z+2\overline{\theta})\overline{\theta}} +  2\cdot\min\left\{ 
    \frac{\frac{wz-\overline{\theta}^2}{(w+z+2\overline{\theta})\overline{\theta}} + \frac{1}{2}}{z}, 
    \frac{\frac{wz-\overline{\theta}^2}{(w+z+2\overline{\theta})\overline{\theta}} + 1}{z+\overline{\theta}}
    \right\}
    - \tau. 
\end{align*}
Note that 
\begin{align*}
	& \frac{\frac{wz-\overline{\theta}^2}{(w+z+2\overline{\theta})\overline{\theta}} + \frac{1}{2}}{z} \le  
    \frac{\frac{wz-\overline{\theta}^2}{(w+z+2\overline{\theta})\overline{\theta}} + 1}{z+\overline{\theta}}
	\\
	\Longleftrightarrow \quad 
	& 
	\frac{(wz-\overline{\theta}^2) z}{(w+z+2\overline{\theta})\overline{\theta}}  + \frac{z}{2} 
	+
	\frac{(wz-\overline{\theta}^2)\overline{\theta}}{(w+z+2\overline{\theta})\overline{\theta}}  + \frac{\overline{\theta}}{2}
	\le 
	\frac{(wz-\overline{\theta}^2) z}{(w+z+2\overline{\theta})\overline{\theta}}  + z
	\\
	\Longleftrightarrow \quad 
	&   
	\frac{z}{2} \ge 
	\frac{wz-\overline{\theta}^2}{w+z+2\overline{\theta}}  + \frac{\overline{\theta}}{2}
	= 
	\frac{wz-\overline{\theta}^2 + (w+z+2\overline{\theta}) \frac{\overline{\theta}}{2}}{w+z+2\overline{\theta}}
	= 
	\frac{wz + (w+z) \frac{\overline{\theta}}{2} }{w+z+2\overline{\theta}}
	\\
	\Longleftrightarrow \quad 
	& 
	(w+z)z + 2 z \overline{\theta} \ge 2 wz + (w+z) \overline{\theta}
	\\
	\Longleftrightarrow \quad 
	&  z(z-w) + \overline{\theta}(z - w) \ge 0, 
\end{align*}
which must hold because $z \ge w \ge 0$ and $\overline{\theta} \ge 0$.
Thus, we have 
\begin{align*}
	& \quad \ (1-2\tau) \frac{A}{B|\tm|} +  2\cdot\min\left\{ 
    \frac{\frac{A}{B|\tm|} + \frac{1}{2}}{z}, 
    \frac{\frac{A}{B|\tm|} + 1}{z+\overline{\theta}}
    \right\}
    - \tau 
	\\
	& \ge
	(1-2\tau) \frac{wz-\overline{\theta}^2}{(w+z+2\overline{\theta})\overline{\theta}} +  
    \frac{2 \frac{wz-\overline{\theta}^2}{(w+z+2\overline{\theta})\overline{\theta}} + 1}{z}
    - \tau\\
	& = C_1, 
\end{align*}
where the last equality holds by definition.
From \eqref{eq:wz_h_tilde_bound}, we then have
\begin{align*}
(w + s)(z + s) \tilde h(\bm{x}, s)
    & 
    \ge 
    |\tm| B
   C_1
    + (1-\tau)Bs
\end{align*}
This immediately implies that 
\begin{align*}
    F(s) & \equiv \frac{s^{-\nu} (w+s) (z+s) \tilde h(\bm{x}, s)}{\sqrt{|\tm| - s}}
	\ge 
	\frac{s^{-\nu} }{\sqrt{|\tm| - s}} 
	\left[ |\tm| B C_1
    + (1-\tau)Bs \right]\\
	& = 
	\frac{Bs^{-\nu} }{\sqrt{|\tm| - s}} 
	\left[ |\tm| C_1
    + (1-\tau)s \right].
\end{align*}

From the above, Lemma \ref{lem:F_bound_w_min}(a) holds. 
\end{proof}
}

\hideif{
\begin{proof}[Proof of Lemma \ref{lem:F_bound_w_min}(b)]
From Lemma \ref{lem:18}, for $\nu < 1$, we have
\begin{align*}
	& \quad \ \int_0^{|\tm|} 
		\frac{s^{-\nu} }{\sqrt{|\tm| - s}} 
		\left[ |\tm| C_1
		+ (1-\tau)s \right] \ \text{d} s
	\\
	& = 
	|\tm| C_1 \cdot 
	\int_0^{|\tm|} 
	\frac{s^{-\nu}}{\sqrt{|\tm| - s}}\ \text{d} s
	+ (1-\tau)
	\int_0^{|\tm|} \frac{s^{-(\nu-1)}}{\sqrt{|\tm| - s}}\ \text{d} s
	\\
	& = 
	C_1 
	\sqrt{\pi} |\tm|^{\frac{3}{2}-\nu} \frac{\Gamma(1-\nu)}{\Gamma(\frac{3}{2}-\nu)}	
	+ (1-\tau)
	\sqrt{\pi} |\tm|^{\frac{3}{2}-\nu} \frac{\Gamma(2-\nu)}{\Gamma(\frac{5}{2}-\nu)}\\
	& = 
	\sqrt{\pi} |\tm|^{\frac{3}{2}-\nu}
	\frac{\Gamma(1-\nu)}{\Gamma(\frac{3}{2}-\nu)}
	\left[ 
	C_1 
	+ (1-\tau) \frac{1-\nu}{\frac{3}{2}-\nu}
	\right],
\end{align*}
where the last equality uses the property of gamma functions. 
Note that, by definition,  
\begin{align}\label{eq:ratio_wz_w_z}
	\frac{wz-\overline{\theta}^2}{(w+z+2\overline{\theta})\overline{\theta}}
	+ \frac{1}{2}
	& = 
	\frac{wz
	+ \frac{1}{2} (w+z)\overline{\theta}
	}{(w+z+2\overline{\theta})\overline{\theta}} > 0,
\end{align}
which immediately implies that 
\begin{align*}
	C_1 + (1-\tau) 
	& = 
	(1-2\tau) \frac{wz-\overline{\theta}^2}{(w+z+2\overline{\theta})\overline{\theta}} +  
    \frac{2 \frac{wz-\overline{\theta}^2}{(w+z+2\overline{\theta})\overline{\theta}} + 1}{z}
    - \tau + (1-\tau) \\
	& = 
	(1-2\tau) \left[ \frac{wz-\overline{\theta}^2}{(w+z+2\overline{\theta})\overline{\theta}} + 1\right] +  
    \frac{2 \frac{wz-\overline{\theta}^2}{(w+z+2\overline{\theta})\overline{\theta}} + 1}{z}>0.
\end{align*}
Consequently, when $\nu < 1$ and $\tau < \frac{1}{2}$, we have 
\begin{align*}
	& C_1 
	+ (1-\tau) \frac{1-\nu}{\frac{3}{2}-\nu} > 0   \\
	\Longleftrightarrow  \quad & 
	C_1 
	+ (1-\tau) \left( 1 - \frac{\frac{1}{2}}{\frac{3}{2}-\nu} \right) > 0  \\
	\Longleftrightarrow \quad & 
	\frac{C_1}{1-\tau} + 1 > \frac{\frac{1}{2}}{\frac{3}{2}-\nu}
	\\
	\Longleftrightarrow \quad 
	& 
	\frac{3}{2} - \nu > \frac{1}{2(1+\frac{C_1}{1-\tau})}
	\\
	\Longleftrightarrow \quad 
        &
	\nu < \frac{3}{2} - \frac{1}{2(1+\frac{C_1}{1-\tau})},
\end{align*}
where the second last equivalence holds because $\frac{C_1}{1-\tau} + 1 = \frac{C_1+1-\tau}{1-\tau} > 0$. 
Therefore, Lemma \ref{lem:F_bound_w_min}(b) holds. 
\end{proof}
	
}

\begin{lem}
\label{lem:F_bound_w_gen}
Consider the same setup and notations in Lemma \ref{lem:rej_prob_U}, and adopt the notations from Lemmas \ref{lem:deri_PQ} and \ref{lem:deri_U}. 
Suppose that $z \equiv \max_{1 \leq i \leq m} x_i > 0$, $|\theta_{m+1}| \le \overline{\theta}$ for some positive $\overline{\theta}$, and $\tau \le 1/2$. 
Define 
\begin{align*}
    F(s) \equiv \frac{s^{-\nu} (w+s) (z+s) \tilde h(\bm{x}, s)}{\sqrt{|\tm| - s}}, 
\end{align*}
and 
\begin{align*}
	C_2 & \equiv 
	(1-2\tau) \frac{wz-\overline{\theta}^2}{(w+z+2\overline{\theta})\overline{\theta}} 
    - \tau.
\end{align*}

\begin{itemize}
	\item[(a)] 
	For any $s\in (0, |\tm|)$, 
	\begin{align*}
	F(s) & \ge
	\frac{Bs^{-\nu} }{\sqrt{|\tm| - s}} 
	\left[ |\tm| C_2
    + (1-\tau)s \right].
	\end{align*}

	\item[(b)] 
	If 
	\begin{align*}
		\nu < \min\left\{ \frac{3}{2} - \frac{1}{2(1+\frac{C_2}{1-\tau})}, \ 1 \right\}
	\end{align*}
	then 
	\begin{align*}
		\int_0^{|\tm|} 
		\frac{s^{-\nu} }{\sqrt{|\tm| - s}} 
		\left[ |\tm| C_2
    	+ (1-\tau)s \right] \ \text{d} s > 0. 
	\end{align*}
\end{itemize}
\end{lem}

\hideif{
\begin{proof}[Proof of Lemma \ref{lem:F_bound_w_gen}]
The proof of
Lemma \ref{lem:F_bound_w_gen}
follows almost the same steps as the proof of Lemma \ref{lem:F_bound_w_min}, except for a few differences discussed below. 
For (a), we no longer have \eqref{eq:lem:11b-3} and instead bound the second term in \eqref{eq:lem:11b-1} by zero. 
That is, 
$\frac{2(A + \frac{s - \theta_{m+1}}{2} B)}{(w+s)(z+s)P(\bm{x}, s)}\ge 0$, 
which follows immediately from \eqref{eq:lem:11b-5}.
For (b), we also have 
\begin{align*}
	C_2 + (1-\tau) 
	& = 
	(1-2\tau) \frac{wz-\overline{\theta}^2}{(w+z+2\overline{\theta})\overline{\theta}} 
    - \tau + (1-\tau) \\
        & = 
	(1-2\tau) \left[ \frac{wz-\overline{\theta}^2}{(w+z+2\overline{\theta})\overline{\theta}} + 1\right] >0,
\end{align*}
which follows immediately from \eqref{eq:ratio_wz_w_z}. 
\end{proof}

}

\subsubsection{Bounds on the derivative of the rejection probability}

\begin{lem}[Modified from Lemma 1 of \citet{bakirov1989jms}] \label{lem:integral_product_bound}
Let $F$ and $G$ be two functions of $x$ that satisfy the following properties, where $\tm$ can be any given negative number in $\mathbb{R}$: 
\begin{enumerate}
	\item[(i)] $G(x)$ is not identically equal to zero, continuous and nonnegative on $(0, |\tm|]$ and $\frac{\deri}{\deri x}G(x)$ is continuous and nonnegative 
    on 
    $(0, |\tm|)$. 
	\item[(ii)] $F(x)$ is continuous 
    on $(0, |\tm|)$, and 
    there exists some $x_0 < |\tm|$ such that 
    $F(x)(x-x_0) \geq 0$ and $F(x) > 0$ for $x_0 < x < |\tm|$. 
\end{enumerate}
If $\int^{|\tm|}_0 F(x) \ \text{d} x$ converges and is positive,
then $\int^{|\tm|}_0 F(x) G(x) \ \text{d} x$ converges, and 
\[
	\int^{|\tm|}_0 F(x) G(x) \ \text{d} x > 0.
\]
\end{lem}

\hideif{
\begin{proof}[Proof of Lemma \ref{lem:integral_product_bound}]
The result mainly follows from the proof of Lemma 1 in \citet{bakirov1989jms}, with $1$ in the integral limit replaced by $|\tm|$ and slightly revised conditions. 
For completeness purposes, we prove the lemma with the revised conditions below, although most of the proof are taken from \citet{bakirov1989jms}. \par

First, from condition (i), $G(\cdot)$ is nonnegative and monotone nondecreasing on $(0, |\tm|]$. Thus, $G(x)$ must have a nonnegative limit as $x \longrightarrow 0+$. 
Define the value of $G(\cdot)$ evaluated at $0$ as this limit, 
i.e., $G(0) \equiv \lim_{x \longrightarrow 0+} G(x) \ge 0$. 
Consequently, $G(\cdot)$ becomes a continuous and nonnegative function on $[0, |\tm|]$. 
In addition, there must exist a finite $M$ such that $0\le G(x)\le M$ for all $x\in [0, |\tm|]$.

Second, we prove that $\overline{F}(x) \equiv \int^{|\tm|}_x F(s) \ \text{d}s > 0$ for any $x \in [0, |\tm|)$. 
If condition (ii) holds for some $x_0\le 0$, then this holds obviously. 
Below we consider only the case where condition (ii) holds for some $x_0 \in (0, |\theta_{m+1}|)$. 
For any $x\ge x_0$, we obviously have $\int^{|\tm|}_x F(s) \ \text{d}s > 0$, since $F(x) > 0$ for $x>x_0$. 
For any $x<x_0$, we have 
\begin{align*}
    \int^{|\tm|}_x F(s) \ \text{d}s
    & = 
    \int^{|\tm|}_0 F(s) \ \text{d}s
    -
    \int^{x}_0 F(s) \ \text{d}s
    \ge \int^{|\tm|}_0 F(s) \ \text{d}s > 0, 
\end{align*}
where the second last inequality holds because $F(s)\le 0$ for $0< s \le x < x_0$. 

Third, we prove that $\int^{|\tm|}_0 F(x) G(x) \ \text{d} x$ converges. From condition (ii), there exists $0<x_1<|\tm|$ such that $F(x)\ge 0$ for all $x\in [x_1, |\tm|)$. From the first part, for any $x\in [x_1, |\tm|)$, we then have $0\le F(x)G(x) \le M F(x)$. This implies that 
$\int_{x_1}^{|\tm|} F(x)G(x) \ \deri x$ is bounded from the above by $ M\int_{x_1}^{|\tm|} F(x) \ \deri x$ and thus converges. 
We consider then two cases depending on the value of $x_0$ in condition (ii). 
\begin{itemize}
    \item Consider first the case where condition (ii) holds for some $x_0 \le 0$. 
    Let $x_2$ be any number in $(0, |\tm|)$.
    We then have $F(x) \ge 0$ for all $x \in (0, x_2]$.
    From the first part, for any $x \in (0, x_2]$, we then have 
    $0 \le F(x) G(x) \le M F(x)$.  This implies that 
$\int_{0}^{x_2} F(x)G(x) \ \deri x$ is bounded from the above by $ M\int_{0}^{x_2} F(x) \ \deri x$ and thus converges. 

    \item Consider then the case where condition (ii) holds for some $x_0 > 0$. Let $x_3 < x_0$ be a number in  $(0, |\tm|)$. We then have $F(x) \le 0$ for all $x \in (0, x_3]$.  From the first part, for any $x \in (0, x_3]$, we then have 
    $M F(x) \le F(x) G(x) \le 0$. This implies that 
$\int_{0}^{x_3} F(x)G(x) \ \deri x$ is bounded from the below by $ M\int_{0}^{x_3} F(x) \ \deri x$ and thus converges.
\end{itemize}
The above discussion then implies that $\int^{|\tm|}_0 F(x) G(x) \ \text{d} x$ converges.

Fourth, 
recalling the definition of $\overline{F}(x)$ in the second part, we have $\frac{\deri}{\deri x} \overline{F}(x) = - F(x)$. For any positive $\epsilon_1$ and $\epsilon_2$ such that $\epsilon_1+\epsilon_2 < |\tm|$,
using integration by parts,  we have 
\begin{align}\label{eq:FG_int_part}
    & \quad \int^{|\tm|-\epsilon_2}_{\epsilon_1} F(x) G(x)\ \text{d}x
    \nonumber
    \\
    & =
    - \int^{|\tm|-\epsilon_2}_{\epsilon_1} G(x) 
    \ \text{d} \overline{F}(x)   
    \nonumber
    \\
    & = 
    - \left. G(x) \overline{F}(x) \right|_{\epsilon_1}^{|\tm|-\epsilon_2}
    +  
    \int^{|\tm|-\epsilon_2}_{\epsilon_1} \overline{F}(x)
    \ \text{d} G(x)    \nonumber
    \\
    & =
    - G(|\tm|-\epsilon_2) \overline{F}(|\tm|-\epsilon_2) +  G(\epsilon_1) \overline{F}(\epsilon_1) 
    + 
    \int^{|\tm|-\epsilon_2}_{\epsilon_1} \frac{\deri}{\deri x} G(x) \overline{F}(x)
     \ \deri x.
\end{align}
From the first part and the condition that $\int^{|\tm|}_0 F(x) \ \text{d} x$ converges, we can know that  
\begin{align*}
    G(|\tm|-\epsilon_2) \overline{F}(|\tm|-\epsilon_2)
    = 
    G(|\tm|-\epsilon_2) \int^{|\tm|}_{|\tm|-\epsilon_2} F(s) \ \text{d}s
    \longrightarrow G(|\tm|) \cdot 0 = 0 
\end{align*}
as $\epsilon_2 \longrightarrow 0+$, and 
\begin{align*}
    G(\epsilon_1) \overline{F}(\epsilon_1)  
    \longrightarrow 
    G(0) \int^{|\tm|}_0 F(x) \ \text{d} x = G(0) \overline{F}(0) 
\end{align*}
as $\epsilon_1 \longrightarrow 0+$. From the third part, $\int^{|\tm|}_0 F(x) G(x) \ \text{d} x$ converges. 
These imply that $\int^{|\tm|}_{0} \frac{\deri}{\deri x} G(x) \overline{F}(x) \ \deri x$ converges. 
By letting $\epsilon_1$ and $\epsilon_2$ in \eqref{eq:FG_int_part} converge to zero from the right, we then have 
\begin{align}\label{eq:int_part0}
    \int^{|\tm|}_{0} F(x) G(x)\ \text{d}x
    & =
    G(0) \overline{F}(0)
    + 
    \int^{|\tm|}_{0} \frac{\deri}{\deri x} G(x) \overline{F}(x)
     \ \deri x.
\end{align}
From the second part, we know that $\overline{F}(x) > 0$ for any $x \in [0, |\tm|)$. 
From condition (i) and the first part, we know that $G(0)\ge 0$ and $\frac{\deri}{\deri x} G(x)$ is continuous and nonnegative for any $x \in (0, |\tm|)$. 
Thus, the right-hand side of \eqref{eq:int_part0} nonnegative, 
and it becomes zero if and only if $G(0)=0$ and $\frac{\deri}{\deri x} G(x) = 0$ for any $x \in (0, |\tm|)$, 
under which $G(\cdot)$ becomes a zero function on $[0, |\tm|]$. 
Because $G(x)$ is not identically equal to zero as in condition (i), the right-hand side of \eqref{eq:int_part0} must be positive. Therefore, $\int^{|\tm|}_{0} F(x) G(x)\ \text{d}x > 0$.

From the above, Lemma \ref{lem:integral_product_bound} holds. 
\end{proof}
}

\begin{lem}\label{lem:deri_rej_prob_bound}
	Consider the same setup and notations in Lemmas \ref{lem:rej_prob_U}--\ref{lem:deri_AB}. 
Let $z = x_{(m)}$, $w$ be any one of $(x_1, x_2, \ldots, x_m)$, $\overline{\theta} = m+z/\kappa$, 
$\underline{\theta} =  \max_{1\le k \le m} \theta(x_{(k)}, k)$, 
and 
\begin{align*}
	C \equiv 
	\begin{cases}
	(1-2\tau) \frac{wz-\overline{\theta}^2}{(w+z+2\overline{\theta})\overline{\theta}} +  
    \frac{2 \frac{wz-\overline{\theta}^2}{(w+z+2\overline{\theta})\overline{\theta}} + 1}{z}
    - \tau, & \text{if } w = x_{(1)}, \\
	(1-2\tau) \frac{wz-\overline{\theta}^2}{(w+z+2\overline{\theta})\overline{\theta}}
    - \tau, & \text{otherwise}.
\end{cases}
\end{align*}
Fix the value of $c$, the values of $\{x_i\}_{i=1}^m$ excluding $(z,w)$, and the value of $\tm$, 
and view $\bP[|T_m|>c]$ as a function of $z$. 
If $\tau \le 1/2$, $z>w > 0$, and 
\begin{align*}
	\frac{\overline{\theta}}{z + \overline{\theta}}
	+ 
	\frac{2\overline{\theta}}{w + \overline{\theta}} 
	-\frac{\underline{\theta} - m}{\tau \underline{\theta} - 1}
	+ \frac{1-\tau}{1-\tau+\min\{C, 0\}}
	<  
	1 
\end{align*}
then 
$
\frac{\partial \bP[|T_m|>c]}{\partial z} < 0. 
$
\end{lem}

\hideif{
\begin{proof}[Proof of Lemma \ref{lem:deri_rej_prob_bound}]
	Let 
	\begin{align}\label{eq:nu}
		\nu = 1+\frac{\overline{\theta}}{2(z + \overline{\theta})}
		+ 
		\frac{\overline{\theta}}{w + \overline{\theta}} 
		-\frac{\underline{\theta} - m}{2(\tau \underline{\theta} - 1)}, 
	\end{align}
	and 
	define $G(s)$ and $F(s)$ in  the same say as in 
	Lemmas \ref{lem:deriv_G}--\ref{lem:F_bound_w_gen}:
	\begin{align*}
		G(s) & = 
		\frac{s^{\nu} U(\bm{x},s)}{(w+s)(z+s)} 
           , \\
		F(s) & = \frac{s^{-\nu}(w+s)(z+s)\tilde h(\bm{x}, s) }{\sqrt{|\theta_{m+1}|-s}} 
	\end{align*}
        for $s \in (0, |\tm|)$.
	From Lemma \ref{lem:deri_rej_prob}, we have 
	\begin{align*}
		\frac{\partial \bP[|T_m|>c]}{\partial z}
		& = -\frac{\Delta}{2} L(\tilde h(\bm{x}, s))
		= 
		-\frac{\Delta}{2} 
		\frac{1}{\pi}
			\int^{|\theta_{m+1}|}_0
			\frac{\tilde h(\bm{x}, s) U(\bm{x},s)}{\sqrt{|\theta_{m+1}|-s}} \ \mathrm{d} s
		\\
		& = 
		-\frac{z-w}{2(z - \theta_{m+1})^2}
		\frac{1}{\pi}
			\int^{|\theta_{m+1}|}_0
			\frac{s^{\nu} U(\bm{x},s)}{(w+s)(z+s)}
			\frac{s^{-\nu}(w+s)(z+s)\tilde h(\bm{x}, s) }{\sqrt{|\theta_{m+1}|-s}} \ \mathrm{d} s
		\\
		& = 
		-\frac{z-w}{2(z - \theta_{m+1})^2}
		\frac{1}{\pi}
			\int^{|\theta_{m+1}|}_0
			G(s)
			F(s) \ \mathrm{d} s. 
	\end{align*}
	Note that $z>w$. To prove Lemma \ref{lem:deri_rej_prob_bound}, 
	it suffices to prove that $\int^{|\theta_{m+1}|}_0
	G(s)
	F(s) \ \mathrm{d} s > 0$ under the conditions in Lemma \ref{lem:deri_rej_prob_bound}.  

	First, from Lemma \ref{lem:neg_root_abs_lower}, $\underline{\theta} = \max_{1\le k \le m} \theta(x_{(k)}, k) \ge m$. 
    From Lemmas \ref{lem:neg_root_abs_lower}
	and 
	\ref{lem:neg_root_abs_upper}, 
	we must have $m \le \underline{\theta} \le |\tm| \le \overline{\theta}$. 
	
	Second, 
	by definition, we can know that $G(s)$ is positive and continuous on $(0, |\tm|]$.  
	Moreover, by the definition in \eqref{eq:nu} and using Lemma \ref{lem:deriv_G}, we can know that  
	$\frac{\text{d} \log G(s)}{\text{d} s} \ge 0$
	for $s \in (0, |\theta_{m+1}|]$. 
	These then imply that $\frac{\text{d} G(s)}{\text{d} s} \ge 0$ for $s \in (0, |\theta_{m+1}|]$.

	Third, by the definition of $C$ and using Lemmas \ref{lem:F_bound_w_min} and \ref{lem:F_bound_w_gen} for $s\in (0, |\tm|)$, we have 
	$F(s) \ge B \tilde{F}(s)$, where 
    \begin{align*}
		\tilde{F}(s) & \equiv
		\frac{s^{-\nu} }{\sqrt{|\tm| - s}} 
		\left[ |\tm| C
		+ (1-\tau)s \right].
	\end{align*}
	By the definition in \eqref{eq:nu} and the conditions in Lemma \ref{lem:deri_rej_prob_bound}, we have
	\begin{align*}
		 & 2(\nu - 1) + \frac{1-\tau}{1-\tau+\min\{C, 0\}}
		<  
		1     \\
		\Longrightarrow
        \qquad & 
		\nu < \frac{3}{2} - \frac{1}{2(1+\frac{\min\{C, 0\}}{1-\tau})}
		= 
		\min\left\{ \frac{3}{2} - \frac{1}{2(1+\frac{C}{1-\tau})}, \ 1 \right\}.
	\end{align*}
	From Lemmas \ref{lem:F_bound_w_min} and \ref{lem:F_bound_w_gen}, 
	we then have 
	$
	\int_0^{|\tm|} 
	\tilde{F}(s) \ \text{d} s > 0. 
	$
	Moreover, because $1-\tau>0$, 
	there must exists $s_0 < |\tm|$ such that 
	$\tilde{F}(s)\cdot (s-s_0)\ge 0$ for $s\in (0, |\tm|)$ and $F(s) >0$ for $s_0 < s < |\tm|$.  

	From the above, Lemma \ref{lem:integral_product_bound} and the fact that $B>0$ from its definition, we then have
	\begin{align*}
		\int^{|\theta_{m+1}|}_0
	G(s)
	F(s) \ \mathrm{d} s 
	\ge 
	B \cdot \int^{|\theta_{m+1}|}_0
	G(s)
	\tilde{F}(s) \ \mathrm{d} s 
	> 0. 
	\end{align*}
	Therefore, Lemma \ref{lem:deri_rej_prob_bound} holds. 
\end{proof}
}

\section{Proof of theorems and lemmas}

\subsection{Proof of Theorem \ref{thm:size_asymp}}

To prove Theorem \ref{thm:size_asymp}, we need the following lemma. 

\begin{lem}\label{lem:rej_prob_point_mass}
Under Assumption \ref{assu:normal}, 
for any $\delta \in \mathbb{R}$ and $c>0$, 
if $c\ne m^{-1/2}$ and $\{\sigma^2_j\}_{j=1}^{m+1}$ are not all zero, then 
$\bP[ (\nob_{m+1} - \overline\nob_m + \delta)^2 =  c^2 S_m^2 ]=0$. 
\end{lem}

\hideif{
\begin{proof}[Proof of Lemma \ref{lem:rej_prob_point_mass}]
If $\sigma_{m+1}>0$, then, by the law of iterated expectation, we must have 
\begin{align}\label{eq:conditioning}
    \bP[ (\nob_{m+1} - \overline\nob_m+ \delta)^2 =  c^2 S_m^2 ]
    & = \bE \left\{ \bP\left[ (\nob_{m+1} - \overline\nob_m+ \delta)^2 =  c^2 S_m^2 \mid \nob_1, \ldots, \nob_m \right] \right\} \notag \\
    & = 0, 
\end{align} 
where the last equality holds because $\nob_{m+1}$ is a continuous random variable and has zero probability mass at any real value. 
Below it suffices to consider the case where $\sigma_{m+1}=0$.

From the proof of Lemma \ref{lem:t_test_matrix}, 
$mc^2 S_m^2 - m(\nob_{m+1} - \overline\nob_m)^2 = \tilde{\bm{\nob}}^\top \bm{V} \tilde{\bm{\nob}}$, 
and 
$\nob_{m+1} -  \overline\nob_m = \bs{a}^\top  \tilde{\bm{\nob}}$, 
where $\bm{V}$ and $\tilde{\bm{\nob}}$ are defined the same as in Lemma \ref{lem:t_test_matrix}, 
and $\bs{a}^\top = (1, -m^{-1} \bs{1}_m^\top)$. 
Consequently, 
\begin{align*}
     \bP[ (\nob_{m+1} - \overline\nob_m + \delta)^2 =  c^2 S_m^2 ]
    \nonumber
    & =
    \bP[ m(\nob_{m+1} - \overline\nob_m + \delta)^2 =  mc^2 S_m^2 ] \\ 
    &  =
    \bP[ mc^2 S_m^2  - m(\nob_{m+1} - \overline\nob_m)^2 - 2 m \delta (\nob_{m+1} - \overline\nob_m) - m\delta^2 = 0]
    \nonumber
    \\
    & = 
    \bP\left[ \tilde{\bm{\nob}}^\top \bm{V} \tilde{\bm{\nob}} - 2m \delta \bs{a}^\top \tilde{\bm{\nob}}  - m \delta^2 = 0 \right]   \\
    & =
    \bP\left[ \bs{\xi}^\top \bs{D}\bm{V} \bs{D} \bs{\xi} - 2m \delta \bs{a}^\top \bs{D} \bs{\xi} - m \delta^2 = 0  \right], 
\end{align*}
where $\bs{D}$ is defined as in Lemma \ref{lem:character_sigma}, $\bs{\xi} \equiv (\xi_1, \xi_2, \ldots, \xi_{m+1})^\top$, 
and $\{\xi_i\}^{m+1}_{i=1}$ are i.i.d. standard normal random variables.
Let $\bs{D}\bm{V} \bs{D} = \bs{\Gamma} \bs{\Lambda} \bs{\Gamma}^\top$ be the eigendecomposition of $\bs{D}\bm{V} \bs{D}$, where $\bs{\Gamma}$ is an orthogonal matrix, $\bs{\Lambda}$ is a diagonal matrix with diagonal elements $\{\lambda_i\}^{m+1}_{i=1}$, and $\{\lambda_i\}^{m+1}_{i=1}$ are the eigenvalues of $\bs{D}\bm{V} \bs{D}$, or equivalently the root of the characteristic polynomial $f(\lambda)$ in Lemma \ref{lem:character_sigma}. 
Let $\bs{\zeta} = (\zeta_1, \ldots, \zeta_{m+1})^\top = \bs{\Gamma}^\top \bs{\xi}$. 
We can verify that $\{\zeta_i\}^{m+1}_{i=1}$ are i.i.d. standard normal random variables, and $\bs{\xi} =\bs{\Gamma}\bs{\zeta}$. 
We then have 
\begin{align}\label{eq:prob_zero_lambda}
    & \quad \ \bP[ (\nob_{m+1} - \overline\nob_m + \delta)^2 =  c^2 S_m^2 ]
    \nonumber
    \\
    &
    =
    \bP\left[ \bs{\xi}^\top \bs{D}\bm{V} \bs{D} \bs{\xi} - 2m \delta \bs{a}^\top \bs{D} \bs{\xi} - m \delta^2 = 0  \right]
    =
    \bP\left[ \bs{\zeta}^\top \bs{\Lambda} \bs{\zeta} - 2m \delta \bs{a}^\top \bs{D} \bs{\Gamma}\bs{\zeta} - m \delta^2 = 0  \right]
    \nonumber
    \\
    & = 
    \bP\left[ \sum^{m+1}_{i=1} \lambda_i \zeta_i^2 - 2m \delta \bs{a}^\top \bs{D} \bs{\Gamma}\bs{\zeta} - m \delta^2 = 0  \right].
\end{align}
Because $\sigma_{m+1}=0$, the characteristic polynomial $f(\lambda)$ simplifies to 
\begin{align*}
    f(\lambda)
    & = 
    - \lambda
        \prod_{i=1}^m ( \kappa\sigma_i^2 - \lambda)  
        +
        \frac{\kappa+1}{m} \lambda
        \cdot 
        \sum_{i=1}^m 
        \left[
        \sigma_i^2 \prod_{j\ne i, 1\le j \le m} ( \kappa\sigma_j^2 - \lambda) 
        \right],
\end{align*}
where $\kappa = \frac{mc^2}{m-1}$. 
Let $\sigma_{(1)} \le \sigma_{(2)} \le \ldots \le \sigma_{(m)}$ be the sorted values of  $\{\sigma_i\}^{m}_{i=1}$. We then have 
\begin{align*}
    (-1)^{m-1} f(\kappa \sigma_{(m)}^2) 
    & =
    \frac{\kappa+1}{m} \lambda
    \cdot 
    \sigma_{(m)}^2 \kappa^{m-1} \prod_{j\ne m, 1\le j \le m} \{ \sigma_{(m)}^2 -\sigma_{(j))}^2 \} \ge 0, 
    \\
    (-1)^{m-1} f(\kappa \sigma_{(m-1)}^2) 
    & =
    \frac{\kappa+1}{m} \lambda
    \cdot 
    \sigma_{(m-1)}^2 \kappa^{m-1} \prod_{j\ne m-1, 1\le j \le m} \{ \sigma_{(m-1)}^2 -\sigma_{(j))}^2 \} \le 0,
\end{align*}
which imply that $f(\lambda)$ must have a root in $[\kappa \sigma_{(m-1)}^2, \kappa \sigma_{(m)}^2]$. 
If $\sigma_{(m-1)}^2>0$, then $f(\lambda)$ must have a positive root. 
By a conditioning argument similar to \eqref{eq:conditioning}, 
we can know that the quantity in \eqref{eq:prob_zero_lambda} must be zero. 
If $\sigma_{(m-1)}^2=0$, then we have must $\sigma_{(j)}^2=0$ for all $1\le j < m$, under which $f(\lambda)$ further simplifies to 
\begin{align*}
    f(\lambda)
    & = 
    (- \lambda)^m 
    \cdot ( \kappa\sigma_{(m)}^2 - \lambda)  
    +
    \frac{\kappa+1}{m} \lambda
    \cdot 
    \sigma_{(m)}^2 ( - \lambda)^{m-1}   \\ 
    & =
    (- \lambda)^m 
    \left( \kappa\sigma_{(m)}^2 - \lambda - \frac{\kappa+1}{m}
    \sigma_{(m)}^2  \right)\\
    & = 
    (- \lambda)^m 
    \frac{1}{m}
    \left\{ [ (m-1)\kappa - 1 ] \sigma_{(m)}^2 - m\lambda  \right\} \\
    & = 
    (- \lambda)^m 
    \frac{1}{m}
    \left[ ( mc^2 - 1 ) \sigma_{(m)}^2 - m\lambda  \right],
\end{align*}
where the last equality follows from the definition of $\kappa$. 
Because $c\ne m^{-1/2}$ and $\{\sigma^2_j\}_{j=1}^{m+1}$ are not all zero, 
we can know that $\frac{ (mc^2 - 1 )\sigma_{(m)}^2}{m}$ must be a nonzero root of $f(\lambda)$. By a conditioning argument similar to \eqref{eq:conditioning}, the quantity in \eqref{eq:prob_zero_lambda} must be zero. 

From the above, Lemma \ref{lem:rej_prob_point_mass} holds. 
\end{proof}
}

\begin{proof}[Proof of Theorem \ref{thm:size_asymp}]
Let $\tilde{\cp}_{n,j} \equiv \sqrt{n} (\widehat\cp_j - \mu_0)$ for $1\le j \le m$, $\tilde{\cp}_{n, m+1} \equiv \sqrt{n} (\widehat\cp_{m+1} - \mu_1)$, and $\delta_n \equiv \sqrt{n}(\mu_1 - \mu_0)$. 
By definition, we can verify that 
\begin{align*}
    |\widehat{T}_m| > c 
    & \ \ \Longleftrightarrow  \ \
    (\widehat\cp_{m+1} - \overline{\widehat{\cp}}_m)^2 
    - c^2 \widehat S_m^2 > 0
     \ \ \Longleftrightarrow  \ \
    (\tilde\cp_{n, m+1} - \overline{\tilde{\cp}}_{n,m} + \delta_n)^2 
    - c^2 \tilde S_{n,m}^2 > 0, 
\end{align*}
where $\overline{\tilde{\cp}}_{n,m} = m^{-1} \sum_{j=1}^m \tilde{\cp}_{n,j}$ and $\tilde S_{n,m}^2$ is the sample variance of $\{\tilde{\cp}_{n,j}\}_{j=1}^m$. 
From Assumption \ref{assu:CLT}, the condition that $\sqrt{n}(\mu_1 - \mu_0) \longrightarrow \delta$ as $n\longrightarrow \infty$, and using continuous mapping theorem, we have 
\begin{align*}
    (\tilde\cp_{n, m+1} - \overline{\tilde{\cp}}_{n,m} + \delta_n)^2 
    - c^2 \tilde S_{n,m}^2 
    \converged 
    (\nob_{m+1} - \overline\nob_m + \delta)^2 - c^2 S_m^2,
\end{align*}
where $\nob_i \sim \mathcal{N}(0, \sigma^2_j)$ for $1\le j \le m+1$, $\{\nob_i\}_{i=1}^{m+1}$ are mutually independent, 
and $\overline\nob_m$ and $S_m^2$ are the sample average and sample variance of $\{\nob_i\}_{i=1}^{m}$. 
From Lemma \ref{lem:rej_prob_point_mass}, when $c\ne m^{-1/2}$,  the distribution function of $(\nob_{m+1} - \overline\nob_m + \delta)^2 - c^2 S_m^2$ is continuous at $0$. Consequently, we must have 
\begin{align*}
    \bP[|\widehat{T}_m| > c ]
    & = 
    \bP \left[ (\tilde\cp_{n, m+1} - \overline{\tilde{\cp}}_{n,m} + \delta_n)^2 
    - c^2 \tilde S_{n,m}^2 > 0 \right]
    \\
    & \longrightarrow 
    \bP \left[ (\nob_{m+1} - \overline\nob_m + \delta)^2 - c^2 S_m^2 > 0 \right]
    = \bP[|T_m| > c ], 
\end{align*}
where the last equality holds by definition. 
From the above, we derive Theorem \ref{thm:size_asymp}. 
\end{proof}

\subsection{Proof of Lemmas \ref{lem:rej_prob_integral_main} and \ref{lem:first_second_deriv_gamma}
and Theorem \ref{thm:max_rej_prob}}

\begin{proof}[Proof of Lemma \ref{lem:rej_prob_integral_main}]
    Lemma \ref{lem:rej_prob_integral_main} follows directly from Lemmas \ref{lem:roots}, \ref{lem:rej_prob_integral}, \ref{lem:neg_root_abs_lower} and \ref{lem:neg_root_abs_upper}. 
\end{proof}

To prove Lemma \ref{lemma:max_at_finite}, we need the following lemma. 

\begin{lem}\label{lemma:sup_exist}
	Let $\{\xi_i\}^{m+1}_{i=1}$ be i.i.d.~standard normal random variables. For any $c>0$ and $(\sigma_{m+1}, \sigma_1, \ldots, \sigma_m) \in \mathbb{R}^{m+1}_{\ge 0}$, define 
	\begin{align*}
		D_c(\sigma_{m+1}, \sigma_1, \ldots, \sigma_m)
		& \equiv \Big( \sigma_{m+1} \xi_{m+1} - m^{-1}\sum^m_{i=1} \sigma_i \xi_{i} \Big)^2 \\ 
        & \qquad- c^2 \cdot (m-1)^{-1} \sum_{j=1}^m \Big( \sigma_j \xi_{j}  
         - m^{-1}\sum^m_{i=1} \sigma_i \xi_{i} \Big)^2,\\
         p_c(\sigma_{m+1}, \sigma_1, \ldots, \sigma_m)
	   & \equiv  
	\bP [D_c(\sigma_{m+1}, \sigma_1, \ldots, \sigma_m) > 0].
	\end{align*} 
    Consider any given $c>0, \rho>0$ and $1\le k \le m$, 
    and let $\tilde{p}$ denote 
    the supremum of $p_c(\sigma_{m+1}, \sigma_1, \ldots, \sigma_m)$ over all possible values of $(\sigma_{m+1}, \sigma_1, \ldots, \sigma_m)\in \mathbb{R}^{m+1}_{\ge 0}$ such that
    $\sigma_{m+1} \le \rho \sigma_{(k)}$, where $\sigma_{(1)} \le \sigma_{(2)} \le \ldots \le \sigma_{(m)}$ denote the sorted values of $\{\sigma_i\}_{i=1}^m$. 
    That is, 
    $$
    \tilde{p} \equiv \sup_{ (\sigma_{m+1}, \sigma_1, \ldots, \sigma_m)\in \mathbb{R}^{m+1}_{\ge 0}: \sigma_{m+1} \le \rho \sigma_{(k)}  } p_c(\sigma_{m+1}, \sigma_1, \ldots, \sigma_m).
    $$
    If $c\ne m^{-1/2}$, then one of the following two must hold:
    \begin{enumerate}[label=(\alph*)]
    \item $\tilde{p} = \sup_{(\sigma_1, \ldots, \sigma_m)\in \mathbb{R}^{m}_{\ge 0}} p_c(0, \sigma_1, \ldots, \sigma_m)$, 
    \item $\tilde{p} = p_c(\tilde{\sigma}_{m+1}, \tilde{\sigma}_1, \ldots, \tilde{\sigma}_m)$ for some $(\tilde{\sigma}_{m+1}, \tilde{\sigma}_1, \ldots, \tilde{\sigma}_m) \in \mathbb{R}_{+} \times \mathbb{R}_{\ge 0}^{m}$ such that 
    $\tilde{\sigma}_{m+1} \le \rho \tilde{\sigma}_{(k)}$, where $\tilde{\sigma}_{(1)} \le \tilde{\sigma}_{(2)} \le \ldots \le \tilde{\sigma}_{(m)}$ denote the sorted values of $\{\tilde{\sigma}_i\}_{i=1}^m$. 
    \end{enumerate} 
\end{lem}

\hideif{
\begin{proof}[Proof of Lemma \ref{lemma:sup_exist}]
    Below we consider any given $c>0, \rho>0$ and $1\le k \le m$. 
    Below we state two properties about $\tilde{p}$, followed from its definition.
    \begin{itemize}
        \item For any $(\sigma_{m+1}, \sigma_1, \ldots, \sigma_m)\in \mathbb{R}^{m+1}_{\ge 0}$, if $\sigma_{m+1} = 0$, then we must have $\sigma_{m+1} \le \rho \sigma_{(k)}$. 
    By the definition of $\tilde{p}$, this implies that $\tilde{p} \ge  \sup_{(\sigma_1, \ldots, \sigma_m)\in \mathbb{R}^{m}_{\ge 0}} p_c(0, \sigma_1, \ldots, \sigma_m)$. 

        \item By the definition of $\tilde{p}$ and note that the value of $p_c(\sigma_{m+1}, \sigma_1, \ldots, \sigma_m)$ is invariant under permutations of $(\sigma_1, \ldots, \sigma_m)$, there exists a sequence $\{ (\sigma_{n,m+1}, \sigma_{n1}, \ldots, \sigma_{nm})\}_{n=1}^\infty$ such that $p_n \equiv p_c(\sigma_{n,m+1}, \sigma_{n1}, \ldots, \sigma_{nm}) \converge \tilde{p}$ as $n\converge \infty$ and $\sigma_{n,m+1} \le \rho \sigma_{ni}$ for all $n$ and all $k\le i \le n$. 
    \end{itemize}

    First, we consider the case where there are infinitely many $n$ such that $\sigma_{n,m+1}^2 = 0$. In this case, we must have $\tilde{p} \le \sup_{(\sigma_1, \ldots, \sigma_m)\in \mathbb{R}^{m}_{\ge 0}} p_c(0, \sigma_1, \ldots, \sigma_m)$. 
    From the discussion before, we must have $\tilde{p} = \sup_{(\sigma_1, \ldots, \sigma_m)\in \mathbb{R}^{m}_{\ge 0}} p_c(0, \sigma_1, \ldots, \sigma_m)$, i.e., (a) in Lemma \ref{lemma:sup_exist} holds. 
    
    Second, we consider the case where there are only finitely many $n$ such that $\sigma_{n,m+1}^2 = 0$. We can therefore assume $\sigma_{n,m+1}^2 > 0$ for all $n$ without losing any generality.  In the following, we consider two cases, depending on whether the limit superior of $\max_{1\le i \le m} \frac{\sigma_{ni}}{\sigma_{n,m+1}}$ is finite. 
    \begin{itemize}
        \item[(i)] We consider the case where the limit superior of $\max_{1\le i \le m} \frac{\sigma_{ni}}{\sigma_{n,m+1}}$ is finite. 
	By the Bolzano--Weierstrass theorem, there exists a subsequence $\{ (\sigma_{n_{j},m+1}, \sigma_{n_{j}1}, \ldots, \sigma_{n_jm})\}_{j=1}^\infty$ such that $\frac{\sigma_{ni}}{\sigma_{n,m+1}}$ converges to some $a_i$ for each $1\le i \le m$. 
	This implies that, as $j\converge \infty$,  
	\begin{align*} 
        D_c(1, 
        \tfrac{\sigma_{n_j1}}{\sigma_{n_{j},m+1}}, \ldots, 
        \tfrac{\sigma_{n_jm}}{\sigma_{n_{j},m+1}}) 
        \convergeas D_c(1, a_1, \ldots, a_n).
	\end{align*}
    Because $c\ne m^{-1/2}$, Lemma \ref{lem:rej_prob_point_mass} implies that     
 
    the distribution function of $D_c(1, a_1, \ldots, a_n)$ is continuous at zero. 
    By the property of weak convergence, we then have
    \begin{align*}
        \tilde{p} & = \lim_{j\converge \infty} p_{n_j} =  \lim_{j\converge \infty} \bP[D_c(1, \tfrac{\sigma_{n_j1}}{\sigma_{n_{j},m+1}}, \ldots, 
        \tfrac{\sigma_{n_jm}}{\sigma_{n_{j},m+1}})>0] 
        = \bP[D_c(1, a_1, \ldots, a_n) > 0] \\
        & = p_c(1, a_1, \ldots, a_n).
    \end{align*}
    Furthermore, it is easy to verify that $\rho a_i = \rho \lim_{j\converge \infty} \frac{\sigma_{n_ji}}{\sigma_{n_{j},m+1}} \ge 1$ for $k\le i \le n$. 
    Consequently, (b) in Lemma \ref{lemma:sup_exist} holds. 
    
    \item[(ii)] We consider the case where the limit superior of $\max_{1\le i \le m} \frac{\sigma_{ni}}{ \sigma_{n,m+1}}$ is infinite. 
    There then exists a subsequence such that  $\frac{\sigma_{n,m+1}}{\max_{1\le i \le m} \sigma_{ni}} \converge 0$ along this subsequence. 
    By the Bolzano--Weierstrass theorem, there exists a subsequence $\{ (\sigma_{n_{j},m+1}, \sigma_{n_{j}1}, \ldots, \sigma_{n_jm})\}_{j=1}^\infty$ such that, as $j\converge \infty$, $\frac{\sigma_{n_j,m+1}}{\max_{1\le l \le m} \sigma_{n_j l}} \converge 0$ and 
    $\frac{\sigma_{n_j i}}{\max_{1\le l \le m} \sigma_{n_j l}} \converge b_i \le 1$ for all $1\le i\le m$. Moreover, at least one of $\{b_i\}_{i=1}^m$ is $1$. 
    
    These then imply that 
    \begin{align*}
        D_c(\tfrac{\sigma_{n_{j},m+1}}{\max_{1\le l \le m} \sigma_{n_j l}},\  \tfrac{\sigma_{n_j1}}{\max_{1\le l \le m} \sigma_{n_j l}},\  \ldots, \tfrac{\sigma_{n_jm}}{\max_{1\le l \le m} \sigma_{n_j l}}) 
        \convergeas D_c(0, b_1, \ldots, b_n).
    \end{align*}
    Since $c\ne m^{-1/2}$ holds, Lemma \ref{lem:rej_prob_point_mass} implies that     
    the distribution function of $D_c(0, b_1, \ldots, b_n)$ is continuous at zero. 
    By the property of weak convergence, we then have
    \begin{align*}
        \tilde{p} & = \lim_{j\converge \infty} p_{n_j} =  \lim_{j\converge \infty} \bP[ D_c(\tfrac{\sigma_{n_{j},m+1}}{\max_{1\le l \le m} \sigma_{n_j l}},\  \tfrac{\sigma_{n_j1}}{\max_{1\le l \le m} \sigma_{n_j l}},\  \ldots, \tfrac{\sigma_{n_jm}}{\max_{1\le l \le m} \sigma_{n_j l}}) >0  ] \\
        & = \bP[D_c(0, b_1, \ldots, b_n) > 0] = p_c(0, b_1, \ldots, b_n) \le  \sup_{(\sigma_1, \ldots, \sigma_m)\in \mathbb{R}^{m}_{\ge 0}} p_c(0, \sigma_1, \ldots, \sigma_m). 
    \end{align*}
    From the discussion before, this implies that $\tilde{p} = \sup_{(\sigma_1, \ldots, \sigma_m)\in \mathbb{R}^{m}_{\ge 0}} p_c(0, \sigma_1, \ldots, \sigma_m)$, i.e., (a) in Lemma \ref{lemma:sup_exist} holds. 
    \end{itemize}

    From the above, Lemma \ref{lemma:sup_exist} holds. 
\end{proof}
}

\begin{lem}
For any $1\le k \le m$ and $c\ne m^{-1/2}$,\footnote{Similar to the footnote for Theorem \ref{thm:size_asymp} and as discussed in Remark \ref{re:c_greater_1_over_sq_m}, we will consider values of $c$ greater than $m^{-1/2}$ for most conventional significance levels.} 
the maximum rejection probability $p_m(c; k, \rho)$ in \eqref{eq:max_rej_prob} must be obtained at some $(\sigma_1, \ldots, \sigma_m, \sigma_{m+1}) \in \cSS$. 
\end{lem}

\begin{proof}[\bf Proof of Lemma \ref{lemma:max_at_finite}]
    From Lemma \ref{lemma:sup_exist}, we can know that one of the following two must hold: 
    \begin{enumerate}[label=(\alph*)]
    \item $p_m(c; k, \rho) = \sup_{(\sigma_1, \ldots, \sigma_m)\in \mathbb{R}^{m}_{\ge 0}} p_c(0, \sigma_1, \ldots, \sigma_m)$, 
    \item $p_m(c; k, \rho) = p_c(\tilde{\sigma}_{m+1}, \tilde{\sigma}_1, \ldots, \tilde{\sigma}_m)$ for some $(\tilde{\sigma}_{m+1}, \tilde{\sigma}_1, \ldots, \tilde{\sigma}_m) \in \mathbb{R}_{+} \times \mathbb{R}_{\ge 0}^{m}$ such that 
    $\tilde{\sigma}_{m+1} \le \rho \tilde{\sigma}_{(k)}$, where $\tilde{\sigma}_{(1)} \le \tilde{\sigma}_{(2)} \le \ldots \le \tilde{\sigma}_{(m)}$ denote the sorted values of $\{\tilde{\sigma}_i\}_{i=1}^m$. 
    
    \end{enumerate} 
    When (b) holds, Lemma \ref{lemma:max_at_finite} holds obviously. Below we consider only the case when (a) holds. 
    When $\sigma_{m+1}=0$, our $t$ statistic essentially reduces to a one-sample $t$ statistic, except for a constant scaling term of $\sqrt{m}$. 
    From \citet{Bakirov:2006aa}, the supremum of $p_c(0, \sigma_1, \ldots, \sigma_m)$ over $(\sigma_1, \ldots, \sigma_m)\in \mathbb{R}^{m}_{\ge 0}$ must be achieved when some of $\{\sigma_j\}_{j=1}^m$ are zero and the remaining take a some common positive value, such as 1. 
    From the above, Lemma \ref{lemma:max_at_finite} holds. 
\end{proof}

\begin{proof}[Proof of Lemma \ref{lem:first_second_deriv_gamma}]
    Lemma \ref{lem:first_second_deriv_gamma} follows directly from 
    Lemma \ref{lem:deri_rej_prob}, 
    noting that $z=\kappa\gamma_1^2$, where $\kappa$ is a constant depending only on $m$ and $c$. 
\end{proof}

\begin{proof}[Proof of Theorem \ref{thm:form_maximizer_main}]
    Consider any given $1\le k \le m$, $\rho \ge 0$, and $c>0$ with $c\ne m^{-1/2}$. 
When $\rho = 0$, 
Theorem \ref{thm:form_maximizer_main} follows immediately from \citet{Bakirov:2006aa}. 
We then consider the case where $\rho>0$. 
From Lemma \ref{lemma:max_at_finite}, the maximum rejection probability $p_m(c; k, \rho)$ is achieved at some finite $\{\sigma_j\}_{j=1}^m$. 
If $\sigma_{m+1} = 0$, then the maximizer must have the form in (i), as shown in \citet{Bakirov:2006aa}. 
Below we consider only the case where $\sigma_{m+1} > 0$. 

Because the rejection probability is unchanged when we scale all the variances $\{\sigma_j\}_{j=1}^m$ by a positive constant, we must have
$p_m(c; k, \rho)=\overline{p}_m(c; \gamma_1, \ldots, \gamma_m)$ for some $(\gamma_1, \ldots, \gamma_m) \in \mathbb{R}^m_{\ge 0}$, 
where $\gamma_{(k)} \ge \rho^{-1}$ and $\gamma_{(1)}\le \gamma_{(2)}\le \ldots \le \gamma_{(m)}$ are the sorted values of $\{\gamma_i\}_{i=1}^m$. 
We prove that, besides $0$ and $\rho^{-1}$, $\{\gamma_i\}_{i=1}^m$ cannot take more than one distinct values; 
equivalently, there exists $\gamma\ge 0$ such that $\gamma_i \in \{0, \rho^{-1}, \gamma\}$ for all $1\le i \le m$. 
We prove this by contradiction. Assume that, without loss of generality, $\gamma_1, \gamma_2 \notin \{0, \rho^{-1}\}$, and $\gamma_1 \ne \gamma_2$. 
Fix the values of $c$, $\gamma_3, \ldots, \gamma_m$ and $\theta_{m+1}$, and view $\overline{p}_m(c; \gamma_1, \gamma_2, \gamma_3, \ldots, \gamma_m)$ as a function of only $\gamma_1$, where $\gamma_2$ is uniquely determined by $\gamma_1, c, \gamma_3, \ldots, \gamma_m$ and $\theta_{m+1}$. 
Note that the condition $\gamma_{(k)} \ge \rho^{-1}$ still holds when we slightly change the values of $\gamma_1$ and consequently $\gamma_2$. 
Thus, $\gamma_1$ is at least a local maximizer of $\overline{p}_m(c; \gamma_1, \gamma_2, \gamma_3, \ldots, \gamma_m)$ over a sufficiently small neighborhood of $\gamma_1$. 
This implies that the first order derivative of $\overline{p}_m(c; \gamma_1, \gamma_2, \gamma_3, \ldots, \gamma_m)$ over $\gamma_1^2$ is zero, and the second order derivative of $\overline{p}_m(c; \gamma_1, \gamma_2, \gamma_3, \ldots, \gamma_m)$ over $\gamma_1^2$ is less than or equal to zero. 
This, however, contradicts Lemma \ref{lem:first_second_deriv_gamma}.

Therefore, we must have 
$p_m(c; k, \rho)=\overline{p}_m(c; \gamma_1, \ldots, \gamma_m)$ for some $\{\gamma_j\}_{j=1}^m$ such that 
$\gamma_i \in \{0, \rho^{-1}, \gamma\}$ for $1\le i \le m$ and some $\gamma \ge 0$. 
Again, because the rejection probability is unchanged when we scale all the variances $\{\sigma_j\}_{j=1}^m$ by a positive constant, the maximum rejection probability must be obtained at some  $\{\sigma_j\}_{j=1}^{m+1}$ such that $\sigma_{m+1} = \rho$, and $\sigma_i \in \{0,1, \gamma\}$ for $1\le i \le m$ and some $\gamma \ge 0$.

From the above, Theorem \ref{thm:form_maximizer_main} holds. 
\end{proof}

\begin{proof}[Proof of Theorem \ref{thm:max_rej_prob}]
Consider any given $1\le k \le m$, $\rho > 0$, and $c>0$ with $c\ne m^{-1/2}$. 
From Theorem \ref{thm:form_maximizer_main}, we have either 
$p_m(c; k, \rho)=p_{m,0}(c)$, or 
$p_m(c; k, \rho)=\overline{p}_m(c; \gamma_1, \ldots, \gamma_m)$ for some $(\gamma_1, \ldots, \gamma_m) \in \mathbb{R}^m_{\ge 0}$ such that $\gamma_i \in \{0, \rho^{-1}, \gamma\}$ for some $\gamma\in \mathbb{R}$ and $\gamma_{(k)} \ge \rho^{-1}$, where $\gamma_{(1)}\le \gamma_{(2)}\le \ldots \le \gamma_{(m)}$ are the sorted values of $\{\gamma_i\}_{i=1}^m$.

Now suppose that the latter holds, i.e., $p_m(c; k, \rho)=\overline{p}_m(c; \gamma_1, \ldots, \gamma_m)$  some $(\gamma_1, \ldots, \gamma_m) \in \mathbb{R}^m_{\ge 0}$ satisfying the conditions discussed before. 

Let $m_1$ and $m_0$ denote the numbers of $\gamma_i$s that take values $\rho^{-1}$ and $0$, respectively. 
That is, $m_1 = \sum_{i=1}^m \I(\gamma_i = \rho^{-1})$ and $m_0 = \sum_{i=1}^m \I(\gamma_i = 0)$. 
Consequently, the number of  $\gamma_i$s that take the value $\gamma$ is $m-m_1-m_0$, 
and $\overline{p}_m(c; \gamma_1, \gamma_2, \gamma_3, \ldots, \gamma_m)$ simplifies to $\overline{p}_m(c; \rho, \gamma; m_1, m_0)$ defined in \eqref{eq:p_m_given_num_gamma}. 
Note that $\{\gamma_i\}_{i=1}^m$ needs to satisfy the constraint that $\gamma_{(k)} \ge \rho^{-1}$. 
Thus, we must have $0\le m_0\le k-1$, $m_1 \le m-m_0$, and 
$\gamma \in \mathbb{R}_{\ge 0}$ if $m_1 \ge m-k+1$ and $\gamma \in [\rho^{-1}, \infty)$ if $m_1 < m-k+1.$
In sum, $p_m(c; k, \rho)$ equals to $\overline{p}_m(c; \rho, \gamma; m_1, m_0)$ for some $0\le m_0\le k-1$, $m_1 \le m-m_0$, and 
$\gamma \in \mathbb{R}_{\ge 0}$ if $m_1 \ge m-k+1$ and $\gamma \in [\rho^{-1}, \infty)$ if $m_1 < m-k+1.$

From the above, we list possible cases where the rejection probability obtains its supremum value $p_m(c; k, \rho)$. 
Thus, $p_m(c; k, \rho)$ must be the supremum over all the cases we discussed above. 
We can therefore derive Theorem \ref{thm:max_rej_prob}. 
\end{proof}

\subsection{Proof of Lemma \ref{lem:sign_deriv_gamma}}

Below we first give the form of $H_m(c; \gamma_1, \ldots, \gamma_m)$ in Lemma \ref{lem:sign_deriv_gamma}. 
For any $c>0$ and $(\gamma_1, \gamma_2, \ldots, \gamma_m) \in \mathbb{R}^m_{\ge 0}$, 
let $\kappa \equiv \frac{mc^2}{m-1}$, $\tau \equiv \frac{\kappa+1}{\kappa m}$, 
$x_{i} \equiv \kappa \gamma_{i}^2$ for $1\le i \le m$, 
$x_{(1)}\le x_{(2)} \le \ldots \le x_{(m)}$ be the sorted values of $\{x_i\}_{i=1}^m$, 
and 
\begin{align}\label{eq:H_prime}
    \begin{split}
    & \hspace{-15pt} H_{m}'(w; c; \gamma_1, \ldots, \gamma_m) \\
    & \equiv \frac{\overline{\theta}}{x_{(m)} + \overline{\theta}}
	+ 
	\frac{2\overline{\theta}}{w + \overline{\theta}} 
	-\frac{\underline{\theta} - m}{\tau \underline{\theta} - 1}
	+ \frac{1-\tau}{1-\tau+\min\{C_m(w; c; \gamma_1, \ldots, \gamma_m), 0\}}
    -
	1, 
    \end{split}
\end{align}
where $\overline{\theta} \equiv m+ \frac{x_{(m)}}{\kappa}$, $\underline{\theta} \equiv \max_{1\le j \le m} \theta(x_{(j)}, j)$ with $\theta(\cdot)$ defined as in Lemma \ref{lem:neg_root_abs_lower}, 
and 
\begin{align*}
    C_m(w; c; \gamma_1, \ldots, \gamma_m) \equiv 
    \begin{cases}
    \left(1-2\tau+\frac{2}{x_{(m)}}\right) \frac{w x_{(m)}-\overline{\theta}^2}{(w+x_{(m)}+2\overline{\theta})\overline{\theta}} 
    + \frac{1}{x_{(m)}}
    - \tau, & \text{if } w = x_{(1)}, \\
	(1-2\tau) \frac{wz-\overline{\theta}^2}{(w+z+2\overline{\theta})\overline{\theta}}
    - \tau, & \text{if } w > x_{(1)}.
\end{cases}
\end{align*}
Then $H_m(c; \gamma_1, \ldots, \gamma_m)$ is defined as 
\begin{align}\label{eq:H_H_prime}
    H_m(c; \gamma_1, \ldots, \gamma_m)
    & = 
    H_{m}'(x_2; c; \gamma_1, \ldots, \gamma_m).
\end{align}

\begin{proof}[Proof of Lemma \ref{lem:sign_deriv_gamma}]
Note that if $c\ge \sqrt{2(m-1)/[m(m-2)]}$, then $\tau$ defined as in Lemma \ref{lem:deri_rej_prob_bound} satisfies that 
\begin{align*}
    \tau = \frac{\kappa+1}{\kappa m} = \frac{1}{m} + \frac{1}{\kappa m} 
    = \frac{1}{m} + \frac{m-1}{m^2 c^2}
    \le \frac{1}{m} + \frac{(m-1)m(m-2)}{m^2 \cdot 2(m-1)}
    = \frac{1}{m} + \frac{m-2}{2m} = \frac{1}{2},
\end{align*}
where $\kappa$ is defined as in  Lemma \ref{lem:deri_rej_prob_bound}. 
Lemma \ref{lem:sign_deriv_gamma} follows immediately from Lemma \ref{lem:deri_rej_prob_bound}. 
\end{proof}

\subsection{Simplifying the optimization under general relative heterogeneity assumption}
Below we consider simplification on optimization under general relative heterogeneity constraint. 
Specifically,  Lemma \ref{lem:sign_deriv_gamma} can help to simplify the optimization in \eqref{eq:max_rej_m1_m0}. 
Consider any $1\le m_1\le m-1$, $0 \le m_0 \le m-m_1-1$, and the optimization of $ \overline{p}_m(c; \rho, \gamma; m_1, m_0)$ over either $\gamma \in [0, \infty)$ or $\gamma \in [\rho^{-1}, \infty)$. 
Let 
\begin{align*}
	(\gamma_1, \gamma_2, \ldots, \gamma_m)
	& = 
	(\underbrace{\gamma, \ldots, \gamma}_{m-m_1-m_0}, \underbrace{\rho^{-1}, \ldots, \rho^{-1}}_{m_1}, \underbrace{0, \ldots, 0}_{m_0}). 
\end{align*}
Define 
\begin{align*}
	\check{H}_{m}(c; \rho, \gamma; m_1, m_0)
    & = H_{m}'(\rho^{-1}; c; \gamma_1, \ldots, \gamma_m) \}
\end{align*}
For any $\rho>0$ and $\gamma > \rho^{-1}$,  if $\check{H}_{m}(c; \rho, \gamma; m_1, m_0)<0$,  then we can strictly increase the rejection probability $\overline{p}_m(\gamma_1, \gamma_2, \ldots, \gamma_m)$ by  slightly decreasing one of $\{\gamma_j\}^m_{j=1}$ that is equal to $\gamma$ and slightly increasing one of $\{\gamma_j\}^m_{j=1}$ that is equal to
$\rho^{-1}$.  
Importantly, if the original $\{\gamma_j\}^m_{j=1}$ satisfy the relative heterogeneity assumption for some $k$, then the slight changes of the $\{\gamma_j\}^m_{j=1}$ will maintain this relative heterogeneity assumption. 
Consequently,  $\overline{p}_m(\gamma_1, \gamma_2, \ldots, \gamma_m)$ cannot be the maximum rejection probability under the relative heterogeneity assumption.

From the above, 
if $\check{H}_m(c; \rho, \gamma; m_1, m_0)<0$ for all $\gamma > \rho^{-1}$, 
then $\sup_{\gamma \in [\rho^{-1}, \infty)}\overline{p}_m(c; \rho, \gamma; m_1, m_0)$ must be obtained at $\gamma = \rho^{-1}$, 
and $\sup_{\gamma \in \mathbb{R}_{\ge 0}}\overline{p}_m(c; \rho, \gamma; m_1, m_0)$ must be obtained at some $\gamma \in [0, \rho^{-1}]$. 
In other words, we can either obtain a closed-form solution for the optimization or restrict the optimization to a smaller range.

\subsection{Simplifying the optimization under relative heterogeneity assumption with $k=1$}

\begin{theorem}\label{thm:max_rej_prob_k_1_closed_form}
For any given $m \ge 4$, $\rho > 0$, $c > \sqrt{\frac{3(m-1)}{m(m-3)}}$, define the following as a function of $\gamma$:
\begin{align*}
	\widetilde{H}_m(\gamma; c, \rho) 
	& \equiv \frac{\overline{\theta}}{\kappa \gamma^2 + \overline{\theta}}
	+ 
	\frac{2\overline{\theta}}{\kappa \rho^{-2} + \overline{\theta}} 
	-\frac{\underline{\theta} - m}{\tau \underline{\theta} - 1}
	+ \frac{1-\tau}{1-\tau+\min\{C, 0\}}
	- 
	1,
\end{align*}
where $\overline{\theta} \equiv m+\gamma^2$, $\underline{\theta} \equiv m + \rho^{-2}$, $\kappa \equiv \frac{m c^2}{m-1}$, $\tau \equiv \frac{\kappa+1}{m\kappa}$, and 
\begin{align*}
	C 
	& \equiv
	\left( 1-2\tau + \frac{2}{\kappa \gamma^2} \right) \frac{\kappa^2 \gamma^2 \rho^{-2} -\overline{\theta}^2}{(\kappa \gamma^2 + \kappa \rho^{-2}+2\overline{\theta})\overline{\theta}} +  
    \frac{1}{\kappa \gamma^2}
    - \tau.
\end{align*}
Suppose that
\[
	\widetilde{H}_m(\gamma; c, \rho)  < 0 \ 
    \text{ for all } \gamma > \rho^{-1}. 
\]
Then, the maximum rejection probability $p_m(c; 1, \rho)$ under Assumption \ref{assu:relative_heter} with $k=1$ and the given $\rho$, $m$, and $c$ has the following equivalent form:
	\begin{align*}
		p_m(c; 1, \rho) = 
		\bP \left[ |t_{m-1}|\sqrt{\rho^2+\frac{1}{m}} > c \right].
	\end{align*}
\end{theorem}

\begin{proof}[Proof of Theorem \ref{thm:max_rej_prob_k_1_closed_form}]
First, for any given $1\le m_1\le m-1$, $\rho>0$ and $\gamma > \rho^{-1}$, we consider 
\begin{align}\label{eq:gamma_rho_m1}
    (\gamma_1, \gamma_2, \ldots, \gamma_m)
    & = 
    (\gamma, \rho^{-1}, \underbrace{\gamma, \ldots, \gamma}_{m-m_1-1}, \underbrace{\rho^{-1}, \ldots, \rho^{-1}}_{m_1-1}). 
\end{align}
Recall the definition in \eqref{eq:H_prime} and \eqref{eq:H_H_prime}, we have 
\begin{align*}
    & \hspace{-15pt} H_m(c; \gamma_1, \ldots, \gamma_m)
    \\
    & = 
    H_{m}'(x_2; c; \gamma_1, \ldots, \gamma_m)
    \\ 
    & = 
    \frac{\overline{\theta}}{x_{(m)} + \overline{\theta}}
	+ 
	\frac{2\overline{\theta}}{x_2 + \overline{\theta}} 
	-\frac{\underline{\theta}' - m}{\tau \underline{\theta}' - 1}
	+ \frac{1-\tau}{1-\tau+\min\{C_m(x_2; c; \gamma_1, \ldots, \gamma_m), 0\}}
    -
	1,
\end{align*}
where $\kappa \equiv \frac{mc^2}{m-1}$, $\tau \equiv \frac{\kappa+1}{\kappa m}$, 
$x_{i} \equiv \kappa \gamma_{i}^2$ for $1\le i \le m$, 
$x_{(1)}\le x_{(2)} \le \ldots \le x_{(m)}$ are the sorted values of $\{x_i\}_{i=1}^m$, 
$\overline{\theta} = m+x_{(m)}/\kappa$, $\underline{\theta}' = \max_{1\le j \le m} \theta(x_{(j)}, j)$ with $\theta(\cdot)$ defined as in Lemma \ref{lem:neg_root_abs_lower}, 
and 
\begin{align*}
    C_m(x_2; c; \gamma_1, \ldots, \gamma_m)
    = 
     \left(1-2\tau+\frac{2}{x_{(m)}}\right) \frac{x_2 x_{(m)}-\overline{\theta}^2}{(x_2+x_{(m)}+2\overline{\theta})\overline{\theta}} 
    + \frac{1}{x_{(m)}}
    - \tau. 
\end{align*}
Note that $(\theta - m)/(\tau \theta - 1)$ is increasing in $\theta$, due to the fact that $-1 + m\tau = 1/\kappa >0$, 
and $\underline{\theta}' = \max_{1\le j \le m} \theta(x_{(j)}, j) \ge \theta(x_{(1)}, 1) = m+\rho^{-2} = \underline{\theta}$, where the last equality follows from Lemma \ref{lem:neg_root_abs_lower}(c). 
We then have 
\begin{align*}
    H_m(c; \gamma_1, \ldots, \gamma_m)
    & = 
    \frac{\overline{\theta}}{x_{(m)} + \overline{\theta}}
	+ 
	\frac{2\overline{\theta}}{x_2 + \overline{\theta}} 
	-\frac{\underline{\theta}' - m}{\tau \underline{\theta}' - 1}
	+ \frac{1-\tau}{1-\tau+\min\{C_m(x_2; c; \gamma_1, \ldots, \gamma_m), 0\}}
    -
	1\\
    & \le
    \frac{\overline{\theta}}{x_{(m)} + \overline{\theta}}
	+ 
	\frac{2\overline{\theta}}{x_2 + \overline{\theta}} 
	-\frac{\underline{\theta} - m}{\tau \underline{\theta} - 1}
	+ \frac{1-\tau}{1-\tau+\min\{C_m(x_2; c; \gamma_1, \ldots, \gamma_m), 0\}}
    -
	1\\
    & = 
    \widetilde{H}_m(c; \rho, \gamma),
\end{align*}
where the last equality follows by definition. 

Second, from the first part and the condition that $\widetilde{H}_m(c; \rho, \gamma) < 0$ for all $\gamma > \rho^{-1}$, 
we know that $H_m(c; \gamma_1, \ldots, \gamma_m)<0$, for all $1\le m_1 \le m$, $\rho>0$, $\gamma>\rho^{-1}$ and $(\gamma_1, \ldots, \gamma_m)$ defined as in \eqref{eq:gamma_rho_m1}. 
In addition, we have 
\begin{align*}
    \frac{c}{\sqrt{\frac{2(m-1)}{m(m-2)}}} > 
    \frac{\sqrt{\frac{3(m-1)}{m(m-3)}} }{\sqrt{\frac{2(m-1)}{m(m-2)}}} 
    = 
    \sqrt{
    \frac{3(m-2)}{2(m-3)}
    }
    = 
    \sqrt{
    \frac{3m-6}{2m-6}
    }
    \ge 1. 
\end{align*}
From Lemma \ref{lem:sign_deriv_gamma}, we know that for $1\le m_1\le m-1$ and any $c, \rho>0$, 
$\overline{p}_m(c; \rho, \gamma; m_1, 0)$ cannot be $p_m(c;1,\rho)$ for any $\gamma > \rho^{-1}$, since we can strictly increase the rejection probability $\overline{p}_{m}(\gamma_1, \ldots, \gamma_m)$ by  slightly decreasing one of $\{\gamma_i\}_{i=1}^m$ that is equal to $\gamma$ and slightly increasing one of $\{\gamma_i\}_{i=1}^m$ that is equal to $\rho^{-1}$.

Third, we consider $\overline{p}_m(c; \gamma \bs{1}_{m}^\top)$ for any $\gamma > 0$. In this case, $\overline{p}_m(c; \gamma \bs{1}_{m}^\top) = \bP[|T_m|>c]$, where $T_m$ is defined as in \eqref{eq:pop-test-1} with $\nob_{m+1} \sim \mathcal{N}(0, \gamma^{-2})$ and $\nob_i \sim \mathcal{N}(0, 1)$ for $1\le i \le m$. We can verify that $T_m \sim \sqrt{\gamma^{-2}+m^{-1}} t_{m-1}$, where $t_{m-1}$ is a $t$-distributed random variable with $m-1$ degrees of freedom.
Consequently, 
\begin{align*}
    \overline{p}_m(c; \gamma \bs{1}_{m}^\top) = 
    \bP[|T_m|>c]
    = \bP\left[ |t_{m-1}|> \frac{c}{\sqrt{\gamma^{-2}+m^{-1}}} \right],
\end{align*}
which is decreasing in $\gamma$. 
This implies that 
$\sup_{\gamma \in [\rho^{-1}, \infty)}\overline{p}_m(c; \gamma \bs{1}_{m}^\top) =  \overline{p}_m(c; \rho^{-1} \bs{1}_{m}^\top)$. 

Fourth, from the second and third parts, we know that $p_m(c; 1, \rho)$ defined as in Theorem \ref{thm:max_rej_prob} has the following equivalent forms:
\begin{align*}
	p_m(c; 1, \rho) & 
	= \max\Big\{ \max_{0\le m_1 \le m} \tilde{p}_m(c; k, \rho; m_1, 0), \ p_{m,0}(c) \Big\} \\
    & = 
    \max\Big\{ \overline{p}_m(c; \rho^{-1} \bs{1}_{m}^\top), \ p_{m,0}(c) \Big\} \\
    & = 
    \max\left\{ \bP\left[ |t_{m-1}|> \frac{c}{\sqrt{\rho^2+m^{-1}}} \right], \ p_{m,0}(c) \right\}.
\end{align*}

Fifth, from \citet{Bakirov:2006aa}, we can know that, when $c > \sqrt{3(m-1)/[m(m-3)]}$, 
$$
p_{m,0}(c) = \bP\left[ |t_{m-1}|> \sqrt{m}c \right]
\le \bP\left[ |t_{m-1}|> \frac{c}{\sqrt{\rho^2+m^{-1}}} \right].
$$
This implies that 
\begin{align*}
	p_m(c; 1, \rho) & =
    \max\left\{ \bP\left[ |t_{m-1}|> \frac{c}{\sqrt{\rho^2+m^{-1}}} \right], \ p_{m,0}(c) \right\}
    = \bP\left[ |t_{m-1}|> \frac{c}{\sqrt{\rho^2+m^{-1}}} \right].
\end{align*}

Therefore, Theorem \ref{thm:max_rej_prob_k_1_closed_form} holds.
\end{proof}

\subsection{Proof of Theorem \ref{thm:max_rej_prob_k_1_simple_closed_form}}

We will prove Theorem \ref{thm:max_rej_prob_k_1_simple_closed_form} using Theorem \ref{thm:max_rej_prob_k_1_closed_form}. 
Define 
\begin{align}\label{eq:H_tilde_123}
    \begin{split}
    \tha 
    & \equiv \frac{\overline{\theta}}{\kappa \gamma^2 + \overline{\theta}} + \frac{2\overline{\theta}}{\kappa \rho^{-2} + \overline{\theta}}, 
    \\ 
    \thb 
    &\equiv  \frac{1-\tau}{1-\tau+\min\{C, 0\}}, 
    \\
    \thc 
    &\equiv - \frac{\underline{\theta} - m}{\tau \underline{\theta} - 1}, 
    \end{split}
\end{align}
where $\kappa, \overline{\theta}, \underline{\theta}, \tau$ and $C$ are defined the same as in Theorem \ref{thm:max_rej_prob_k_1_closed_form}:

\begin{align}\label{eq:kappa_theta_over_under_tau_C}
\kappa & \equiv \frac{m c^2}{m-1}, \quad 
\overline{\theta} \equiv m+\gamma^2, \quad 
\underline{\theta} \equiv m + \rho^{-2}, \quad 
\tau \equiv \frac{\kappa+1}{m\kappa}, 
\nonumber
\\
	C 
	& \equiv
	\left( 1-2\tau + \frac{2}{\kappa \gamma^2} \right) \frac{\kappa^2 \gamma^2 \rho^{-2} -\overline{\theta}^2}{(\kappa \gamma^2 + \kappa \rho^{-2}+2\overline{\theta})\overline{\theta}} +  
    \frac{1}{\kappa \gamma^2}
    - \tau.
\end{align}
By definition, we can then write $\thh$ in Theorem \ref{thm:max_rej_prob_k_1_closed_form} as
\begin{align*}
    \thh 
    =
    \tha
    + \thb
    + \thc
    - 1.
\end{align*}

In the following Lemmas \ref{lem:lb_h1a}--\ref{lemma:lb_h2}, we construct upper bounds $\oha$ and $\ohb$ for $\tha$ and $\thb$, respectively, over all $\gamma \in (\rho^{-1}, \infty)$; see \eqref{eq:H_m1_overline} and \eqref{eq:H_m2_overline} for their expressions. 
We can thus bound $\thh$ over all $\gamma > \rho^{-1}$ from the above by 
\[
	\oha
	+ \ohb
	+ \thc
	- 1.
\]
We then establish the monotonicity of $\oha$, $\ohb$, and $\thc$ with respect to $c$ in Lemmas \ref{lem:mon_h1}--\ref{lem:mon_h3}, respectively. 
Finally, we prove Theorem \ref{thm:max_rej_prob_k_1_simple_closed_form}.

\begin{lem} \label{lem:lb_h1a}
Define $\tha$
as in \eqref{eq:H_tilde_123} for $\gamma > \rho^{-1}$. 
One of the following statements is true:
\begin{enumerate}
	\item[(a)] There exists at most one finite $\gamma_0$ such that $\gamma_0 > \rho^{-1}$ and $\left.\frac{\partial\tha}{\partial \gamma}\right|_{\gamma = \gamma_0} = 0$.
	\item[(b)] $\frac{\partial\tha}{\partial \gamma} = 0$ for any $\gamma > \rho^{-1}$.
\end{enumerate}
\end{lem}

\hideif{
\begin{proof}[Proof of Lemma \ref{lem:lb_h1a}] Note that $\tha$ can be written explicitly in terms of $\gamma^2$ as follows:
\begin{equation}
	\label{eq:htilde-1-1}
	\tha \equiv 
	\frac{m + \gamma^2}{(\kappa + 1)\gamma^2 + m}
	+
	2 \left[
		\frac{m + \gamma^2}{\gamma^2 + (\kappa \rho^{-2} + m)}
	\right].
\end{equation}
The derivatives of the above two fractions with respect to $\gamma^2$ are, respectively, 
\begin{align}
	\label{eq:htilde-1-d1}
	\frac{\partial [\frac{m + \gamma^2}{(\kappa + 1)\gamma^2 + m}]}{\partial (\gamma^2)}
	= \frac{-m \kappa}{[(\kappa + 1)\gamma^2 + m]^2},
\end{align}
and
\begin{align}
	\label{eq:htilde-1-d2}
	\frac{\partial [\frac{m + \gamma^2}{\gamma^2 + (\kappa \rho^{-2} + m)}]}{\partial (\gamma^2)}
	= \frac{ \kappa \rho^{-2}}{[\gamma^2 + (\kappa \rho^{-2} + m)]^2}.
\end{align}
Combining \eqref{eq:htilde-1-d1} and \eqref{eq:htilde-1-d2} gives
\begin{align}\label{eq:htilde-1-2}
	& \quad \ \frac{\partial \tha}{\partial (\gamma^2)}
    \nonumber
    \\
	& = 	\frac{\partial [\frac{m + \gamma^2}{(\kappa + 1)\gamma^2 + m}]}{\partial (\gamma^2)}
	+
		2\frac{\partial [\frac{m + \gamma^2}{\gamma^2 + (\kappa \rho^{-2} + m)}]}{\partial (\gamma^2)}		\notag	
    =  \kappa 
	\left\{
		\frac{ 2\rho^{-2}}{[\gamma^2 + (\kappa \rho^{-2} + m)]^2}
		- \frac{m }{[(\kappa + 1)\gamma^2 + m]^2}
	\right\}
    \nonumber
    \\
    & = 
    \frac{\kappa
    \{
    2\rho^{-2}[(\kappa + 1)\gamma^2 + m]^2 - m [\gamma^2 + (\kappa \rho^{-2} + m)]^2
    \}
    }{[\gamma^2 + (\kappa \rho^{-2} + m)]^2\cdot [(\kappa + 1)\gamma^2 + m]^2}
    \nonumber
    \\
    & = 
    \frac{\kappa
    \{
    \sqrt{2\rho^{-2}} [(\kappa + 1)\gamma^2 + m]
    + \sqrt{m} [\gamma^2 + (\kappa \rho^{-2} + m)]
    \}
    }{[\gamma^2 + (\kappa \rho^{-2} + m)]^2\cdot [(\kappa + 1)\gamma^2 + m]^2}
    \times 
    \nonumber
    \\
    & \quad \quad \quad 
    \left\{ 
    \sqrt{2\rho^{-2}} [(\kappa + 1)\gamma^2 + m]
    - \sqrt{m} [\gamma^2 + (\kappa \rho^{-2} + m)]
    \right\}
    \nonumber
    \\
    & = 
    \frac{\kappa
    \{
    \sqrt{2\rho^{-2}} [(\kappa + 1)\gamma^2 + m]
    + \sqrt{m} [\gamma^2 + (\kappa \rho^{-2} + m)]
    \}
    }{[\gamma^2 + (\kappa \rho^{-2} + m)]^2\cdot [(\kappa + 1)\gamma^2 + m]^2}
    \times 
    \nonumber
    \\
    & \quad \quad \quad 
    \left\{ 
    \left[ \sqrt{2\rho^{-2}} (\kappa + 1) - \sqrt{m} \right] \gamma^2 -
    \left[ \sqrt{m} (\kappa \rho^{-2} + m) - \sqrt{2\rho^{-2}} m \right]
    \right\}.
\end{align}

The goal of this lemma is to show either there exists at most one $\gamma_0>\rho^{-1}$ that satisfy statement (a), or the partial derivative is constant for all $\gamma$.
Since $\frac{\partial \tha}{\partial (\gamma^2)} = 2\gamma \frac{\partial \tha}{\partial \gamma}$ and we require $\gamma > \rho^{-1} > 0$, finding a $\gamma_0$ such that $\left.\frac{\partial\tha}{\partial \gamma}\right|_{\gamma = \gamma_0} = 0$ is equivalent to finding $\gamma_0$ that satisfy $\left.\frac{\partial\tha}{\partial (\gamma^2)}\right|_{\gamma^2 = \gamma_0^2} = 0$. To this end, it suffices to consider the scenario such that 
\eqref{eq:htilde-1-2} equals 0. 
Because $\rho$, $\kappa$, $\gamma$ and $m$ are all positive, 
\eqref{eq:htilde-1-2} equals 0 if and only if 
\begin{align}
	\label{eq:htilde-1-3-1a}
     \left[ \sqrt{2\rho^{-2}} (\kappa + 1) - \sqrt{m} \right] \gamma^2 = \sqrt{m} (\kappa \rho^{-2} + m) - \sqrt{2\rho^{-2}} m. 
\end{align} 
There are two scenarios to consider. 

First, consider the case where $\sqrt{2\rho^{-2}}(\kappa + 1) - \sqrt{m} \neq 0$. Then, \eqref{eq:htilde-1-3-1a} implies that
\begin{align}
	\label{eq:htilde-1-3-1b}
	\gamma^2 = 
	\frac{\sqrt{m}(\kappa \rho^{-2} + m) - \sqrt{2\rho^{-2}}m}{\sqrt{2\rho^{-2}}(\kappa + 1) - \sqrt{m}}.
\end{align}
In \eqref{eq:htilde-1-3-1b}, if the RHS is positive, then there is one positive $\gamma$ that satisfy the equation. If the RHS equals 0, then this cannot be the $\gamma_0$ we wish to find because we require $\gamma_0 > \rho^{-1} > 0$. If the RHS is negative, then there is no real solution. Thus, 
in this case, there is at most one $\gamma_0 > \rho^{-1}$ 
such that $\left.\frac{\partial\tha}{\partial \gamma}\right|_{\gamma = \gamma_0} = 0$
. \par 

Now, consider the second scenario where $\sqrt{2\rho^{-2}}(\kappa + 1) - \sqrt{m} = 0$ under the first case. Then, the derivative is either always positive, negative, or zero for all $\gamma > \rho^{-1}$. Thus, either statement (a) or (b) is true. \par

From the above, Lemma \ref{lem:lb_h1a} holds. 
\end{proof}
}

\begin{lem} \label{lemma:lb_h1}
For any given values of positive $c$, $\rho$, and $m$, 
define $\tha$
as in \eqref{eq:H_tilde_123} for $\gamma > \rho^{-1}$. 
For any $\gamma > \rho^{-1}$, we have 
$
\tha \le 
\oha,
$
where 
\begin{align}\label{eq:H_m1_overline}
    \oha & \equiv \max\left\{\widetilde H_{m,1}(\rho^{-1}; c, \rho), \lim_{\gamma \to \infty} \tha \right\}
    \notag \\
    & = 
    \max\left\{
    \frac{3(m + \rho^{-2})}{(\kappa + 1)\rho^{-2} + m}, 
    \frac{2\kappa + 3}{\kappa + 1}
    \right\}. 
\end{align}
\end{lem}

\hideif{
\begin{proof}[Proof of Lemma \ref{lemma:lb_h1}]

Following Lemma \ref{lem:lb_h1a}, we consider the following three scenarios. 

First, suppose there is no $\gamma_0 > \rho^{-1}$ such that $\left.\frac{\partial \tha}{\partial \gamma}\right|_{\gamma = \gamma_0} = 0$. This means $\frac{\partial \tha}{\partial \gamma}$ is either positive or negative for all $\gamma > \rho^{-1}$ because $\frac{\partial \tha}{\partial \gamma}$ is continuous for all real $\gamma$ from \eqref{eq:htilde-1-2}. Thus, $\tha$ is monotone in $\gamma$ for all $\gamma > \rho^{-1}$.

Second, suppose that $\frac{\partial \tha}{\partial \gamma} = 0$ for any $\gamma > \rho^{-1}$. This means $\tha$ is constant in $\gamma$ for all $\gamma > \rho^{-1}$.

Third, suppose there is one $\gamma_0 > \rho^{-1}$ such that $\left.\frac{\partial \tha}{\partial \gamma}\right|_{\gamma = \gamma_0} = 0$. For this $\gamma_0$, it has to satisfy \eqref{eq:htilde-1-3-1b}. There are two scenarios to consider as follows:
\begin{enumerate}
	\item[(i)] $\sqrt{m}(\kappa \rho^{-2} + m) - \sqrt{2\rho^{-2}}m > 0$ and $\sqrt{2\rho^{-2}}(\kappa + 1) - \sqrt{m} > 0$.
	\item[(ii)] $\sqrt{m}(\kappa \rho^{-2} + m) - \sqrt{2\rho^{-2}}m < 0$ and $\sqrt{2\rho^{-2}}(\kappa + 1) - \sqrt{m} < 0$.
\end{enumerate}
We first show that scenario (ii) is impossible. To see this, assume to the contrary that $\sqrt{m}(\kappa \rho^{-2} + m) - \sqrt{2\rho^{-2}}m < 0$ and $\sqrt{2\rho^{-2}}(\kappa + 1) - \sqrt{m} < 0$. The first inequality implies that 
\begin{equation}
	\label{eq:htilde-1-shape-1}
	\kappa \rho^{-2} + m < \sqrt{2 \rho^{-2}m},
\end{equation}
whereas the second inequality implies that 
\begin{equation}
	\label{eq:htilde-1-shape-2}
	\sqrt{2\rho^{-2}} < \frac{\sqrt{m}}{\kappa + 1}.
\end{equation}
Combining inequalities \eqref{eq:htilde-1-shape-1} and \eqref{eq:htilde-1-shape-2} gives
\begin{equation*}
	\kappa \rho^{-2} + m
	< 
    \sqrt{m} \cdot \frac{\sqrt{m}}{\kappa + 1}
    =
    \frac{m}{\kappa + 1} < m,
\end{equation*}
leading to a contradiction. 
Therefore, it remains to consider scenario (i), i.e., $\sqrt{m}(\kappa \rho^{-2} + m) - \sqrt{2\rho^{-2}}m > 0$ and $\sqrt{2\rho^{-2}}(\kappa + 1) - \sqrt{m} > 0$. 
From \eqref{eq:htilde-1-2}, we can know that $\frac{\partial \tha}{\partial (\gamma^2)}$ is negative when $\gamma< \gamma_0$ and positive when $\gamma > \gamma_0$. 
Consequently, 
$\tha$ is decreasing for $\gamma < \gamma_0$ and increasing for $\gamma > \gamma_0$.  \par 

From the above, we have:
\begin{enumerate}
	\item[(a)] in the first and second scenarios, $\tha$ is either increasing, decreasing, or constant in $\gamma \in (\rho^{-1}, \infty)$; 
	\item[(b)] in the third scenario, $\tha$ is decreasing for $\gamma \in (\rho^{-1}, \gamma_0)$ and increasing for $\gamma \in (\gamma_0, \infty)$, for some $\gamma_0 > \rho^{-1}$. 
\end{enumerate}
This means $\tha$ is bounded from above by its endpoints, i.e., it is bounded from  above by 
\[
	\oha \equiv \max\left\{\widetilde H_{m,1}(\rho^{-1}; c, \rho), \lim_{\gamma \to \infty} \tha \right\}, 
\]
where
\begin{align*}
	\widetilde H_{m,1}(\rho^{-1}; c, \rho)
	& = \frac{m + \rho^{-2}}{(\kappa + 1)\rho^{-2} + m}
	+
	\frac{2(m + \rho^{-2})}{(\kappa + 1)\rho^{-2} + m}
	=
	\frac{3(m + \rho^{-2})}{(\kappa + 1)\rho^{-2} + m},	\\
	\lim_{\gamma \to \infty} \tha
	& = \lim_{\gamma \to \infty}\frac{m + \gamma^2}{(\kappa + 1)\gamma^2 + m}
	+
	2 \lim_{\gamma \to \infty} \left[
		\frac{m + \gamma^2}{\gamma^2 + (\kappa \rho^{-2} + m)}
	\right]    \\
	& = \frac{1}{\kappa + 1} + 2 = \frac{2\kappa + 3}{\kappa + 1}.
\end{align*}
Therefore, we have proved Lemma \ref{lemma:lb_h1}. 
\end{proof}

}

\begin{lem} \label{lem:lb_h2a}
For any given positive $c$, $\rho$, $m$ and any  $\gamma > \rho^{-1}$, define
\begin{align*}
	D_1(\gamma) 
	& = \frac{1}{2} \gamma^4 + \left(\kappa \rho^{-2} + \frac{1}{2}\rho^{-2} + \frac{1}{2} m\right)\gamma^2 + \frac{1}{2} \rho^{-2} m,	\\
	D_2(\gamma)
	& = (\kappa + 2)\gamma^4 + (\kappa m + \kappa \rho^{-2} + 4m)\gamma^2 + m(\kappa \rho^{-2} + 2m), \\
    Z(\gamma) & = \frac{D_1(\gamma) }{D_2(\gamma)}. 
\end{align*}
For all $\gamma > \rho^{-1}$, we have $Z(\gamma) \ge \underline Z$, where 
\begin{align}\label{eq:Z_underline}
    \underline Z
	\equiv
	\min \left\{
		\frac{1}{2(m\rho^2 + 1)},
		\frac{1}{2(\kappa + 2)}
	\right\}
    = \frac{1}{2 \cdot \max\{m\rho^2 + 1, \kappa + 2\}}
    > 0.
\end{align}
\end{lem}

\hideif{
\begin{proof}[Proof of Lemma \ref{lem:lb_h2a}]
To derive the lower bound on $Z(\gamma)$, we will study $\frac{\partial Z(\gamma)}{\partial \gamma^2}$. 
The 
derivatives of $D_1(\gamma)$ and $D_2(\gamma)$ with respect to $\gamma^2$ are, respectively:
\begin{align*}
	\frac{\partial D_1(\gamma)}{\partial \gamma^2} 
	& = 
    \gamma^2 + \kappa \rho^{-2} + \frac{1}{2} (m + \rho^{-2}), 
    \\
	\frac{\partial D_2(\gamma)}{\partial \gamma^2}
	& = 2(\kappa + 2)\gamma^2 + (\kappa m + \kappa \rho^{-2} + 4m) .
\end{align*}

Using the above, we have
\begin{align}
	\frac{\partial Z(\gamma)}{\partial \gamma^2}
	& = \frac{D_2(\gamma) \frac{\partial D_1(\gamma)}{\partial (\gamma^2)} - D_1(\gamma) \frac{\partial D_2(\gamma)}{\partial (\gamma^2)}}{\{D_2(\gamma)\}^2}		\notag	\\
	& = 
	\frac{[m - \rho^{-2} (\kappa + 1)^2] \gamma^4 + 2m (m - \rho^{-2})\gamma^2 + m [\kappa^2 \rho^{-4} + m^2 + m\rho^{-2} (2 \kappa - 1)]}{\{D_2(\gamma)\}^2}		\notag	\\
	& 
	=
	\frac{Z_1 \gamma^4 + Z_2 \gamma^2 + Z_3}{\{D_2(\gamma)\}^2},
	\label{eq:ineq-h-f-1}
\end{align}
where $Z_1 \equiv m - \rho^{-2} (\kappa + 1)^2$, $Z_2 \equiv 2m(m - \rho^{-2})$, and $Z_3 \equiv m[\kappa^2 \rho^{-4} + m^2 + m \rho^{-2}(2\kappa - 1)]$. 
Let $f(\gamma^2) \equiv Z_1 \gamma^4 + Z_2 \gamma^2 + Z_3$ be a quadratic function of $\gamma^2$. In the following, we consider 
three different cases depending on the sign of $Z_1$.

First, suppose that $Z_1 > 0$. Since $\kappa > 0$, it follows that 
\begin{equation}
	Z_2 =
	2m(m - \rho^{-2})
	> 2m [m - \rho^{-2}(\kappa + 1)^2 ]
	= 2m Z_1 > 0.
	\label{eq:ineq-h-f-2}
\end{equation}
In addition, $Z_1 > 0$ implies that $m\rho^2 > (\kappa + 1)^2$, or equivalently, 
\begin{equation}
	\label{eq:kappa-range-z1-1}
	\kappa < \rho\sqrt{m} - 1.
\end{equation}
Note that $Z_3$ can be written as
\begin{align}
	Z_3
	& = m [\kappa^2 \rho^{-4} + 2m\kappa \rho^{-2} + m^2 - m\rho^{-2}]	\notag
        = m [(\kappa \rho^{-2} + m)^2 - m\rho^{-2}]	\notag \\
	& = m (\kappa \rho^{-2} + m + \sqrt{m} \rho^{-1}) (\kappa \rho^{-2} + m - \sqrt{m}\rho^{-1})	\notag \\
	& = m \rho^{-2} (\kappa \rho^{-2} + m + \sqrt{m} \rho^{-1}) (\kappa + m\rho^2 - \sqrt{m}\rho) \notag \\
	& = m \rho^{-2} (\kappa \rho^{-2} + m + \sqrt{m} \rho^{-1}) [\kappa + \sqrt{m}\rho(\sqrt{m}\rho-1)].
	\label{eq:kappa-range-z1-2}
\end{align}
Combining \eqref{eq:kappa-range-z1-1} and \eqref{eq:kappa-range-z1-2}, we have
\begin{align}\label{eq:Z_3_case_1}
    Z_3 > m \rho^{-2} (\kappa \rho^{-2} + m + \sqrt{m} \rho^{-1}) (\kappa + \sqrt{m}\rho\kappa) > 0.
\end{align}
\eqref{eq:ineq-h-f-2} and \eqref{eq:Z_3_case_1} then imply that 
$-\frac{Z_2}{Z_1} < 0$ and $\frac{Z_3}{Z_1} > 0$. 
If the quadratic function $f(\cdot)$ have real roots, then, by Vieta's formulas, the sum of the two roots are negative while the product of roots is positive, implying that both roots of $f(\gamma^2) = 0$ are negative. 
Otherwise, the quadratic function $f(\cdot)$ takes positive values on the whole real line. 
In both cases, we have $f(\gamma^2) > 0$ for all $\gamma > 0$. \par 

Second, consider the case where $Z_1 = 0$.  In this case, $f(\gamma^2) = Z_2 \gamma^2 + Z_3$.  By the same argument as in \eqref{eq:ineq-h-f-2}, we have $Z_2 > 0$. Note that $Z_1 = 0$ implies $\kappa =\sqrt{m} \rho-1$. Using \eqref{eq:kappa-range-z1-2}, we have $Z_3 > 0$. Therefore, $f(\gamma^2) > 0$ for all $\gamma > 0$.\par 

Third, consider the case where $Z_1 < 0$. We consider the following cases. 
\begin{itemize}
    \item Suppose that $f(\cdot)$ has real roots and $Z_2<0$. Then the sum of the two roots of $f(\cdot)$ equals $-\frac{Z_2}{Z_1} < 0$. Thus, there is at most one positive root of $f(\cdot)$.   
    If $f(\cdot)$ has a positive root $\gamma_0^2$, i.e., $f(\gamma_0^2) = 0$, then $f(\gamma^2)$ must be positive when $0<\gamma<\gamma_0$ and negative when $\gamma>\gamma_0$. 
    Otherwise, $f(\gamma^2)<0$ for all $\gamma>0$. 
    
    \item Suppose that $f(\cdot)$ has real roots and $Z_2\ge 0$. 
    By definition, this means $\sqrt{m} \geq \rho^{-1}$, or equivalently, $\sqrt{m}\rho - 1 \geq 0$. From \eqref{eq:kappa-range-z1-2}, because $\kappa > 0$, this implies that $Z_3 > 0$.
    Thus, the product of the two roots of $f(\cdot)$ equals $\frac{Z_3}{Z_1} < 0$. Consequently, $f(\cdot)$ has one positive root and one negative root.  
    Let $\gamma_0^2$ be the positive root of $f(\cdot)$, i.e., $f(\gamma_0^2) = 0$. Then $f(\gamma^2)$ must be positive when $0<\gamma<\gamma_0$ and negative when $\gamma>\gamma_0$. 

    \item Suppose that $f(\cdot)$ does not have real roots. 
    Then the quadratic function $f(\cdot)$ takes negative values on the whole real line. 
\end{itemize}

From the above, one of the following must hold for $f(\cdot)$: 
\begin{itemize}
    \item[(i)] $f(\gamma^2) > 0$ for all $\gamma>0$; 
    \item[(ii)] $f(\gamma^2) < 0$ for all $\gamma>0$;
    \item[(iii)] for some $\gamma_0>0$, $f(\gamma^2) > 0$ for $0<\gamma<\gamma_0$, $f(\gamma^2_0) = 0$, and $f(\gamma^2) < 0$ for $\gamma>\gamma_0$. 
\end{itemize}
From \eqref{eq:ineq-h-f-1}, we can know that 
\begin{itemize}
    \item[(i)] $Z(\gamma)$ is either increasing or decreasing in $\gamma\in (\rho^{-1}, \infty)$; 
    \item[(ii)] for some $\gamma_0 > \rho^{-1}$, $Z(\gamma)$ is increasing for $\gamma \in (\rho^{-1}, \gamma_0)$ and decreasing for $\gamma \in (\gamma_0, \infty)$
\end{itemize}
Thus, the infimum of $Z(\gamma)$ over $\gamma\in  (\rho^{-1}, \infty)$ must be obtained at the endpoints, which immediately implies that, for all $\gamma>\rho^{-1}$, 
\begin{equation}
	\label{eq:ineq-h-f-3}
	Z(\gamma) \geq 
	\min \left\{
		Z(\rho^{-1}),
		\ \lim_{\gamma \to \infty} Z(\gamma)
	\right\},
\end{equation}
where
\begin{align*}
	Z(\rho^{-1})
	& 
	= 
	\frac{\rho^{-4}(\kappa + \rho^2m+1)}{2\rho^{-4}(m\rho^2 + 1)(\kappa + m\rho^2 + 1)}	
	= \frac{1}{2(m\rho^2 + 1)},	\\
	\lim_{\gamma \to \infty} Z(\gamma)
	& =  \frac{1}{2(\kappa + 2)}.
\end{align*}
By definition, this then implies that $Z(\gamma)\ge \underline{Z}$ for all $\gamma>\rho^{-1}$. 
Moreover, we have $\underline Z > 0$ because $\kappa, m, \rho > 0$. 
Therefore, we derive Lemma \ref{lem:lb_h2a}. 
\end{proof}
}

\begin{lem} \label{lemma:lb_h2}
For any given values of positive $m\ge 4$, $c> \sqrt{\frac{2(m-1)}{m(m-2)}}$ and $\rho$, 
define $\thb$ as in \eqref{eq:H_tilde_123} for $\gamma>\rho^{-1}$. 
For any $\gamma>\rho^{-1}$, we have 
$\thb \leq \ohb$, where 
\begin{align}\label{eq:H_m2_overline}
	\ohb
	\equiv
\frac{1-\tau}{1-\tau +\min\{\left[(1 - 2\tau)\kappa \underline Z - \frac{1}{2} \right] , 0\}}, 
\end{align}
and $\underline Z$ is defined as in \eqref{eq:Z_underline}. 
\end{lem}

\hideif{
\begin{proof}[Proof of Lemma \ref{lemma:lb_h2}]
Recall that $\thb \equiv  \frac{1-\tau}{1-\tau+\min\{C, 0\}}$ with 
\begin{align*}
	C 
	& \equiv 
	\left( 1-2\tau + \frac{2}{\kappa \gamma^2} \right) \frac{\kappa^2 \gamma^2 \rho^{-2} -\overline{\theta}^2}{(\kappa \gamma^2 + \kappa \rho^{-2}+2\overline{\theta})\overline{\theta}} +  
    \frac{1}{\kappa \gamma^2}
    - \tau
    \ \ 
    \text{ and } \ \ 
    \overline{\theta} = m+\gamma^2. 
\end{align*}
Define $D_1(\gamma)$ and $D_2(\gamma)$ as in \eqref{lem:lb_h2a}. 
We have 
\begin{align*}
    (\kappa \gamma^2 + \kappa \rho^{-2}+2\overline{\theta})\overline{\theta}
    & = 
    [(\kappa+2)\gamma^2 + 2m + \kappa \rho^{-2}](\gamma^2+m)
    \\
    & = D_2(\gamma), 
\end{align*}
and 
\begin{align}
    (\kappa^2 \gamma^2 \rho^{-2} - \overline\theta^2)
    + \frac{1}{2}(\kappa \gamma^2 + \kappa \rho^{-2} + 2\overline\theta)\overline\theta 
    & = \kappa^2 \gamma^2 \rho^{-2} + \frac{1}{2}\kappa \gamma^2 \overline\theta + \frac{1}{2}\kappa \rho^{-2} \overline\theta
    \equiv D_1(\gamma) \cdot \kappa > 0.
    \label{eq:ineq-h-c-1}
\end{align}
By the same logic as the proof of Lemma \ref{lem:sign_deriv_gamma}, 
because $c\ge \sqrt{\frac{2(m-1)}{m(m-2)}}$, we have $\tau < 1/2$.
Hence, we have
\begin{align}
	C 
	& = 
	\left( 1-2\tau + \frac{2}{\kappa \gamma^2} \right) \frac{D_1(\gamma)\kappa - \frac{1}{2}D_2(\gamma)}{D_2(\gamma)} +  
    \frac{1}{\kappa \gamma^2} 
    - \tau \notag   = \left( 1-2\tau + \frac{2}{\kappa \gamma^2} \right) 
 \frac{D_1(\gamma) \kappa}{D_2(\gamma)}
 - \frac{1}{2} \notag   \\
 & \geq (1 - 2\tau)  \frac{D_1(\gamma) \kappa}{D_2(\gamma)}
 - \frac{1}{2}
 = (1 - 2\tau)\kappa  Z(\gamma)
 - \frac{1}{2}
 \ge (1 - 2\tau)\kappa  \underline{Z} - \frac{1}{2},  \label{eq:ineq-h-c-1b}
\end{align}
where the last inequality follows from Lemma \ref{lem:lb_h2a} and that $\tau < 1/2$.
This then implies that 
\begin{align*}
    1-\tau + \min\{C, 0\} \ge 1 - \tau + \min\left\{ \left[(1 - 2\tau)\kappa \underline Z - \frac{1}{2} \right] , 0 \right\} > 0, 
\end{align*}
and consequently
\begin{equation}
	\label{eq:ineq-h-f-5}
	\thb
	=
	\frac{1-\tau}{1-\tau + \min\{C, 0\}} 
	\leq \frac{1-\tau}{1-\tau +\min\{\left[(1 - 2\tau)\kappa \underline Z - \frac{1}{2} \right] , 0\}}
	=
	\ohb.
\end{equation}
Therefore, Lemma \ref{lemma:lb_h2} holds. 
\end{proof}

}

\begin{lem}\label{lem:mon_h1}
For any given values of positive $c$, $\rho$, and $m\ge 2$, define $\oha$ as in \eqref{eq:H_m1_overline}. 
We have that $\oha$ is decreasing in $c>0$.
\end{lem}

\hideif{
\begin{proof}[Proof of Lemma \ref{lem:mon_h1}]
By definition, 
\[
	\oha
	= \max
	\left\{ \frac{3(m + \rho^{-2})}{(\kappa + 1)\rho^{-2} + m},
	\frac{2\kappa + 3}{\kappa + 1}\right\},
\]
where $\kappa = \frac{m c^2}{m-1}$ is increasing in $c>0$. 
Note that both $\frac{3(m + \rho^{-2})}{(\kappa + 1)\rho^{-2} + m}$ and $\frac{2\kappa + 3}{\kappa + 1} = 2 + \frac{1}{\kappa + 1}$ are decreasing in $\kappa$. 
These then imply that $\oha$ is decreasing in $\kappa$. 
Consequently, $\oha$ is decreasing in $c>0$. 
Therefore, Lemma \ref{lem:mon_h1} holds. 
\end{proof}
}

\begin{lem} \label{lem:mon_h2}
For any given values of positive $m\ge 4$, $c> \sqrt{\frac{2(m-1)}{m(m-2)}}$ and $\rho$, define $\ohb$ as in \eqref{eq:H_m2_overline}. 
We have that $\ohb$ is nonincreasing in $c$. 
\end{lem}

\hideif{
\begin{proof}[Proof of Lemma \ref{lem:mon_h2}]
Because $\kappa$ is increasing in $c>0$, as shown in the proof of Lemma \ref{lem:mon_h1}, 
it suffices to prove that $\ohb$ is nonincreasing in $\kappa$. 
By definition, 
\[
	\ohb
	= 
	\frac{1-\tau}{1-\tau +\min\{(1 - 2\tau) \kappa \underline Z - \frac{1}{2} , 0\}},
\]
where
\begin{equation*}
    \kappa = \frac{m c^2}{m-1}, \qquad 
    \tau = \frac{\kappa+1}{m\kappa}, \qquad 
	\underline Z
	=
	\min \left\{
		\frac{1}{2(m\rho^2 + 1)},  \frac{1}{2(\kappa + 2)}
	\right\}. 
\end{equation*}
Note that 
\[
	(1 - 2\tau)\kappa \underline Z - \frac{1}{2}
	=
	(1 - 2\tau)\min \left\{
		\frac{\kappa}{2(m\rho^2 + 1)},  \frac{\kappa}{2(\kappa + 2)}
	\right\}
	-\frac{1}{2}.
\]
Obviously, $\frac{\kappa}{2(m\rho^2 + 1)}$ and $\frac{\kappa}{2(\kappa + 2)}$ are increasing in $\kappa$, which implies that $\kappa \underline{Z}$ is increasing in $\kappa$. 
In addition, $\tau = \frac{\kappa+1}{m\kappa} = \frac{1}{m}+\frac{1}{m\kappa}$ is decreasing in $\kappa$. 
Thus, $(1 - 2\tau)\kappa \underline Z - \frac{1}{2}$ is increasing in $\kappa$. 
Below we consider two cases 
depending on whether there exists $\kappa_0>0$ such that $(1 - 2\tau) \kappa \underline Z - \frac{1}{2} = 0$ when evaluated at $\kappa = \kappa_0$. 

First, assume that such $\kappa_0>0$ exists. Since $(1 - 2\tau)\kappa \underline Z - \frac{1}{2}$ is increasing in $\kappa$, we can know that $(1 - 2\tau) \kappa \underline Z - \frac{1}{2} < 0$ for $\kappa < \kappa_0$ and $(1 - 2\tau) \kappa \underline Z - \frac{1}{2} > 0$ for $\kappa > \kappa_0$. Thus,
for $\kappa\ge \kappa_0$, 
\begin{equation}
	\label{eq:ohb-k0-1}
	\ohb
	= \frac{1-\tau}{1-\tau}  = 1,
\end{equation}
which is a constant over all $\kappa \geq \kappa_0$. 
For $\kappa < \kappa_0$, 
\begin{equation}
	\label{eq:ohb-k0-2}
	\ohb
	=
	\frac{1-\tau}{1-\tau +\left[(1 - 2\tau) \kappa \underline Z - \frac{1}{2} \right]} 
    = \frac{2(1-\tau)}{(1-2\tau)(1+2\kappa\underline{Z})}
    = 
    \frac{1}{1+2\kappa \underline{Z}} \left( 1 + \frac{1}{1-2\tau} \right).
\end{equation}
From the discussion before, $\kappa \underline{Z}$ is increasing in $\kappa$, and $\tau$ is decreasing in $\kappa$. 
These imply that \eqref{eq:ohb-k0-2} must be decreasing in $\kappa \in (0, \kappa_0)$. 
Therefore, $\ohb$  is nonincreasing in $\kappa$.

Second, assume that such $\kappa$ does not exist. 
Then $(1 - 2\tau) \kappa \underline Z - \frac{1}{2}$ is either positive for any $\kappa > 0$, or negative for any $\kappa > 0$. 
Below we consider the two cases, separately. 
\begin{itemize}
    \item Suppose that $(1 - 2\tau) \kappa \underline Z - \frac{1}{2} > 0$ for any $\kappa > 0$. Then $\ohb$ has the form in \eqref{eq:ohb-k0-1}. This implies that $\ohb$ is constant over all $\kappa>0$.

    \item Suppose that $(1 - 2\tau) \kappa \underline Z - \frac{1}{2} < 0$ for any $\kappa > 0$. 
    Then $\ohb$ has the form in \eqref{eq:ohb-k0-2}. For the same reason as in the first case, $\ohb$ is decreasing in $\kappa>0$. 
\end{itemize}

From the above, Lemma \ref{lem:mon_h2} holds. 
\end{proof}

}

\begin{lem} \label{lem:mon_h3}
For given values of positive $c$, $\rho$, and $m\ge 2$,
define $\thc$ as in \eqref{eq:H_tilde_123}. 
We have that $\thc$ is decreasing in $c>0$
\end{lem}

\hideif{
\begin{proof}[Proof of Lemma \ref{lem:mon_h3}]
Because $\kappa$ is increasing in $c>0$, as shown in the proof of Lemma \ref{lem:mon_h1}, 
it suffices to prove that $\thc$ is decreasing in $\kappa$.
By definition, 
\begin{align*}
    \frac{\partial}{\partial \kappa} \thc \equiv - (\underline{\theta} - m) \cdot \frac{\partial}{\partial \kappa} \left( \frac{1}{\tau \underline{\theta} - 1} \right)
    = 
    (\underline{\theta} - m) \cdot \frac{\underline{\theta}}{(\tau \underline{\theta} - 1)^2} \frac{\partial \tau}{\partial \kappa}, 
\end{align*}
where $\underline \theta = m + \rho^{-2}$ and $\tau = \frac{\kappa+1}{m\kappa}$. 
Because $\underline\theta - m = \rho^{-2} > 0$, $\underline\theta > 0$, and $\tau = \frac{1}{m}+\frac{1}{m\kappa}$ is decreasing in $\kappa$, 
we can know that $\frac{\partial}{\partial \kappa} \thc < 0$. 
Consequently,  $\thc$ is decreasing in $\kappa$. 
Therefore, Lemma \ref{lem:mon_h3} holds. 
\end{proof}
}

\begin{proof}[Proof of Theorem \ref{thm:max_rej_prob_k_1_simple_closed_form}]
Define $\thh$ as in Theorem \ref{thm:max_rej_prob_k_1_closed_form}, $\tha$, $\thb$ and $\thc$ as in \eqref{eq:H_tilde_123}, $\oha$ as in Lemma \ref{lemma:lb_h1}, and $\ohb$ as in Lemma \ref{lemma:lb_h2}. 
By definition and from Lemmas \ref{lemma:lb_h1} and \ref{lemma:lb_h2}, for any $c> \sqrt{\frac{2(m-1)}{m(m-2)}}$, 
\begin{align}\label{eq:bound_H_tilde_gamma}
    \thh
    & =
    \tha
    + \thb
    + \thc
    - 1
    \nonumber
    \\
    & \le \oha
	+ \ohb
	+ \thc
	- 1. 
\end{align}
In addition, by definition, 
\begin{align}\label{eq:H_overline}
    & \quad \ \oha
	+ \ohb
	+ \thc
	- 1
    \nonumber
    \\
    & = 
    \max\left\{
    \frac{3(m + \rho^{-2})}{(\kappa + 1)\rho^{-2} + m}, 
    \frac{2\kappa + 3}{\kappa + 1}
    \right\}
    + 
    \frac{1-\tau}{1-\tau +\min\{(1 - 2\tau)\kappa \underline Z - \frac{1}{2} , 0\}}
    - \frac{\underline{\theta} - m}{\tau \underline{\theta} - 1}
    - 1
    \nonumber
    \\
    & = 
    \max\left\{
    \frac{3(m\rho^2 + 1)}{m\rho^{2} + \kappa + 1}, 
    \frac{2\kappa + 3}{\kappa + 1}
    \right\}
    + 
    \frac{1-\tau}{1-\tau +\min\{(1 - 2\tau)\kappa \underline Z - \frac{1}{2}, 0\}}
    - \frac{\rho^{-2}}{\tau m + \tau \rho^{-2}  - 1}
    - 1
    \nonumber
    \\
    & = 
    \max\left\{
    \frac{3(m\rho^2 + 1)}{m\rho^{2} + \kappa + 1}, 
    \frac{2\kappa + 3}{\kappa + 1}
    \right\}
    + 
    \frac{1-\tau}{1-\tau +\min\{(1 - 2\tau)\kappa \underline Z - \frac{1}{2}, 0\}}
    - \frac{m\kappa}{m\rho^{2} + \kappa+1}
    - 1
    \nonumber
    \\
    & = \ohh, 
\end{align}
where $\kappa$, $\tau$, $\underline{\theta}$ is defined as in \eqref{eq:kappa_theta_over_under_tau_C}, 
$\underline{Z}$ is defined as in \eqref{eq:Z_underline}, 
the first and last equalities follow by definition, 
and  
the second and third equalities follow by some algebra.

First, from Lemmas \ref{lem:mon_h1}--\ref{lem:mon_h3}, 
$\ohh = \oha
	+ \ohb
	+ \thc
	- 1$
is decreasing in $c$. 
Moreover, as $c\longrightarrow \infty$, 
we have $\kappa = mc^2/(m-1) \longrightarrow \infty$, 
$\tau = (\kappa+1)/(m\kappa) \longrightarrow 1/m$, 
\begin{align*}
    \kappa \underline{Z} & = \frac{\kappa}{2 \cdot \max\{m\rho^2 + 1, \kappa + 2\}}
    \longrightarrow \frac{1}{2}, 
\end{align*}
and consequently 
\begin{align}\label{eq:limit_H_overline}
    & \hspace{-10pt} \ohh 
    \notag \\ 
    & = 
    \max\left\{
    \frac{3(m\rho^2 + 1)}{m\rho^{2} + \kappa + 1}, 
    \frac{2\kappa + 3}{\kappa + 1}
    \right\}
    + 
    \frac{1-\tau}{1-\tau +\min\{(1 - 2\tau)\kappa \underline Z - \frac{1}{2}, 0\}}
    - \frac{m\kappa}{m\rho^{2} + \kappa+1}
    - 1
    \nonumber
    \\
    & 
    \longrightarrow 
    \max\left\{
    0, 
    2
    \right\}
    + 
    \frac{1-m^{-1}}{1-m^{-1} +\min\{(1 - 2m^{-1}) \frac{1}{2} - \frac{1}{2}, 0\}}
    - m
    - 1
    \nonumber
    \\
    & = 1 - m + \frac{1-m^{-1}}{1-m^{-1} - m^{-1}}
    = 1 - m + \frac{m-1}{m-2}
    = \frac{(1-m)(m-3)}{m-2} < 0, 
\end{align}
where the last inequality holds because $m\ge 4$.

Second, by definition and from \eqref{eq:limit_H_overline}, we know that $\underline{c}_{m, \rho} = \inf\{c > \sqrt{\frac{3(m-1)}{m(m-3)}}: \ohh \le 0 \}$ must be finite.
Because $\ohh$ is decreasing in $c$, we must have $\ohh < 0$ for any $c> \underline{c}_{m, \rho}$. 
From \eqref{eq:bound_H_tilde_gamma} and \eqref{eq:H_overline}, we then have, for any given $m\ge 4$, $\rho>0$,  $c> \underline{c}_{m, \rho} \ge \sqrt{\frac{3(m-1)}{m(m-3)}} > \sqrt{\frac{2(m-1)}{m(m-2)}}$,  
\begin{align*}
    \thh \le \ohh < 0 \ \ \text{ for all } \gamma > \rho^{-1}. 
\end{align*}
From Theorem \ref{thm:max_rej_prob_k_1_closed_form}, this implies that, for any given $m\ge 4$, $\rho>0$,  $c> \underline{c}_{m, \rho}$, the maximum rejection probability $p_m(c; 1, \rho)$ under Assumption \ref{assu:relative_heter} with $k=1$ and the given $\rho$, $m$, and $c$ has the following equivalent form:
\begin{align*}
    p_m(c; 1, \rho) = 
    \bP \left[ |t_{m-1}|\sqrt{\rho^2+\frac{1}{m}} > c \right].
\end{align*}

Third, we consider the case when $c = \underline{c}_{m, \rho}$. 
From Lemma \ref{lem:first_second_deriv_gamma}, there exist $\{\sigma_1, \ldots, \sigma_{m+1}\} \in \mathcal{S}_m(1, \rho)$ such that 
$\bP_0[|T_m| > \underline{c}_{m, \rho}] = p_m(\underline{c}_{m, \rho}; 1, \rho)$. 
By the right continuity of distribution functions, we can know 
\begin{align}\label{eq:max_rej_c_m_rho}
    p_m(\underline{c}_{m, \rho}; 1, \rho)
    & = 
    \bP_0[|T_m| > \underline{c}_{m, \rho}]
    \notag \\
    & = 
    \lim_{\epsilon \rightarrow 0+} \bP_0[|T_m| > \underline{c}_{m, \rho}+\epsilon]
   \notag  \\
    & \le 
    \lim_{\epsilon \rightarrow 0+} 
    \bP \left[ |t_{m-1}|\sqrt{\rho^2+\frac{1}{m}} > \underline{c}_{m, \rho}+\epsilon \right]
    \nonumber
    \\
    & = \bP \left[ |t_{m-1}|\sqrt{\rho^2+\frac{1}{m}} > \underline{c}_{m, \rho}\right],
\end{align}
where the last inequality follows from the second part. 
Moreover, the rejection probability can reach the upper bound on the right-hand side of \eqref{eq:max_rej_c_m_rho} when $\sigma_{m+1} = \rho$ and $\sigma_1 = \ldots = \sigma_m = 1$. 
Thus, \eqref{eq:max_rej_c_m_rho} must hold with equality. 

From the above, Theorem \ref{thm:max_rej_prob_k_1_simple_closed_form} holds. 
\end{proof}

\subsection{Proof of Theorem \ref{thm:power}}

\begin{proof}[Proof of Theorem \ref{thm:power}]
Because $c>0$, we obviously have $\bP[|T_m|>c] \ge \bP[T_m>c] = 1 - \bP[T_m \le c]$. 
Note that 
\begin{align*}
    T_m \le c 
    & \quad \Longleftrightarrow \quad
    (\psi_{m+1} - \delta) 
    + \delta  - \overline{\psi} \le cS_m \\
    & \quad \Longleftrightarrow \quad
    cS_m + \overline{\psi} - (\psi_{m+1} - \delta) \ge \delta\\
    &  \quad\Longrightarrow  \quad
    \{ cS_m + \overline{\psi} - (\psi_{m+1} - \delta) \}^2 \ge \delta^2, 
\end{align*}
where the last step holds because $\delta>0$. 
By the Markov inequality, we then have 
\begin{align*}
    \bP[T_m \le c ] \le 
    \bP[
    \{ cS_m + \overline{\psi} - (\psi_{m+1} - \delta) \}^2 \ge \delta^2    
    ]
    \le \delta^{-2} 
    \bE [
    \{ cS_m + \overline{\psi} - (\psi_{m+1} - \delta) \}^2
    ]. 
\end{align*}
Under Assumption \ref{assu:normal}, 
we have 
\begin{align*}
    \bE [
    \{ cS_m + \overline{\psi} - (\psi_{m+1} - \delta) \}^2
    ]
    & = \bE [
    ( cS_m + \overline{\psi} )^2
    ]
    + 
    \bE[ (\psi_{m+1} - \delta)^2 ]
    \\
    & \le 2 \bE [
    c^2 S_m^2
    ]
    + 2 \bE [
    \overline{\psi}^2
    ]
    + 
    \bE[ (\psi_{m+1} - \delta)^2 ],
\end{align*}
where the first equality holds because $\psi_{m+1} - \delta\sim \mathcal{N}(0, \sigma^2_{m+1})$ and is independent of $\{\psi_j\}_{j=1}^m$, 
and the second equality follows from the Cauchy–Schwarz inequality.  
By some algebra, 
$\bE[ (\psi_{m+1} - \delta)^2 ] = \Var(\psi_{m+1}) = \sigma^2_{m+1}$, 
$\bE [
    \overline{\psi}^2
    ]
    =
\Var(\overline{\psi}) = \frac{1}{m^2} \sum_{j=1}^m \sigma_j^2, 
    $
\begin{align*}
    \bE [
    S_m^2
    ]
    & = \frac{1}{m-1} \sum_{j=1}^m \bE[\psi_j^2] - \frac{m}{m-1} \bE[ \overline{\psi}^2 ]
    = \frac{1}{m-1} \sum_{j=1}^m \sigma_j^2 - \frac{m}{m-1} \frac{1}{m^2} \sum_{j=1}^m \sigma_j^2
    = \frac{1}{m} \sum_{j=1}^m \sigma_j^2. 
\end{align*}
These then imply that 
\begin{align*}
    \bE [
    \{ cS_m + \overline{\psi} - (\psi_{m+1} - \delta) \}^2
    ]
    & \le 2 \bE [
    c^2 S_m^2
    ]
    + 2 \bE [
    \overline{\psi}^2
    ]
    + 
    \bE[ (\psi_{m+1} - \delta)^2 ]\\
    & = 2c^2  \frac{1}{m} \sum_{j=1}^m \sigma_j^2
    + 
    \frac{2}{m^2} \sum_{j=1}^m \sigma_j^2 + \sigma^2_{m+1}
    \\
    & = 
    \sigma^2_{m+1} + \frac{2(c^2 + m^{-1})}{m} \sum_{j=1}^m \sigma_j^2.  
\end{align*}
From the above, we have 
\begin{align*}
    \bP[T_m \le c ] 
    \le \delta^{-2} 
    \bE [
    \{ cS_m + \overline{\psi} - (\psi_{m+1} - \delta) \}^2
    ]
    \le 
    \frac{1}{\delta^2} \left[ \sigma^2_{m+1} + \frac{2(c^2 + m^{-1})}{m} \sum_{j=1}^m \sigma_j^2 \right], 
\end{align*}
and consequently 
\begin{align*}
    \bP[|T_m|>c] \ge \bP[T_m>c] = 1 - \bP[T_m \le c]
    \ge 
    1 - \frac{1}{\delta^2} \left[ \sigma^2_{m+1} + \frac{2(c^2 + m^{-1})}{m} \sum_{j=1}^m \sigma_j^2 \right]. 
\end{align*}
Therefore, Theorem \ref{thm:power} holds. 
\end{proof}

\begin{proof}[Comment on the number of one-dimensional optimizations needed]
We first count the number of one-dimensional optimizations needed at a given $k$. 
From Theorem \ref{thm:max_rej_prob}, we need to consider the one-dimensional optimization in \eqref{eq:max_rej_m1_m0} for $0\le m_0 \le k-1$ and $0\le m_1 \le m-m_0-1$; note that when $m_1 = m-m_0$ no optimization is needed. Thus, the number of  one-dimensional optimization needed is 
\begin{align*}
    \sum_{m_0 = 0}^{k-1} (m-m_0) 
    & = \sum_{j=m-k+1}^{m} j = \frac{1}{2}k(2m+1-k). 
\end{align*}

We then count the total number of one-dimensional optimizations needed to compute $p_m(c; k, \rho)$ for all $1 \le k \le m$. For each combination of $(m_1, m_0, m_{\text{c}})$ satisfying $m_1 \ge 0$, $m_0 \ge 0$, $m_{\text{c}} \ge 1$, and $m_1 + m_0 + m_{\text{c}} = m$, we solve the one-dimensional optimization in \eqref{eq:max_rej_m1_m0} at most twice. Intuitively, $m_1$, $m_0$, and $m_{\text{c}}$ correspond to the numbers of ${\gamma_j}_{j=1}^m$ that are equal to $\rho^{-1}$, $0$, and a common value $\gamma$, respectively. The optimization is solved at most twice because $\gamma$ may vary from $0$ or $\rho^{-1}$ to infinity. By some algebra, the total number of one-dimensional optimizations needed is at most 
\begin{align*}
    2 \cdot \binom{m+1}{2} = m(m+1).
\end{align*}
\end{proof}

\subsection{Proof of Theorem \ref{thm:simu_ci_rho}}

\begin{proof}[Proof of Theorem \ref{thm:simu_ci_rho}]
By the definition of $p_m(c; k, \rho)$ in \eqref{eq:max_rej_prob}, we know that it is nondecreasing in $\rho$. Thus, the set $\mathcal{I}_{m, \alpha, k}$ in \eqref{eq:interval_rho_k} must be an one-sided interval. 
To prove Theorem \ref{thm:simu_ci_rho}, it suffices to prove the simultaneous validity of the confidence intervals $\mathcal{I}_{m, \alpha, k}$s as stated in (ii). 
Let $G_0(c) = \bP_0[|T_m|>c]$ denote the true tail probability of $|T_m|$ at the true values of $\{\sigma_j\}_{j=1}^{m+1}$. 

Suppose that $\rho^\star_k \notin \mathcal{I}_{m, \alpha, k}$ for some $k\in \{1,2,\ldots, m\}$. 
By definition, we then have $p_m(|T_m|; k, \rho^\star_k) \le \alpha$ and $\rho^\star_k < \infty$. 
By the definition of $\rho^\star_k$ in \eqref{eq:rho_star}, 
we must have $\sigma_{m+1} \le \rho_k^\star \sigma_{(k)}$
for the true standard deviation of the treated cluster and that of the control cluster at rank $k$. 
By the definition in \eqref{eq:max_rej_prob}, this then implies that  
$
p_m(c; k, \rho^\star_k) \ge G_0(c) 
$
for any $c\ge 0$. 
Consequently, we must have 
$\alpha \ge p_m(|T_m|; k, \rho^\star_k) \ge G_0(|T_m|)$. 

From the discussion before, we then have 
\begin{align*}
    \bP\left[ \rho^\star_k \notin \mathcal{I}_{m, \alpha, k} \text{ for some } 1\le k \le m \right] 
    \le 
    \bP\left[ G_0(|T_m|) \le \alpha \right] \le \alpha,  
\end{align*}
where the last inequality holds because $G_0(\cdot)$ is the tail probability function of $|T_m|$ (see, e.g., Lemma A4 in \citet{wuli2025jasa} for a proof). 
This then implies that  
\begin{align*}
    \bP\left[ \rho^\star_k \in \mathcal{I}_{m, \alpha, k} \text{ for all } 1\le k \le m \right]
    & = 1 - \bP\left[ \rho^\star_k \notin \mathcal{I}_{m, \alpha, k} \text{ for some } 1\le k \le m \right] 
    \ge 1-\alpha,  
\end{align*}
Therefore, we derive Theorem \ref{thm:simu_ci_rho}. 
\end{proof}

\subsection{Proof of Theorem \ref{thm:limit_T_m_large}}

\begin{proof}[Proof of Theorem \ref{thm:limit_T_m_large}]
From Condition \ref{cond:large_m}, as $m\longrightarrow \infty$,  
\begin{align*}
    \Var(\overline{\psi}) & = \frac{1}{m^2} \sum_{j=1}^m \sigma^2_j = \frac{1}{m} \cdot \frac{1}{m} \sum_{j=1}^m \sigma^2_j \longrightarrow 0, 
\end{align*}
and 
\begin{align*}
    \Var\left( \frac{1}{m} \sum_{j=1}^m \psi_j^2 \right)
    & = \frac{1}{m^2} \sum_{j=1}^m \Var (\psi_j^2) = \frac{1}{m^2} \sum_{j=1}^m (2 \sigma^4) = \frac{2}{m^2} \sum_{j=1}^m \sigma^4 \longrightarrow 0, 
\end{align*}
where the second last equality uses the fact that the variance of a chi-squared random variable with degree of freedom 1 is 2. 
By Chebyshev's inequality, these then imply that 
\begin{align*}
    \overline{\psi} = o_{\bP}(1), 
    \quad 
    \frac{1}{m} \sum_{j=1}^m \psi_j^2 = \bE \left[ \frac{1}{m} \sum_{j=1}^m \psi_j^2 \right] + o_{\bP}(1) =  \frac{1}{m} \sum_{j=1}^m \sigma_j^2+ o_{\bP}(1). 
\end{align*}
and consequently
\begin{align*}
    S_m^2 = \frac{m}{m-1} \frac{1}{m} \sum_{j=1}^m \psi_j^2 - \frac{m}{m-1} \overline{\psi}^2 
    = \frac{1}{m} \sum_{j=1}^m \sigma_j^2+ o_{\bP}(1). 
\end{align*}
These further imply that 
\begin{align*}
    T_m = \frac{\psi_{m+1} - o_{\bP}(1)}{\sqrt{m^{-1}\sum_{j=1}^m \sigma_j^2+o_{\bP}(1)}} = 
    \frac{\psi_{m+1}}{\sqrt{m^{-1}\sum_{j=1}^m \sigma_j^2}} + o_{\bP}(1), 
\end{align*}
where the last equality holds under Condition \ref{cond:large_m}. 
Let $\varepsilon \equiv \frac{\psi_{m+1} - \delta}{\sigma_{m+1}} \sim \mathcal{N}(0,1)$. Note that 
\begin{align*}
    \frac{\psi_{m+1}}{\sqrt{m^{-1}\sum_{j=1}^m \sigma_j^2}} 
    = 
    \frac{\sigma_{m+1} \varepsilon + \delta}{\sqrt{m^{-1}\sum_{j=1}^m \sigma_j^2}}
    = \frac{\sigma_{m+1} \varepsilon}{\sqrt{m^{-1}\sum_{j=1}^m \sigma_j^2}}
    + 
    \frac{\delta}{\sqrt{m^{-1}\sum_{j=1}^m \sigma_j^2}}.
\end{align*}
Using \citet[Lemma A27][]{wang2022rerandomization}, we can then derive that, as $m\longrightarrow \infty$, 
\begin{align*}
    \sup_{c\in \mathbb{R}}\left| \bP[T_m \le c] - 
    \bP\Bigg[ \frac{\sigma_{m+1} \varepsilon + \delta}{\sqrt{m^{-1}\sum_{j=1}^m \sigma_j^2}} \le c \Bigg] \right|
    \longrightarrow 0.
\end{align*}
Therefore, Theorem \ref{thm:limit_T_m_large} holds. 

\end{proof}

\section{Supplemental results for Section \ref{sec:sims}} \label{sec:app:sims}

\subsection{Supplemental results for Section \ref{sec:sims-1}} \label{sec:app:sims-1}
This section reports additional results from simulations on the first simulation design with $k = 1$ and $k = 2$ at various significance levels.

\subsubsection{Results for $k = 1$}
This subsection reports the simulation results for $k = 1$ that are not contained in the main paper. Figures \ref{fig:sim1-rho-a01} and \ref{fig:sim1-rho-a10} report the results at the 1\% and 10\% levels respectively for \textbf{DGP 1}. Figures \ref{fig:sim1-dgp2-rho-a01} and \ref{fig:sim1-dgp2-rho-a10} report the results at the 1\% and 10\% levels respectively for \textbf{DGP 2}.

\begin{figure}
    \centering
    \includegraphics[scale=1]{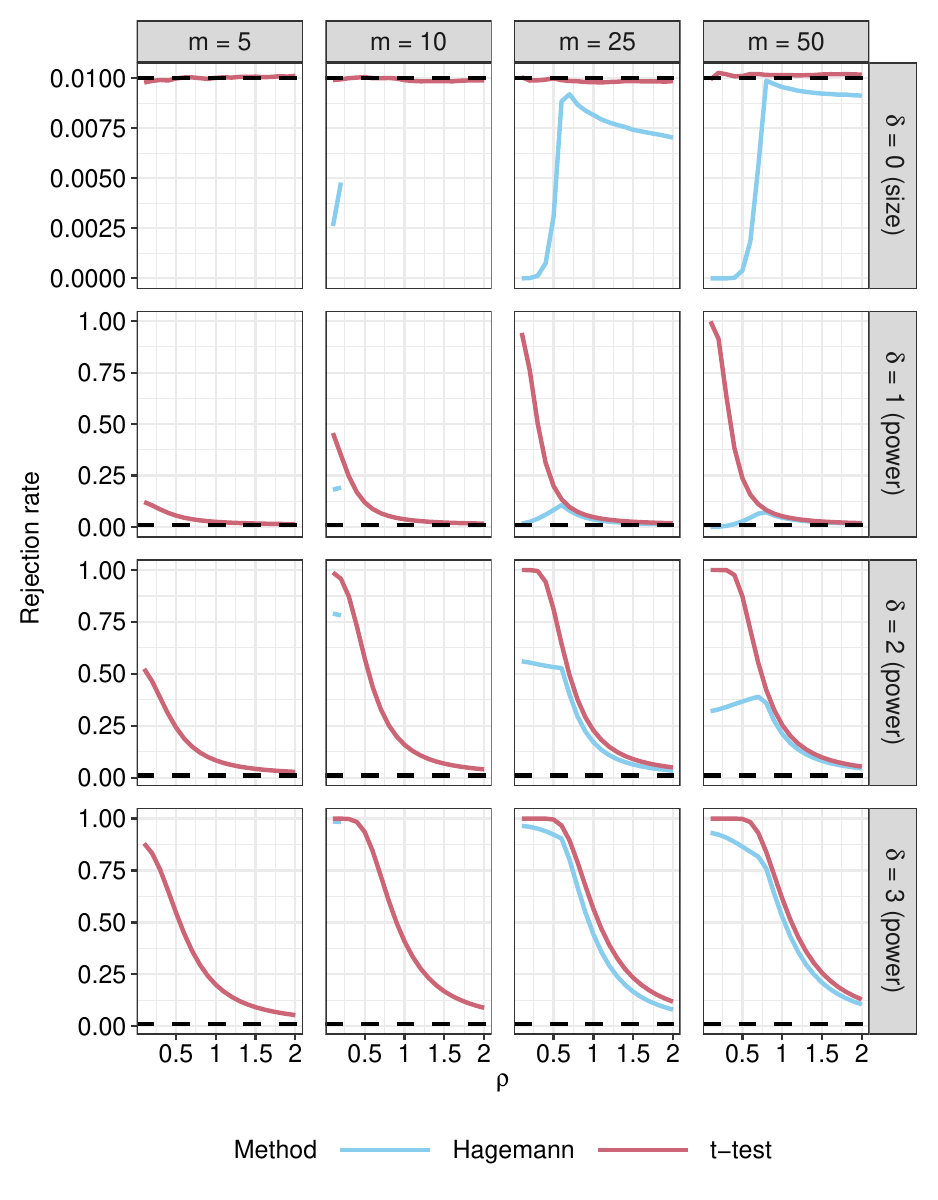}
    \caption{Probability of rejection against heterogeneity parameter $\rho$ for various cluster size $m$ and alternatives $\delta$ at $\alpha = 0.01$ for \textbf{DGP 1} of simulation design 1 with $k = 1$.}
    \label{fig:sim1-rho-a01}
\end{figure}

\begin{figure}
    \centering
    \includegraphics[scale=1]{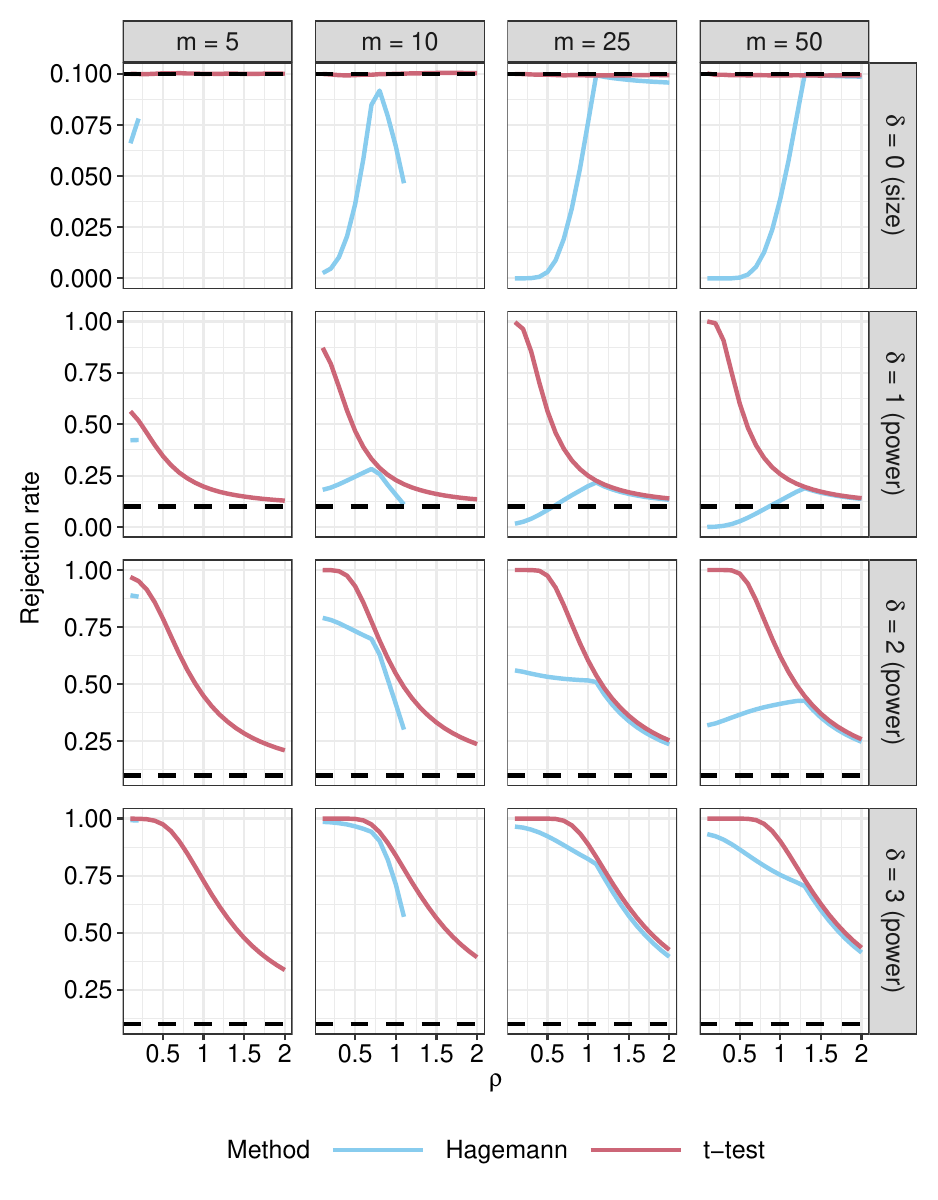}
    \caption{Probability of rejection against heterogeneity parameter $\rho$ for various cluster size $m$ and alternatives $\delta$ at $\alpha = 0.1$ for \textbf{DGP 1} of simulation design 1 with $k = 1$.}
    \label{fig:sim1-rho-a10}
\end{figure}

\begin{figure}
    \centering
    \includegraphics[scale=1]{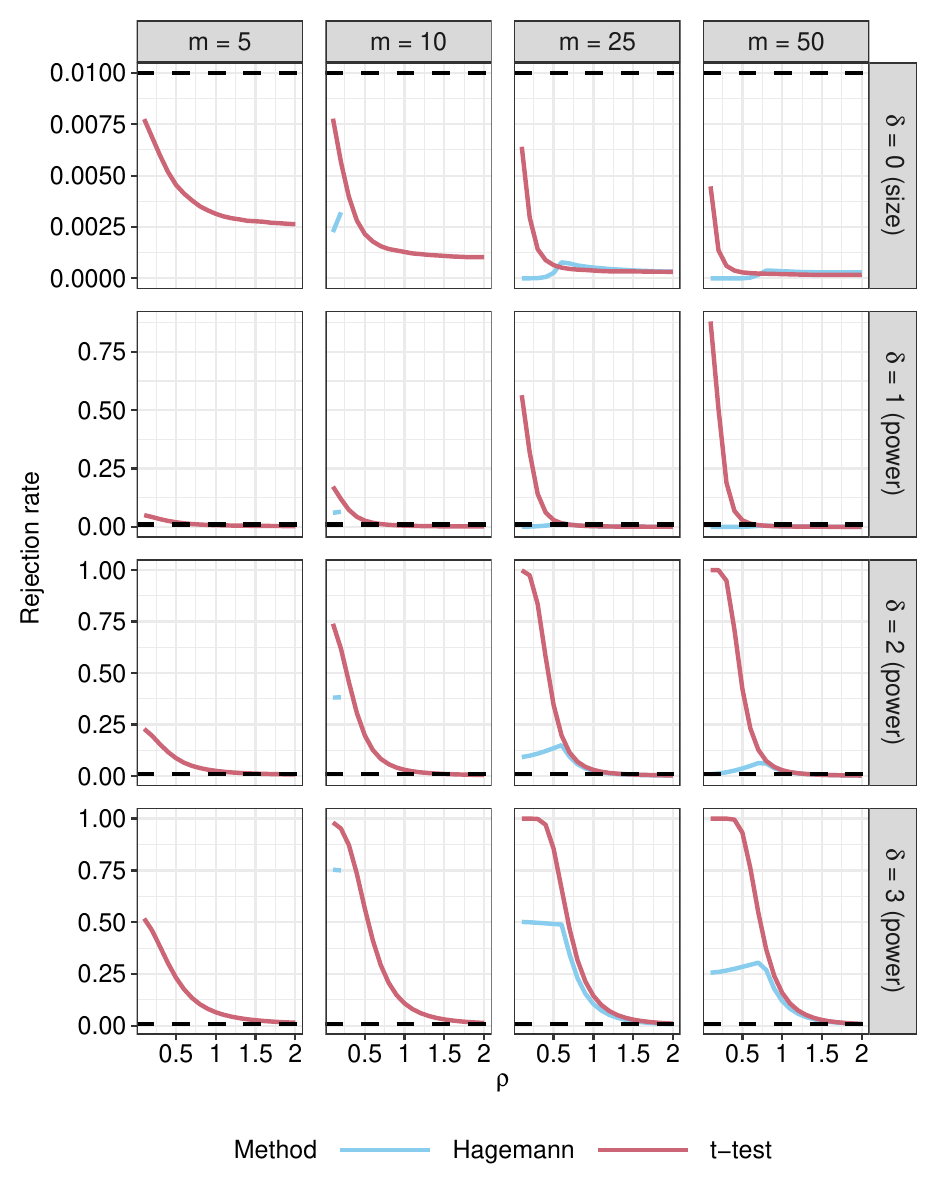}
     \caption{Probability of rejection against heterogeneity parameter $\rho$ for various cluster size $m$ and alternatives $\cp$ at $\alpha = 0.01$ for \textbf{DGP 2} of simulation design 1 with $k = 1$.}
    \label{fig:sim1-dgp2-rho-a01}
\end{figure}

\begin{figure}
    \centering
    \includegraphics[scale=1]{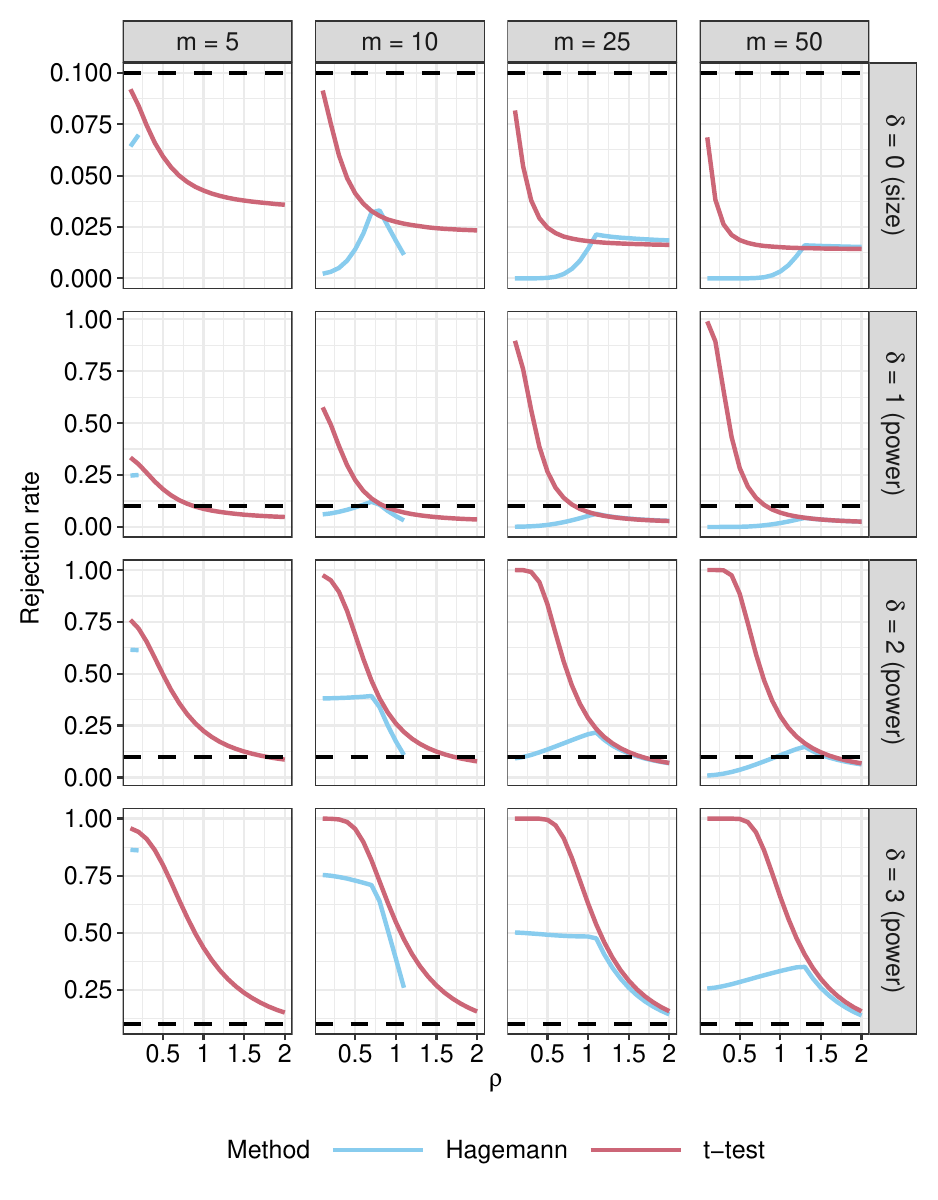}
    \caption{Probability of rejection against heterogeneity parameter $\rho$ for various cluster size $m$ and alternatives $\cp$ at $\alpha = 0.1$ for \textbf{DGP 2} of simulation design 1 with $k = 1$.}
    \label{fig:sim1-dgp2-rho-a10}
\end{figure}

\subsubsection{Results for $k = 2$}
This subsection reports the results for $k = 2$. Figures \ref{fig:sim1-dgp1-k2-rho-a01} to \ref{fig:sim1-dgp1-k2-rho-a10} report the results at the 1\%, 5\% and 10\% levels respectively for \textbf{DGP 1}. Figures \ref{fig:sim1-dgp2-k2-rho-a01} to \ref{fig:sim1-dgp2-k2-rho-a10} report the results at the 1\%, 5\% and 10\% levels respectively for \textbf{DGP 2}. The figures show that the $t$-test continue to perform favorably in other significance levels.

\begin{figure}
    \centering
    \includegraphics[scale=1]{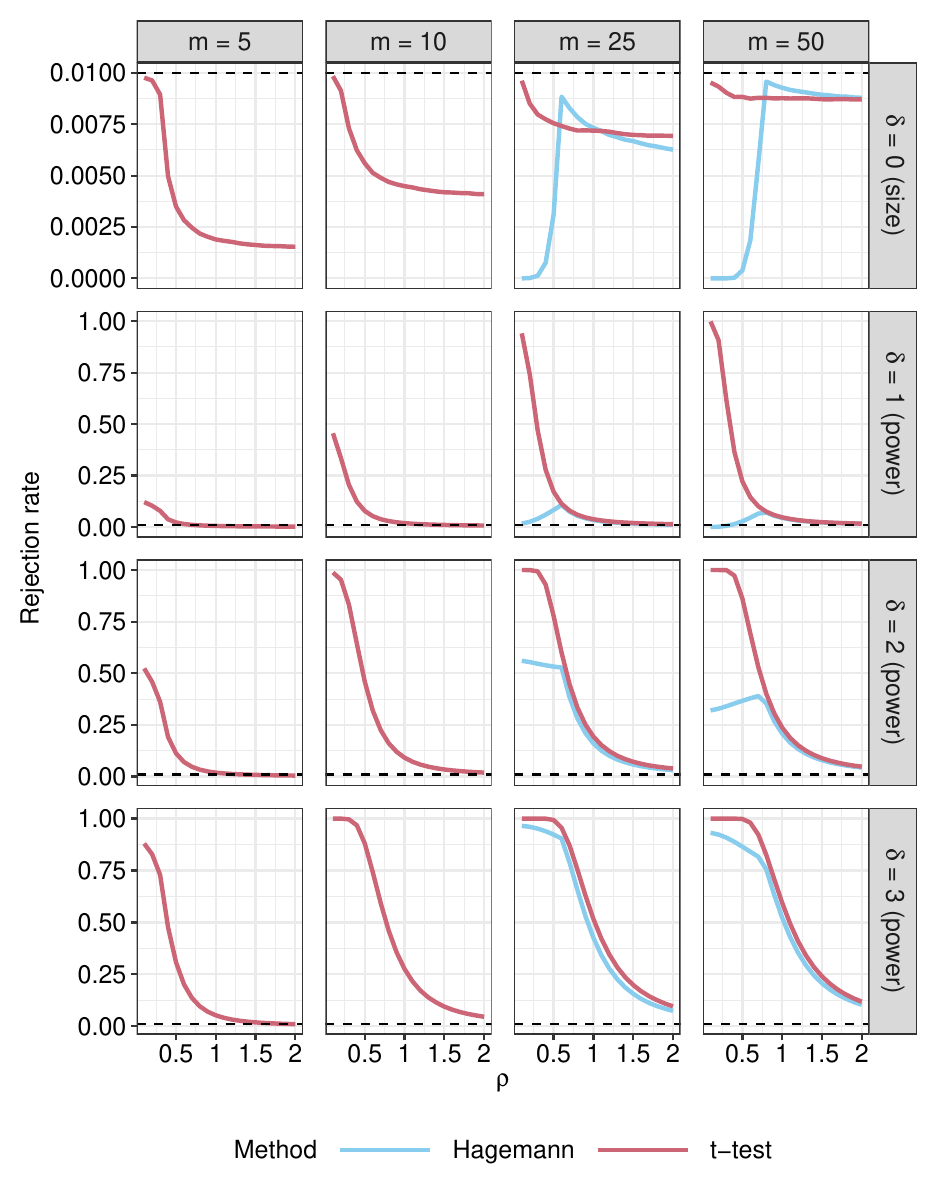}
    \caption{Probability of rejection against heterogeneity parameter $\rho$ for various cluster size $m$ and alternatives $\delta$ at $\alpha = 0.01$ for \textbf{DGP 1} of simulation design 1 with $k = 2$.}
    \label{fig:sim1-dgp1-k2-rho-a01}
\end{figure}

\begin{figure}
    \centering
    \includegraphics[scale=1]{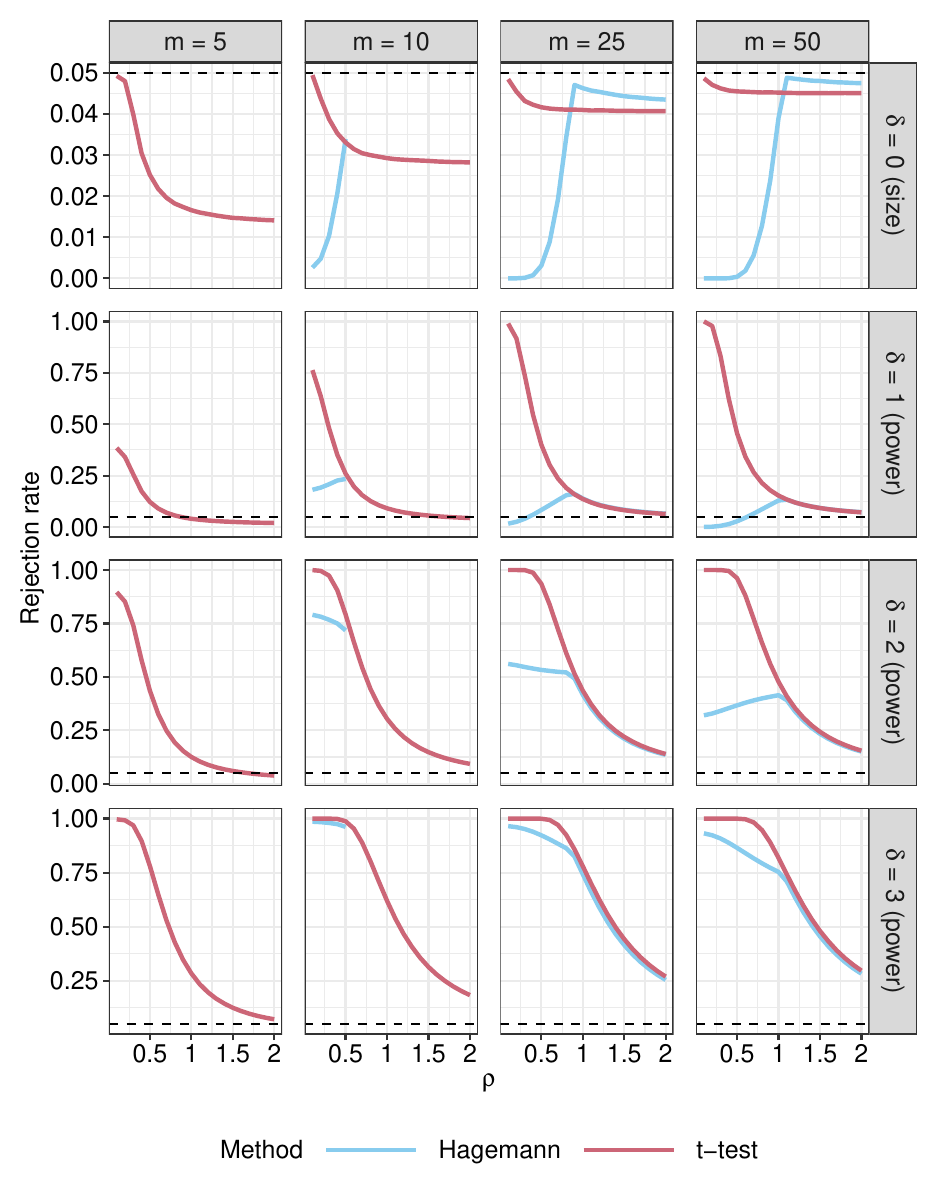}
    \caption{Probability of rejection against heterogeneity parameter $\rho$ for various cluster size $m$ and alternatives $\delta$ at $\alpha = 0.05$ for \textbf{DGP 1} of simulation design 1 with $k = 2$.}
    \label{fig:sim1-dgp1-k2-rho-a05}
\end{figure}

\begin{figure}
    \centering
    \includegraphics[scale=1]{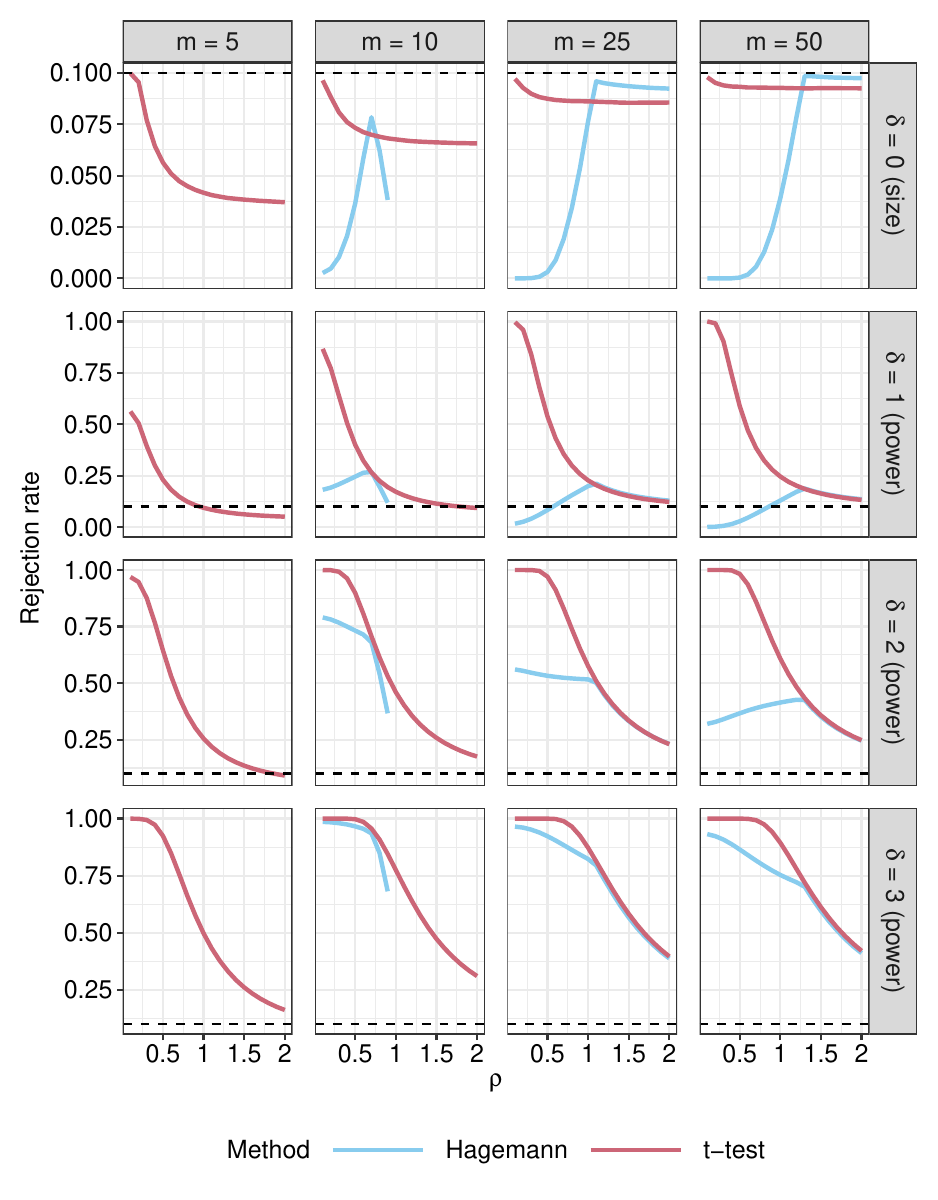}
    \caption{Probability of rejection against heterogeneity parameter $\rho$ for various cluster size $m$ and alternatives $\delta$ at $\alpha = 0.1$ for \textbf{DGP 1} of simulation design 1 with $k = 2$.}
    \label{fig:sim1-dgp1-k2-rho-a10}
\end{figure}

\begin{figure}
    \centering
    \includegraphics[scale=1]{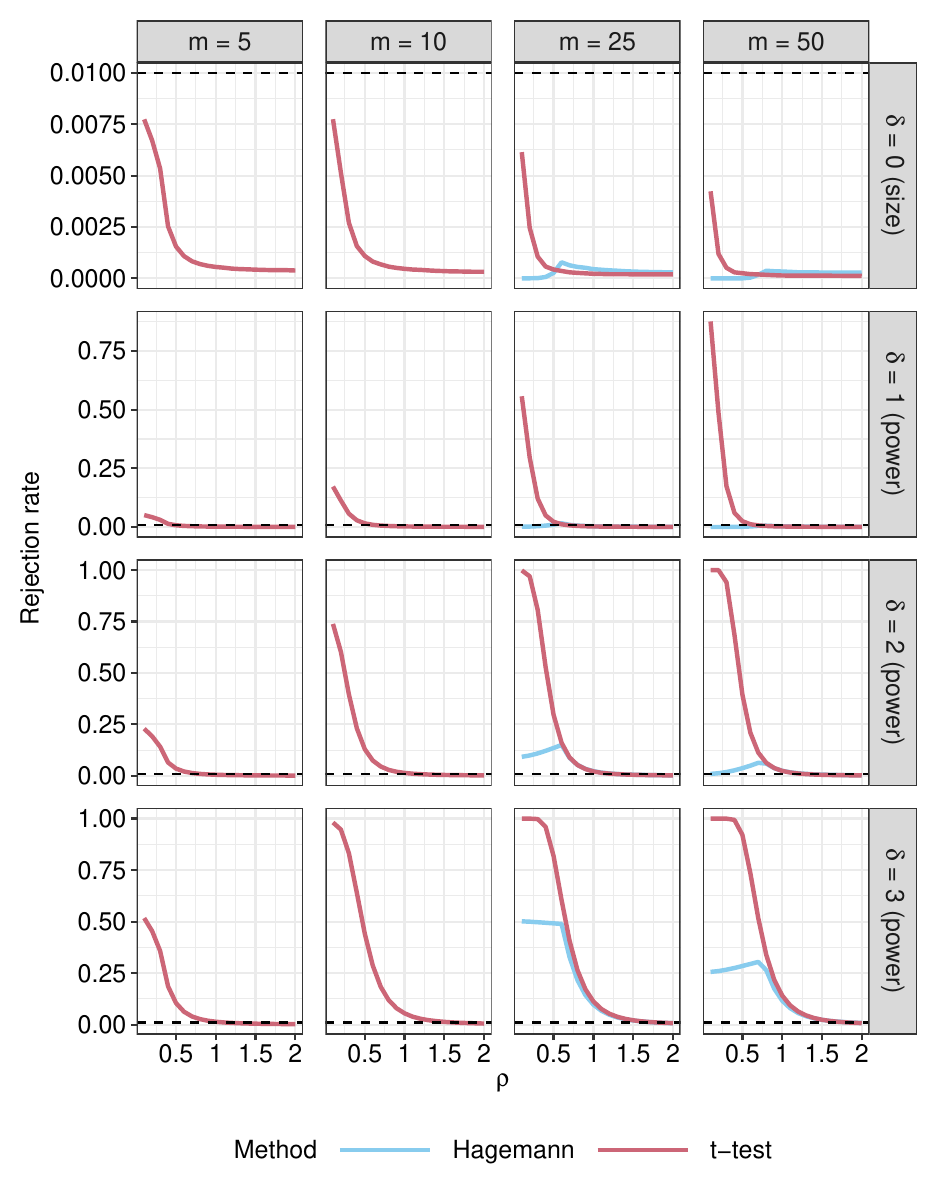}
    \caption{Probability of rejection against heterogeneity parameter $\rho$ for various cluster size $m$ and alternatives $\delta$ at $\alpha = 0.01$ for \textbf{DGP 2} of simulation design 1 with $k = 2$.}
    \label{fig:sim1-dgp2-k2-rho-a01}
\end{figure}

\begin{figure}
    \centering
    \includegraphics[scale=1]{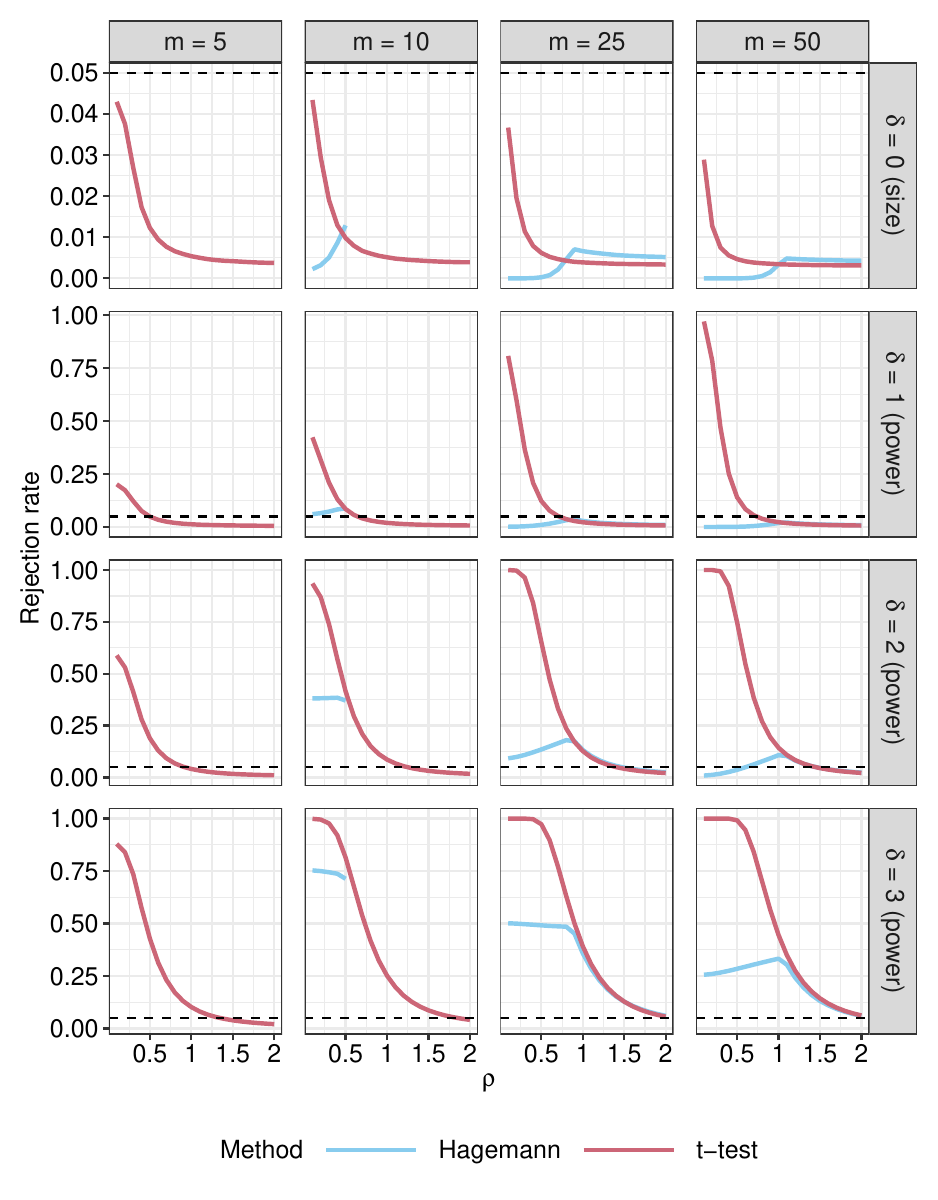}
    \caption{Probability of rejection against heterogeneity parameter $\rho$ for various cluster size $m$ and alternatives $\delta$ at $\alpha = 0.05$ for \textbf{DGP 2} of simulation design 1 with $k = 2$.}
    \label{fig:sim1-dgp2-k2-rho-a05}
\end{figure}

\begin{figure}
    \centering
    \includegraphics[scale=1]{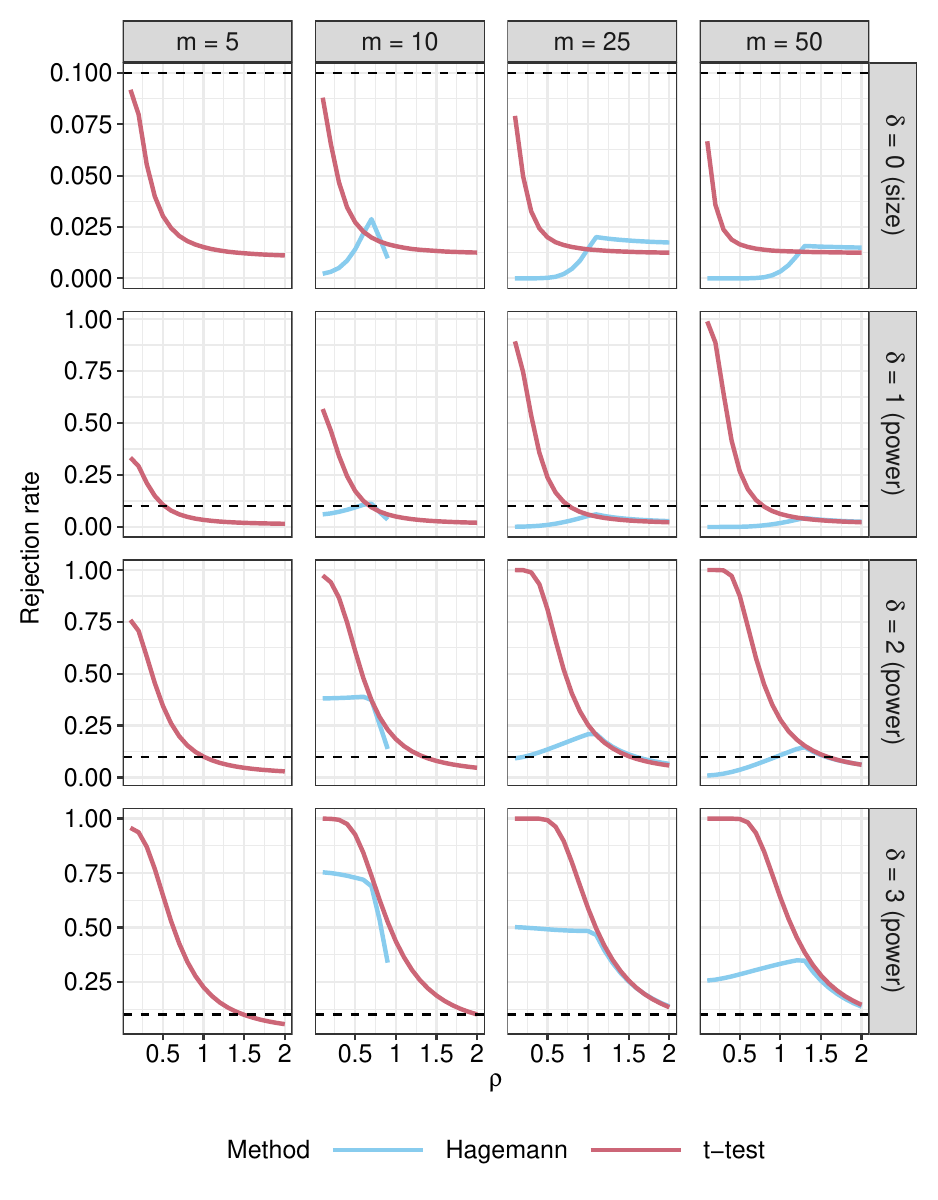}
    \caption{Probability of rejection against heterogeneity parameter $\rho$ for various cluster size $m$ and alternatives $\delta$ at $\alpha = 0.1$ for \textbf{DGP 2} of simulation design 1 with $k = 2$.}
    \label{fig:sim1-dgp2-k2-rho-a10}
\end{figure}

\subsection{Supplemental results for Section \ref{sec:sims-2}} \label{sec:app:sims-2}

Figures \ref{fig:sim2-dgp3-a05} to \ref{fig:sim2-dgp5-a05} report the results for DGPs 3 to 5 for various $m$ and $\rho$ at $\alpha = 0.05$.

\begin{figure}[!ht]
    \centering
    \includegraphics[scale=1]{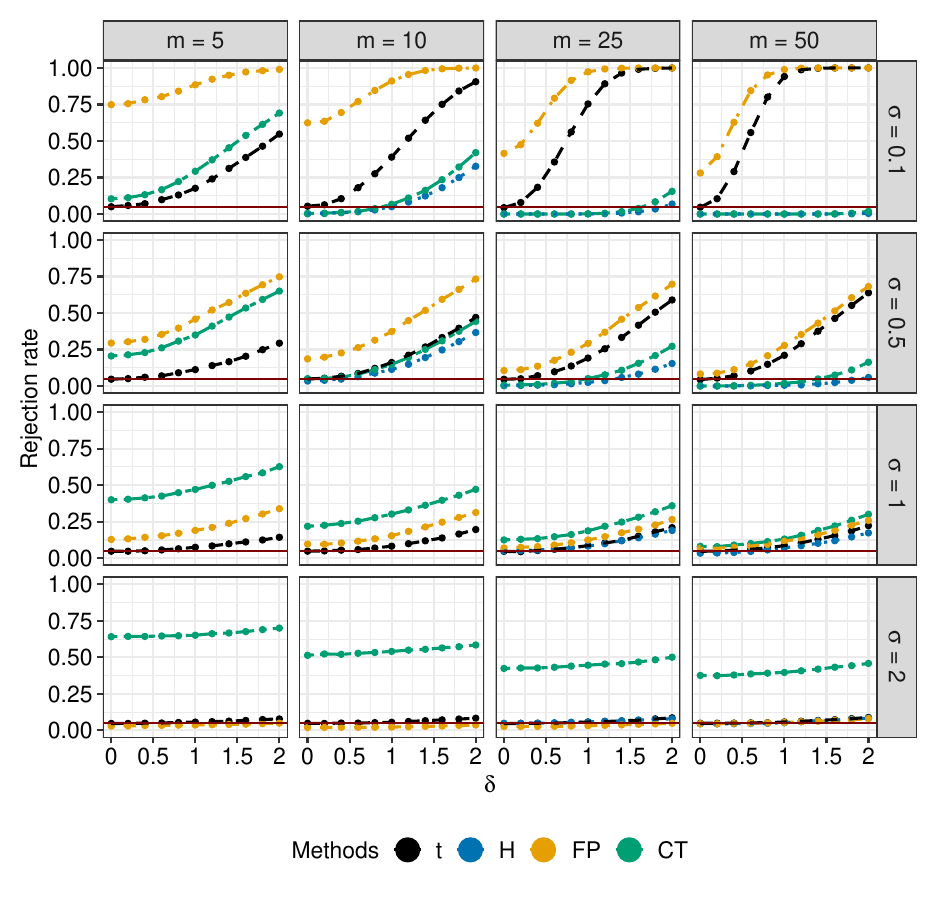}
    \caption{Simulation results for \textbf{DGP 3} of simulation design 2 at $\alpha = 0.05$.}
    \label{fig:sim2-dgp3-a05}
\end{figure}

\begin{figure}[!ht]
    \centering
    \includegraphics[scale=1]{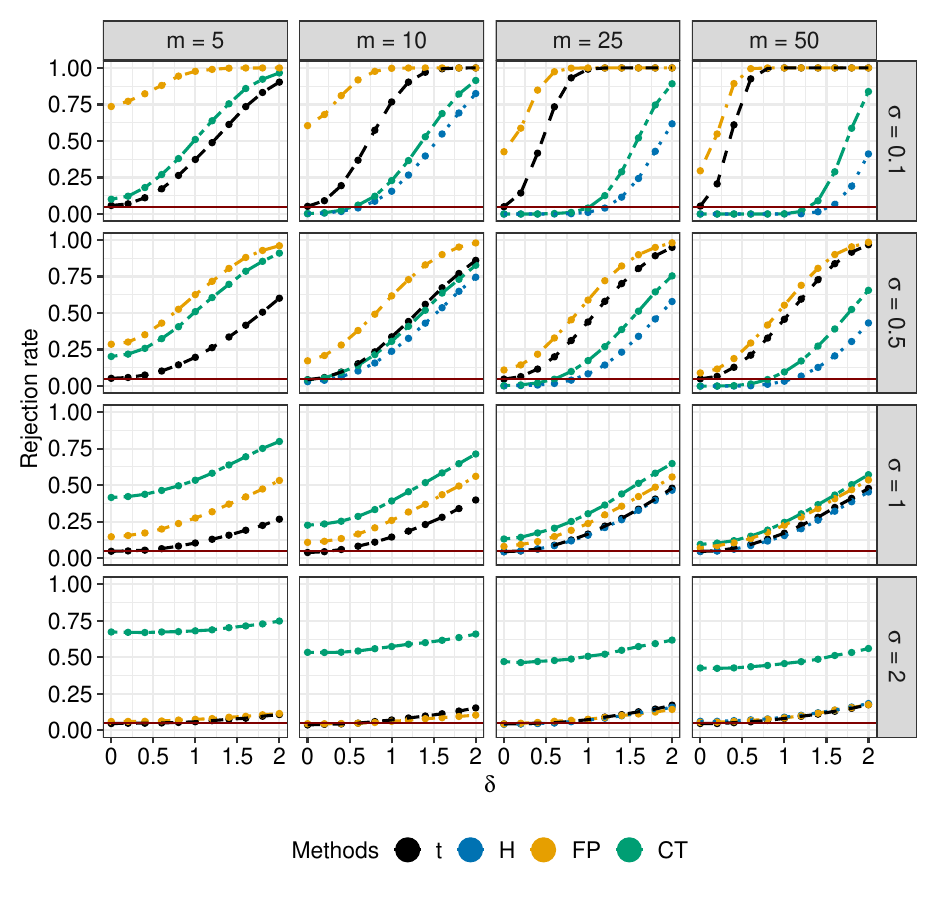}
    \caption{Simulation results for \textbf{DGP 4} of simulation design 2 at $\alpha = 0.05$.}
    \label{fig:sim2-dgp4-a05}
\end{figure}

\begin{figure}[!ht]
    \centering
    \includegraphics[scale=1]{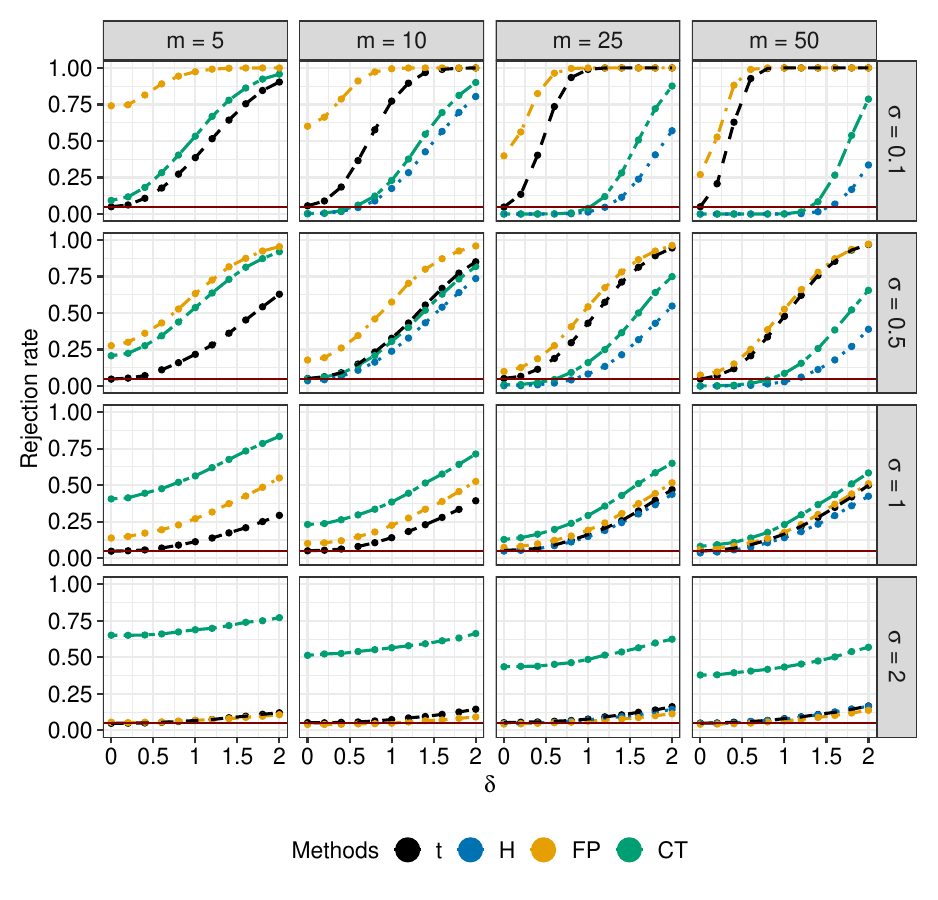}
    \caption{Simulation results for \textbf{DGP 5} of simulation design 2 at $\alpha = 0.05$.}
    \label{fig:sim2-dgp5-a05}
\end{figure}

\section{Supplemental tables: critical values for $k = 2$} \label{sec:cv-k2-app}
Table \ref{tab:cv-k2} reports the critical values for different numbers of clusters $m$ and heterogeneity parameters $\rho$ for $\alpha = 0.01$ and $\alpha = 0.05$ when $k = 2$.

\begin{table}[!ht]
    \centering
    \caption{Critical values for different values of $\alpha$, $m$ and $\rho$ for $k = 2$.}
    \label{tab:cv-k2}
    {\small \include{tables/paper_cv_k2}}
\end{table}

\end{document}

%% file: tables/paper_cv_k1.tex
\setlength\tabcolsep{3.9pt}

\resizebox{\textwidth}{!}{
\begin{tabular}[t]{lrrrrrrrrrrrr}
\toprule
\multicolumn{1}{c}{ } & \multicolumn{6}{c}{$\alpha = 0.01$} & \multicolumn{6}{c}{$\alpha = 0.05$} \\
\cmidrule(l{3pt}r{3pt}){2-7} \cmidrule(l{3pt}r{3pt}){8-13}
$\rho \backslash m$ & 5 & 10 & 15 & 20 & 25 & 50 & 5 & 10 & 15 & 20 & 25 & 50\\
\midrule
0.2 & 2.256 & 1.216 & 0.972 & 0.858 & 0.791 & 0.656 & 1.360 & 0.846 & 0.700 & 0.628 & 0.584 & 0.492\\
0.4 & 2.762 & 1.657 & 1.417 & 1.311 & 1.251 & 1.137 & 1.666 & 1.153 & 1.021 & 0.959 & 0.923 & 0.853\\
0.6 & 3.445 & 2.204 & 1.944 & 1.832 & 1.769 & 1.652 & 2.078 & 1.534 & 1.401 & 1.340 & 1.305 & 1.239\\
0.8 & 4.220 & 2.796 & 2.502 & 2.376 & 2.306 & 2.177 & 2.545 & 1.946 & 1.803 & 1.739 & 1.702 & 1.633\\
1.0 & 5.044 & 3.408 & 3.074 & 2.932 & 2.852 & 2.707 & 3.041 & 2.373 & 2.215 & 2.145 & 2.105 & 2.030\\
\addlinespace
1.2 & 5.896 & 4.033 & 3.654 & 3.492 & 3.403 & 3.238 & 3.556 & 2.807 & 2.633 & 2.555 & 2.511 & 2.428\\
1.4 & 6.767 & 4.664 & 4.238 & 4.056 & 3.955 & 3.771 & 4.081 & 3.247 & 3.053 & 2.967 & 2.919 & 2.828\\
1.6 & 7.649 & 5.300 & 4.825 & 4.622 & 4.510 & 4.305 & 4.613 & 3.689 & 3.476 & 3.381 & 3.328 & 3.228\\
1.8 & 8.539 & 5.939 & 5.413 & 5.189 & 5.065 & 4.839 & 5.150 & 4.134 & 3.900 & 3.796 & 3.738 & 3.628\\
2.0 & 9.436 & 6.580 & 6.003 & 5.758 & 5.622 & 5.373 & 5.690 & 4.581 & 4.325 & 4.212 & 4.148 & 4.029\\
\addlinespace
2.2 & 10.336 & 7.223 & 6.594 & 6.326 & 6.179 & 5.908 & 6.233 & 5.028 & 4.751 & 4.628 & 4.559 & 4.430\\
2.4 & 11.240 & 7.867 & 7.186 & 6.896 & 6.736 & 6.443 & 6.778 & 5.476 & 5.177 & 5.045 & 4.971 & 4.831\\
2.6 & 12.146 & 8.512 & 7.778 & 7.466 & 7.294 & 6.978 & 7.325 & 5.925 & 5.604 & 5.462 & 5.382 & 5.233\\
2.8 & 13.055 & 9.157 & 8.371 & 8.036 & 7.851 & 7.513 & 7.873 & 6.374 & 6.031 & 5.879 & 5.794 & 5.634\\
3.0 & 13.965 & 9.804 & 8.964 & 8.607 & 8.409 & 8.049 & 8.421 & 6.824 & 6.458 & 6.296 & 6.205 & 6.035\\
\addlinespace
3.2 & 14.876 & 10.450 & 9.557 & 9.177 & 8.968 & 8.584 & 8.971 & 7.274 & 6.886 & 6.714 & 6.617 & 6.437\\
3.4 & 15.789 & 11.097 & 10.150 & 9.748 & 9.526 & 9.120 & 9.521 & 7.725 & 7.313 & 7.132 & 7.029 & 6.838\\
3.6 & 16.702 & 11.744 & 10.744 & 10.319 & 10.085 & 9.655 & 10.072 & 8.175 & 7.741 & 7.549 & 7.441 & 7.240\\
3.8 & 17.616 & 12.392 & 11.338 & 10.890 & 10.643 & 10.191 & 10.623 & 8.626 & 8.169 & 7.967 & 7.854 & 7.642\\
4.0 & 18.531 & 13.040 & 11.932 & 11.462 & 11.202 & 10.727 & 11.175 & 9.077 & 8.597 & 8.385 & 8.266 & 8.043\\
\addlinespace
4.2 & 19.447 & 13.688 & 12.526 & 12.033 & 11.760 & 11.262 & 11.727 & 9.528 & 9.025 & 8.803 & 8.678 & 8.445\\
4.4 & 20.362 & 14.336 & 13.121 & 12.604 & 12.319 & 11.798 & 12.279 & 9.979 & 9.453 & 9.221 & 9.091 & 8.847\\
4.6 & 21.279 & 14.985 & 13.715 & 13.176 & 12.878 & 12.334 & 12.832 & 10.430 & 9.882 & 9.639 & 9.503 & 9.248\\
4.8 & 22.195 & 15.633 & 14.310 & 13.747 & 13.437 & 12.869 & 13.385 & 10.882 & 10.310 & 10.057 & 9.915 & 9.650\\
5.0 & 23.112 & 16.282 & 14.904 & 14.319 & 13.996 & 13.405 & 13.938 & 11.333 & 10.738 & 10.476 & 10.328 & 10.052\\
\bottomrule
\end{tabular}
}

%% file: tables/paper_alphas_v6.tex
\begin{tabular}[t]{lrrrrrrrrr}
\toprule
$m \ \backslash \ \rho$   & 0.1 & 0.2 & 0.5 & 1 & 2 & 3 & 4 & 5 & 10\\
\midrule
5 & 3.950 & 4.418 & 6.438 & 9.456 & 11.866 & 12.505 & 12.647 & 12.693 & 12.770\\
10 & 3.911 & 4.510 & 7.313 & 9.026 & 9.404 & 9.435 & 9.480 & 9.486 & 9.504\\
20 & 3.684 & 4.615 & 6.337 & 6.711 & 6.829 & 6.850 & 6.874 & 6.866 & 6.874\\
25 & 3.491 & 4.442 & 6.068 & 6.278 & 6.354 & 6.350 & 6.357 & 6.365 & 6.389\\
50 & 3.768 & 4.663 & 5.308 & 5.330 & 5.391 & 5.437 & 5.435 & 5.436 & 5.441\\
\bottomrule
\end{tabular}

%% file: tables/depewswensen2022ej-rcc-v5-step0.025.tex
\begin{tabular}[t]{ccccc}
\toprule
 & Homicide rate & Suicide rate & Gun suicide rate & Non-gun suicide rate\\
\midrule
Point estimate & -0.04 & -0.85 & -1.42 & 0.54\\
\cmidrule{1-5}
Wild $p$-value & 0.87 & 0.29 & 0 & 0.28\\
\cmidrule{1-5}
\addlinespace[0.3em]
\multicolumn{5}{l}{\textbf{Rearrangement test (Hagemann, 2024)}}\\
\hspace{1em}$\alpha = 0.01$ & NA & NA & NA & NA\\
\cmidrule{1-5}
\hspace{1em}$\alpha = 0.05$ & NA & NA & NA & NA\\
\cmidrule{1-5}
\hspace{1em}$\alpha = 0.1$ & NA & NA & NA & NA\\
\cmidrule{1-5}
\addlinespace[0.3em]
\multicolumn{5}{l}{\textbf{$t$-test}}\\
\hspace{1em}$\alpha = 0.01$ & NA & NA & 0.65 & NA\\
\cmidrule{1-5}
\hspace{1em}$\alpha = 0.05$ & NA & NA & 1.03 & NA\\
\cmidrule{1-5}
\hspace{1em}$\alpha = 0.1$ & NA & NA & 1.3 & NA\\
\bottomrule
\end{tabular}

%% file: tables/hiralwaetal2024restat-v2.tex
\begin{tabular}[t]{ccccc}
\toprule
\multicolumn{1}{c}{ } & \multicolumn{2}{c}{Focal year 2020} & \multicolumn{2}{c}{Focal year 2021} \\
\cmidrule(l{3pt}r{3pt}){2-3} \cmidrule(l{3pt}r{3pt}){4-5}
 & \makecell{Above \\ threshold \\ \$100--101.389k} & \makecell{Below \\ threshold \\ \$98.61--100k} & \makecell{Above \\ threshold \\ \$101.39--102.78k} & \makecell{Below \\ threshold \\ \$100--101.39k}\\
\midrule
Point estimate & -0.001 & -0.002 & -0.038  & -0.024\\
\cmidrule{1-5}
\addlinespace[0.3em]
\multicolumn{5}{l}{\textbf{Rearrangement test (Hagemann, 2024)}}\\
\hspace{1em}$\alpha = 0.01$ & NA & NA & NA & NA\\
\cmidrule{1-5}
\hspace{1em}$\alpha = 0.05$ & NA & NA & NA & NA\\
\cmidrule{1-5}
\hspace{1em}$\alpha = 0.1$ & NA & NA & NA & NA\\
\cmidrule{1-5}
\addlinespace[0.3em]
\multicolumn{5}{l}{\textbf{$t$-test}}\\
\hspace{1em}$\alpha = 0.01$ & NA & NA & 0.31 & 0.14\\
\cmidrule{1-5}
\hspace{1em}$\alpha = 0.05$ & NA & NA & 0.45 & 0.25\\
\cmidrule{1-5}
\hspace{1em}$\alpha = 0.1$ & NA & NA &  0.56 & 0.33\\
\bottomrule
\end{tabular}

%% file: tables/paper_cv_k2.tex
\setlength\tabcolsep{3.9pt}
\begin{tabular}[t]{lrrrrrrrrrrrr}
\toprule
\multicolumn{1}{c}{ } & \multicolumn{6}{c}{$\alpha = 0.01$} & \multicolumn{6}{c}{$\alpha = 0.05$} \\
\cmidrule(l{3pt}r{3pt}){2-7} \cmidrule(l{3pt}r{3pt}){8-13}
$\rho \backslash m$ & 5 & 10 & 15 & 20 & 25 & 50 & 5 & 10 & 15 & 20 & 25 & 50\\
\midrule
0.2 & 2.260 & 1.224 & 0.984 & 0.869 & 0.800 & 0.661 & 1.365 & 0.860 & 0.711 & 0.636 & 0.590 & 0.496\\
0.4 & 2.820 & 1.722 & 1.460 & 1.341 & 1.274 & 1.148 & 1.729 & 1.200 & 1.051 & 0.981 & 0.940 & 0.861\\
0.6 & 3.652 & 2.328 & 2.017 & 1.882 & 1.807 & 1.669 & 2.264 & 1.616 & 1.450 & 1.375 & 1.332 & 1.251\\
0.8 & 4.668 & 2.977 & 2.604 & 2.445 & 2.358 & 2.200 & 2.849 & 2.061 & 1.870 & 1.786 & 1.738 & 1.649\\
1.0 & 5.724 & 3.644 & 3.204 & 3.019 & 2.918 & 2.736 & 3.459 & 2.521 & 2.301 & 2.205 & 2.151 & 2.051\\
\addlinespace
1.2 & 6.794 & 4.323 & 3.811 & 3.598 & 3.482 & 3.273 & 4.082 & 2.988 & 2.737 & 2.627 & 2.566 & 2.454\\
1.4 & 7.874 & 5.007 & 4.423 & 4.180 & 4.049 & 3.812 & 4.713 & 3.460 & 3.176 & 3.052 & 2.984 & 2.858\\
1.6 & 8.960 & 5.696 & 5.037 & 4.765 & 4.617 & 4.351 & 5.350 & 3.935 & 3.616 & 3.479 & 3.403 & 3.262\\
1.8 & 10.049 & 6.387 & 5.653 & 5.350 & 5.186 & 4.892 & 5.990 & 4.411 & 4.058 & 3.906 & 3.822 & 3.667\\
2.0 & 11.142 & 7.080 & 6.270 & 5.936 & 5.756 & 5.432 & 6.633 & 4.889 & 4.501 & 4.334 & 4.242 & 4.072\\
\addlinespace
2.2 & 12.236 & 7.775 & 6.888 & 6.524 & 6.326 & 5.973 & 7.278 & 5.368 & 4.945 & 4.763 & 4.662 & 4.477\\
2.4 & 13.332 & 8.470 & 7.507 & 7.111 & 6.897 & 6.513 & 7.924 & 5.848 & 5.389 & 5.192 & 5.083 & 4.883\\
2.6 & 14.429 & 9.166 & 8.126 & 7.699 & 7.468 & 7.054 & 8.572 & 6.329 & 5.834 & 5.621 & 5.504 & 5.288\\
2.8 & 15.527 & 9.864 & 8.746 & 8.288 & 8.039 & 7.596 & 9.220 & 6.810 & 6.279 & 6.051 & 5.925 & 5.694\\
3.0 & 16.626 & 10.561 & 9.366 & 8.876 & 8.611 & 8.137 & 9.869 & 7.291 & 6.724 & 6.481 & 6.346 & 6.100\\
\addlinespace
3.2 & 17.726 & 11.259 & 9.987 & 9.465 & 9.183 & 8.678 & 10.519 & 7.773 & 7.169 & 6.911 & 6.767 & 6.505\\
3.4 & 18.826 & 11.957 & 10.607 & 10.054 & 9.754 & 9.220 & 11.169 & 8.254 & 7.614 & 7.340 & 7.189 & 6.911\\
3.6 & 19.926 & 12.656 & 11.228 & 10.643 & 10.326 & 9.761 & 11.819 & 8.736 & 8.060 & 7.771 & 7.610 & 7.317\\
3.8 & 21.027 & 13.355 & 11.849 & 11.232 & 10.898 & 10.302 & 12.470 & 9.219 & 8.506 & 8.201 & 8.032 & 7.723\\
4.0 & 22.128 & 14.054 & 12.470 & 11.821 & 11.471 & 10.844 & 13.121 & 9.701 & 8.952 & 8.631 & 8.454 & 8.129\\
\addlinespace
4.2 & 23.229 & 14.753 & 13.092 & 12.411 & 12.043 & 11.386 & 13.772 & 10.184 & 9.398 & 9.061 & 8.875 & 8.535\\
4.4 & 24.331 & 15.452 & 13.713 & 13.000 & 12.615 & 11.927 & 14.424 & 10.666 & 9.844 & 9.492 & 9.297 & 8.941\\
4.6 & 25.432 & 16.152 & 14.334 & 13.590 & 13.187 & 12.469 & 15.075 & 11.149 & 10.290 & 9.922 & 9.719 & 9.347\\
4.8 & 26.534 & 16.851 & 14.956 & 14.180 & 13.760 & 13.010 & 15.727 & 11.632 & 10.736 & 10.353 & 10.141 & 9.753\\
5.0 & 27.636 & 17.551 & 15.578 & 14.769 & 14.332 & 13.552 & 16.379 & 12.115 & 11.182 & 10.783 & 10.562 & 10.159\\
\bottomrule
\end{tabular}

%% file: refs.bib
@article{wang2022rerandomization,
	author = {Y. Wang and X. Li},
	journal = {The Annals of Statistics},
	pages = {3439 -- 3465},
	title = {{Rerandomization with diminishing covariate imbalance and diverging number of covariates}},
	volume = {50},
	year = {2022}}

@article{Bakirov:2006aa,
	author = {Bakirov, N. K. and Sz{\'e}kely, G. J.},
	journal = {Journal of Mathematical Sciences},
	pages = {6497--6505},
	title = {Student's t-test for Gaussian scale mixtures},
	volume = {139},
	year = {2006}}

@article{hagemann2024wp,
	author = {Hagemann, Andres},
	journal = {Working Paper},
	title = {{Inference with a single treated cluster}},
	year = {2024}}

@article{makshanovshalaevskii1977some,
	author = {Makshanov, AV and Shalaevskii, OV},
	journal = {Zapiski Nauchnykh Seminarov POMI},
	pages = {118--138},
	publisher = {St. Petersburg Department of Steklov Institute of Mathematics, Russian~{\ldots}},
	title = {Some problems of asymptotic approximations of distributions},
	volume = {74},
	year = {1977}}

@article{bakirov1989jms,
	author = {Bakirov, NK},
	journal = {Journal of Mathematical Sciences},
	number = {4},
	pages = {433--440},
	publisher = {Springer},
	title = {An extremal property of the student distribution},
	volume = {44},
	year = {1989}}

@article{bakirov1998jms,
	author = {Bakirov, Nail K},
	journal = {Journal of Mathematical Sciences},
	pages = {1460--1467},
	publisher = {Springer},
	title = {Nonhomogeneous samples in the Behrens-Fisher problem},
	volume = {89},
	year = {1998}}

@article{kumarliang2024aejep,
	author = {Kumar, Anil and Liang, Che-Yuan},
	doi = {10.1257/pol.20200683},
	journal = {American Economic Journal: Economic Policy},
	month = {August},
	number = {3},
	pages = {1--26},
	title = {Labor Market Effects of Credit Constraints: Evidence from a Natural Experiment},
	url = {https://www.aeaweb.org/articles?id=10.1257/pol.20200683},
	volume = {16},
	year = {2024}}

@article{dillenderetal2023jpube,
	author = {Marcus Dillender and Lu Jinks and Anthony T. {Lo Sasso}},
	doi = {https://doi.org/10.1016/j.jpubeco.2022.104781},
	issn = {0047-2727},
	journal = {Journal of Public Economics},
	pages = {104781},
	title = {When (and why) providers do not respond to changes in reimbursement rates},
	url = {https://www.sciencedirect.com/science/article/pii/S0047272722001839},
	volume = {217},
	year = {2023}}

@article{alpertetal2024aejep,
	author = {Alpert, Abby and Dykstra, Sarah and Jacobson, Mireille},
	doi = {10.1257/pol.20200579},
	journal = {American Economic Journal: Economic Policy},
	month = {February},
	number = {1},
	pages = {87--123},
	title = {Hassle Costs versus Information: How Do Prescription Drug Monitoring Programs Reduce Opioid Prescribing?},
	url = {https://www.aeaweb.org/articles?id=10.1257/pol.20200579},
	volume = {16},
	year = {2024}}

@article{harrislarsen2023jhr,
	author = {Harris, Douglas N and Larsen, Matthew F},
	journal = {Journal of Human Resources},
	number = {5},
	pages = {1608--1643},
	publisher = {University of Wisconsin Press},
	title = {Taken by storm: The effects of Hurricane Katrina on medium-term student outcomes in New Orleans},
	volume = {58},
	year = {2023}}

@article{besteretal2011joe,
	author = {Bester, Alan and Conley, Timothy and Hansen, Christian},
	journal = {Journal of Econometrics},
	number = {2},
	pages = {137-151},
	title = {Inference with dependent data using cluster covariance estimators},
	url = {https://EconPapers.repec.org/RePEc:eee:econom:v:165:y:2011:i:2:p:137-151},
	volume = {165},
	year = {2011}}

@article{canayetal2017ecta,
	author = {Ivan A. Canay and Joseph P. Romano and Azeem M. Shaikh},
	issn = {00129682, 14680262},
	journal = {Econometrica},
	number = {3},
	pages = {1013--1030},
	publisher = {[Wiley, The Econometric Society]},
	title = {Randomization Tests under an Approximate Symmetry Assumption},
	url = {http://www.jstor.org/stable/44955148},
	volume = {85},
	year = {2017}}

@article{ibraimovmuller2016restat,
	author = {Ibragimov, Rustam and M{\"u}ller, Ulrich K.},
	doi = {10.1162/REST_a_00545},
	eprint = {https://direct.mit.edu/rest/article-pdf/98/1/83/1918100/rest\_a\_00545.pdf},
	issn = {0034-6535},
	journal = {The Review of Economics and Statistics},
	month = {03},
	number = {1},
	pages = {83-96},
	title = {{Inference with Few Heterogeneous Clusters}},
	url = {https://doi.org/10.1162/REST\_a\_00545},
	volume = {98},
	year = {2016}}

@article{hagemann2022wp,
	author = {Andreas Hagemann},
	journal = {Working Paper},
	title = {{Permutation inference with a finite number of heterogeneous clusters}},
	year = {2022}}

@article{cameronetal2008restat,
	author = {Cameron, A. and Gelbach, Jonah and Miller, Douglas},
	journal = {The Review of Economics and Statistics},
	number = {3},
	pages = {414-427},
	title = {Bootstrap-Based Improvements for Inference with Clustered Errors},
	url = {https://EconPapers.repec.org/RePEc:tpr:restat:v:90:y:2008:i:3:p:414-427},
	volume = {90},
	year = {2008}}

@article{ibraimovmuller2010jbes,
	author = {Ibragimov, Rustam and M{\"u}ller, Ulrich K.},
	journal = {Journal of Business \& Economic Statistics},
	number = {4},
	pages = {453-468},
	title = {t-Statistic Based Correlation and Heterogeneity Robust Inference},
	url = {https://EconPapers.repec.org/RePEc:bes:jnlbes:v:28:i:4:y:2010:p:453-468},
	volume = {28},
	year = {2010}}

@article{conleytaber2011restat,
	author = {Conley, Timothy G. and Taber, Christopher R.},
	doi = {10.1162/REST_a_00049},
	eprint = {https://direct.mit.edu/rest/article-pdf/93/1/113/1919185/rest\_a\_00049.pdf},
	issn = {0034-6535},
	journal = {The Review of Economics and Statistics},
	month = {02},
	number = {1},
	pages = {113-125},
	title = {{Inference with ``Difference in Differences'' with a Small Number of Policy Changes}},
	url = {https://doi.org/10.1162/REST\_a\_00049},
	volume = {93},
	year = {2011}}

@article{fermanpinto2019restat,
	author = {Ferman, Bruno and Pinto, Cristine},
	doi = {10.1162/rest_a_00759},
	eprint = {https://direct.mit.edu/rest/article-pdf/101/3/452/1916793/rest\_a\_00759.pdf},
	issn = {0034-6535},
	journal = {The Review of Economics and Statistics},
	month = {07},
	number = {3},
	pages = {452-467},
	title = {{Inference in Differences-in-Differences with Few Treated Groups and Heteroskedasticity}},
	url = {https://doi.org/10.1162/rest\_a\_00759},
	volume = {101},
	year = {2019}}

@article{wangburke2022jfe,
	author = {Jialan Wang and Kathleen Burke},
	doi = {https://doi.org/10.1016/j.jfineco.2021.09.024},
	issn = {0304-405X},
	journal = {Journal of Financial Economics},
	number = {2, Part B},
	pages = {489-507},
	title = {The effects of disclosure and enforcement on payday lending in Texas},
	url = {https://www.sciencedirect.com/science/article/pii/S0304405X21004372},
	volume = {145},
	year = {2022}}

@article{canaysantosshaikh2021restat,
	author = {Ivan A. Canay and Andres Santos and Azeem M. Shaikh},
	issue = {2},
	journal = {The Review of Economics and Statistics},
	month = {January},
	title = {The Wild Bootstrap with a ``small'' number of ``large'' Clusters},
	volume = {2021},
	year = {2021}}

@article{oldenmoen2023ej,
	author = {Olden, Andreas and M{\o}en, Jarle},
	doi = {10.1093/ectj/utac010},
	eprint = {https://academic.oup.com/ectj/article-pdf/25/3/531/45842047/utac010.pdf},
	issn = {1368-4221},
	journal = {The Econometrics Journal},
	month = {03},
	number = {3},
	pages = {531-553},
	title = {The triple difference estimator},
	url = {https://doi.org/10.1093/ectj/utac010},
	volume = {25},
	year = {2022}}

@article{hiraiwaetal2024wp,
	author = {Hiraiwa, Takuya and Lipsitz, Michael and Starr, Evan},
	journal = {The Review of Economics and Statistics},
	pages = {1-47},
	title = {Do Firms Value Court Enforceability of Noncompete Agreements? A Revealed Preference Approach},
	year = {2024}}

@article{depewswensen2022ej,
	author = {Depew, Briggs and Swensen, Isaac},
	doi = {10.1093/ej/ueac004},
	eprint = {https://academic.oup.com/ej/article-pdf/132/646/2118/51753860/ueac004.pdf},
	issn = {0013-0133},
	journal = {The Economic Journal},
	month = {01},
	number = {646},
	pages = {2118-2140},
	title = {The Effect of Concealed-Carry and Handgun Restrictions on Gun-Related Deaths: Evidence from the Sullivan Act of 1911},
	url = {https://doi.org/10.1093/ej/ueac004},
	volume = {132},
	year = {2022}}

@article{rothetal2023joe,
title = {What’s trending in difference-in-differences? A synthesis of the recent econometrics literature},
journal = {Journal of Econometrics},
volume = {235},
number = {2},
pages = {2218-2244},
year = {2023},
issn = {0304-4076},
doi = {https://doi.org/10.1016/j.jeconom.2023.03.008},
url = {https://www.sciencedirect.com/science/article/pii/S0304407623001318},
author = {Jonathan Roth and Pedro H.C. Sant’Anna and Alyssa Bilinski and John Poe},
}

@article{dechaisemartindhaultfœuille2023ej,
    author = {de Chaisemartin, Clément and D’Haultfœuille, Xavier},
    title = {Two-way fixed effects and differences-in-differences with heterogeneous treatment effects: a survey},
    journal = {The Econometrics Journal},
    volume = {26},
    number = {3},
    pages = {C1-C30},
    year = {2022},
    month = {06},
    issn = {1368-4221},
    doi = {10.1093/ectj/utac017},
    url = {https://doi.org/10.1093/ectj/utac017},
    eprint = {https://academic.oup.com/ectj/article-pdf/26/3/C1/51707976/utac017.pdf},
}

@misc{bakeretal2025wp,
      title={Difference-in-Differences Designs: A Practitioner's Guide}, 
      author={Andrew Baker and Brantly Callaway and Scott Cunningham and Andrew Goodman-Bacon and Pedro H. C. Sant'Anna},
      year={2025},
      eprint={2503.13323},
      archivePrefix={arXiv},
      primaryClass={econ.EM},
      url={https://arxiv.org/abs/2503.13323}, 
}

@article{alvarezetal2025wp,
  title={Inference with few treated units},
  author={Alvarez, Luis and Ferman, Bruno and W{\"u}thrich, Kaspar},
  journal={arXiv preprint arXiv:2504.19841},
  year={2025}
}

@article{cameronmiller2015jhr,
  author =        {Cameron, A Colin and Miller, Douglas L},
  journal =       {Journal of human resources},
  number =        {2},
  pages =         {317--372},
  publisher =     {University of Wisconsin Press},
  title =         {A practitioner’s guide to cluster-robust inference},
  volume =        {50},
  year =          {2015},
}

@article{conleyetal2018jar,
  author =        {Conley, Timothy and Gonçalves, Silva and
                   Hansen, Christian},
  journal =       {Journal of Accounting Research},
  number =        {4},
  pages =         {1139-1203},
  title =         {Inference with Dependent Data in Accounting and
                   Finance Applications},
  volume =        {56},
  year =          {2018},
  doi =           {https://doi.org/10.1111/1475-679X.12219},
  url =           {https://onlinelibrary.wiley.com/doi/abs/10.1111/1475-
                  679X.12219},
}

@article{mackinnonetal2023joe,
  author =        {James G. MacKinnon and Morten Ørregaard Nielsen and
                   Matthew D. Webb},
  journal =       {Journal of Econometrics},
  number =        {2},
  pages =         {272-299},
  title =         {Cluster-robust inference: A guide to empirical
                   practice},
  volume =        {232},
  year =          {2023},
  doi =           {https://doi.org/10.1016/j.jeconom.2022.04.001},
  issn =          {0304-4076},
  url =           {https://www.sciencedirect.com/science/article/pii/
                  S0304407622000781},
}

@misc{fred2025website,
  author = {{Federal Reserve Economic Data}},
  title = {Resident Population by State, Annual},
  howpublished = {\url{https://fred.stlouisfed.org/release/tables?rid=118&eid=259194&od=1911-01-01#}},
  year = {2025},
  note = {Accessed: 2025-05-25}
}

@article{lau2025wp,
	author = {Lau, Chun Pong},
	journal = {Working Paper},
	title = {{Combining Clusters for the Approximate Randomization Test}},
        url = {"https://arxiv.org/abs/2502.03865"},
	year = {2025}}

@misc{alvarezferman2023wp,
      title={Inference in Difference-in-Differences with Few Treated Units and Spatial Correlation}, 
      author={Luis Alvarez and Bruno Ferman},
      year={2023},
      eprint={2006.16997},
      archivePrefix={arXiv},
      primaryClass={econ.EM},
      url={https://arxiv.org/abs/2006.16997}, 
}

@article{horowitzmanski1995ecta,
  author =        {Horowitz, Joel and Manski, Charles},
  journal =       {Econometrica},
  number =        {2},
  pages =         {281-302},
  title =         {Identification and Robustness with Contaminated and
                   Corrupted Data},
  volume =        {63},
  year =          {1995},
  url =           {https://EconPapers.repec.org/RePEc:ecm:emetrp:v:63:y:1995:i:
                  2:p:281-302},
}

@article{klinesantos2013qe,
  author =        {Kline, Patrick and Santos, Andres},
  journal =       {Quantitative Economics},
  number =        {2},
  pages =         {231-267},
  title =         {Sensitivity to missing data assumptions: Theory and
                   an evaluation of the U.S. wage structure},
  volume =        {4},
  year =          {2013},
  doi =           {https://doi.org/10.3982/QE176},
  url =           {https://onlinelibrary.wiley.com/doi/abs/10.3982/QE176},
}

@article{mastenpoirier2020qe,
  author =        {Masten, Matthew A. and Poirier, Alexandre},
  journal =       {Quantitative Economics},
  number =        {1},
  pages =         {41-111},
  title =         {Inference on breakdown frontiers},
  volume =        {11},
  year =          {2020},
  doi =           {https://doi.org/10.3982/QE1288},
  url =           {https://onlinelibrary.wiley.com/doi/abs/10.3982/QE1288},
}

@article{wuli2025jasa,
  title={Sensitivity analysis for quantiles of hidden biases in matched observational studies},
  author={Wu, Dongxiao and Li, Xinran},
  journal={Journal of the American Statistical Association},
  pages={1--12},
  year={2025},
  publisher={Taylor \& Francis}
}

@article{cuili2025wp,
  title={Robust Sensitivity Analysis via Augmented Percentile Bootstrap under Simultaneous Violations of Unconfoundedness and Overlap},
  author={Cui, Han and Li, Xinran},
  journal={arXiv preprint arXiv:2509.13169},
  year={2025}
}
